\documentclass[11pt,letterpaper]{article}

\usepackage[utf8]{inputenc}
\usepackage[english]{babel}
\usepackage{csquotes}
\usepackage[T1]{fontenc}

\usepackage[margin=1in]{geometry} 

\usepackage{setspace}
\usepackage{graphicx}
\usepackage{tikz}
\usetikzlibrary{arrows.meta,positioning,calc,fit}

\usepackage{amsmath, amssymb, amsfonts}
\usepackage{amsthm, mathrsfs, mathtools}

\usepackage[linktoc=page]{hyperref}
\hypersetup{
  colorlinks=false,
  pdfborder={0 0 0.3},
  linkbordercolor={0.5 0.5 0.5},
  citebordercolor={0.5 0.5 0.5},
  urlbordercolor={0.5 0.5 0.5}
}

\usepackage{zref-clever}

\zcsetup{
    nameinlink = false,
    cap        = false,
    abbrev     = false,
}

\theoremstyle{plain}
\newtheorem{theorem}{Theorem}[section]
\newtheorem{lemma}[theorem]{Lemma}
\newtheorem{proposition}[theorem]{Proposition}
\newtheorem{corollary}[theorem]{Corollary}

\newtheorem{definition}[theorem]{Definition}
\newtheorem{maintheorem}{Theorem}

\theoremstyle{definition}

\newtheorem{remark}[theorem]{Remark}

\newcommand{\zcsharedtype}[1]{%
    \AddToHook{env/#1/begin}{%
        \zcsetup{countertype={theorem=#1}}%
    }%
}

\zcsetup{
    countertype={maintheorem=theorem}
}

\zcsharedtype{lemma}
\zcsharedtype{proposition}
\zcsharedtype{corollary}
\zcsharedtype{definition}
\zcsharedtype{example}
\zcsharedtype{remark}

\usepackage[backend=biber,style=alphabetic, minbibnames=99,sorting=nty,sortcites=true,maxbibnames=99]{biblatex}
\usepackage{titling}

\usepackage{authblk}

\usepackage{enumitem} 
\usepackage{caption} 
\usepackage{subcaption} 
\usepackage{microtype} 
\usepackage{dsfont} 

\usepackage[baseline]{euflag}

\let\epsilon\varepsilon

\let\phi\varphi

\newcommand*{\Tr}[1]{\mathop{}\!\mathrm{Tr}{\left(#1\right)}}

\def\id{\mathds{1}}

\newcommand{\bra}[1]{\ensuremath{\left\langle#1\right|}}
\newcommand{\ket}[1]{\ensuremath{\left|#1\right\rangle}}
\newcommand{\braket}[2]{\langle #1  |#2\rangle}
\newcommand{\ketbra}[2]{|#1\rangle\langle #2|  }
\newcommand{\sandwich}[3]{\langle #1|#2 |#3\rangle  }
\newcommand{\norm}[1]{\left\lVert#1\right\rVert}

\newcommand{\identitymap}{\operatorname{id}}

\newcommand{\cA}{\mathcal{A}}
\newcommand{\cE}{\mathcal{E}}
\newcommand{\cH}{\mathcal{H}}
\newcommand{\cM}{\mathcal{M}}

\newcommand{\cR}{\mathcal{R}}

\newcommand{\cX}{\mathcal{X}}
\newcommand{\cY}{\mathcal{Y}}
\newcommand{\cZ}{\mathcal{Z}}

\newcommand{\fA}{\mathfrak{A}}
\newcommand{\fQ}{\mathfrak{Q}}
\newcommand{\fL}{\mathfrak{L}}
\newcommand{\Tree}{\mathsf{T}}
\newcommand{\Vertices}[1]{\mathrm{V}\!\left(#1\right)}
\newcommand{\Edges}[1]{\mathrm{E}\!\left(#1\right)}
\newcommand{\PartyVertices}[1]{\mathrm{P}\!\left(#1\right)}
\newcommand{\SourceVertices}[1]{\mathrm{S}\!\left(#1\right)}

\newcommand{\Parent}[1]{%
    \operatorname{par}\!\left(#1\right)%
}

\newcommand{\Children}[1]{%
    \operatorname{ch}\!\left(#1\right)%
}

\newcommand{\Neighbors}[1]{%
    N_{\Tree}\!\left(#1\right)%
}

\newcommand{\meas}[3]{%
    m^{#1}_{#2\mid #3}%
}

\newcommand{\Meas}[3]{%
    M^{#1}_{#2\mid #3}%
}

\newcommand{\instr}[3]{%
    \Phi^{#1}_{#2\mid #3}%
}

\newcommand{\rstate}[3]{%
    \omega^{#1}_{#2\mid #3}%
}

\newcommand{\TransferMsg}[1]{Q_{#1}}

\newcommand{\InAlg}[1]{\cA_{#1}^{\mathrm{in}}}
\newcommand{\InState}[1]{\omega_{#1}^{\mathrm{in}}}
\newcommand{\InHilb}[1]{\cH_{#1}^{\mathrm{in}}}
\newcommand{\InVec}[1]{\Omega_{#1}^{\mathrm{in}}}
\newcommand{\InRep}[1]{\pi_{#1}^{\mathrm{in}}}

\newcommand{\JoinAlg}[1]{\cA_{#1}}

\newcommand{\InPolyAlg}[1]{\fA_{#1}^{\mathrm{in}}}
\newcommand{\rPolystate}[2]{%
    \varsigma_{#1\mid #2}^{r}%
}
\newcommand{\JoinPolyAlg}[1]{\fA_{#1}}

\newcommand{\thinS}{\mathscr{S}}

\newcommand{\fatS}{\pmb{\mathscr{S}}}

\newcommand{\thinTreeS}{\thinS_{(\Tree, r)}}
\newcommand{\InfatS}[1]{\fatS_{#1}^{\mathrm{in}}}

\newcommand{\Monomial}[2]{\operatorname{Mon}_{\leq #1}(#2)}

\newcommand{\fLsf}{\fL^{\mathrm{ST}}}

\newcommand{\formalInstr}[3]{%
    \Psi^{#1}_{#2\mid #3}%
}

\newcommand{\Bof}[1]{%
    \mathcal{B}\!\left(#1\right)%
}

\newcommand{\maxotimes}{%
    \mathbin{\otimes_{\max}}%
}

\newcommand{\vnotimes}{%
    \mathbin{\overline{\otimes}}%
}

\newcommand{\vnbigotimes}{%
    \mathop{\overline{\bigotimes}}%
}

\newcommand{\gns}[1]{%
    \bigl(\cH_{#1},\pi_{#1}, \ket{\Omega_{#1}}\bigr)%
}

\newcommand{\thmpart}[1]{\par\addvspace{\medskipamount}\noindent\textsc{#1}}

\usepackage{xcolor}

\hypersetup{
  pdftitle = {Quantum inflation and source-transfer hierarchies are complete for tree networks},
  pdfauthor = {Xiangling Xu, Igor Klep, Marc-Olivier Renou},
  pdfsubject = {Completeness of quantum inflation and source-transfer SDP hierarchies for quantum tree-network correlations, with a sufficient criterion for finite-dimensional extraction},
  pdfkeywords = {quantum tree networks, quantum correlations, quantum causal compatibility, network nonlocality, quantum inflation, inflation-NPA hierarchy, semidefinite programming, noncommutative real algebraic geometry, multi-state polynomial optimization, source-transfer model, mixed quantum model, completely positive maps, Radon-Nikodym derivative, archimedean quadratic modules, flat extension, quantum de Finetti theorem, coRE-completeness, device-independent certification},
  pdflang = {en},
  pdfdisplaydoctitle = true
}

\title{\Large\bfseries{Quantum inflation and source-transfer hierarchies\\are complete for tree networks}}

\author[1,$*$]{Xiangling Xu}
\author[2,3]{Igor Klep}
\author[1]{Marc-Olivier Renou}

\affil[1]{Inria, CPHT, LIX, CNRS, \'Ecole Polytechnique, Institut Polytechnique de Paris, Palaiseau, France}
\affil[2]{University of Ljubljana, Faculty of Mathematics and Physics, Jadranska 21, 1000 Ljubljana, Slovenia}
\affil[3]{University of Primorska, Faculty of Mathematics, Natural Sciences and Information Technologies, Glagolja\v{s}ka 8, 6000 Koper, Slovenia}
\affil[*]{\scriptsize\texttt{xu.xiangling@inria.fr}}

\date{\vspace{-2.80em}}

\begin{document}
\maketitle
\begin{abstract}
Understanding the capabilities and limits of quantum systems is central to quantum information processing.
This includes determining which correlations can arise in networks of quantum systems distributed by independent sources.
Yet for general quantum networks, no systematic, dimension-free method is known to bound these correlations arbitrarily well.
We resolve this problem for all quantum tree networks through a new \emph{source-transfer model}, which we prove equivalent to a \emph{mixed quantum model}.
The mixed model places no restriction on dimension and reduces to the standard tensor-product model in finite dimensions; the source-transfer model instead describes these correlations through states on $C^*$-algebras and completely positive maps composed recursively along the tree.
We prove the equivalence by reconstructing local measurements as Radon--Nikodym derivatives in the spatial tensor products prescribed by the network.
Building on this characterization, we establish two convergent outer hierarchies of semidefinite programs (SDPs): a new \emph{source-transfer hierarchy} and the \emph{inflation-NPA hierarchy}.
For the source-transfer hierarchy, we extend noncommutative real algebraic geometry to \emph{multi-state polynomials with completely positive maps}, prove a reconstruction theorem establishing its convergence, and give a sufficient stopping criterion for extracting finite-dimensional realizations.
For the inflation-NPA hierarchy, we reconstruct source-transfer realizations from weak-operator limits of averages over source copies produced by quantum inflation.
Consequently, the game-value promise problem for quantum tree networks in the mixed model is $\mathsf{coRE}$-complete.
\end{abstract}

\par\addvspace{1.2\baselineskip}
\noindent\makebox[\textwidth][c]{%
\begin{minipage}{0.9\textwidth}
\par
\end{minipage}%
}
\par\addvspace{1.2\baselineskip}

\tableofcontents

\section{Introduction}
Understanding what quantum systems can and cannot achieve is a central goal of quantum information theory.
For information processing protocols in a quantum network, these possibilities depend on how the parties are connected and which quantum sources they share.
Here, we characterize the correlations achievable when separated parties perform local measurements on systems distributed by independent quantum sources in a prescribed network structure (\zcref{fig:IntroductionNetworkScope}).
This is a first step towards understanding more general quantum causal structures that arise in tasks such as distributed quantum computation.
In a device-independent setting, the source states and measurement operators are unspecified, so this characterization provides bounds valid for every choice consistent with the network structure~\cite{Lee2018TowardsDeviceIndependent}.

Bell scenarios, in which all parties share a single quantum source (\zcref{fig:IntroductionScopeBell}), provide the canonical setting for studying quantum nonlocality~\cite{brunner2014bell}.
The standard tensor-product model assigns separate Hilbert spaces to the parties, whereas the commuting-operator model expresses locality through commutation relations between their measurements.
These axiomatizations agree in finite dimensions, but the commuting-operator model is strictly more general: it admits correlations outside even the closure of finite-dimensional tensor-product correlations~\cite{ji2021mip}.

\begin{figure}[!tbp]
    \centering
    \begingroup
    \tikzset{
        scope party/.style={draw, rectangle, minimum size=6mm, inner sep=0pt, fill=white, font=\small},
        scope source/.style={draw, circle, minimum size=6mm, inner sep=0pt, fill=black!7, font=\small},
        scope arrow/.style={-{Stealth[length=1.5mm]}, line width=0.85pt},
        scope message/.style={scope arrow, draw=blue!55!black},
        scope memory/.style={line width=0.55pt, draw=black!45}
    }
    \begin{subfigure}[t]{0.48\linewidth}
        \centering
        \begin{tikzpicture}[line width=0.55pt]
            \path[use as bounding box] (-3.10,-1.55) rectangle (3.10,1.15);
            \begin{scope}[scale=1.12]
            \node[scope source] (s) at (0,-0.45) {$\alpha$};
            \node[scope party] (a) at (0,0.75) {$A$};
            \node[scope party] (b) at (-1.03923,-1.05) {$B$};
            \node[scope party] (c) at (1.03923,-1.05) {$C$};
            \draw (s)--(a) (s)--(b) (s)--(c);
            \end{scope}
        \end{tikzpicture}
        \caption{Bell scenarios}
        \label{fig:IntroductionScopeBell}
        \medskip
        \begin{tikzpicture}[line width=0.55pt]
            \path[use as bounding box] (-3.10,-5.00) rectangle (3.10,-1.45);
            \begin{scope}[shift={(0,-2.85)}]
                \draw[black, dashed, line width=0.7pt] (0,0) ellipse (1.98 and 1.15);
                \filldraw[draw=black, fill=black!4] (0,0) ellipse (1.55 and 0.80);
                \filldraw[draw=black, fill=black!12] (0,-0.31) ellipse (1.15 and 0.49);
                \node[font=\small] at (0,0.46) {$C_{\mathrm{qc}}$};
                \node[font=\small] at (0,-0.31) {$C_{\mathrm{qa}}$};
                \node[font=\footnotesize, align=center] at (0,-1.80) {NPA hierarchy\\converges from outside};
                \draw[scope arrow, black] (0,1.15)--(0,0.86);
                \draw[scope arrow, black] (0,-1.15)--(0,-0.86);
            \end{scope}
        \end{tikzpicture}
    \end{subfigure}\hfill
    \begin{subfigure}[t]{0.48\linewidth}
        \centering
        \begin{tikzpicture}[line width=0.55pt]
            \path[use as bounding box] (-3.10,-1.55) rectangle (3.10,1.15);
            \begin{scope}[xscale=1.40, yscale=1.12]
            \node[scope party] (a) at (-1.95,-0.70) {$A$};
            \node[scope source] (ab) at (-1.30,-0.70) {$\alpha$};
            \node[scope party] (b) at (-0.65,-0.70) {$B$};
            \node[scope source] (bce) at (0,-0.70) {$\beta$};
            \node[scope party] (c) at (0.65,-0.70) {$C$};
            \node[scope source] (cd) at (1.30,-0.70) {$\gamma$};
            \node[scope party] (d) at (1.95,-0.70) {$D$};
            \node[scope party] (e) at (0,0.75) {$E$};
            \draw (a)--(ab)--(b)--(bce)--(c)--(cd)--(d) (bce)--(e);
            \end{scope}
        \end{tikzpicture}
        \caption{Tree networks: this work}
        \label{fig:IntroductionScopeTree}
        \medskip
        \begin{tikzpicture}[line width=0.55pt]
            \path[use as bounding box] (-3.10,-5.00) rectangle (3.10,-1.45);
            \begin{scope}[shift={(0,-2.85)}]
                \begin{scope}[xscale=1.08, yscale=0.85, xshift=0.24cm]
                \draw[black, dashed, line width=0.7pt] (1.48,0.88) .. controls (0.12,1.80) and (-2.04,1.42) .. (-2.04,0) .. controls (-2.04,-1.42) and (0.12,-1.80) .. (1.48,-0.88) -- (0.42,-0.035) .. controls (-0.04,-0.22) and (-0.08,-0.16) .. (-0.08,0) .. controls (-0.08,0.16) and (-0.04,0.22) .. (0.42,0.035) -- cycle;
                \filldraw[draw=black, fill=black!4] (0.98,0.67) .. controls (0.05,1.20) and (-1.46,0.93) .. (-1.46,0) .. controls (-1.46,-0.93) and (0.05,-1.20) .. (0.98,-0.67) -- (0.48,-0.33) .. controls (-0.14,-0.65) and (-0.48,-0.40) .. (-0.48,0) .. controls (-0.48,0.40) and (-0.14,0.65) .. (0.48,0.33) -- cycle;
                \node[font=\small] at (-1.05,0) {$C_{\mathrm{mix}}$};
                \end{scope}
                \draw[scope arrow, black] (0,1.15)--(0,0.86);
                \draw[scope arrow, black] (0,-1.15)--(0,-0.86);
                \node[font=\footnotesize, align=center] at (0,-1.80) {Source-transfer and inflation-NPA\\both converge from outside};
            \end{scope}
        \end{tikzpicture}
    \end{subfigure}
    \endgroup
    \caption{From Bell scenarios to quantum tree networks.
    Squares denote parties and circles denote independent sources.
    Dashed contours illustrate outer approximations, with arrows indicating convergence.
    In (a), the NPA hierarchy converges from outside to the commuting-operator set $C_{\mathrm{qc}}$, which in general strictly contains $C_{\mathrm{qa}}$, the closure of finite-dimensional tensor-product correlations.
    In (b), the source-transfer and inflation-NPA hierarchies both converge from outside to the nonconvex mixed-model correlation set $C_{\mathrm{mix}}$ on tree networks.
    }
    \label{fig:IntroductionNetworkScope}
\end{figure}
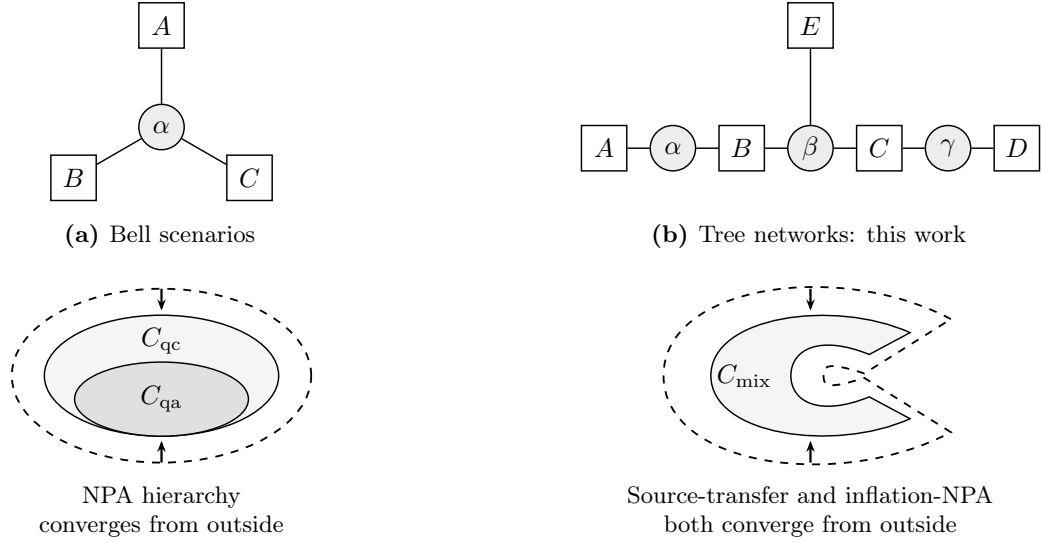

This distinction raises a fundamental computational question: which model admits a systematic characterization through inner or outer approximations?
The $\mathsf{RE}$-completeness of the tensor-product model and the $\mathsf{coRE}$-completeness of the commuting-operator model rule out general algorithms producing convergent outer approximations of tensor-product correlations or convergent inner approximations of commuting-operator correlations~\cite{ji2021mip,lin2025mipcocore}.
Inner approximations establish achievable performance and can demonstrate quantum advantages over classical strategies, while outer approximations establish limits valid for every strategy allowed by the model.
To bound quantum performance and adversarial capabilities in device-independent protocols, the commuting-operator model is therefore a natural target for complete outer hierarchies~\cite{Pironio2010RandomNumber,Primaatmaja2023securityofdevice}.

The Navascu\'es--Pironio--Ac\'in (NPA) hierarchy provides such a convergent outer approximation~\cite{navascues2008convergent}.
It is an instance of noncommutative polynomial optimization, in which positivity and algebraic relations among measurement operators are tested through a hierarchy of semidefinite programs (SDPs)~\cite{pironio2010convergent}.
The hierarchy is \emph{outer}: every correlation in the model passes every level.
It is also \emph{convergent}: every correlation outside the model is rejected at some finite level.
We call a hierarchy \emph{complete} when it has both properties, as illustrated by \zcref{fig:IntroductionScopeBell}.
Thus, the NPA hierarchy is complete and provides rigorous, arbitrarily tight bounds on optimal performance in the commuting-operator model of Bell scenarios.

Quantum networks with several independent sources generalize the Bell setting of a single shared source~\cite{tavakoli2022bellnetworks}.
The natural next question is which axiomatizations of measurement locality and source independence admit complete outer approximations that characterize network correlations and bound their performance.
Such bounds have applications to device-independent cryptography, randomness generation, and the certification of entangled states and measurements in networks~\cite{wolfe2021quantum,minati2025randomness,supic2023networkselftest,Renou2018SelftestingEntangled}.
More broadly, SDP hierarchies also bound the performance of quantum provers in cryptographically compiled nonlocal games~\cite{klep2026quantitativequantumsoundnessbipartite,baroni2026quantitativequantumsoundnessmultipartite}.
Quantum networks also reveal foundational distinctions, such as separations between real and complex quantum theory and between fermionic and bosonic information theory~\cite{Renou2021QuantumTheory,kalarde2026fermionsfundamentallynonlocalbosons}.
Questions about correlations permitted by locality and causal constraints also arise in distributed quantum computing and in online and dynamic graph models~\cite{balliu2024distributedquantumadvantagelocal,akbari2026onlinelocalitymeetsdistributed}.

The main difficulty in extending complete outer hierarchies to quantum networks is to enforce source independence.
The network problem is generally nonconvex (\zcref{fig:IntroductionScopeTree}): a mixture of two valid quantum network strategies need not be a valid strategy for the same network~\cite{tavakoli2022bellnetworks,renou2019limits}.
The original NPA hierarchy does not enforce source independence and is therefore too weak for this task.
The challenge is to identify additional constraints that can be imposed in finite SDP relaxations and prove that their limiting solutions reconstruct realizations with the required network structure.

For correlations generated by independent classical sources in a network, inflation provides a complete hierarchy of linear-programming relaxations based on source duplications~\cite{navascues2020inflation}.
Its quantum counterpart---quantum inflation and the associated inflation-NPA hierarchy~\cite{wolfe2021quantum}---provides outer relaxations, but convergence remains a major open question for general quantum networks.
The authors of~\cite{ligthart2023convergent} proved convergence of a modified inflation-NPA hierarchy with additional bounds on operator Schmidt ranks and norms of local operators.
These bounds, however, mean that the resulting relaxations are no longer outer for the unrestricted model.

Prior to this work, completeness of inflation-NPA and two alternative SDP hierarchies had been established for bilocal networks, where three parties are connected by two independent sources~\cite{ligthart2023inflation,renou2026two}, with extensions to star networks, where independent bipartite sources connect a central party to the other parties~\cite[Corollary~3.3.17]{ligthart2024semidefinite}.
A key step in these completeness proofs is an operator-algebraic reconstruction that ``splits'' the central party into subsystems associated with its connected sources~\cite{ligthart2023inflation,renou2026two,xu2023quantum}.
However, already on the four-party line in \zcref{fig:IntroductionLine}, the same argument does not guarantee consistent splittings of the two internal parties $2$ and $3$.
Extending the above results therefore requires new tools to make these local splittings consistent across the network.

In this work, we develop these tools and resolve the problem for all quantum tree networks, that is, connected networks whose graph of sources and measurement parties contains no cycles (\zcref{fig:IntroductionScopeTree}).
Central to our contribution is a new \emph{source-transfer model}, which we prove equivalent to the \emph{mixed quantum model} on every such tree network.
The model is called \emph{mixed} because it uses tensor products of states to describe source independence and commuting observable algebras to describe measurement locality; in finite dimensions, it is equivalent to the standard tensor-product model.
The source-transfer model describes the same correlations through states on source $C^*$-algebras and completely positive maps composed recursively along the tree, without the source duplication used by quantum inflation.
We prove the equivalence by reconstructing local measurements as Radon--Nikodym derivatives in the spatial tensor products prescribed by the network.
Building on this equivalence, we establish \emph{two complete SDP hierarchies}.
For the first, we extend noncommutative real algebraic geometry to a framework for jointly optimizing over measurement operators, states on different algebras, and completely positive maps between them, subject to the network structure.
This leads to a new \emph{source-transfer hierarchy} based on \emph{multi-state polynomials with completely positive maps}.
We prove a reconstruction theorem establishing its completeness and give a sufficient stopping criterion for extracting finite-dimensional realizations.
The second hierarchy is inflation-NPA itself: we resolve its convergence for all quantum tree networks by reconstructing source-transfer realizations.
A direct computational consequence is that the game-value promise problem remains $\mathsf{coRE}$-complete for quantum tree networks in the mixed model.

\subsection{Mixed and source-transfer models for quantum tree networks}\label{sec:IntroductionModels}
To obtain device-independent bounds valid for systems of unrestricted dimension, including infinite-dimensional realizations, we must specify the network model to which our outer approximations converge.
We formulate the \emph{mixed quantum model} for general networks, extending the bilocal formulation of~\cite[Corollary~7]{ligthart2023inflation}, and then introduce the \emph{source-transfer model}.

To illustrate the two models, we consider the line-of-four network $1-\alpha-2-\beta-3-\gamma-4$ in \zcref{fig:IntroductionLine}.
We describe two models for correlations $p(\mathbf{a}|\mathbf{x})$ arising in this network, where $\mathbf{a}=(a_1,a_2,a_3,a_4)$ and $\mathbf{x}=(x_1,x_2,x_3,x_4)$ denote the outcomes and measurement settings.

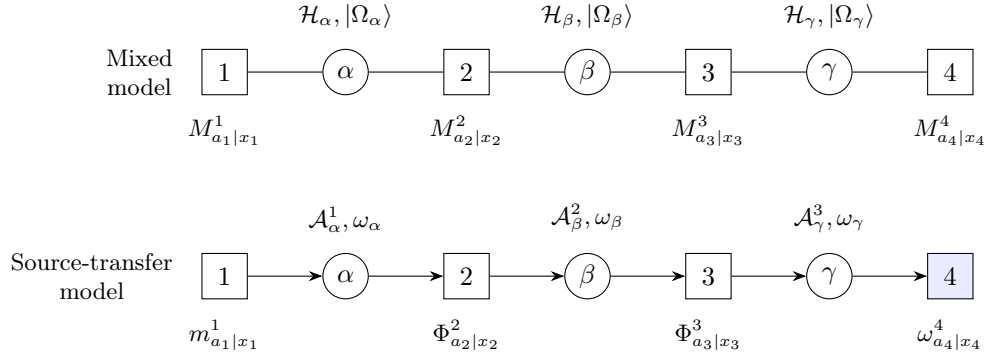
\begin{figure}[htbp]
    \centering
    \begin{tikzpicture}[
        party/.style={draw,rectangle,minimum size=6mm,fill=white},
        root/.style={party,fill=blue!8},
        source/.style={draw,circle,minimum size=6mm,inner sep=1pt},
        every node/.style={font=\small},
        annotation/.style={font=\footnotesize,align=center},
        transfer/.style={-{Stealth[length=1.8mm]}}
    ]
        \node[annotation,anchor=east] at (-0.55,0) {Mixed\\model};
        \node[party] (p1) at (0,0) {$1$};
        \node[source] (a) at (1.6,0) {$\alpha$};
        \node[party] (p2) at (3.2,0) {$2$};
        \node[source] (b) at (4.8,0) {$\beta$};
        \node[party] (p3) at (6.4,0) {$3$};
        \node[source] (c) at (8.0,0) {$\gamma$};
        \node[party] (p4) at (9.6,0) {$4$};
        \draw (p1)--(a)--(p2)--(b)--(p3)--(c)--(p4);
        \node[annotation,above=4pt of a] {$\cH_\alpha,\ket{\Omega_\alpha}$};
        \node[annotation,above=4pt of b] {$\cH_\beta,\ket{\Omega_\beta}$};
        \node[annotation,above=4pt of c] {$\cH_\gamma,\ket{\Omega_\gamma}$};
        \node[annotation,below=4pt of p1] {$\Meas{1}{a_1}{x_1}$};
        \node[annotation,below=4pt of p2] {$\Meas{2}{a_2}{x_2}$};
        \node[annotation,below=4pt of p3] {$\Meas{3}{a_3}{x_3}$};
        \node[annotation,below=4pt of p4] {$\Meas{4}{a_4}{x_4}$};

        \node[annotation,anchor=east] at (-0.55,-2.7) {Source-transfer\\model};
        \node[party] (s1) at (0,-2.7) {$1$};
        \node[source] (sa) at (1.6,-2.7) {$\alpha$};
        \node[party] (s2) at (3.2,-2.7) {$2$};
        \node[source] (sb) at (4.8,-2.7) {$\beta$};
        \node[party] (s3) at (6.4,-2.7) {$3$};
        \node[source] (sc) at (8.0,-2.7) {$\gamma$};
        \node[root] (s4) at (9.6,-2.7) {$4$};
        \draw[transfer] (s1)--(sa);
        \draw[transfer] (sa)--(s2);
        \draw[transfer] (s2)--(sb);
        \draw[transfer] (sb)--(s3);
        \draw[transfer] (s3)--(sc);
        \draw[transfer] (sc)--(s4);
        \node[annotation,above=4pt of sa] {$\cA_\alpha^1,\omega_\alpha$};
        \node[annotation,above=4pt of sb] {$\cA_\beta^2,\omega_\beta$};
        \node[annotation,above=4pt of sc] {$\cA_\gamma^3,\omega_\gamma$};
        \node[annotation,below=4pt of s1] {$\meas{1}{a_1}{x_1}$};
        \node[annotation,below=4pt of s2] {$\instr{2}{a_2}{x_2}$};
        \node[annotation,below=4pt of s3] {$\instr{3}{a_3}{x_3}$};
        \node[annotation,below=4pt of s4] {$\rstate{4}{a_4}{x_4}$};
    \end{tikzpicture}
    \caption{Two quantum models for the line-of-four network.
    Squares denote parties and circles denote independent sources.
    In the mixed model, local measurements act on the source systems.
    In the source-transfer model, arrows indicate the recursive evaluation of the correlation from party $1$ towards the root at party $4$.}
    \label{fig:IntroductionLine}
\end{figure}

In the mixed quantum model, each source $\delta\in\{\alpha,\beta,\gamma\}$ has a Hilbert space $\cH_\delta$ and a unit vector $\ket{\Omega_\delta}$.
Source independence is expressed by the product vector $\ket{\Omega_\alpha}\otimes\ket{\Omega_\beta}\otimes\ket{\Omega_\gamma}$.
The two parties receiving each source have von Neumann algebras $\cM_\alpha^1,\cM_\alpha^2\subset\Bof{\cH_\alpha}$, $\cM_\beta^2,\cM_\beta^3\subset\Bof{\cH_\beta}$, and $\cM_\gamma^3,\cM_\gamma^4\subset\Bof{\cH_\gamma}$, respectively, satisfying
\begin{align}\label{eq:IntroductionMixedLocality}
    [\cM_\alpha^1,\cM_\alpha^2]=[\cM_\beta^2,\cM_\beta^3]=[\cM_\gamma^3,\cM_\gamma^4]=0.
\end{align}
Each party measures the systems from its neighboring sources with positive operator-valued measure (POVM) effects
\begin{equation}\label{eq:IntroductionMixedMeasurements}
    \begin{aligned}
        \Meas{1}{a_1}{x_1}&\in\cM_\alpha^1, & \Meas{2}{a_2}{x_2}&\in\cM_\alpha^2\vnotimes\cM_\beta^2,\\
        \Meas{3}{a_3}{x_3}&\in\cM_\beta^3\vnotimes\cM_\gamma^3, & \Meas{4}{a_4}{x_4}&\in\cM_\gamma^4,
    \end{aligned}
\end{equation}
where $\vnotimes$ denotes the spatial tensor product of von Neumann algebras.
We regard each POVM effect as an operator on $\cH_\alpha\otimes\cH_\beta\otimes\cH_\gamma$ by tensoring with identities on the remaining source spaces.
Born's rule gives
\begin{align}\label{eq:IntroductionMixedCorrelation}
    p(\mathbf{a}\mid\mathbf{x})=\sandwich{\Omega_\alpha\otimes\Omega_\beta\otimes\Omega_\gamma}{\prod_{v=1}^{4}\Meas{v}{a_v}{x_v}}{\Omega_\alpha\otimes\Omega_\beta\otimes\Omega_\gamma}.
\end{align}

The source-transfer model describes the same correlation recursively through abstract $C^*$-algebras, states, and completely positive maps, without specifying Hilbert-space representations.
Choosing party $4$ as the root, we evaluate the correlation by transferring party $1$'s measurement operator from source $\alpha$ to $\beta$ through the map at party $2$, and then to $\gamma$ through the map at party $3$.
The three sources are now described by unital $C^*$-algebras $\cA_\alpha^1$, $\cA_\beta^2$, and $\cA_\gamma^3$, with states $\omega_\alpha$, $\omega_\beta$, and $\omega_\gamma$.
A leaf POVM $\meas{1}{a_1}{x_1}\in\cA_\alpha^1$ is followed by completely positive maps and a positive root functional:
\begin{equation}\label{eq:IntroductionTransferTypes}
    \begin{aligned}
        \instr{2}{a_2}{x_2}:\cA_\alpha^1\longrightarrow\cA_\beta^2,\qquad \instr{3}{a_3}{x_3}:\cA_\beta^2\longrightarrow\cA_\gamma^3, \qquad \rstate{4}{a_4}{x_4}:\cA_\gamma^3\longrightarrow\mathbb{C}.
    \end{aligned}
\end{equation}
They recover the correlation by
\begin{align}\label{eq:IntroductionLineTransfer}
    p(\mathbf{a}|\mathbf{x})=\rstate{4}{a_4}{x_4}\!\left(\instr{3}{a_3}{x_3}\!\left(\instr{2}{a_2}{x_2}\!\left(\meas{1}{a_1}{x_1}\right)\right)\right).
\end{align}

Source independence requires
\begin{equation}\label{eq:IntroductionSourceReplacer}
    \begin{aligned}
        \sum_{a_2}\instr{2}{a_2}{x_2}(q)&=\omega_\alpha(q)\,\id_{\cA_\beta^2}, &&q\in\cA_\alpha^1,\\
        \sum_{a_3}\instr{3}{a_3}{x_3}(q)&=\omega_\beta(q)\,\id_{\cA_\gamma^3}, &&q\in\cA_\beta^2,\\
        \sum_{a_4}\rstate{4}{a_4}{x_4}&=\omega_\gamma.
    \end{aligned}
\end{equation}
We call these identities \emph{source-replacer conditions}.
In the Schr\"{o}dinger picture, summing over a party's outcomes gives a channel that discards the input state and prepares the fixed source state.

The general definitions of the mixed and source-transfer models are given in \zcref{def:MixedQuantumModelSourcePartyGraph,def:SourceTransferModelTree}, respectively.
Further examples of mixed and source-transfer realizations are worked out in \zcref{sec:SimpleExample}.
We prove that the two models give the same set of correlations on every source-party tree (\zcref{thm:SourceTransferTreeEquivalence}); this set is compact (\zcref{rem:MixedTreeCompactness}).
In finite dimensions, both are equivalent to the standard tensor-product model of quantum information (\zcref{rem:FiniteDimensionalTransferToMixed}).
Note that a characterization based on measurement-algebra factorization alone~\cite{PozasKerstjens2019Bounding,yang2026mutuallycommutingvonneumannalgebra} is not sufficient on tree networks; see \zcref{sec:AppendixWhySourceTransfer} for a counterexample already on the line-of-four network.

\subsection{Main results}\label{sec:IntroductionResults}
We now state the results for an arbitrary finite quantum tree network.
The source-transfer characterization yields two complete SDP hierarchies: the \emph{source-transfer hierarchy}, which directly imposes polynomial constraints on states and transfer maps, and the \emph{inflation-NPA hierarchy}, which uses quantum inflation.
Their constructions are introduced in \zcref{sec:IntroductionMultiState,sec:IntroductionInflation}, respectively.
The following statement summarizes \zcref{thm:SourceTransferTreeEquivalence,thm:SourceTransferSDPConvergence,thm:InflationSourceTransferTreeEquivalence,cor:MixedTreeSDPConvergence,cor:InflationNPATreeConvergence,rem:FiniteDimensionalTransferToMixed,thm:SourceTransferFlatExtraction,cor:MixedTreeGameComputability}.

\begin{maintheorem}[Informal]\label{thm:IntroductionCompleteness}
Let $\Tree$ be a finite quantum tree network with finite setting and outcome sets.

\thmpart{Equivalence and completeness.}
The source-transfer hierarchy and the inflation-NPA hierarchy are complete for the mixed model on $\Tree$.
Specifically, for a correlation $p$, the following are equivalent:
\begin{enumerate}[label=(\roman*)]
    \item $p$ admits a mixed quantum realization on $\Tree$ as in \zcref{def:MixedQuantumModelSourcePartyGraph};
    \item $p$ admits a source-transfer realization on $(\Tree,r)$ for every root party $r$, as in \zcref{def:SourceTransferModelTree};
    \item $p$ passes every level of the source-transfer SDP hierarchy on $(\Tree,r)$ for every root party $r$, as in \zcref{def:SourceTransferSDPRelaxation};
    \item $p$ passes every level of the inflation-NPA hierarchy on $\Tree$ as in \zcref{def:InflationNPATest}.
\end{enumerate}

\thmpart{Finite-dimensional realizations and extraction.}
\begin{enumerate}[label=(\roman*)]
    \item In finite dimensions, the mixed, source-transfer, and tensor-product models give the same correlations on $\Tree$.
    \item From a finite-level feasible solution of the source-transfer SDP hierarchy satisfying the stopping criterion of \zcref{thm:SourceTransferFlatExtraction}, we can extract a finite-dimensional realization in each of the models in (i), reproducing the prescribed correlations.
\end{enumerate}

\thmpart{Optimization.}
Let $f$ be a real polynomial in the correlation entries, and write $f(\cR)$ for its value on the correlation generated by a mixed realization $\cR$ on $\Tree$.
For either hierarchy, let $\beta_d^{\min}(f)$ and $\beta_d^{\max}(f)$ denote its lower and upper SDP bounds for this objective, as the level increases through the hierarchy, with a fixed root $r$ in the source-transfer case.
Then
\begin{equation}\label{eq:IntroductionOptimization}
    \begin{aligned}
        \beta_1^{\min}(f)\leq\beta_2^{\min}(f)\leq\cdots&\nearrow\inf_{\cR}f(\cR),\\
        \beta_1^{\max}(f)\geq\beta_2^{\max}(f)\geq\cdots&\searrow\sup_{\cR}f(\cR),
    \end{aligned}
\end{equation}
where $\cR$ ranges over mixed realizations on $\Tree$.

For the source-transfer hierarchy, an optimal finite-level solution satisfying the stopping criterion for $f$ yields an optimal finite-dimensional realization in each of the models above.

\thmpart{Computability.}
Consequently, given a finite tree network and a game, the promise problem
\begin{equation}\label{eq:IntroductionGamePromise}
    \begin{aligned}
        \mathrm{YES}:\quad &\text{optimal mixed-model winning probability}=1,\\
        \mathrm{NO}:\quad &\text{optimal mixed-model winning probability}\leq\frac{1}{2},
    \end{aligned}
\end{equation}
is $\mathsf{coRE}$-complete.
\end{maintheorem}

Although both hierarchies characterize the same correlation set in the limit, their finite-level bounds can differ.
The size and strength of the source-transfer relaxations can also depend on the chosen root (\zcref{rem:FiniteLevelRootDependence}).
In addition, we extend the source-transfer characterization to networks in which removing the classical sources leaves a disjoint union of quantum trees, by treating the classical values as additional measurement settings~(\zcref{sec:BeyondTree}).

\subsection{Multi-state polynomials with completely positive maps and their SDP hierarchy}\label{sec:IntroductionMultiState}
We now construct the first of the two hierarchies in \zcref{thm:IntroductionCompleteness}.
The source-transfer model raises a new optimization problem: how can we impose source independence while optimizing jointly over states and completely positive maps?
Existing SDP hierarchies for post-measurement states and quantum instruments characterize convex sets of Bell correlations~\cite{klep2026quantitativequantumsoundnessbipartite,baroni2026quantitativequantumsoundnessmultipartite}, but their constraints do not enforce independence between several sources.
For networks, the source-replacer identities couple the maps to products of independent source-state evaluations, and these additional constraints must survive reconstruction from the SDP data.

We resolve this problem by introducing a new framework of \emph{multi-state polynomials with completely positive maps}, extending noncommutative real algebraic geometry to multiple states linked by completely positive maps.
Going beyond state-polynomial optimization~\cite{klep2024state}, the framework incorporates recursively composed maps between distinct source algebras and polynomial relations among their source states.
We prove a reconstruction theorem that recovers these states and maps simultaneously while preserving complete positivity and source independence (\zcref{lem:SourceTransferCharacterReconstruction}).
This establishes completeness of the resulting source-transfer hierarchy for the nonconvex network models considered here (\zcref{thm:SourceTransferSDPConvergence}).
Furthermore, the stopping criterion in \zcref{thm:SourceTransferFlatExtraction} extracts a finite-dimensional source-transfer realization from a feasible solution at a finite level of the hierarchy.

Returning to the line-of-four example in \zcref{sec:IntroductionModels}, we replace the source algebras by formal $*$-algebras $\fA_\alpha^1$, $\fA_\beta^2$, and $\fA_\gamma^3$.
We introduce formal letters $\meas{1}{a_1}{x_1}$ for the measurement effects at party $1$, which generate $\fA_\alpha^1$.
For each word $q\in\fA_\alpha^1$ and each outcome and setting $(a_2,x_2)$, we introduce a formal symbol $\formalInstr{2}{a_2}{x_2}(q)$.
These symbols generate $\fA_\beta^2$; repeating the construction with symbols $\formalInstr{3}{a_3}{x_3}(q)$ for words $q\in\fA_\beta^2$ generates $\fA_\gamma^3$.
We extend the formal maps linearly and require adjoint preservation.
In particular, the construction includes products of map outputs, such as $\formalInstr{2}{a_2}{x_2}(q)^*\formalInstr{2}{a_2}{x_2}(q)$, as well as compositions of different maps.

For each word $q$ in $\fA_\alpha^1$, $\fA_\beta^2$, or $\fA_\gamma^3$, we introduce a scalar symbol $\varsigma_\delta(q)$ representing its expectation under the corresponding source state $\omega_\delta$, with $\delta=\alpha,\beta,\gamma$, respectively.
Likewise, for each word $q\in\fA_\gamma^3$, we introduce symbols $\rPolystate{a_4}{x_4}(q)$ representing the values of the root functionals $\rstate{4}{a_4}{x_4}$.
These scalar symbols commute and generate the commutative multi-state polynomial $*$-algebra $\thinS$.
The correlation and a source-replacer condition are represented by the following polynomial expressions, where $q\in\fA_\alpha^1$ and $w_1,w_2\in\fA_\beta^2$:
\begin{equation}\label{eq:IntroductionFormalCorrelationSourceReplacer}
    \begin{aligned}
        h_{\mathbf{a}\mid\mathbf{x}}\coloneqq\rPolystate{a_4}{x_4}\!\left(\formalInstr{3}{a_3}{x_3}\!\left(\formalInstr{2}{a_2}{x_2}\!\left(\meas{1}{a_1}{x_1}\right)\right)\right),\\
        \sum_{a_2}\varsigma_\beta\!\left(w_1^*\formalInstr{2}{a_2}{x_2}(q)w_2\right)=\varsigma_\alpha(q)\,\varsigma_\beta(w_1^* w_2).
    \end{aligned}
\end{equation}
We can \emph{evaluate} these expressions on a given source-transfer realization by replacing the formal measurement letters with the realization's measurement effects, $\formalInstr{v}{a_v}{x_v}$ with $\instr{v}{a_v}{x_v}$, $\varsigma_\delta$ with $\omega_\delta$, and $\rPolystate{a_4}{x_4}$ with $\rstate{4}{a_4}{x_4}$.
This recovers $h_{\mathbf{a}|\mathbf{x}}=p(\mathbf{a}|\mathbf{x})$ from \zcref{eq:IntroductionLineTransfer} and the first source-replacer identity in \zcref{eq:IntroductionSourceReplacer}, tested between $w_1^*$ and $w_2$ in the state $\omega_\beta$.

This polynomial framework gives rise to SDP hierarchies: one optimizes finite-dimensional restrictions of a normalized positive linear functional $L:\thinS\to\mathbb{C}$, subject to constraints encoding the required properties of the measurements, states, and transfer maps.
For instance, complete positivity of $\formalInstr{2}{a_2}{x_2}$ is encoded by constraints of the form
\begin{align}\label{eq:IntroductionPolynomialCPMatrix}
    \left[L\!\left(\varsigma_\beta\!\left(w_i^*\formalInstr{2}{a_2}{x_2}(q_i^*q_j)w_j\right)\right)\right]_{i,j=1}^{n}\succeq0,
\end{align}
for finite families of words $q_i\in\fA_\alpha^1$ and $w_i\in\fA_\beta^2$, indexed by $i=1,\ldots,n$.
Together with normalization, state positivity, the archimedean condition, and source-replacer relations, these constraints yield a finite SDP at each level.
We call this the \emph{source-transfer SDP hierarchy}.
The general polynomial construction and SDP hierarchy are given in \zcref{sec:MultiStatePolyRing,sec:SourceTransferSDPHierarchy}.
Explicit examples of the resulting SDP relaxations are given in \zcref{sec:SimpleExample}.

\subsection{Quantum inflation and the inflation-NPA hierarchy}\label{sec:IntroductionInflation}
The second hierarchy in \zcref{thm:IntroductionCompleteness} is built from quantum inflation.
\emph{Quantum inflation} constrains the correlations achievable in a quantum network through the device-duplication principle~\cite{wolfe2021quantum}.
Suppose a correlation $p$ is compatible with the line-of-four network in the mixed quantum model, with POVM effects $\Meas{v}{a_v}{x_v}$ at parties $v=1,2,3,4$.
Parties $1$ and $4$ measure systems from $\alpha$ and $\gamma$, respectively, while parties $2$ and $3$ perform joint measurements on systems from $(\alpha,\beta)$ and $(\beta,\gamma)$.

We now prepare $m$ independent copies of each source, as illustrated for $m=2$ in \zcref{fig:IntroductionInflationLine}.
Each party can apply its original POVM to any choice of copies of the sources it receives.

\begin{figure}[htbp]
    \centering
    \definecolor{inflationAlpha}{HTML}{0077BB}
    \definecolor{inflationBeta}{HTML}{EE7733}
    \definecolor{inflationGamma}{HTML}{AA3377}
    \begin{tikzpicture}[
        party/.style={draw,rectangle,minimum size=6mm,fill=white},
        source/.style={draw,circle,minimum size=7mm,inner sep=1pt,fill=white},
        every node/.style={font=\small},
        connection/.style={draw=black,-{Stealth[length=1.5mm]}},
        selected/.style={line width=0.9pt,-{Stealth[length=1.7mm]}},
        selectedalpha/.style={selected,draw=inflationAlpha},
        selectedbeta/.style={selected,draw=inflationBeta},
        selectedgamma/.style={selected,draw=inflationGamma}
    ]
        \node[party] (p1) at (0,0) {$1$};
        \node[party] (p2) at (3.2,0) {$2$};
        \node[party] (p3) at (6.4,0) {$3$};
        \node[party] (p4) at (9.6,0) {$4$};
        \node[source] (a1) at (1.6,0.85) {$\alpha_1$};
        \node[source] (a2) at (1.6,-0.85) {$\alpha_2$};
        \node[source] (b1) at (4.8,0.85) {$\beta_1$};
        \node[source] (b2) at (4.8,-0.85) {$\beta_2$};
        \node[source] (c1) at (8.0,0.85) {$\gamma_1$};
        \node[source] (c2) at (8.0,-0.85) {$\gamma_2$};
        \draw[connection] (a1)--(p1);
        \draw[connection] (a2)--(p1);
        \draw[connection] (a2)--(p2);
        \draw[connection] (b1)--(p2);
        \draw[connection] (b1)--(p3);
        \draw[connection] (c2)--(p3);
        \draw[connection] (c1)--(p4);
        \draw[connection] (c2)--(p4);
        \draw[selectedalpha] (a1)--(p2);
        \draw[selectedbeta] (b2)--(p2);
        \draw[selectedbeta] (b2)--(p3);
        \draw[selectedgamma] (c1)--(p3);
        \node[font=\footnotesize] at (0,-1.6) {$E^1_{a_1\mid x_1}(i)$};
        \node[font=\footnotesize] at (3.2,-1.6) {$E^2_{a_2\mid x_2}(\textcolor{inflationAlpha}{1},\textcolor{inflationBeta}{2})$};
        \node[font=\footnotesize] at (6.4,-1.6) {$E^3_{a_3\mid x_3}(\textcolor{inflationBeta}{2},\textcolor{inflationGamma}{1})$};
        \node[font=\footnotesize] at (9.6,-1.6) {$E^4_{a_4\mid x_4}(k)$};
    \end{tikzpicture}
    \caption{Quantum inflation of the line-of-four network with two copies of each source.
    Squares denote parties and circles denote independent source copies.
    Colored indices and edges indicate the source copies on which each measurement acts.}
    \label{fig:IntroductionInflationLine}
\end{figure}
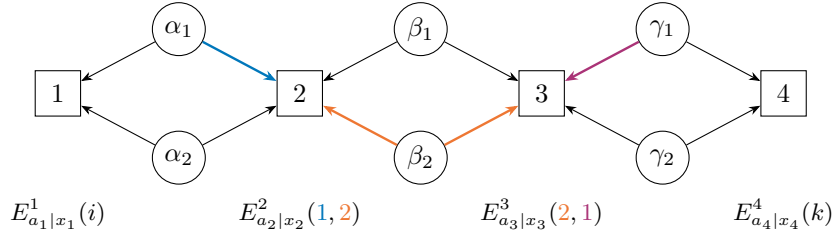

Writing $[m]\coloneqq\{1,\ldots,m\}$, we denote these inflated POVM effects by
\begin{align}\label{eq:IntroductionInflatedLine}
    E^1_{a_1\mid x_1}(i),\qquad E^2_{a_2\mid x_2}(i,j),\qquad E^3_{a_3\mid x_3}(j,k),\qquad E^4_{a_4\mid x_4}(k),\qquad i,j,k\in[m],
\end{align}
where $i,j,k$ label copies of $\alpha,\beta,\gamma$, respectively.
The inflated effects act on the tensor product of all copied source Hilbert spaces.
For example, $E^2_{a_2\mid x_2}(1,2)$ is a copy of $\Meas{2}{a_2}{x_2}$ acting on $\cH_\alpha^{(1)}\otimes\cH_\beta^{(2)}$ and as the identity on the remaining copies.

For each fixed choice of copy indices and measurement setting, the inflated effects form a POVM.
Effects at different parties commute by measurement locality.
At the same party, effects commute whenever they act on disjoint source copies; for example, $[E^2_{a_2\mid x_2}(i,j),E^2_{b_2\mid y_2}(i',j')]=0$ when $i\neq i'$ and $j\neq j'$.

The duplicated source states define a state $\omega_m$ on the algebra generated by the inflated POVM effects.
Let $S_m$ be the permutation group of $[m]$ and $\Gamma_m=S_m\times S_m\times S_m$.
An element $g=(g_\alpha,g_\beta,g_\gamma)\in\Gamma_m$ acts on this algebra by relabeling the copies:
\begin{equation}\label{eq:IntroductionInflationAction}
    \begin{aligned}
        \alpha_g\!\left(E^1_{a_1\mid x_1}(i)\right)&=E^1_{a_1\mid x_1}(g_\alpha(i)),\\
        \alpha_g\!\left(E^2_{a_2\mid x_2}(i,j)\right)&=E^2_{a_2\mid x_2}(g_\alpha(i),g_\beta(j)),\\
        \alpha_g\!\left(E^3_{a_3\mid x_3}(j,k)\right)&=E^3_{a_3\mid x_3}(g_\beta(j),g_\gamma(k)),\\
        \alpha_g\!\left(E^4_{a_4\mid x_4}(k)\right)&=E^4_{a_4\mid x_4}(g_\gamma(k)).
    \end{aligned}
\end{equation}
For each event $e=(\mathbf{a}|\mathbf{x})$ and $t\in[m]$, define the event operator acting on the $t$th copy of every source by
\begin{align}\label{eq:IntroductionInflationDiagonal}
    D_e^{(t)}\coloneqq E^1_{a_1\mid x_1}(t)E^2_{a_2\mid x_2}(t,t)E^3_{a_3\mid x_3}(t,t)E^4_{a_4\mid x_4}(t).
\end{align}
The device-duplication principle implies
\begin{equation}\label{eq:IntroductionInflationConstraints}
    \begin{aligned}
        \omega_m\circ\alpha_g&=\omega_m,\qquad g\in\Gamma_m,\\
        \omega_m\!\left(\prod_{t=1}^{\ell}D_{e_t}^{(t)}\right)&=\prod_{t=1}^{\ell}p(e_t),\qquad 1\leq\ell\leq m,
    \end{aligned}
\end{equation}
for all complete events $e_1,\ldots,e_\ell$.
The first identity holds because permuting identical source copies leaves the state unchanged.
The second follows because the event operators act on disjoint copies of the sources, so independence makes their joint probability factorize.
The general definition of the inflation test is given in \zcref{def:InflationTestPOVM}.

Applying the NPA hierarchy at each fixed inflation level $m$, and then increasing $m$, gives the \emph{inflation-NPA hierarchy}.
The construction above ensures that every mixed-model correlation passes all levels.
To show completeness in \zcref{thm:IntroductionCompleteness}, we prove the converse for all quantum tree networks: every correlation passing all levels admits a mixed-model realization on the given tree (\zcref{cor:InflationNPATreeConvergence}).
Thus, consistency with the device-duplication principle at every level is sufficient to characterize the mixed-model correlations on a tree network.
Note that we use the POVM formulation while the original inflation-NPA is projective; the two agree in the limit (\zcref{rem:InflationPOVMProjective}).

\subsection{Technical outline}\label{sec:IntroductionProofIdeas}
After defining the mixed model in \zcref{sec:SourcePartyGraphMixedModel}, we develop in \zcref{sec:SourceReplacerRadonNikodym} the tools that establish its equivalence with the source-transfer model, which we prove in \zcref{sec:SourceTransferModelMain}.
We then reconstruct source-transfer realizations from the source-transfer and inflation-NPA hierarchies in \zcref{sec:CompleteSDPSourceTransfer,sec:InflationTreeConvergence}, respectively.
This reduces both completeness proofs in \zcref{thm:IntroductionCompleteness} to constructing source states and transfer maps satisfying the source-replacer conditions.

\medskip
\noindent\emph{From transfer maps to mixed-model measurements.}
To prove the model equivalence in \zcref{thm:SourceTransferTreeEquivalence}, we reconstruct mixed-model POVM effects from transfer maps using Arveson's Radon--Nikodym theorem~\cite{arveson1969subalgebras} (see also~\cite[Theorem~III.1]{raginsky2003radon}).
Represent each source algebra on the Gelfand--Naimark--Segal (GNS) space of its state, with weak closures $\mathcal{N}_\alpha$, $\mathcal{N}_\beta$, and $\mathcal{N}_\gamma$.
Under the source-replacer conditions in \zcref{eq:IntroductionSourceReplacer}, \zcref{cor:RadonNikodymSourceProductReplacer} gives POVMs at parties $2$ and $3$ satisfying
\begin{align}\label{eq:IntroductionReconstructedLocality}
    \Meas{2}{a_2}{x_2}\in\mathcal{N}_\alpha'\vnotimes\mathcal{N}_\beta,\qquad \Meas{3}{a_3}{x_3}\in\mathcal{N}_\beta'\vnotimes\mathcal{N}_\gamma,
\end{align}
where the prime denotes the commutant.
Thus, the reconstructed POVMs act on the prescribed sources and satisfy the commutation relations required by the mixed model.
The expectation identities of \zcref{cor:RadonNikodymSourceProductReplacer} recursively recover Born's rule for this model.

\zcref[S]{cor:RadonNikodymSourceProductReplacer} also applies to general trees.
We remark that the localization of the reconstructed POVMs provided by the corollary allows a consistent and simultaneous ``splitting'' of the POVMs on the associated source algebras.
This is the central obstruction to generalizing the bilocal and star-network techniques of~\cite{ligthart2023inflation,ligthart2024semidefinite} to general networks.

\medskip
\noindent\emph{From multi-state SDP solutions to transfer maps.}
A key step in proving completeness of the source-transfer SDP hierarchy in \zcref{sec:MultiStatePolyRing,sec:SourceTransferSDPHierarchy} is reconstructing transfer maps $\instr{v}{a_v}{x_v}$ from the formal symbols $\formalInstr{v}{a_v}{x_v}$.
We express the SDP constraints, including the prescribed correlation, as nonnegativity on a quadratic module $\mathfrak{Q}\subset\thinS$ and prove that $\mathfrak{Q}$ is archimedean (\zcref{sec:QuadraticModuleSDPHierarchy}).
This allows us to extract a subsequence of feasible SDP solutions converging pointwise to a normalized positive functional $L:\thinS\to\mathbb{C}$ nonnegative on $\mathfrak{Q}$.
Applied to $L$, the Kadison--Dubois representation theorem~\cite[Theorem~5.4.4]{marshall2008positive} yields a character $K:\thinS\to\mathbb{C}$, that is, a unital $*$-homomorphism nonnegative on $\mathfrak{Q}$.
The reconstruction in \zcref{lem:SourceTransferCharacterReconstruction} turns $K$ into a source-transfer realization.
For the line-of-four example, the source-state functionals $K\circ\varsigma_\alpha$, $K\circ\varsigma_\beta$, and $K\circ\varsigma_\gamma$ give GNS representations $\pi_\alpha$, $\pi_\beta$, and $\pi_\gamma$ of $\fA_\alpha^1$, $\fA_\beta^2$, and $\fA_\gamma^3$, respectively, leading to source $C^*$-algebras $\cA_\alpha^1$, $\cA_\beta^2$, and $\cA_\gamma^3$.
We define
\begin{align}\label{eq:IntroductionReconstructedTransfer}
    \instr{2}{a_2}{x_2}\!\left(\pi_\alpha(q)\right)\coloneqq\pi_\beta\!\left(\formalInstr{2}{a_2}{x_2}(q)\right),\qquad q\in\fA_\alpha^1.
\end{align}
Due to the source-replacer condition in \zcref{eq:IntroductionFormalCorrelationSourceReplacer}, we prove that this expression is well-defined and extends to a completely positive map $\cA_\alpha^1\to\cA_\beta^2$.
The extension satisfies the source-replacer conditions.
The full reconstruction produces all the data of a source-transfer realization reproducing the prescribed correlation, establishing completeness in \zcref{thm:SourceTransferSDPConvergence}.
We refer to \zcref{sec:SimpleExample} for explicit calculations on concrete examples that may help the reader follow the general proof.

In \zcref{sec:SourceTransferFlatness}, we also establish a sufficient flatness criterion for extracting a finite-dimensional realization from a finite-level solution.
The computational consequence in \zcref{sec:ComputabilitySourceTransferSDP} is $\mathsf{coRE}$-completeness of the mixed-model game-value promise problem on quantum tree networks.
The character reconstruction also yields compactness of the mixed-model correlation set on trees (\zcref{rem:MixedTreeCompactness}).

\medskip
\noindent\emph{From inflation to transfer maps.}
In \zcref{sec:InflationTreeConvergence}, we reconstruct a source-transfer realization for a correlation $p$ compatible with quantum inflation at every level $m$.
Following the de Finetti argument of~\cite{ligthart2023convergent}, we obtain a permutation-invariant state $\omega$ on the algebra generated by infinitely many copies of the inflated measurements, reproducing $p$ and factorizing on disjoint source labels (\zcref{lem:InflationFactorizingState}).
We use $\omega$ to reconstruct the source-transfer model as follows.
For the line-of-four example, work in the GNS representation of $\omega$ on $\cH$, and let $\cA_\alpha^1$ be the $C^*$-algebra generated by $E^1_{a_1\mid x_1}(1)$.
For $q\in\cA_\alpha^1$, write $q^{(i)}$ for its copy with source-$\alpha$ label $i$.
For each fixed source-$\beta$ label $j$, define maps $\cA_\alpha^1\to\Bof{\cH}$ by averaging over the input label $i$:
\begin{align}\label{eq:IntroductionInflationTransferAverage}
    \Phi^{2,j;m}_{a_2\mid x_2}(q)\coloneqq\frac1m\sum_{i=1}^{m}q^{(i)}E^2_{a_2\mid x_2}(i,j).
\end{align}
We show that, along a subsequence as $m\to\infty$, these maps converge pointwise in the weak operator topology to completely positive maps, with the limit at $j=1$ defining $\instr{2}{a_2}{x_2}$.
Summing $\Phi^{2,j;m}_{a_2\mid x_2}(q)$ over $a_2$ gives $m^{-1}\sum_{i=1}^{m}q^{(i)}$, which we show converges strongly to $\omega(q)I$ as $m\to\infty$, recovering the source-replacer condition in \zcref{eq:IntroductionSourceReplacer}.
In \zcref{lem:InflationRecursiveConstruction}, we show that the construction is well-defined with the required commutation and permutation relations, allowing us to iterate it along the tree network.
Finally, we verify that the resulting source-transfer realization reproduces $p$, completing the proof of \zcref{thm:InflationSourceTransferTreeEquivalence}.

\subsection{Discussion and outlook}\label{sec:IntroductionDiscussion}
The source-transfer model we introduce gives an operator-algebraic characterization of mixed-model correlations on every quantum tree network.
Its recursive formulation expresses source independence through states and completely positive maps and allows us to reconstruct local measurements in the tensor products prescribed by the network.
This characterization underlies both the new source-transfer hierarchy, based on multi-state polynomials with completely positive maps, and our completeness proof for inflation-NPA.
In finite dimensions, the mixed, source-transfer, and tensor-product models give the same correlations (\zcref{rem:FiniteDimensionalTransferToMixed}).
The stopping criterion in \zcref{thm:SourceTransferFlatExtraction} connects this equivalence to finite-level solutions of the source-transfer SDP hierarchy: a solution satisfying the stopping criterion yields a finite-dimensional realization, and an optimal solution yields an optimal realization.

The source-transfer hierarchy has two structural advantages: it avoids source duplication and provides a sufficient finite-dimensional extraction stopping criterion.
For inflation-NPA, reconstruction of independent sources relies on a de Finetti theorem as $m\to\infty$, so it remains unclear whether current techniques can yield a sufficient flatness condition at a fixed level $(m,d)$ (\zcref{rem:InflationFlatnessStopping}).

Our results provide strong evidence that the mixed quantum model is a natural target for complete outer hierarchies on general quantum networks.
It keeps the commuting-operator description of measurement locality used in Bell scenarios, while representing source independence by tensor products of source states.
On tree networks, we prove that every correlation passing all levels of either hierarchy admits a realization with both structures.
In particular, inflation-NPA completeness shows that consistency with device duplication at every level suffices to recover such a realization.
The mixed-model game-value promise problem on trees is consequently $\mathsf{coRE}$-complete, retaining the computational classification of its commuting-operator Bell counterpart.

The main open problem is to extend this characterization to general networks containing cycles, such as the triangle network (\zcref{fig:BeyondTreeTriangle}).
Resolving this problem is an important first step towards the rigorous characterization of more general causal models arising from, for example, distributed quantum computing.
The extension in \zcref{sec:BeyondTree} already covers networks whose quantum sources form disjoint trees after removing the classical sources.
Beyond this case, our reconstruction faces a structural difficulty: it uses the independence of source states in disjoint branches of a tree, whereas on a network with cycles those branches can remain connected through the rest of the network.
Hence, a general characterization must account for these connections while preserving source independence.
Completeness of the original inflation-NPA hierarchy for the mixed model on arbitrary networks remains open.

A practical next step is to implement both hierarchies and compare the bounds they give on quantities such as game values for concrete tree networks.
One can also introduce additional constraints at the level of states and completely positive maps for specific applications.
The comparison should include SDP size and the dependence of the source-transfer bounds on the chosen root.
For the stopping criterion, useful improvements would include smaller degree requirements for specific tasks and procedures for extracting states and measurements from numerical solutions.
Symmetries in the network structure and sparsity may also help reduce the computational cost.

Another direction is to incorporate quantum inputs or partial information about the devices, such as trusted input states, specified measurements, or dimension bounds.
These extensions could connect the framework to semiquantum games, steering, and semi-device-independent certification on networks.
Existing SDP hierarchies for operationally non-signaling sequential strategies~\cite{klep2026quantitativequantumsoundnessbipartite,baroni2026quantitativequantumsoundnessmultipartite} also motivate extending our framework to sequential measurements with multiple independent quantum sources.
The recursive use of completely positive maps also suggests an algebraic connection with finitely correlated states~\cite{fannes1992finitely}.
We plan to investigate this connection and the extensions described above in future work.

The operator-algebraic tools may be of independent interest beyond the study of quantum network correlations.
Our Radon--Nikodym lemmas reconstruct measurement operators in prescribed spatial tensor products from completely positive maps satisfying source-replacer conditions.
It would be useful to determine whether analogous conclusions hold for maps dominated by more general reference maps.
The new polynomial optimization framework that we introduce may also be of independent interest to noncommutative real algebraic geometry, extending the noncommutative polynomial optimization framework to incorporate multiple states linked by completely positive maps.

\section{Source-party graphs and mixed quantum models}\label{sec:SourcePartyGraphMixedModel}
In this section, we introduce the mixed quantum model for general quantum networks.
The model is called \emph{mixed} because it uses commuting observable algebras for measurement parties sharing a source, while independent sources are described by tensor-product Hilbert spaces and product states.
We describe the network structure using source-party graphs.
Since our main results concern quantum tree networks, we also fix the rooted-tree notation used in the source-transfer construction in \zcref{sec:SourceTransferModelMain}.

\subsection{Notation for tree graphs}\label{sec:TreeGraph}
Consider a finite undirected graph $\Tree = ( \Vertices{\Tree}, \Edges{\Tree} )$, where $\Vertices{\Tree}$ is its set of vertices and $\Edges{\Tree}$ is its set of edges.
We say that the graph $\Tree$ is a \emph{tree} if it is connected and contains no cycles.
A \emph{rooted tree} is a pair $(\Tree, r)$, where $r \in \Vertices{\Tree}$ is a distinguished vertex called the \emph{root}.

For a rooted tree $(\Tree, r) $, let $v \in \Vertices{\Tree}$ with $v \neq r$.
Since $\Tree$ is a tree, there is a unique simple path $v = v_0 \sim v_1 \sim v_2 \sim \cdots \sim v_l = r$ from $v$ to $r$.
The vertex next to $v$ on the path to $r$ is the \emph{parent} of $v$, denoted by $\Parent{v}: = v_1$.
The \emph{children} of $v$ are the vertices whose parent is $v$, forming the set $\Children{v} = \{ u \in \Vertices{\Tree} \mid \Parent{u} = v \}$.
A vertex $v$ is called a \emph{leaf} if $\Children{v} = \emptyset$, and the root has no parent.

For any vertex $v\in \Vertices{\Tree}$, the \emph{subtree rooted at $v$}, denoted by $\Tree_v$, is the subgraph consisting of $v$ together with all vertices $u$ whose unique path to the root $r$ passes through $v$.

For $v\in\Vertices{\Tree}$, define its \emph{neighborhood} by $\Neighbors{v} \coloneqq \{ u\in\Vertices{\Tree} \mid \{u,v\} \in\Edges{\Tree} \}$.
The \emph{degree} of the vertex $v$ is $\deg_{\Tree}(v)=\lvert\Neighbors{v}\rvert$.

Throughout this manuscript, the words parent, child, and leaf are understood relative to the chosen root $r$.

\subsection{Source-party graphs and source-party trees}\label{sec:SourcePartyGraphs}
Let $\Tree = ( \Vertices{\Tree}, \Edges{\Tree} )$ be a finite graph whose vertex set is partitioned as
\begin{align}
    \Vertices{\Tree} = \PartyVertices{\Tree} \cup \SourceVertices{\Tree}.
\end{align}
The vertices $v$ in $\PartyVertices{\Tree}$ represent measurement parties, while the vertices $\alpha$ in $\SourceVertices{\Tree}$ represent independent quantum sources.
We require every edge in $\Edges{\Tree}$ to connect a party to a source, i.e.,
\begin{align}
    \Edges{\Tree} \subset \left\{ \{v, \alpha\} \mid v \in \PartyVertices{\Tree}, \ \alpha \in \SourceVertices{\Tree} \right\}.
\end{align}
We call $\Tree$ a \emph{source-party graph}, and a \emph{source-party tree} if $\Tree$ is also a tree; see \zcref{fig:InternalMultipartiteSourceTree} for a pictorial illustration.

The neighbors of vertices alternate between parties and sources, i.e., $\Neighbors{v} \subset \SourceVertices{\Tree}$ and $\Neighbors{\alpha} \subset \PartyVertices{\Tree}$.
If a party $r \in \PartyVertices{\Tree}$ is chosen as the root, then the parent and children also alternate: $\Parent{v} \in \SourceVertices{\Tree}$, $\Children{v} \subset \SourceVertices{\Tree}$ for a non-root party $v$, while $\Parent{\alpha} \in \PartyVertices{\Tree}$, $\Children{\alpha} \subset \PartyVertices{\Tree}$ for a source $\alpha$.
(The choice of root will be used later for the source-transfer model in \zcref{sec:SourceTransferModelMain}.)

\subsection{Mixed quantum models}\label{sec:MixedQuantumModelSourcePartyGraph}
Recall that $\vnotimes$ denotes the von Neumann spatial tensor product, namely the weak-operator closure of the algebraic tensor product~\cite[\nopp III.1.5.4]{blackadar2006operator}.
We now introduce an infinite-dimensional generalization of the quantum information description for source-party graphs, in which we model independent sources with tensor products and measurement locality with commutation.
See \zcref{fig:InternalMultipartiteSourceTree} for an illustration.

\begin{definition}[Mixed quantum model on a source-party graph]\label{def:MixedQuantumModelSourcePartyGraph}
    Let $\Tree = ( \PartyVertices{\Tree} \cup \SourceVertices{\Tree}, \Edges{\Tree} )$ be a finite source-party graph.
    For each party $v \in \PartyVertices{\Tree}$, let $\mathsf{X}_v$ and $\mathsf{A}_v$ denote its finite set of measurement settings and outcomes.
    A correlation
    \begin{align}
        p(\mathbf{a} | \mathbf{x}), \qquad \mathbf{a}=(a_v)_{v \in \PartyVertices{\Tree}}, \qquad \mathbf{x}=(x_v)_{v \in \PartyVertices{\Tree}},
    \end{align}
    admits a \emph{mixed quantum realization} $\cR_{\Tree}^{\mathrm{mixed}}$ on $\Tree$ if the following objects exist.
    
    For every independent source $\alpha \in \SourceVertices{\Tree}$, there are a Hilbert space $\cH_\alpha$, a unit state vector $\ket{\Omega_\alpha} \in \cH_\alpha$, and, for every party $v \in \Neighbors{\alpha}$, a von Neumann algebra $\cM_{\alpha}^{v} \subset \Bof{\cH_\alpha}$.
    The algebras $\cM_{\alpha}^{v}$ associated with the same source are mutually commuting:
    \begin{align}
        [\cM_{\alpha}^u, \cM_{\alpha}^v] = 0
    \end{align}
    for any $u \neq v$ and $u, v \in \Neighbors{\alpha}$.

    For every measurement party $v \in \PartyVertices{\Tree}$, define its local observable algebra by
    \begin{align}
        \cM^v \coloneqq \vnbigotimes_{\alpha \in \Neighbors{v}} \cM_{\alpha}^{v} \subset \vnbigotimes_{\alpha \in \Neighbors{v}} \Bof{\cH_\alpha}.
    \end{align}
    For every setting $x_v \in \mathsf{X}_v$, there is a positive operator-valued measure (POVM) $\bigl\{ \Meas{v}{a_v}{x_v} \bigr\}_{a_v\in\mathsf{A}_v} \subset \cM^v$.
    
    Finally, denote the global Hilbert space and global vector state by
    \begin{align}
        \cH \coloneqq \bigotimes_{\alpha \in \SourceVertices{\Tree}} \cH_\alpha, \qquad \ket{\Omega} \coloneqq \bigotimes_{\alpha \in \SourceVertices{\Tree}} \ket{\Omega_\alpha}.
    \end{align}
    The observed correlation satisfies
    \begin{align}\label{eq:MixedModelCorrelation}
        p(\mathbf{a} | \mathbf{x}) = \sandwich{\Omega}{\prod_{v\in\PartyVertices{\Tree}} \Meas{v}{a_v}{x_v}}{\Omega},
    \end{align}
    where we identify each $\cM^v$ with its natural amplification to $\cH$, acting trivially on $\cH_{\alpha}$ for $\alpha \notin \Neighbors{v}$ so that $[\cM^u, \cM^v] = 0$.
    Denote by $C_{\mathrm{mix}}(\Tree)$ the set of correlations $p$ admitting a mixed quantum realization on the graph $\Tree$.
\end{definition}

\begin{figure}[htbp]
    \centering
    \begin{tikzpicture}[
        scale=0.95,
        transform shape,
        every node/.style={font=\small},
        party/.style={
            rectangle,
            draw,
            thick,
            minimum width=8mm,
            minimum height=8mm,
            inner sep=0pt
        },
        source/.style={
            circle,
            draw,
            thick,
            minimum size=9mm,
            inner sep=0pt
        },
        source info/.style={
            font=\scriptsize,
            align=center
        },
        algebra/.style={
            fill=white,
            inner sep=1.2pt,
            font=\scriptsize
        }
    ]

    \node[party]  (p1) at (0,0) {$1$};
    \node[source] (a1) at (2.1,0) {$\alpha_1$};
    \node[party]  (p2) at (4.2,0) {$2$};
    \node[source] (a2) at (6.8,0) {$\alpha_2$};
    \node[party]  (p4) at (9.4,0) {$4$};

    \node[source] (a4) at (4.4,-2.4) {$\alpha_4$};
    \node[party]  (p3) at (6.8,-2.4) {$3$};
    \node[source] (a3) at (9.2,-2.4) {$\alpha_3$};
    \node[party]  (p5) at (11.4,-2.4) {$5$};

    \node[source info, above=2pt of a1]
        {$\bigl(\cH_{\alpha_1},\ket{\Omega_{\alpha_1}}\bigr)$};

    \node[source info, above=2pt of a2]
        {$\bigl(\cH_{\alpha_2},\ket{\Omega_{\alpha_2}}\bigr)$};

    \node[source info, below=2pt of a3]
        {$\bigl(\cH_{\alpha_3},\ket{\Omega_{\alpha_3}}\bigr)$};

    \node[source info, below=2pt of a4]
        {$\bigl(\cH_{\alpha_4},\ket{\Omega_{\alpha_4}}\bigr)$};

    \draw[thick] (p1) --
        node[midway, below=3pt, algebra]
            {$\cM_{\alpha_1}^{1}$}
        (a1);

    \draw[thick] (a1) --
        node[midway, below=3pt, algebra]
            {$\cM_{\alpha_1}^{2}$}
        (p2);

    \draw[thick] (p2) --
        node[midway, below=3pt, algebra]
            {$\cM_{\alpha_2}^{2}$}
        (a2);

    \draw[thick] (a2) --
        node[midway, below=3pt, algebra]
            {$\cM_{\alpha_2}^{4}$}
        (p4);

    \draw[thick] (a2) --
        node[midway, right=4pt, algebra]
            {$\cM_{\alpha_2}^{3}$}
        (p3);

    \draw[thick] (a4) --
        node[midway, above=3pt, algebra]
            {$\cM_{\alpha_4}^{3}$}
        (p3);

    \draw[thick] (p3) --
        node[midway, above=3pt, algebra]
            {$\cM_{\alpha_3}^{3}$}
        (a3);

    \draw[thick] (a3) --
        node[midway, above=3pt, algebra]
            {$\cM_{\alpha_3}^{5}$}
        (p5);

    \end{tikzpicture}

    \caption{
    A source-party tree with parties $1,\ldots,5$ represented by squares and independent sources $\alpha_1,\ldots,\alpha_4$ represented by circles.
    The source $\alpha_2$ is tripartite and is shared by parties $2,3,4$, while $\alpha_1$ and $\alpha_3$ are bipartite.
    The source $\alpha_4$ is a local source connected only to party $3$.
    In the mixed model of \zcref{def:MixedQuantumModelSourcePartyGraph}, each source $\alpha_j$ carries a Hilbert space $\cH_{\alpha_j}$ and a unit vector $\ket{\Omega_{\alpha_j}}$.
    For every adjacent party $i\in\Neighbors{\alpha_j}$, the von Neumann algebra $\cM_{\alpha_j}^{i} \subset \Bof{\cH_{\alpha_j}}$ describes the observables of party $i$ on source $\alpha_j$; for each fixed source, these algebras mutually commute.
    The local observable algebras are $\cM^1=\cM_{\alpha_1}^{1}$, $\cM^2=\cM_{\alpha_1}^{2}\vnotimes\cM_{\alpha_2}^{2}$, $\cM^3=\cM_{\alpha_2}^{3}\vnotimes\cM_{\alpha_3}^{3}\vnotimes\cM_{\alpha_4}^{3}$, $\cM^4=\cM_{\alpha_2}^{4}$, and $\cM^5=\cM_{\alpha_3}^{5}$.
    The global mixed realization acts on $\cH_{\alpha_1}\otimes\cH_{\alpha_2}\otimes\cH_{\alpha_3}\otimes\cH_{\alpha_4}$ with product source vector $\ket{\Omega_{\alpha_1}\otimes\Omega_{\alpha_2}\otimes\Omega_{\alpha_3}\otimes\Omega_{\alpha_4}}$.
    }
    \label{fig:InternalMultipartiteSourceTree}
\end{figure}
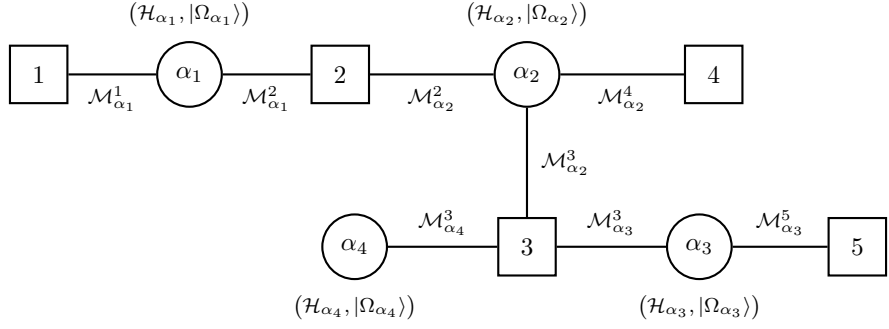

The \emph{quantum causal compatibility problem in the mixed quantum model} asks whether a prescribed correlation $p$ admits a realization as in \zcref{def:MixedQuantumModelSourcePartyGraph}, i.e., is $p \in C_{\mathrm{mix}}(\Tree)$ or $p \notin C_{\mathrm{mix}}(\Tree)$.
As we shall show in \zcref{rem:MixedTreeCompactness}, the set $C_{\mathrm{mix}}(\Tree)$ is compact.

\begin{remark}[Relation to the standard tensor-product model]\label{rem:MixedTensorProductModels}
The mixed model of \zcref{def:MixedQuantumModelSourcePartyGraph} is a direct generalization of the standard tensor-product model in quantum information theory~\cite{nielsen2010quantum}; it differs only in the notion of locality within an individual source space $\Bof{\cH_\alpha}$ for $\alpha \in \SourceVertices{\Tree}$.

Indeed, in the standard tensor-product model, every source $\alpha$ shared by parties $v \in \Neighbors{\alpha}$ is described by a multipartite Hilbert-space decomposition $\cH_\alpha = \bigotimes_{v \in\Neighbors{\alpha}} \cH_{\alpha}^v$.
The associated von Neumann algebras are
\begin{align}
    \cM_{\alpha}^v = \Bof{\cH_{\alpha}^v} \otimes \bigotimes_{u\in\Neighbors{\alpha}\setminus\{v\}} I_{\cH_{\alpha}^u}.
\end{align}
By contrast, the mixed model replaces only this tensor-factor requirement with the weaker commutation condition $[\cM_{\alpha}^u,\cM_{\alpha}^v]=0$ for $u, v \in \Neighbors{\alpha}$.

In particular, the mixed model is a priori more general than the standard tensor-product model in infinite dimensions~\cite{ji2021mip}.
However, in finite dimensions, the commuting-algebra formulation reduces to the tensor-product one~\cite{xu2025quantitativetsirelsonstheoremsapproximate}, and by the same inductive argument in the proof of~\cite[Corollary~9]{ligthart2023inflation}, a finite-dimensional mixed model for trees reduces to the standard quantum-information model for tree networks.
\end{remark}

For the rest of the manuscript, we restrict to the case where $\Tree$ is a source-party tree.
We refer to these as \emph{quantum tree network scenarios}.
When a root is needed, it will always be a party vertex in $\PartyVertices{\Tree}$.

\section{Source-replacer maps and their Radon--Nikodym derivatives}\label{sec:SourceReplacerRadonNikodym}
In this section, we make some technical observations on Arveson's Radon--Nikodym derivatives for replacer channels and source-replacer maps, which are the key tools for the proof of equivalence between the mixed and source-transfer models on quantum tree scenarios.
We refer to~\cite{paulsen2002completely} for background on completely positive and completely bounded maps, and to~\cite{blackadar2006operator,takesaki2002theory} for general $C^*$- and von Neumann algebra theory.

\subsection{Stinespring dilation of replacer channels}\label{sec:StinespringReplacer}
Let $\cX, \cY$ be unital $C^*$-algebras with units $\id_{\cX}, \id_{\cY}$, respectively.
Let $\omega_{\cX}: \cX \to \mathbb{C}$ be a state.
The \emph{replacer channel} (sometimes called \emph{erasure}) with respect to $\omega_{\cX}$ is a completely positive map
\begin{align}
    \cX \to \cY, \quad x \mapsto \omega_{\cX}(x) \id_{\cY}.
\end{align}

Physically, in the Schr\"{o}dinger picture, the replacer channel discards any input and prepares a fixed state; in the Heisenberg algebraic picture, it evaluates the input in a given state and returns a scalar multiple of the identity.

We make the following observation on the Stinespring dilations~\cite[\nopp II.6.9.7]{blackadar2006operator} of replacer channels on the Gelfand--Naimark--Segal (GNS) representations~\cite[\nopp II.6.4]{blackadar2006operator}.

\begin{lemma}[Stinespring dilations of replacer channels]\label{lem:StinespringReplacer}
    Let $\cX, \cY$ be unital $C^*$-algebras with units $\id_{\cX}, \id_{\cY}$.
    Let $\omega_{\cX}$ and $\omega_{\cY}$ be states on $\cX$ and $\cY$, with the GNS representations $\gns{\cX}$ and $\gns{\cY}$, respectively.

    Then $(\cH_{\cX} \otimes \cH_{\cY}, \pi_{\cX} \otimes I_{\cH_{\cY}}, V_{\omega_{\cX}})$ is the minimal Stinespring dilation of the replacer channel represented on $\cH_{\cY}$
    \begin{align}
        x \mapsto \pi_{\cY} (\omega_{\cX}(x) \id_{\cY}) = \omega_{\cX}(x) I_{\cH_{\cY}}.
    \end{align}
    That is, the map $\pi_{\cX} \otimes I_{\cH_{\cY}}$ defined by
    \begin{align}
        x \mapsto \pi_{\cX}(x) \otimes I_{\cH_{\cY}} \in \Bof{\cH_{\cX} \otimes \cH_{\cY}} = \Bof{\cH_{\cX}} \vnotimes \Bof{\cH_{\cY}}
    \end{align}
    is a unital $*$-representation of $\cX$ on $\Bof{\cH_{\cX} \otimes \cH_{\cY}}$, the map
    \begin{align}
        V_{\omega_{\cX}}: \cH_{\cY} \to \cH_{\cX} \otimes \cH_{\cY}, \quad \ket{y} \mapsto \ket{\Omega_{\cX}} \otimes \ket{y}
    \end{align}
    is an isometry satisfying
    \begin{align}
        V_{\omega_{\cX}}^* (\pi_{\cX}(x) \otimes I_{\cH_{\cY}}) V_{\omega_{\cX}} = \omega_{\cX}(x) I_{\cH_{\cY}},
    \end{align}
    and $(\pi_{\cX} \otimes I_{\cH_{\cY}})(\cX)V_{\omega_{\cX}}\cH_{\cY}$ is dense in $\cH_{\cX} \otimes \cH_{\cY}$.
    Thus the dilation is minimal and is unique up to unitary equivalence.
\end{lemma}
\begin{proof}
    The replacer channel represented on $\cH_{\cY}$, $x \mapsto \omega_{\cX}(x) I_{\cH_{\cY}}$, is unital and completely positive since $\omega_{\cX}$ is a state and $\lambda \mapsto \lambda I_{\cH_{\cY}}$ is a unital $*$-representation of $\mathbb{C}$.
    It also follows immediately from the fact that $\pi_{\cX}$ is a unital $*$-representation that $\pi_{\cX} \otimes I_{\cH_{\cY}}$ is also a unital $*$-representation of $\cX$ on $\cH_{\cX} \otimes \cH_{\cY}$.
    
    Let $V_{\omega_{\cX}}: \cH_{\cY} \to \cH_{\cX} \otimes \cH_{\cY}$, $\ket{y} \mapsto \ket{\Omega_{\cX}} \otimes \ket{y}$.
    Since for every $\ket{y} \in \cH_{\cY}$ one has
    \begin{align}
        \norm{V_{\omega_{\cX}} \ket{y}}^2 = \norm{\ket{\Omega_{\cX}}}^2 \norm{\ket{y}}^2 = \norm{\ket{y}}^2,
    \end{align}
    $V_{\omega_{\cX}}$ is an isometry.
    Furthermore, for every $x \in \cX$ and $\ket{y_1}, \ket{y_2} \in \cH_{\cY}$, we calculate
    \begin{equation}
        \begin{aligned}
            \sandwich{y_1}{V_{\omega_{\cX}}^* (\pi_{\cX}(x) \otimes I_{\cH_{\cY}}) V_{\omega_{\cX}}}{y_2} &= (\bra{\Omega_\cX} \otimes \bra{y_1}) (\pi_{\cX}(x) \otimes I_{\cH_{\cY}}) (\ket{\Omega_\cX} \otimes \ket{y_2}) \\
            &= \sandwich{\Omega_{\cX}}{\pi_{\cX}(x)}{\Omega_{\cX}} \cdot \braket{y_1}{y_2} 
            = \omega_{\cX}(x) \cdot \braket{y_1}{y_2} = \sandwich{y_1}{ \omega_{\cX}(x) I_{\cH_{\cY}} }{y_2}.
        \end{aligned}
    \end{equation}
    Therefore,
    $V_{\omega_{\cX}}^* (\pi_{\cX}(x) \otimes I_{\cH_{\cY}}) V_{\omega_{\cX}} = \omega_{\cX}(x) I_{\cH_{\cY}}$.

    The minimality is due to
    \begin{align}
        (\pi_{\cX} \otimes I_{\cH_{\cY}})(x) V_{\omega_{\cX}} \ket{y} = \pi_{\cX}(x)\ket{\Omega_{\cX}} \otimes \ket{y}
    \end{align}
    for every $x \in \cX$ and $\ket{y} \in \cH_{\cY}$ and the cyclicity of the GNS triple $\gns{\cX}$.
    The unitary equivalence then follows from the standard Stinespring dilation theorem for completely positive maps; see~\cite[\nopp II.6.9.7]{blackadar2006operator} and also~\cite[Chapter~4]{paulsen2002completely}.
\end{proof}

\subsection{Arveson's Radon--Nikodym derivatives for source-replacer maps}\label{sec:RadonNikodymReplacer}
Let $\cZ \subset \cY$ be a unital $C^*$-subalgebra sharing the common unit $\id_{\cZ}=\id_{\cY}$.
Consider completely positive maps $\Phi_i: \cX \to \cZ$ indexed by $i$ in a finite index set $I$.
We say that $\{\Phi_i\}_{I}$ satisfies the \emph{source-replacer condition} with respect to the state $\omega_{\cX}: \cX \to \mathbb{C}$ if $\sum_{i \in I} \Phi_i(x) = \omega_{\cX}(x) \id_{\cZ}$ for all $x \in \cX$.

Physically, each $\Phi_i$ transfers algebraic information from the input source algebra $\cX$ to the specified subalgebra $\cZ$ of the output source algebra $\cY$, conditioned on outcome $i$.
The source-replacer condition states that, upon marginalizing/forgetting the outcome $i$, the input is erased and replaced by the source state $\omega_{\cX}$.
These $\Phi_i$ are the objects with which we shall reconstruct the measurement POVMs via Arveson's Radon--Nikodym derivatives~\cite{arveson1969subalgebras} (see also~\cite[Theorem~III.1]{raginsky2003radon}) in the source-transfer model.
The introduction of $\cZ \subset \cY$ shall correspond to $\cM^v_{\alpha} \subset \Bof{\cH_{\alpha}}$, which localizes the derivatives to the correct observable algebras of a multipartite quantum source.

\begin{lemma}[Arveson's Radon--Nikodym derivatives for source-replacer maps]\label{lem:RadonNikodymSourceReplacer}
    Let $\cX, \cY$ be unital $C^*$-algebras with units $\id_{\cX}, \id_{\cY}$, and let $\cZ \subset \cY$ be a unital $C^*$-subalgebra with $\id_{\cZ} = \id_{\cY}$.
    Let $\omega_{\cX}$ and $\omega_{\cY}$ be states on $\cX$ and $\cY$, with the GNS representations $\gns{\cX}$ and $\gns{\cY}$.
    Define the von Neumann algebras $\cM_{\cX} \coloneqq \pi_{\cX}(\cX)'' \subset \Bof{\cH_{\cX}}$, $\cM_{\cY} \coloneqq \pi_{\cY}(\cY)'' \subset \Bof{\cH_{\cY}}$, and $\cM_{\cZ} \coloneqq \pi_{\cY}(\cZ)'' \subset \cM_{\cY}$ completed by double commutant, respectively.
    
    Suppose that $\Phi_i: \cX \to \cZ$ are completely positive maps over a finite index set $I$ satisfying the source-replacer condition with respect to the state $\omega_{\cX}$:
    \begin{align}
        \sum_{i \in I} \Phi_i(x) = \omega_{\cX}(x) \id_{\cZ}, \quad \text{ for all } x \in \cX.
    \end{align}

    Then, there exists an isometry $V: \cH_{\cY} \to \cH_{\cX} \otimes \cH_{\cY}$ and a unique positive operator
    \begin{align}
        F_i \in \cM_{\cX}' \vnotimes \cM_{\cZ} \subset \cM_{\cX}' \vnotimes \cM_{\cY} \subset \Bof{\cH_{\cX}} \vnotimes \Bof{\cH_{\cY}}
    \end{align}
    for every $i \in I$, such that
    \begin{equation}\label{eq:RadonNikodymIdentity}
        \begin{aligned}
            \pi_{\cY} \circ \Phi_i(x) &= V^* F_i (\pi_{\cX}(x) \otimes I_{\cH_{\cY}}) V \\
            &= V^*  (\pi_{\cX}(x) \otimes I_{\cH_{\cY}}) F_i V \\
            &= V^* F_i^{1/2} (\pi_{\cX}(x) \otimes I_{\cH_{\cY}}) F_i^{1/2} V
        \end{aligned}
    \end{equation}
    for all $x \in \cX$ and
    \begin{align}\label{eq:POVMCompleteness}
        \sum_{i \in I} F_i = I_{\cH_{\cX} \otimes \cH_{\cY}} = I_{\cH_{\cX}} \otimes I_{\cH_{\cY}}.
    \end{align}

    Moreover, for every $x\in\cX$ and every $Z\in\cM_{\cZ}'$, a useful expectation-value form of \zcref{eq:RadonNikodymIdentity} is
    \begin{align}\label{eq:RadonNikodymExpectationIdentity}
        \sandwich{\Omega_{\cY}}{ Z\,\pi_{\cY}\!\left(\Phi_i(x)\right) }{\Omega_{\cY}} = \sandwich{\Omega_{\cX}\otimes\Omega_{\cY}}{ F_i \left( \pi_{\cX}(x)\otimes Z \right) }{\Omega_{\cX}\otimes\Omega_{\cY}}.
    \end{align}
\end{lemma}
That is, $\{F_i\}_{I}$ form a POVM and each $F_i$ is the Arveson's Radon--Nikodym derivative of $\Phi_i$.
\begin{proof}
    Applying \zcref{lem:StinespringReplacer} to $\pi_{\cY} \circ (\sum_{i \in I} \Phi_i) = \omega_{\cX} I_{\cH_{\cY}}$, where we have used $\id_{\cZ}=\id_{\cY}$, we obtain its minimal Stinespring dilation $(\cH_{\cX} \otimes \cH_{\cY}, \pi_{\cX} \otimes I_{\cH_{\cY}}, V_{\omega_{\cX}})$.
    The desired isometry is then $V\coloneqq V_{\omega_{\cX}}: \ket{y} \mapsto \ket{\Omega_{\cX}} \otimes \ket{y}$ such that $ V_{\omega_{\cX}}^* (\pi_{\cX}(x) \otimes I_{\cH_{\cY}}) V_{\omega_{\cX}} = \omega_{\cX}(x) I_{\cH_{\cY}}$.

    Then, since every $\pi_{\cY} \circ \Phi_i \leq \pi_{\cY} \circ (\sum_{i \in I} \Phi_i)$ in the sense of complete positivity, we can apply Arveson's Radon--Nikodym derivative theorem (\cite[Theorem~III.1]{raginsky2003radon}). 
    It follows that there exists a unique operator $F_i \in (\pi_{\cX} \otimes I_{\cH_{\cY}})(\cX)' = (\pi_{\cX}(\cX) \otimes I_{\cH_{\cY}})'$ satisfying $0 \leq F_i \leq I_{\cH_{\cX}} \otimes I_{\cH_{\cY}}$ and \zcref{eq:RadonNikodymIdentity}.
    It remains to check that $F_i \in \cM_{\cX}' \vnotimes \cM_{\cZ}$, the expectation-value identity \zcref{eq:RadonNikodymExpectationIdentity}, and the POVM completeness \zcref{eq:POVMCompleteness}.

    To see this, first note that
    \begin{align}
        F_i \in (\pi_{\cX}(\cX) \otimes I_{\cH_{\cY}})' = ((\pi_{\cX}(\cX) \otimes I_{\cH_{\cY}})'')' = (\pi_{\cX}(\cX)'' \vnotimes \mathbb{C}I_{\cH_{\cY}} )' = \cM_{\cX}' \vnotimes \Bof{\cH_{\cY}},
    \end{align}
    where the first equality is due to the fact that commutants are invariant under weak closure (by the double commutant), the second equality is due to the definition of von Neumann spatial tensor product~\cite[\nopp III.1.5.4]{blackadar2006operator}, and the last equality is due to the commutation theorem for tensor products~\cite[\nopp III.4.5.8]{blackadar2006operator}.
    
    We now further show that $F_i \in \cM_{\cX}' \vnotimes \cM_{\cZ} \subset \cM_{\cX}' \vnotimes \Bof{\cH_{\cY}}$.
    Take any $Z \in \cM_{\cZ}'$ and note that
    \begin{align}\label{eq:IsometryCommutantCommutationIdentity}
        (I_{\cH_{\cX}} \otimes Z) V_{\omega_{\cX}} = V_{\omega_{\cX}} Z, \quad V_{\omega_{\cX}}^*(I_{\cH_{\cX}} \otimes Z)  = Z V_{\omega_{\cX}}^*.
    \end{align}
    Indeed,
    \begin{align}
        (I_{\cH_{\cX}} \otimes Z) V_{\omega_{\cX}} \ket{y} = (I_{\cH_{\cX}} \otimes Z) (\ket{\Omega_{\cX}} \otimes \ket{y}) = \ket{\Omega_{\cX}} \otimes Z\ket{y} = V_{\omega_{\cX}} Z \ket{y}
    \end{align}
    for all $\ket{y} \in \cH_{\cY}$; an analogous calculation proves the $V_{\omega_{\cX}}^*$ identity.
    Then, for every $x_1, x_2 \in \cX$ and $\ket{y_1}, \ket{y_2} \in \cH_{\cY}$, we have 
    \begin{equation}
        \begin{aligned}
            &\sandwich{y_1}{V_{\omega_{\cX}}^* (\pi_{\cX}(x_1^*) \otimes I_{\cH_{\cY}}) [I_{\cH_{\cX}} \otimes Z, F_i] (\pi_{\cX}(x_2) \otimes I_{\cH_{\cY}}) V_{\omega_{\cX}}}{y_2} \\
            &\quad = \sandwich{y_1}{ \underbrace{V_{\omega_{\cX}}^* (I_{\cH_{\cX}} \otimes Z)}_{Z V_{\omega_{\cX}}^*} F_i (\pi_{\cX}(x_1^* x_2) \otimes I_{\cH_{\cY}}) V_{\omega_{\cX}}}{y_2} 
            \\&\qquad\qquad\qquad\qquad\qquad\qquad\qquad\quad - \sandwich{y_1}{V_{\omega_{\cX}}^* (\pi_{\cX}(x_1^* x_2) \otimes I_{\cH_{\cY}}) F_i  \underbrace{(I_{\cH_{\cX}} \otimes Z) V_{\omega_{\cX}}}_{V_{\omega_{\cX}} Z} }{y_2} \\
            &\quad = \sandwich{y_1}{ Z \underbrace{V_{\omega_{\cX}}^* F_i (\pi_{\cX}(x_1^* x_2) \otimes I_{\cH_{\cY}}) V_{\omega_{\cX}}}_{\pi_{\cY} \circ \Phi_i(x_1^* x_2)} }{y_2} - \sandwich{y_1}{ \underbrace{V_{\omega_{\cX}}^* (\pi_{\cX}(x_1^* x_2) \otimes I_{\cH_{\cY}}) F_i  V_{\omega_{\cX}}}_{\pi_{\cY} \circ \Phi_i(x_1^* x_2)} Z}{y_2} \\
            &\quad = \sandwich{y_1}{ [Z, \pi_{\cY} \circ \Phi_i(x_1^* x_2)] }{y_2} = 0,
        \end{aligned}
    \end{equation}
    since $F_i \in \cM_{\cX}' \vnotimes \Bof{\cH_{\cY}}$, and $\pi_{\cY} \circ \Phi_i(x_1^* x_2) \in \pi_{\cY}(\cZ) \subset \cM_{\cZ}$, hence commutes with $Z$.
    By Stinespring minimality of $(\cH_{\cX} \otimes \cH_{\cY}, \pi_{\cX} \otimes I_{\cH_{\cY}}, V_{\omega_{\cX}})$, i.e., density of $(\pi_{\cX} \otimes I_{\cH_{\cY}})(\cX)V_{\omega_{\cX}}\cH_{\cY}$, it follows that $[I_{\cH_{\cX}} \otimes Z, F_i] = 0$.
    So $F_i \in (I_{\cH_{\cX}} \otimes \cM_{\cZ}')'$.
    Thus, by~\cite[\nopp III.4.5.9]{blackadar2006operator} on intersections of spatial tensor products, we have
    \begin{align}
        F_i \in (\cM_{\cX}' \vnotimes \Bof{\cH_{\cY}}) \cap (I_{\cH_{\cX}} \otimes \cM_{\cZ}')' = (\cM_{\cX}' \vnotimes \Bof{\cH_{\cY}}) \cap (\Bof{\cH_{\cX}} \vnotimes \cM_{\cZ}'') = \cM_{\cX}' \vnotimes \cM_{\cZ}.
    \end{align}

    We now prove the expectation-value identity \zcref{eq:RadonNikodymExpectationIdentity}.
    Let $x\in\cX$ and $Z\in\cM_{\cZ}'$.
    Using \zcref{eq:RadonNikodymIdentity,eq:IsometryCommutantCommutationIdentity}, $V_{\omega_{\cX}} \ket{\Omega_{\cY}} = \ket{\Omega_{\cX} \otimes \Omega_{\cY}}$, and $[F_i, I_{\cH_{\cX}} \otimes Z] = 0$, we obtain
    \begin{equation}
        \begin{aligned}
            &\sandwich{\Omega_{\cY}}{ Z\,\pi_{\cY}\!\left(\Phi_i(x)\right) }{\Omega_{\cY}} \\
            &\quad= \sandwich{\Omega_{\cY}}{ ZV^* F_i \left(\pi_{\cX}(x)\otimes I_{\cH_{\cY}}\right)V}{\Omega_{\cY}} \\
            &\quad= \sandwich{\Omega_{\cX}\otimes\Omega_{\cY}}{ \left( I_{\cH_{\cX}}\otimes Z \right) F_i \left( \pi_{\cX}(x)\otimes I_{\cH_{\cY}} \right) }{\Omega_{\cX}\otimes\Omega_{\cY}} \\
            &\quad= \sandwich{\Omega_{\cX}\otimes\Omega_{\cY}}{ F_i \left( \pi_{\cX}(x)\otimes Z \right) }{\Omega_{\cX}\otimes\Omega_{\cY}}.
        \end{aligned}
    \end{equation}
    This proves \zcref{eq:RadonNikodymExpectationIdentity}.

    Finally, for every $x_1, x_2 \in \cX$ and $\ket{y_1}, \ket{y_2} \in \cH_{\cY}$,
    \begin{equation}
        \begin{aligned}
            &\sandwich{y_1}{V_{\omega_{\cX}}^* (\pi_{\cX}(x_1^*) \otimes I_{\cH_{\cY}}) (\sum_{i \in I} F_i - I_{\cH_{\cX} \otimes \cH_{\cY}}) (\pi_{\cX}(x_2) \otimes I_{\cH_{\cY}}) V_{\omega_{\cX}}}{y_2} \\
            &\quad = \sum_{i \in I}  \sandwich{y_1}{ \underbrace{V_{\omega_{\cX}}^*  F_i (\pi_{\cX}(x_1^* x_2) \otimes I_{\cH_{\cY}}) V_{\omega_{\cX}}}_{\pi_{\cY} \circ \Phi_i(x_1^* x_2)} }{y_2} - \sandwich{y_1}{ \underbrace{V_{\omega_{\cX}}^* (\pi_{\cX}(x_1^* x_2) \otimes I_{\cH_{\cY}}) V_{\omega_{\cX}}}_{\omega_{\cX}(x_1^* x_2) I_{\cH_{\cY}}} }{y_2} \\
            &\quad = \sandwich{y_1}{\pi_{\cY}(\sum_{i \in I} \Phi_i(x_1^* x_2) - \omega_{\cX}(x_1^* x_2)\id_{\cY})}{y_2} = 0,
        \end{aligned}
    \end{equation}
    due to the source-replacer condition with respect to the state $\omega_{\cX}$.
    \zcref[S]{eq:POVMCompleteness} then follows from the minimality/density again.
\end{proof}

\begin{remark}[Finite-dimensional Choi--Jamio\l kowski form]
\label{rem:RadonNikodymChoi}
In finite dimensions, \zcref{lem:RadonNikodymSourceReplacer} has a direct interpretation through the Choi--Jamio\l kowski isomorphism~\cite{choi1975completely}.
Let $\cX=M_n(\mathbb{C})$, let $\cZ\subset M_m(\mathbb{C})$ be a unital $C^*$-subalgebra, and let $\omega_{\cX}(x)=\Tr{\rho x}$ for a density operator $\rho \in M_n(\mathbb{C})$.
Fixing a basis $\{\ket{j}\}_{j=1}^n$, define the Choi matrix of $\Phi_i$ by
\begin{align}
    J(\Phi_i) \coloneqq \sum_{j,k=1}^n \ketbra{j}{k}\otimes \Phi_i\!\left(\ketbra{j}{k}\right) \in M_n(\mathbb{C})\otimes\cZ.
\end{align}
The source-replacer condition becomes
\begin{align}\label{eq:ChoiSourceReplacer}
    \sum_{i\in I}J(\Phi_i) = \rho^{\mathsf{T}}\otimes\id_{\cZ}.
\end{align}
If $\rho$ is faithful, define
\begin{align}\label{eq:NormalizedChoiRadonNikodym}
    E_i \coloneqq \left( (\rho^{\mathsf{T}})^{-1/2}\otimes\id_{\cZ} \right) J(\Phi_i) \left( (\rho^{\mathsf{T}})^{-1/2}\otimes\id_{\cZ} \right).
\end{align}
Then $E_i \succeq 0$ and $\sum_{i\in I}E_i=I_n\otimes\id_{\cZ}$ by construction.
(If $\rho$ is not faithful, the same statement holds after restricting to the support of $\rho^{\mathsf{T}}$.)
These operators $E_i$ can be identified with the Radon--Nikodym derivatives $F_i$ of $\Phi_i$, and thus, \zcref{lem:RadonNikodymSourceReplacer} can be seen as a basis-free and dimension-independent version of this Choi--Jamio\l kowski description.
\end{remark}

\subsection{Source-replacer maps for product states}\label{sec:ReplacerChannelProductState}
We make the following observation on replacer completely positive maps for product states on the maximal tensor product of $C^*$-algebras, which shall be the key technical result for recovering the mixed quantum model on trees from the source-transfer model.
The proved appropriate localization of the Radon--Nikodym derivatives is precisely what allows us to go beyond bilocal and star networks to arbitrary trees.

\begin{corollary}[Arveson's Radon--Nikodym derivatives for product-source-replacer maps]\label{cor:RadonNikodymSourceProductReplacer}
    Let $\cX_1, \dots, \cX_k$ and $\cY$ be unital $C^*$-algebras, and let $\cZ \subset \cY$ be a unital $C^*$-subalgebra with $\id_{\cZ} = \id_{\cY}$.
    Let $\omega_{\cX_j}: \cX_j \to \mathbb{C}$ for $j = 1, \dots, k$ and $\omega_{\cY}: \cY \to \mathbb{C}$ be states, with GNS representations $\gns{\cX_{j}}$ for $j = 1, \dots, k$ and $\gns{\cY}$, respectively.
    Define the von Neumann algebras $\cM_{\cX_{j}} \coloneqq \pi_{\cX_{j}}(\cX_{j})'' \subset \Bof{\cH_{\cX_{j}}}$ for $j = 1, \dots, k$, $\cM_{\cY} \coloneqq \pi_{\cY}(\cY)'' \subset \Bof{\cH_{\cY}}$, and $\cM_{\cZ} \coloneqq \pi_{\cY}(\cZ)'' \subset \cM_{\cY}$ completed by double commutant, respectively.

    Suppose that $\Phi_i: \cX_1 \maxotimes \cdots \maxotimes \cX_k \to \cZ$ are completely positive maps over a finite index set $I$ satisfying the source-replacer condition with respect to the product state of $\omega_{\cX_j}$:
    \begin{align}
        \sum_{i \in I} \Phi_i(x) = (\omega_{\cX_1} \otimes \cdots \otimes \omega_{\cX_k})(x) \id_{\cZ}, \quad \text{ for all } x \in \cX_1 \maxotimes \cdots \maxotimes \cX_k.
    \end{align}

    Then, there exists an isometry $V: \cH_{\cY} \to \cH_{\cX_1} \otimes \cdots \otimes \cH_{\cX_k} \otimes \cH_{\cY}$ and a unique operator
    \begin{equation}
        \begin{aligned}
            F_i \in \cM_{\cX_1}' \vnotimes \cdots \vnotimes \cM_{\cX_k}' \vnotimes \cM_{\cZ} &\subset \cM_{\cX_1}' \vnotimes \cdots \vnotimes \cM_{\cX_k}' \vnotimes \cM_{\cY} \\
            &\subset \Bof{\cH_{\cX_1}} \vnotimes \cdots \vnotimes \Bof{\cH_{\cX_k}} \vnotimes \Bof{\cH_{\cY}}
        \end{aligned}
    \end{equation}
    for every $i \in I$, such that
    \begin{equation}\label{eq:ProductRadonNikodymIdentity}
        \begin{aligned}
            \pi_{\cY} \circ \Phi_i(x) &= V^* F_i \bigl((\pi_{\cX_1} \otimes \cdots \otimes \pi_{\cX_k})(x) \otimes I_{\cH_{\cY}}\bigr) V \\
            &= V^* \bigl((\pi_{\cX_1} \otimes \cdots \otimes \pi_{\cX_k})(x) \otimes I_{\cH_{\cY}}\bigr) F_i V \\
            &= V^* F_i^{1/2} \bigl((\pi_{\cX_1} \otimes \cdots \otimes \pi_{\cX_k})(x) \otimes I_{\cH_{\cY}}\bigr) F_i^{1/2} V
        \end{aligned}
    \end{equation}
    for all $x \in \cX_1 \maxotimes \cdots \maxotimes \cX_k$ and
    \begin{align}\label{eq:ProductPOVMCompleteness}
        \sum_{i \in I} F_i = I_{\cH_{\cX_1} \otimes \cdots \otimes \cH_{\cX_k} \otimes \cH_{\cY}} = I_{\cH_{\cX_1}} \otimes \cdots \otimes I_{\cH_{\cX_k}} \otimes I_{\cH_{\cY}}.
    \end{align}

    Moreover, for every $x \in \cX_1 \maxotimes \cdots \maxotimes \cX_k$ and every $Z \in \cM_{\cZ}'$, a useful expectation-value form of \zcref{eq:ProductRadonNikodymIdentity} is
    \begin{align}\label{eq:ProductRadonNikodymExpectationIdentity}
        \sandwich{\Omega_{\cY}}{ Z\,\pi_{\cY}\!\left(\Phi_i(x)\right)}{\Omega_{\cY}} = \sandwich{ \bigotimes_{j = 1}^k \Omega_{\cX_j} \otimes \Omega_{\cY}}{ F_i \left( \left( \pi_{\cX_1} \otimes\cdots\otimes \pi_{\cX_k} \right)(x) \otimes Z \right) }{\bigotimes_{j = 1}^k \Omega_{\cX_j} \otimes \Omega_{\cY}}.
    \end{align}
\end{corollary}
\begin{proof}
    Denote the maximal tensor product $C^*$-algebra and the product state by
    \begin{align}
        \cX \coloneqq \cX_1 \maxotimes \cdots \maxotimes \cX_k, \quad \omega_{\cX} = \omega_{\cX_1} \otimes \cdots \otimes \omega_{\cX_k},
    \end{align}
    respectively.
    Then we claim that 
    \begin{align}
        \cH_{\cX} \coloneqq \cH_{\cX_1} \otimes \cdots \otimes \cH_{\cX_k}, \quad \pi_{\cX} \coloneqq \pi_{\cX_1} \otimes \cdots \otimes \pi_{\cX_k}, \quad \ket{\Omega_{\cX}} \coloneqq \ket{\Omega_{\cX_1}} \otimes \cdots \otimes \ket{\Omega_{\cX_k}}
    \end{align}
    is a GNS representation of $(\cX, \omega_{\cX})$.
    Indeed, the $*$-representations $I\otimes\cdots\otimes\pi_{\cX_j}(x_j)\otimes\cdots\otimes I$ of $\cX_j$ have commuting ranges for $j = 1, \dots, k$.
    Hence, by the universal property of the maximal tensor product, they induce a $*$-representation $\pi_{\cX}$ of $\cX$ satisfying $\pi_{\cX}(x_1\otimes\cdots\otimes x_k) = \pi_{\cX_1}(x_1)\otimes\cdots\otimes\pi_{\cX_k}(x_k)$.
    Also,
    \begin{align}
        \omega_{\cX}(x_1 \otimes \cdots \otimes x_k) = \prod_{j=1}^k \omega_{\cX_j}(x_j) = \prod_{j=1}^k \sandwich{\Omega_{\cX_j}}{\pi_{\cX_j}(x_j)}{\Omega_{\cX_j}} =  \sandwich{\Omega_{\cX}}{\pi_{\cX}(x_1 \otimes \cdots \otimes x_k)}{\Omega_{\cX}},
    \end{align}
    thus agreeing on $\cX$ since $\cX_1\odot\cdots\odot\cX_k$ is norm dense.
    Cyclicity of $\ket{\Omega_{\cX}}$ follows since each $\ket{\Omega_{\cX_j}}$ is cyclic in $\cH_{\cX_j}$, proving the claim.

    Furthermore, the norm density of $\cX_1\odot\cdots\odot\cX_k$ in $\cX$ gives
    \begin{align}
        \cM_{\cX} \coloneqq \pi_{\cX}(\cX)'' = \left(\pi_{\cX_1}(\cX_1) \odot\cdots\odot \pi_{\cX_k}(\cX_k) \right)'' = \cM_{\cX_1} \vnotimes\cdots\vnotimes \cM_{\cX_k}
    \end{align}
    by the definition of spatial von Neumann tensor product \cite[\nopp III.1.5.4]{blackadar2006operator} and the commutation theorem for tensor products~\cite[\nopp III.4.5.8]{blackadar2006operator}.
    Note that $\cM_{\cX}' = \cM_{\cX_1}' \vnotimes \cdots \vnotimes \cM_{\cX_k}'$ by~\cite[\nopp III.4.5.8]{blackadar2006operator} as well.

    Finally, applying \zcref{lem:RadonNikodymSourceReplacer} to $\Phi_i: \cX_1 \maxotimes \cdots \maxotimes \cX_k = \cX \to \cZ$ gives \zcref{eq:ProductRadonNikodymIdentity,eq:ProductPOVMCompleteness}.
    Under the product-GNS identification above, \zcref{eq:ProductRadonNikodymExpectationIdentity} is exactly \zcref{eq:RadonNikodymExpectationIdentity}.
\end{proof}

We finish the section with two remarks on \zcref{cor:RadonNikodymSourceProductReplacer} that will be used later.

\begin{remark}[Partial sandwich form of the Radon--Nikodym identity]\label{rem:RadonNikodymSandwichForm}
    We first record a useful reformulation of the Radon--Nikodym identity \zcref{eq:ProductRadonNikodymIdentity}.
    This form gives the main intuition behind the source-transfer constructions in \zcref{sec:SourceTransferModelMain}, and will also be useful in the proof of equivalence to the mixed model in \zcref{thm:SourceTransferTreeEquivalence}.
    
    Recall that the Stinespring isometry $V$ in \zcref{cor:RadonNikodymSourceProductReplacer} is given by $\cH_{\cY} \ni \ket{y} \mapsto \ket{\Omega_{\cX_1}} \otimes \cdots \otimes \ket{\Omega_{\cX_k}} \otimes \ket{y}$.
    Thus, we may identify
    \begin{align}
        V = \ket{\Omega_{\cX_1}} \otimes \cdots \otimes \ket{\Omega_{\cX_k}} \otimes I_{\cH_{\cY}} = \ket{\Omega_{\cX_1} \otimes \cdots \otimes \Omega_{\cX_k} } \otimes I_{\cH_{\cY}}.
    \end{align}
    Hence, \zcref{eq:ProductRadonNikodymIdentity} can equivalently be written as
    \begin{equation}\label{eq:ProductRadonNikodymPartialSandwich}
        \begin{aligned}
            \pi_{\cY}\!\left(\Phi_i(x)\right) =& \left( \bra{\Omega_{\cX_1} \otimes \cdots \otimes \Omega_{\cX_k}} \otimes I_{\cH_{\cY}} \right) F_i \left( (\pi_{\cX_1} \otimes \cdots \otimes \pi_{\cX_k})(x)\otimes I_{\cH_{\cY}} \right) \left( \ket{\Omega_{\cX_1} \otimes \cdots \otimes \Omega_{\cX_k}} \otimes I_{\cH_{\cY}} \right) \\
            \coloneqq& \ \sandwich{ \Omega_{\cX_1} \otimes \cdots \otimes \Omega_{\cX_k} }{ F_i \left( (\pi_{\cX_1} \otimes \cdots \otimes \pi_{\cX_k})(x)\otimes I_{\cH_{\cY}} \right) }{ \Omega_{\cX_1} \otimes \cdots \otimes \Omega_{\cX_k} } \\
            =& \sandwich{ \Omega_{\cX_1} \otimes \cdots \otimes \Omega_{\cX_k} }{ \left( (\pi_{\cX_1} \otimes \cdots \otimes \pi_{\cX_k})(x)\otimes I_{\cH_{\cY}} \right) F_i }{ \Omega_{\cX_1} \otimes \cdots \otimes \Omega_{\cX_k} } \\
            =& \sandwich{ \Omega_{\cX_1} \otimes \cdots \otimes \Omega_{\cX_k} }{ (F_i)^{1/2} \left( (\pi_{\cX_1} \otimes \cdots \otimes \pi_{\cX_k})(x)\otimes I_{\cH_{\cY}} \right) (F_i)^{1/2} }{ \Omega_{\cX_1} \otimes \cdots \otimes \Omega_{\cX_k} }.
        \end{aligned}
    \end{equation}
    That is, the sandwich takes the expectation value of the source state $\ket{\Omega_{\cX_1}} \otimes \cdots \otimes \ket{\Omega_{\cX_k}}$ over $\cH_{\cX_1} \otimes \cdots \otimes \cH_{\cX_k}$, while leaving an operator on $\cH_{\cY}$.
\end{remark}

\begin{remark}[Special case for functionals]\label{rem:FunctionalProductRadonNikodymIdentity}
    When $\cY=\mathbb{C}$, then necessarily $\cZ = \mathbb{C}$ as well.
    In this case, every decomposition $\omega_{\cX_1}\otimes\cdots\otimes\omega_{\cX_k} =\sum_{i\in I}\Phi_i$ into positive functionals is represented by a POVM $(F_i)_{i\in I}$ in $\cM_{\cX_1}'\vnotimes\cdots\vnotimes\cM_{\cX_k}'$ via
    \begin{align}\label{eq:FunctionalProductRadonNikodymIdentity}
        \Phi_i(x) = \sandwich{\Omega_{\cX_1} \otimes \cdots \otimes \Omega_{\cX_k}}{F_i (\pi_{\cX_1} \otimes \cdots \otimes \pi_{\cX_k})(x)}{\Omega_{\cX_1} \otimes \cdots \otimes \Omega_{\cX_k}}.
    \end{align}
    This special case can also be proven directly from the Radon--Nikodym derivative theorem for functionals~\cite[\nopp II.6.4.6]{blackadar2006operator}.
\end{remark}

\section{Source-transfer model on source-party trees}\label{sec:SourceTransferModelMain}
The mixed quantum model of \zcref{def:MixedQuantumModelSourcePartyGraph} is formulated in terms of source Hilbert spaces and commuting observable algebras.
For source-party trees, the same correlations admit an equivalent description directly in terms of source $C^*$-algebras, source states, and completely positive maps.
We call this description the \emph{source-transfer model}.
In addition, we refer to \zcref{sec:SimpleExample} for the motivating examples illustrated by \zcref{fig:SourceTransferExamples}.

The name emphasizes two features of the construction.
First, its basic objects are algebras and states associated with the independent sources in $\SourceVertices{\Tree}$, rather than algebras generated by the party measurements; compare, e.g., \cite{wolfe2021quantum,ligthart2023convergent,ligthart2024semidefinite,renou2026two}.

Second, fix a party $r$ as the root.
Then the source-transfer model describes quantum tree scenarios as follows: each party transfers algebraic information from some source algebras to some other source algebras, starting from the leaves of the tree $\Tree$ with the endpoint being the root party $r$.
Specifically, a source $\alpha \in \SourceVertices{\Tree}$ receives messages $Q_v$ produced by its child parties $v \in \Children{\alpha}$, and combines them into one message $Q_{\alpha} \coloneqq \bigotimes_{v \in \Children{\alpha}} Q_{v}$.
Then, at each non-root party $v \in \PartyVertices{\Tree}$, it maps the messages $Q_{\alpha}$ from all of its child sources $\alpha \in \Children{v}$ to a new message $Q_v$ and sends it to its parent source $\Parent{v}$, using quantum instruments $\Phi^v$.
The root party $r$ finally evaluates the incoming messages through post-measurement
subnormalized states.
This transfer direction is only part of the mathematical description and does not represent physical communication between the parties.
Similarly, as we shall show, the choice of $r$ is purely mathematical and does not change the set of described correlations.

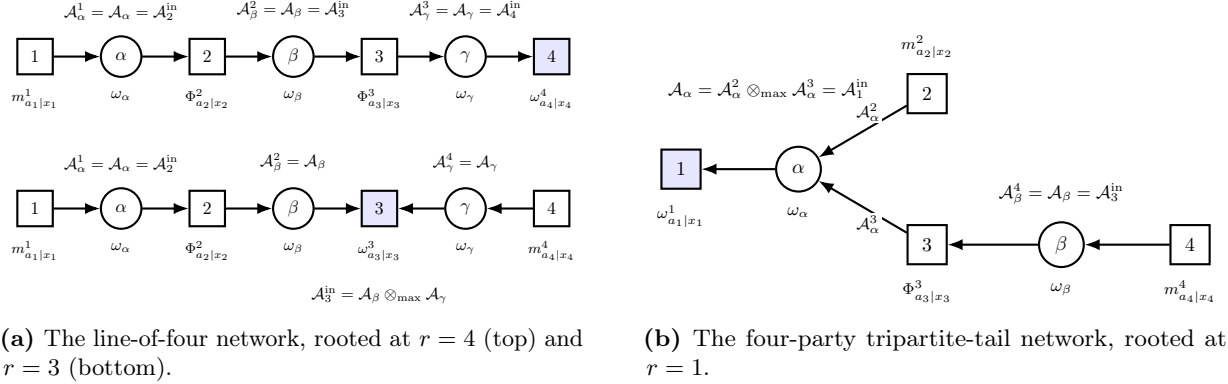
\begin{figure}[htbp]
    \centering
    \begin{subfigure}[t]{0.465\textwidth}
        \centering
        \begin{tikzpicture}[
            scale=0.67,
            transform shape,
            every node/.style={font=\small},
            party/.style={rectangle,draw,thick,minimum width=7mm,minimum height=7mm,inner sep=0pt},
            root/.style={party,fill=blue!10},
            source/.style={circle,draw,thick,minimum size=8mm,inner sep=0pt},
            sourceinfo/.style={font=\footnotesize,align=center},
            partylabel/.style={font=\scriptsize,align=center},
            transfer/.style={-{Latex[length=2mm]},thick}
        ]
            \node[party] (p1t) at (0,1.5) {$1$};
            \node[source] (at) at (1.7,1.5) {$\alpha$};
            \node[party] (p2t) at (3.4,1.5) {$2$};
            \node[source] (bt) at (5.1,1.5) {$\beta$};
            \node[party] (p3t) at (6.8,1.5) {$3$};
            \node[source] (ct) at (8.5,1.5) {$\gamma$};
            \node[root] (p4t) at (10.2,1.5) {$4$};

            \draw[transfer] (p1t)--(at);
            \draw[transfer] (at)--(p2t);
            \draw[transfer] (p2t)--(bt);
            \draw[transfer] (bt)--(p3t);
            \draw[transfer] (p3t)--(ct);
            \draw[transfer] (ct)--(p4t);

            \node[sourceinfo,above=5pt of at] {$\cA_\alpha^1=\JoinAlg{\alpha}=\InAlg{2}$};
            \node[sourceinfo,above=5pt of bt] {$\cA_\beta^2=\JoinAlg{\beta}=\InAlg{3}$};
            \node[sourceinfo,above=5pt of ct] {$\cA_\gamma^3=\JoinAlg{\gamma}=\InAlg{4}$};
            \node[sourceinfo,below=4pt of at] {$\omega_\alpha$};
            \node[sourceinfo,below=4pt of bt] {$\omega_\beta$};
            \node[sourceinfo,below=4pt of ct] {$\omega_\gamma$};

            \node[partylabel,below=4pt of p1t] {$\meas{1}{a_1}{x_1}$};
            \node[partylabel,below=4pt of p2t] {$\instr{2}{a_2}{x_2}$};
            \node[partylabel,below=4pt of p3t] {$\instr{3}{a_3}{x_3}$};
            \node[partylabel,below=4pt of p4t] {$\rstate{4}{a_4}{x_4}$};

            \node[party] (p1b) at (0,-1.5) {$1$};
            \node[source] (ab) at (1.7,-1.5) {$\alpha$};
            \node[party] (p2b) at (3.4,-1.5) {$2$};
            \node[source] (bb) at (5.1,-1.5) {$\beta$};
            \node[root] (p3b) at (6.8,-1.5) {$3$};
            \node[source] (cb) at (8.5,-1.5) {$\gamma$};
            \node[party] (p4b) at (10.2,-1.5) {$4$};

            \draw[transfer] (p1b)--(ab);
            \draw[transfer] (ab)--(p2b);
            \draw[transfer] (p2b)--(bb);
            \draw[transfer] (bb)--(p3b);
            \draw[transfer] (p4b)--(cb);
            \draw[transfer] (cb)--(p3b);

            \node[sourceinfo,above=5pt of ab] {$\cA_\alpha^1=\JoinAlg{\alpha}=\InAlg{2}$};
            \node[sourceinfo,above=5pt of bb] {$\cA_\beta^2=\JoinAlg{\beta}$};
            \node[sourceinfo,above=5pt of cb] {$\cA_\gamma^4=\JoinAlg{\gamma}$};
            \node[sourceinfo,below=4pt of ab] {$\omega_\alpha$};
            \node[sourceinfo,below=4pt of bb] {$\omega_\beta$};
            \node[sourceinfo,below=4pt of cb] {$\omega_\gamma$};

            \node[partylabel,below=4pt of p1b] {$\meas{1}{a_1}{x_1}$};
            \node[partylabel,below=4pt of p2b] {$\instr{2}{a_2}{x_2}$};
            \node[partylabel,below=4pt of p3b] (rootlabel) {$\rstate{3}{a_3}{x_3}$};
            \node[partylabel,below=4pt of p4b] {$\meas{4}{a_4}{x_4}$};
            \node[sourceinfo,below=6pt of rootlabel] {$\InAlg{3}=\JoinAlg{\beta}\maxotimes\JoinAlg{\gamma}$};
        \end{tikzpicture}
        \caption{The line-of-four network, rooted at $r=4$ (top) and $r=3$ (bottom).}
        \label{fig:SourceTransferLineOfFour}
    \end{subfigure}
    \hspace{0.035\textwidth}
    \begin{subfigure}[t]{0.465\textwidth}
        \centering
        \begin{tikzpicture}[
            scale=0.74,
            transform shape,
            every node/.style={font=\small},
            party/.style={rectangle,draw,thick,minimum width=7mm,minimum height=7mm,inner sep=0pt},
            root/.style={party,fill=blue!10},
            source/.style={circle,draw,thick,minimum size=8mm,inner sep=0pt},
            sourceinfo/.style={font=\footnotesize,align=center},
            partylabel/.style={font=\scriptsize,align=center},
            edgeinfo/.style={font=\scriptsize,fill=white,inner sep=1pt},
            transfer/.style={-{Latex[length=2mm]},thick}
        ]
            \node[root] (p1) at (0,0) {$1$};
            \node[source] (a) at (2.1,0) {$\alpha$};
            \node[party] (p2) at (4.4,1.35) {$2$};
            \node[party] (p3) at (4.4,-1.35) {$3$};
            \node[source] (b) at (6.8,-1.35) {$\beta$};
            \node[party] (p4) at (9.1,-1.35) {$4$};

            \draw[transfer] (p2)--(a);
            \draw[transfer] (p3)--(a);
            \draw[transfer] (a)--(p1);
            \draw[transfer] (p4)--(b);
            \draw[transfer] (b)--(p3);

            \node[sourceinfo,anchor=south west,yshift=22pt] at (p1.north west) {$\JoinAlg{\alpha}=\cA_\alpha^2\maxotimes\cA_\alpha^3=\InAlg{1}$};
            \node[sourceinfo,above=5pt of b] {$\cA_\beta^4=\JoinAlg{\beta}=\InAlg{3}$};
            \node[sourceinfo,below=4pt of a] {$\omega_\alpha$};
            \node[sourceinfo,below=4pt of b] {$\omega_\beta$};
            \node[edgeinfo] at (3.35,1.02) {$\cA_\alpha^2$};
            \node[edgeinfo] at (3.35,-0.98) {$\cA_\alpha^3$};

            \node[partylabel,below=4pt of p1] {$\rstate{1}{a_1}{x_1}$};
            \node[partylabel,above=4pt of p2] {$\meas{2}{a_2}{x_2}$};
            \node[partylabel,below=4pt of p3] {$\instr{3}{a_3}{x_3}$};
            \node[partylabel,below=4pt of p4] {$\meas{4}{a_4}{x_4}$};
        \end{tikzpicture}
        \caption{The four-party tripartite-tail network, rooted at $r=1$.}
        \label{fig:SourceTransferTripartiteTail}
    \end{subfigure}
    \caption{
    Source-transfer descriptions of two quantum tree networks.
    Circles denote sources, squares denote parties, and blue squares denote roots.
    Each source $\alpha$ carries a state $\omega_\alpha$ on the joint source algebra $\JoinAlg{\alpha}$, formed from the algebras $\cA_\alpha^v$ of its child parties.
    The input algebra of a party $v$ is $\InAlg{v}=\bigotimes_{\alpha\in\Children{v}}^{\max}\JoinAlg{\alpha}$.
    Leaf parties carry POVM elements, internal non-root parties carry transfer maps, and root parties carry positive functionals.
    Arrows indicate the direction of transfer toward the root.
    }
    \label{fig:SourceTransferExamples}
\end{figure}

\subsection{Definition of the source-transfer model}\label{sec:DefinitionSourceTransferModel}
\begin{definition}[Source-transfer model on a rooted source-party tree]\label{def:SourceTransferModelTree}
Let $\Tree = ( \PartyVertices{\Tree} \cup \SourceVertices{\Tree}, \Edges{\Tree} )$ be a finite source-party tree, rooted at a party $r \in \PartyVertices{\Tree}$.
For each party $v \in \PartyVertices{\Tree}$, let $\mathsf{X}_v$ and $\mathsf{A}_v$ denote its finite set of measurement settings and outcomes, respectively.
A correlation
\begin{align}
    p(\mathbf{a} | \mathbf{x}), \qquad \mathbf{a}=(a_v)_{v \in \PartyVertices{\Tree}}, \qquad \mathbf{x}=(x_v)_{v \in \PartyVertices{\Tree}},
\end{align}
admits a \emph{source-transfer realization} $\cR_{(\Tree,r)}$ on $(\Tree,r)$ if the following objects exist.

For every source $\alpha \in \SourceVertices{\Tree}$ and every child party
$v \in \Children{\alpha}$, there exists a unital $C^*$-algebra $\cA_\alpha^v$.
For each source $\alpha$, define the joint source algebra
\begin{align}\label{eq:TreeSourceAlgebra}
    \JoinAlg{\alpha} \coloneqq \bigotimes_{v\in\Children{\alpha}}^{\max} \cA_\alpha^v,
\end{align}
and let $\omega_\alpha: \JoinAlg{\alpha} \to \mathbb{C}$ be a state.
Each $\cA_\alpha^v$ can be identified with its canonical copy in $\JoinAlg{\alpha}$ sharing the same unit.
If $\alpha$ is a leaf source, we use the convention $\JoinAlg{\alpha}=\mathbb{C}$ and $\omega_\alpha = \identitymap_{\mathbb{C}}$.

For every party $v \in \PartyVertices{\Tree}$, denote its incoming joint source algebra and incoming source state by
\begin{align}\label{eq:TreeIncomingSourceData}
    \InAlg{v} \coloneqq \bigotimes_{\alpha\in\Children{v}}^{\max} \JoinAlg{\alpha}, \qquad \InState{v} \coloneqq \bigotimes_{\alpha\in\Children{v}}\omega_\alpha.
\end{align}
If $v$ is a leaf party, we use the convention $\InAlg{v}=\mathbb{C}$ and $\InState{v}=\identitymap_{\mathbb{C}}$.
The party data are as follows.
\begin{enumerate}[label=(\roman*)]
    \item For every non-root leaf party $v$, and every setting $x_v \in \mathsf{X}_v$, there is a POVM
    \begin{align}\label{eq:TreeLeafPOVM}
        \bigl\{ \meas{v}{a_v}{x_v} \bigr\}_{a_v \in \mathsf{A}_v} \subset \cA_{\Parent{v}}^v.
    \end{align}

    \item For every non-root internal party $v$, there are completely positive maps
    \begin{align}\label{eq:TreeTransferMapType}
        \instr{v}{a_v}{x_v}: \InAlg{v} \to \cA_{\Parent{v}}^v \subset \JoinAlg{\Parent{v}}, \qquad a_v \in \mathsf{A}_v, \quad x_v \in \mathsf{X}_v,
    \end{align}
    satisfying the source-replacer condition
    \begin{align}\label{eq:TreeTransferReplacer}
        \sum_{a_v\in\mathsf{A}_v}\instr{v}{a_v}{x_v}(q)=\InState{v}(q)\,\id_{\cA_{\Parent{v}}^v}
    \end{align}
    for every $q\in\InAlg{v}$ and every $x_v\in\mathsf{X}_v$.

    \item At the root $r$, there are positive functionals
    \begin{align}\label{eq:TreeRootFunctionalType}
        \rstate{r}{a_r}{x_r}:\InAlg{r}\to\mathbb{C}, \qquad a_r\in\mathsf{A}_r, \quad x_r\in\mathsf{X}_r,
    \end{align}
    satisfying the source-replacer condition
    \begin{align}\label{eq:TreeRootFunctionalReplacer}
        \sum_{a_r\in\mathsf{A}_r}\rstate{r}{a_r}{x_r}=\InState{r}
    \end{align}
    for every $x_r\in\mathsf{X}_r$.
\end{enumerate}

Fix outcomes $\mathbf{a}$ and settings $\mathbf{x}$.
We define a \emph{transfer message} $\TransferMsg{v}\in\cA_{\Parent{v}}^v$ for every non-root party $v$ and a \emph{source message} $\TransferMsg{\alpha}\in\JoinAlg{\alpha}$ for every source $\alpha$, by a single recursion running from the leaves of $\Tree$ towards the root.
The recursion starts at the leaf parties, whose messages are their POVM elements:
\begin{align}\label{eq:TreeRecursiveTransfer}
    \TransferMsg{v} \coloneqq \begin{cases}
        \meas{v}{a_v}{x_v}, & v \text{ is a leaf party},\\[1mm]
        \instr{v}{a_v}{x_v}\!\left(\displaystyle\bigotimes_{\alpha\in\Children{v}}\TransferMsg{\alpha}\right), & v \text{ is an internal party}.
    \end{cases}
\end{align}
Each source then combines the messages of its child parties:
\begin{align}\label{eq:TreeSourceMessage}
    \TransferMsg{\alpha} \coloneqq  \bigotimes_{v\in\Children{\alpha}} \TransferMsg{v}.
\end{align}
If $\alpha$ is a leaf source, then $\Children{\alpha}=\emptyset$ and the empty tensor product gives $\TransferMsg{\alpha}=1\in\mathbb{C}$, consistent with the convention $\JoinAlg{\alpha}=\mathbb{C}$ above.
The correlation is recovered at the root by
\begin{align}\label{eq:TreeSourceTransferCorrelation}
    p(\mathbf{a}|\mathbf{x})=\rstate{r}{a_r}{x_r}\!\left(\bigotimes_{\alpha\in\Children{r}}\TransferMsg{\alpha}\right).
\end{align}

We call a source-transfer realization \emph{finite-dimensional} if every joint source algebra $\JoinAlg{\alpha}$ is finite-dimensional, equivalently if every local source algebra $\cA_\alpha^v$ is finite-dimensional.
\end{definition}

As required by the commuting-operator formulation of measurement locality, the maximal tensor product $\bigotimes^{\max}$ used above allows arbitrary commuting representations of the child-party algebras.

Note that the separate treatment of leaves and the root in~\zcref{def:SourceTransferModelTree} is only for readability.
Indeed, a leaf POVM can equivalently be viewed as a completely positive instrument from $\mathbb{C}$ to $\cA_{\Parent{v}}^v$, while the root functionals are completely positive maps from $\InAlg{r}$ to $\mathbb{C}$.
A source leaf represents a local source: we assign its joint source algebra to be $\mathbb{C}$, and absorb the state information into the transfer map of its parent party.

See \zcref{fig:SourceTransferExamples} for an illustration, and \zcref{sec:SimpleExample} for concrete examples.

\subsection{Equivalence with the mixed quantum model}\label{sec:EquivalenceSourceTransferModel}
We now prove that the source-transfer model in \zcref{def:SourceTransferModelTree} is exactly the mixed quantum model on source-party trees.

\begin{theorem}[Equivalence of the source-transfer and mixed models]\label{thm:SourceTransferTreeEquivalence}
Let $\Tree$ be a finite source-party tree and let $r \in \PartyVertices{\Tree}$ be any choice of root party.
A correlation $p(\mathbf{a}|\mathbf{x})$ admits a mixed quantum realization on $\Tree$ as in \zcref{def:MixedQuantumModelSourcePartyGraph} if and only if it admits a source-transfer realization on $(\Tree,r)$ as in \zcref{def:SourceTransferModelTree}.

Consequently, the set of correlations admitting a source-transfer realization on $(\Tree,r)$ is independent of the choice of root party $r \in \PartyVertices{\Tree}$.
\end{theorem}

The equivalence of source-transfer correlations means that the choice of root changes only the algebraic data, but not the set of correlations.
We refer to \zcref{sec:SimpleExample} for explicit calculations on concrete examples that may help the reader follow the general proof below.

\begin{proof}
\emph{Mixed model $\Rightarrow$ source-transfer model.}
Recall $\Tree_v$ is the subtree rooted at $v$, consisting of $v$ as the root together with all vertices whose unique path to the root $r$ passes through $v$.
We first construct a source-transfer realization from a mixed quantum realization of $p$.

\medskip
\noindent
\emph{Constructing source algebras and states.}
For every source $\alpha$ and every child party $v \in \Children{\alpha}$, let
\begin{align}\label{eq:TreeMixedToTransferEdgeData}
    \cA_\alpha^v \coloneqq \cM_\alpha^v \subset \Bof{\cH_\alpha}, \qquad \JoinAlg{\alpha} \coloneqq \bigotimes_{v\in\Children{\alpha}}^{\max}\cA_\alpha^v.
\end{align}
Since the algebras $\cM_\alpha^v$, $v\in\Children{\alpha}$, mutually commute, the universal property of the maximal tensor product gives a canonical unital $*$-representation $\bigotimes_{v\in\Children{\alpha}}q_v \mapsto \prod_{v\in\Children{\alpha}} q_v$.
We omit this canonical representation from the notation throughout this direction for simplicity.
With the omission, define the source state $\omega_\alpha: \cA_{\alpha} \to \mathbb{C}$ by
\begin{align}
    \omega_\alpha(q) \coloneqq \sandwich{\Omega_\alpha}{ q }{\Omega_\alpha}, \qquad q \in \JoinAlg{\alpha}.
\end{align}

\medskip
\noindent
\emph{Constructing data for parties.}
Let $v \in \PartyVertices{\Tree}$ be a non-root party, define the incoming algebras and spaces over its child sources by
\begin{align}
    \InAlg{v} \coloneqq \bigotimes_{\alpha\in\Children{v}}^{\max} \JoinAlg{\alpha}, \qquad \InHilb{v} \coloneqq \bigotimes_{\alpha\in\Children{v}}\cH_\alpha, \qquad \ket{\InVec{v}} \coloneqq \bigotimes_{\alpha\in\Children{v}}\ket{\Omega_\alpha}.
\end{align}
Again, we use the canonical product representation on $\InHilb{v}$ and omit it for simplicity.
We now construct the data for every party.

If $v \in \PartyVertices{\Tree}$ is a non-root leaf, simply define
\begin{align}\label{eq:TreeMixedToTransferLeaf}
    \meas{v}{a_v}{x_v}\coloneqq\Meas{v}{a_v}{x_v}\in\cM_{\Parent{v}}^v=\cA_{\Parent{v}}^v.
\end{align}

If $v \in \PartyVertices{\Tree}$ is a non-root internal party, the transfer maps at $v$, $\instr{v}{a_v}{x_v}: \InAlg{v} \to \cA_{\Parent{v}}^v \subset \JoinAlg{\Parent{v}}$, are defined analogously to \zcref{rem:RadonNikodymSandwichForm} by partial sandwiching:
\begin{align}\label{eq:TreeMixedToTransferMap}
    \instr{v}{a_v}{x_v}(q) \coloneqq \sandwich{\InVec{v}}{\bigl(\Meas{v}{a_v}{x_v}\bigr)^{1/2}\bigl( q \otimes I_{\cH_{\Parent{v}}} \bigr)\bigl(\Meas{v}{a_v}{x_v}\bigr)^{1/2}}{\InVec{v}},
\end{align}
for all $q \in \InAlg{v}$.
Since $\Meas{v}{a_v}{x_v} \in \left( \vnbigotimes_{\alpha\in\Children{v}}\cM_\alpha^v \right) \vnotimes \cM_{\Parent{v}}^v \subset \Bof{\InHilb{v}} \vnotimes \cM_{\Parent{v}}^v$, this partial sandwich indeed gives an operator in $\cM_{\Parent{v}}^v = \cA_{\Parent{v}}^v$ acting on $\cH_{\Parent{v}}$.
In addition, it is completely positive as a composition of the conjugation map by positive operators and the partial trace by a vector state.
Moreover, by construction $\Meas{v}{a_v}{x_v}$ commutes with all $q\otimes I_{\cH_{\Parent{v}}}$; hence its square root also commutes with all $q \otimes I_{\cH_{\Parent{v}}}$~\cite[\nopp I.4.2.1]{blackadar2006operator}.
Summing over $a_v$ and using POVM completeness implies
\begin{align}
    \sum_{a_v}\instr{v}{a_v}{x_v}(q)=\InState{v}(q)\,\id_{\cA_{\Parent{v}}^v},
\end{align}
so the source-replacer condition \zcref{eq:TreeTransferReplacer} holds.

At the root $r$, define $\InHilb{r}$, $\ket{\InVec{r}}$, in the same way and analogously let
\begin{align}\label{eq:TreeMixedToTransferRoot}
    \rstate{r}{a_r}{x_r}(q) \coloneqq \sandwich{\InVec{r}}{\bigl(\Meas{r}{a_r}{x_r}\bigr)^{1/2} q \bigl(\Meas{r}{a_r}{x_r}\bigr)^{1/2}}{\InVec{r}}, \qquad q\in\InAlg{r}.
\end{align}
The same commutation and POVM-completeness argument gives
\begin{align}
    \sum_{a_r}\rstate{r}{a_r}{x_r}=\InState{r},
\end{align}
which is the desired source-replacer condition \zcref{eq:TreeRootFunctionalReplacer}.

\medskip
\noindent
\emph{Recovering the correlations by inductively computing transfer messages.}
It remains to verify the correlation-recovery condition
\zcref{eq:TreeSourceTransferCorrelation}.
Fix outcomes $\mathbf{a}$ and settings $\mathbf{x}$.
For every non-root party $v$ and its subtree $\Tree_v$, we claim that the recursively defined party message satisfies
\begin{align}\label{eq:TreeSubtreeTransferIdentity}
    \TransferMsg{v}=\sandwich{\bigotimes_{\alpha\in\SourceVertices{\Tree_v}}\Omega_\alpha}{\prod_{w\in\PartyVertices{\Tree_v}}\Meas{w}{a_w}{x_w}}{\bigotimes_{\alpha\in\SourceVertices{\Tree_v}}\Omega_\alpha},
\end{align}
where the sandwich is taken over all source Hilbert spaces in the subtree $\Tree_v$, leaving an operator acting on the parent-source Hilbert space $\cH_{\Parent{v}}$.

We prove \zcref{eq:TreeSubtreeTransferIdentity} recursively from the leaves towards the root.
For a leaf party $v$, the subtree $\Tree_v$ contains no source vertex, so the identity is simply \zcref{eq:TreeMixedToTransferLeaf}.
Now let $v$ be an internal party and suppose that \zcref{eq:TreeSubtreeTransferIdentity} holds for every $u\in\Children{\alpha}$ and every $\alpha\in\Children{v}$.
Recall that $\TransferMsg{\alpha} = \bigotimes_{u\in\Children{\alpha}}\TransferMsg{u}$.
For fixed $\alpha\in\Children{v}$, the subtrees $\Tree_u$, $u\in\Children{\alpha}$, are pairwise disjoint, and the same is true for the branches corresponding to distinct $\alpha\in\Children{v}$.
Consequently,
\begin{align}\label{eq:TreeSubtreeSourceDecomposition}
    \bigotimes_{\alpha\in\SourceVertices{\Tree_v}}\ket{\Omega_\alpha} = \underbrace{ \bigotimes_{\alpha\in\Children{v}}\ket{\Omega_\alpha} }_{\ket{\InVec{v}}} \otimes \bigotimes_{\alpha\in\Children{v}} \bigotimes_{u\in\Children{\alpha}} \bigotimes_{\beta\in\SourceVertices{\Tree_u}}\ket{\Omega_\beta},
\end{align}
and
\begin{align}\label{eq:TreeSubtreeMeasurementDecomposition}
    \prod_{w\in\PartyVertices{\Tree_v}} \Meas{w}{a_w}{x_w} = \bigl(\Meas{v}{a_v}{x_v}\bigr)^{1/2} \left( \prod_{\alpha\in\Children{v}} \prod_{u\in\Children{\alpha}} \prod_{w\in\PartyVertices{\Tree_u}} \Meas{w}{a_w}{x_w} \right) \bigl(\Meas{v}{a_v}{x_v}\bigr)^{1/2}
\end{align}
due to the mixed-model commutation relations of $\Meas{v}{a_v}{x_v}$.
Thus, using the induction hypothesis \zcref{eq:TreeSubtreeTransferIdentity}, and \zcref{eq:TreeMixedToTransferMap,eq:TreeSubtreeSourceDecomposition,eq:TreeSubtreeMeasurementDecomposition}, we obtain
\begin{equation}
    \begin{aligned}
        \TransferMsg{v} &= \sandwich{\InVec{v}}{ \bigl(\Meas{v}{a_v}{x_v}\bigr)^{1/2} \left( \bigotimes_{\alpha\in\Children{v}} \TransferMsg{\alpha} \otimes I_{\cH_{\Parent{v}}} \right) \bigl(\Meas{v}{a_v}{x_v}\bigr)^{1/2} }{\InVec{v}} \\
        &= \sandwich{\bigotimes_{\alpha\in\SourceVertices{\Tree_v}}\Omega_\alpha}{\prod_{w\in\PartyVertices{\Tree_v}} \Meas{w}{a_w}{x_w}}{ \bigotimes_{\alpha\in\SourceVertices{\Tree_v}}\Omega_\alpha }.
    \end{aligned}
\end{equation}
This proves \zcref{eq:TreeSubtreeTransferIdentity} for $v$.

Finally, at the root $r$, the same disjoint-subtree argument of internal vertices applies.
Therefore, by \zcref{eq:TreeMixedToTransferRoot} and the claim \zcref{eq:TreeSubtreeTransferIdentity},
\begin{align}
    \rstate{r}{a_r}{x_r}\!\left(\bigotimes_{\alpha\in\Children{r}}\TransferMsg{\alpha}\right)=\sandwich{\bigotimes_{\alpha\in\SourceVertices{\Tree}}\Omega_\alpha}{\prod_{v\in\PartyVertices{\Tree}}\Meas{v}{a_v}{x_v}}{\bigotimes_{\alpha\in\SourceVertices{\Tree}}\Omega_\alpha}=p(\mathbf{a}|\mathbf{x}).
\end{align}
Hence the constructed source-transfer realization reproduces the original mixed correlation.

\bigskip
\noindent
\emph{Source-transfer model $\Rightarrow$ mixed model.}
Conversely, suppose that a source-transfer realization on $(\Tree,r)$ is given.
For every source $\alpha\in\SourceVertices{\Tree}$, let $\gns{\alpha}$ be the GNS representation of $(\JoinAlg{\alpha},\omega_\alpha)$ and define the source von Neumann algebra by
\begin{align}\label{eq:TreeTransferToMixedVNAlgebras}
    \cM_\alpha \coloneqq \pi_\alpha(\JoinAlg{\alpha})'' \subset \Bof{\cH_\alpha}.
\end{align}
For every child party $v \in \Children{\alpha}$, define the observable algebras of $v$ on source $\alpha$ by
\begin{align}\label{eq:TreeTransferToMixedEndpointAlgebras}
    \cM_\alpha^v \coloneqq \pi_\alpha(\cA_\alpha^v)'' \subset \Bof{\cH_\alpha}, \qquad \cM_\alpha^{\Parent{\alpha}} \coloneqq \cM_\alpha' \subset \Bof{\cH_\alpha}.
\end{align}
Note that $\cM_\alpha^v$, $v\in\Children{\alpha}$, mutually commute because they arise from different factors of the maximal tensor product $\cA_\alpha$, and the parent algebra $\cM_\alpha^{\Parent{\alpha}}$ commutes with all of them by construction.
Define the global source Hilbert space and source vector by
\begin{align}\label{eq:TreeTransferToMixedGlobalSource}
    \cH \coloneqq \bigotimes_{\alpha\in\SourceVertices{\Tree}}\cH_\alpha, \qquad \ket{\Omega} \coloneqq \bigotimes_{\alpha\in\SourceVertices{\Tree}}\ket{\Omega_\alpha}.
\end{align}
For every party $v$, the incoming product state $\InState{v} = \bigotimes_{\alpha \in \Children{v}} \omega_{\alpha}$ has the GNS representation
\begin{align}\label{eq:TreeTransferIncomingGNS}
    \InHilb{v} \coloneqq \bigotimes_{\alpha\in\Children{v}}\cH_\alpha, \qquad \InRep{v} \coloneqq \bigotimes_{\alpha\in\Children{v}}\pi_\alpha, \qquad \ket{\InVec{v}} \coloneqq \bigotimes_{\alpha\in\Children{v}}\ket{\Omega_\alpha},
\end{align}
by the product-GNS argument in \zcref{cor:RadonNikodymSourceProductReplacer}.

\medskip
\noindent
\emph{Reconstructing measurement POVMs.}
For every non-root leaf party $v$, define
\begin{align}\label{eq:TreeTransferToMixedLeaf}
    \Meas{v}{a_v}{x_v} \coloneqq \pi_{\Parent{v}}\!\left(\meas{v}{a_v}{x_v}\right) \in \cM_{\Parent{v}}^v.
\end{align}
Now let $v$ be a non-root internal party.
Observe that the completely positive maps $\{ \instr{v}{a_v}{x_v}\}_{a_v \in \mathsf{A}_v}$ satisfy the source-replacer condition with respect to $\InState{v} \id_{\cA_{\Parent{v}}^v}$.
Thus, apply \zcref{cor:RadonNikodymSourceProductReplacer} to its input algebras $\JoinAlg{\alpha}$, $\alpha\in\Children{v}$, ambient algebra $\cY = \cA_{\Parent{v}}$, and output subalgebra $\cZ=\cA_{\Parent{v}}^v$.
This gives a POVM satisfying
\begin{align}\label{eq:TreeTransferToMixedInternalLocation}
    \Meas{v}{a_v}{x_v} \in \left(\vnbigotimes_{\alpha\in\Children{v}}\underbrace{\cM_\alpha'}_{=\cM_\alpha^v}\right)\vnotimes\cM_{\Parent{v}}^v=\vnbigotimes_{\alpha\in\Neighbors{v}}\cM_\alpha^v \coloneqq \cM^v.
\end{align}
The partial-sandwich identity, \zcref{eq:ProductRadonNikodymPartialSandwich}, becomes
\begin{align}\label{eq:TreeTransferToMixedInternalIdentity}
    \pi_{\Parent{v}}\!\left(\instr{v}{a_v}{x_v}(q)\right)=\sandwich{\InVec{v}}{\Meas{v}{a_v}{x_v}\bigl(\InRep{v}(q)\otimes I_{\cH_{\Parent{v}}}\bigr)}{\InVec{v}}, \qquad q\in\InAlg{v}.
\end{align}
(The reason we require $\instr{v}{a_v}{x_v}: \InAlg{v} \to \cA_{\Parent{v}}^v$ instead of $\instr{v}{a_v}{x_v}: \InAlg{v} \to \JoinAlg{\Parent{v}}$ is precisely to localize the reconstructed POVM to the correct $\cM_{\Parent{v}}^v$ in the multipartite parent source $\cM_{\Parent{v}}$.)

At the root $r$, \zcref{rem:FunctionalProductRadonNikodymIdentity} gives a POVM satisfying
\begin{align}\label{eq:TreeTransferToMixedRootLocation}
    \Meas{r}{a_r}{x_r}\in\vnbigotimes_{\alpha\in\Children{r}}\cM_\alpha'=\vnbigotimes_{\alpha\in\Children{r}}\cM_\alpha^r \coloneqq \cM^r,
\end{align}
and
\begin{align}\label{eq:TreeTransferToMixedRootIdentity}
    \rstate{r}{a_r}{x_r}(q)=\sandwich{\InVec{r}}{\Meas{r}{a_r}{x_r}\,\InRep{r}(q)}{\InVec{r}}, \qquad q\in\InAlg{r}.
\end{align}
It follows that every reconstructed POVM belongs to the local observable algebra required in \zcref{def:MixedQuantumModelSourcePartyGraph}.

\medskip
\noindent
\emph{Recovery of the correlation from source-transfer realizations.}
It remains to verify that the reconstructed mixed model reproduces the source-transfer correlation $p(\mathbf{a}|\mathbf{x})$.
For every non-root party $v$, we claim that its recursively defined message $\TransferMsg{v}$ satisfies
\begin{align}\label{eq:TreeTransferToMixedSubtreeIdentity}
    \pi_{\Parent{v}}(\TransferMsg{v})=\sandwich{\bigotimes_{\alpha\in\SourceVertices{\Tree_v}}\Omega_\alpha}{\prod_{w\in\PartyVertices{\Tree_v}}\Meas{w}{a_w}{x_w}}{\bigotimes_{\alpha\in\SourceVertices{\Tree_v}}\Omega_\alpha},
\end{align}
where the partial sandwich is taken over all source Hilbert spaces in subtree $\Tree_v$, leaving an operator acting only on $\cH_{\Parent{v}}$.
This is analogous to \zcref{eq:TreeSubtreeTransferIdentity} from the forward direction.

We prove \zcref{eq:TreeTransferToMixedSubtreeIdentity} by induction from the leaves towards the root.
If $v$ is a leaf party, then $\SourceVertices{\Tree_v}=\emptyset$ and $\pi_{\Parent{v}}(\TransferMsg{v})=\pi_{\Parent{v}}\!\left(\meas{v}{a_v}{x_v}\right)=\Meas{v}{a_v}{x_v}$, leading to \zcref{eq:TreeTransferToMixedSubtreeIdentity}.
Now let $v$ be an internal party and suppose that the claim holds for every $u\in\Children{\alpha}$ and every $\alpha\in\Children{v}$.
For each child source $\alpha\in\Children{v}$, $\TransferMsg{\alpha}=\bigotimes_{u\in\Children{\alpha}}\TransferMsg{u}$ and multiplicativity of the GNS representation give
\begin{align}\label{eq:TreeRepresentedSourceMessage}
    \pi_\alpha(\TransferMsg{\alpha}) =\prod_{u\in\Children{\alpha}}\pi_\alpha(\TransferMsg{u}).
\end{align}
Then, we calculate
\begin{equation}
    \begin{aligned}
        \pi_{\Parent{v}}(\TransferMsg{v})
        &=\pi_{\Parent{v}}\!\left(\instr{v}{a_v}{x_v}\!\left(\bigotimes_{\alpha\in\Children{v}}\TransferMsg{\alpha}\right)\right)\\
        &=\sandwich{\InVec{v}}{\Meas{v}{a_v}{x_v}\left(\InRep{v}\!\left(\bigotimes_{\alpha\in\Children{v}}\TransferMsg{\alpha}\right)\otimes I_{\cH_{\Parent{v}}}\right)}{\InVec{v}}\\
        &=\sandwich{\InVec{v}}{\Meas{v}{a_v}{x_v}\left(\bigotimes_{\alpha\in\Children{v}}\prod_{u\in\Children{\alpha}}\pi_\alpha(\TransferMsg{u})\otimes I_{\cH_{\Parent{v}}}\right)}{\InVec{v}}\\
        &=\sandwich{\bigotimes_{\beta\in\SourceVertices{\Tree_v}}\Omega_\beta}{\Meas{v}{a_v}{x_v}\prod_{\alpha\in\Children{v}}\prod_{u\in\Children{\alpha}}\prod_{w\in\PartyVertices{\Tree_u}}\Meas{w}{a_w}{x_w}}{\bigotimes_{\beta\in\SourceVertices{\Tree_v}}\Omega_\beta}\\
        &=\sandwich{\bigotimes_{\beta\in\SourceVertices{\Tree_v}}\Omega_\beta}{\prod_{w\in\PartyVertices{\Tree_v}}\Meas{w}{a_w}{x_w}}{\bigotimes_{\beta\in\SourceVertices{\Tree_v}}\Omega_\beta}.
    \end{aligned}
\end{equation}
The first equality is the recursive definition of $\TransferMsg{v}$ in \zcref{eq:TreeRecursiveTransfer}, while the second follows from \zcref{eq:TreeTransferToMixedInternalIdentity}.
For the third equality, we use \zcref{eq:TreeRepresentedSourceMessage} and $\InRep{v}\!\left(\bigotimes_{\alpha\in\Children{v}}\TransferMsg{\alpha}\right) =\bigotimes_{\alpha\in\Children{v}}\pi_\alpha(\TransferMsg{\alpha})$.
For the fourth equality, apply the induction hypothesis to every
$u\in\Children{\alpha}$.
Since the associated subtrees are pairwise disjoint, their partial sandwiches can be taken simultaneously.
Using the same source-vector decomposition as in \zcref{eq:TreeSubtreeSourceDecomposition} then gives exactly the fourth line.
The last equality follows because of the commutativity of the reconstructed POVMs and the same party-vertex decomposition \zcref{eq:TreeSubtreeMeasurementDecomposition}.
This proves the claim of \zcref{eq:TreeTransferToMixedSubtreeIdentity}.

Finally, using \zcref{eq:TreeSourceTransferCorrelation,eq:TreeTransferToMixedRootIdentity,eq:TreeTransferToMixedSubtreeIdentity}, and the disjoint subtree decomposition in \zcref{eq:TreeSubtreeSourceDecomposition,eq:TreeSubtreeMeasurementDecomposition}, we analogously compute
\begin{equation}\label{eq:TreeTransferToMixedCorrelation}
    \begin{aligned}
        p(\mathbf{a}|\mathbf{x})
        &=\rstate{r}{a_r}{x_r}\!\left(\bigotimes_{\alpha\in\Children{r}}\TransferMsg{\alpha}\right)\\
        &=\sandwich{\InVec{r}}{\Meas{r}{a_r}{x_r}\InRep{r}\!\left(\bigotimes_{\alpha\in\Children{r}}\TransferMsg{\alpha}\right)}{\InVec{r}}\\
        &=\sandwich{\InVec{r}}{\Meas{r}{a_r}{x_r}\left(\bigotimes_{\alpha\in\Children{r}}\prod_{u\in\Children{\alpha}}\pi_\alpha(\TransferMsg{u})\right)}{\InVec{r}}\\
        &=\sandwich{\Omega}{\Meas{r}{a_r}{x_r}\prod_{\alpha\in\Children{r}}\prod_{u\in\Children{\alpha}}\prod_{w\in\PartyVertices{\Tree_u}}\Meas{w}{a_w}{x_w}}{\Omega} =\sandwich{\Omega}{\prod_{v\in\PartyVertices{\Tree}}\Meas{v}{a_v}{x_v}}{\Omega}.
    \end{aligned}
\end{equation}
Therefore the reconstructed objects form a mixed quantum realization on $\Tree$ and reproduce the original source-transfer correlation $p(\mathbf{a}|\mathbf{x})$.

\medskip
\noindent
\emph{Root-independence of quantum correlations on trees.}
Finally, let $r,s \in \PartyVertices{\Tree}$ be any two choices of the root party.
A source-transfer realization rooted at either party $r$ or $s$ is equivalent to the same mixed quantum model on $\Tree$.
This proves root independence and finishes the proof.
\end{proof}

We remark that the equivalence of both models is preserved in finite dimensions.

\begin{remark}[Equivalence in finite dimensions and the tensor-product model]\label{rem:FiniteDimensionalTransferToMixed}
On source-party trees, finite-dimensional mixed realizations and finite-dimensional source-transfer realizations give exactly the same correlations.
Together with \zcref{rem:MixedTensorProductModels}, this shows that both are equivalent to the finite-dimensional standard tensor-product model of~\cite{nielsen2010quantum}.

Indeed, the mixed-to-source-transfer construction in \zcref{thm:SourceTransferTreeEquivalence} takes $\cA_\alpha^v=\cM_\alpha^v\subset\Bof{\cH_\alpha}$.
If every $\cH_\alpha$ is finite-dimensional, these algebras and their maximal tensor products $\JoinAlg{\alpha}$ are finite-dimensional.
Conversely, if every $\JoinAlg{\alpha}$ is finite-dimensional, its source-state GNS space is finite-dimensional, with $\dim\cH_\alpha\leq\dim_{\mathbb{C}}\JoinAlg{\alpha}$.
The algebras $\cM_\alpha^v$ and $\cM_\alpha'$ in the reconstruction act on these same spaces, and \zcref{cor:RadonNikodymSourceProductReplacer} gives each party's POVM on the tensor product of its neighboring source GNS spaces.
The resulting mixed realization therefore acts on the finite-dimensional global space $\cH=\bigotimes_{\alpha\in\SourceVertices{\Tree}}\cH_\alpha$.
\end{remark}

Before considering an application to general networks with classical sources, we explain where the preceding proof uses the tree structure.

\begin{remark}[What is special about source-party trees?]\label{rem:WhyTreeGraph}
The proof of \zcref{thm:SourceTransferTreeEquivalence} uses the tree structure in two places.
First, every non-root party has a unique parent source, and every source has a unique parent party.
Thus the output algebra of each party's transfer maps is unambiguously determined.

Second, the subtrees below different child sources of a given party are disjoint.
Since distinct source vertices represent independent quantum sources, their states form exactly the product state $\InState{v}$, giving the product-source-replacer condition required by \zcref{cor:RadonNikodymSourceProductReplacer}.
Note that the same product-state mechanism also plays a central role in the previous equivalent results for the bilocal network~\cite{ligthart2023inflation,renou2026two}.

The same argument does not extend directly to source-party graphs containing cycles.
For example, consider the quantum triangle network $1 - \alpha_{12} - 2 - \alpha_{23} - 3 - \alpha_{13} - 1$, which is a six-cycle in the source-party graph.
If party $1$ is chosen as the root, the branches starting with $\alpha_{12}$ and $\alpha_{13}$ remain connected through party $2$, source $\alpha_{23}$, and party $3$.
They are therefore not disjoint source subtrees, and the product-state replacer channel on the right-hand side of the source-replacer condition in \zcref{cor:RadonNikodymSourceProductReplacer} is no longer available.
\end{remark}

\subsection{Application to networks with classical sources}\label{sec:BeyondTree}
A natural question is to what extent \zcref{thm:SourceTransferTreeEquivalence} extends beyond source-party trees.
We show that it also gives a characterization of source-party graphs when removing the classical sources results in a disjoint union of source-party trees.
For simplicity, we demonstrate the argument for a triangle network with a classical source $\gamma$ connecting $A$ and $B$.
The other sources $\alpha$ and $\beta$ connect $A$ to $C$ and $B$ to $C$, respectively.
Removing $\gamma$ leaves the bilocal tree $A-\alpha-C-\beta-B$, which we root at $C$; see \zcref{fig:BeyondTreeTriangle}.

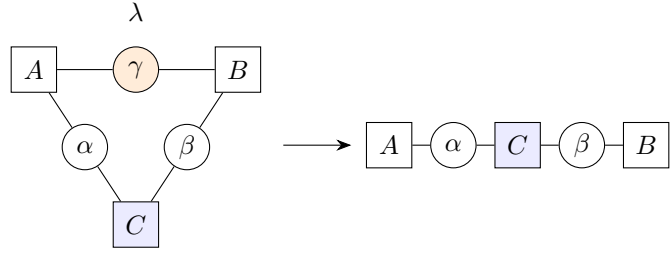
\begin{figure}[htbp]
    \centering
    \begin{tikzpicture}[party/.style={draw,rectangle,minimum size=6mm,fill=white},source/.style={draw,circle,minimum size=6mm,inner sep=1pt},every node/.style={font=\small}]
        \node[party] (A) at (-1.35,0.9) {$A$};
        \node[party] (B) at (1.35,0.9) {$B$};
        \node[party,fill=blue!8] (C) at (0,-1.15) {$C$};
        \node[source,fill=orange!15] (gamma) at (0,0.9) {$\gamma$};
        \node[source] (alpha) at (-0.675,-0.125) {$\alpha$};
        \node[source] (beta) at (0.675,-0.125) {$\beta$};
        \draw (A)--(gamma)--(B)--(beta)--(C)--(alpha)--(A);
        \node[above=2mm of gamma] {$\lambda$};
        \draw[-{Stealth[length=2mm]}] (1.95,-0.1)--(2.85,-0.1);
        \node[party] (Ar) at (3.35,-0.1) {$A$};
        \node[source] (alphar) at (4.2,-0.1) {$\alpha$};
        \node[party,fill=blue!8] (Cr) at (5.05,-0.1) {$C$};
        \node[source] (betar) at (5.9,-0.1) {$\beta$};
        \node[party] (Br) at (6.75,-0.1) {$B$};
        \draw (Ar)--(alphar)--(Cr)--(betar)--(Br);
    \end{tikzpicture}
    \caption{
    A triangle network, an elementary example of a source-party graph with a cycle.
    Here, suppose that $\gamma$ is classical; removing it leaves the bilocal tree $A-\alpha-C-\beta-B$, rooted at $C$.
    }
    \label{fig:BeyondTreeTriangle}
\end{figure}

\begin{proposition}[Triangle networks with one classical source]\label{prop:BeyondTreeClassical}
Let $p(a,b,c\mid x,y,z)$ be a correlation with finite outcome sets $\mathsf{A}_A,\mathsf{A}_B,\mathsf{A}_C$ and finite setting sets $\mathsf{X}_A,\mathsf{X}_B,\mathsf{X}_C$.
Then $p$ admits a mixed triangle realization with a classical source $\gamma$ if and only if there exist $N\in\mathbb{N}$, a probability distribution $(w_t)_{t=1}^{N}$, and a correlation $q$ with the outcome sets $\mathsf{A}_A,\mathsf{A}_B,\mathsf{A}_C$ and setting sets $\mathsf{X}_A\times[N]$, $\mathsf{X}_B\times[N]$, $\mathsf{X}_C$, respectively, such that $q$ admits a source-transfer realization on the bilocal tree $A-\alpha-C-\beta-B$ and
\begin{align}\label{eq:BeyondTreeClassicalRecovery}
    p(a,b,c\mid x,y,z)=\sum_{t=1}^{N}w_t\,q(a,b,c\mid(x,t),(y,t),z)
\end{align}
for every $a,b,c,x,y,z$.
\end{proposition}
\begin{proof}
Suppose that $p$ admits a mixed realization on the triangle such that $\gamma$ distributes a possibly continuous classical variable $\lambda\in\Lambda$ according to a probability measure $\mu$.
The measurements at $A$ and $B$ may depend on $\lambda$, so we write $\Meas{A}{a}{x}(\lambda)$ and $\Meas{B}{b}{y}(\lambda)$.
The remaining source product state $\ket{\Omega_\alpha}\otimes\ket{\Omega_\beta}$ and the measurement $\Meas{C}{c}{z}$ do not depend on $\lambda$.

The correlation $p$ then satisfies
\begin{equation}\label{eq:BeyondTreeClassicalIntegral}
    \begin{aligned}
        p&=\int_\Lambda F_\lambda\,d\mu(\lambda),\\
        F_\lambda(a,b,c\mid x,y,z)&\coloneqq\sandwich{\Omega_\alpha\otimes\Omega_\beta}{\Meas{A}{a}{x}(\lambda)\Meas{B}{b}{y}(\lambda)\Meas{C}{c}{z}}{\Omega_\alpha\otimes\Omega_\beta}.
    \end{aligned}
\end{equation}
Carath\'eodory's theorem~\cite[Proposition~2 and Appendix~A]{rosset2018universal} implies that there exist $N\in\mathbb{N}$, values $\lambda_1,\dots,\lambda_N\in\Lambda$, and weights $w_1,\dots,w_N\geq0$ such that
\begin{align}\label{eq:BeyondTreeFiniteClassicalSource}
    p=\sum_{t=1}^{N}w_tF_{\lambda_t},\qquad \sum_{t=1}^{N}w_t=1.
\end{align}
Thus, $\mu$ can be replaced by a discrete distribution on $\{\lambda_1,\ldots,\lambda_N\}$.
It suffices to use the finitely many measurements $\Meas{A}{a}{x}(\lambda_t)$ and $\Meas{B}{b}{y}(\lambda_t)$, while $\ket{\Omega_\alpha}\otimes\ket{\Omega_\beta}$ and $\Meas{C}{c}{z}$ remain fixed.

Regarding $t,s\in[N]$ as additional measurement settings at $A,B$, respectively, we define a correlation $q$ with the following mixed realization on the bilocal tree:
\begin{align}\label{eq:BeyondTreeClassicalTreeTable}
    q(a,b,c\mid(x,t),(y,s),z)\coloneqq\sandwich{\Omega_\alpha\otimes\Omega_\beta}{\Meas{A}{a}{x}(\lambda_t)\Meas{B}{b}{y}(\lambda_s)\Meas{C}{c}{z}}{\Omega_\alpha\otimes\Omega_\beta}.
\end{align}
Then \zcref{eq:BeyondTreeFiniteClassicalSource} gives \zcref{eq:BeyondTreeClassicalRecovery}.
The correlation $q$ is defined on the bilocal tree $A-\alpha-C-\beta-B$, with finite setting sets indexed by $(x,t)$, $(y,s)$, and $z$.
Hence \zcref{thm:SourceTransferTreeEquivalence} gives a source-transfer realization of $q$.

Conversely, suppose that $q$ admits a source-transfer realization and that $p$ satisfies \zcref{eq:BeyondTreeClassicalRecovery} for a finite probability distribution $(w_t)_{t=1}^{N}$.
By \zcref{thm:SourceTransferTreeEquivalence}, we obtain a mixed realization of $q$ with source states $\ket{\Omega_\alpha},\ket{\Omega_\beta}$ independent of $t,s$.
We also have measurements $\Meas{A}{a}{(x,t)}$, $\Meas{B}{b}{(y,s)}$, and $\Meas{C}{c}{z}$.
Add an independent classical source $\gamma$ that sends the same value $t$ to $A,B$ with probability $w_t$.
Using $\Meas{A}{a}{(x,t)}$ and $\Meas{B}{b}{(y,t)}$ recovers $p$ by \zcref{eq:BeyondTreeClassicalRecovery}.
\end{proof}

\begin{remark}[Generalization beyond triangle networks]\label{rem:GeneralizationClassicalSource}
    The proposition generalizes inductively to networks in which removing the classical sources results in a disjoint union of source-party trees.
    Indeed, we reduce each classical source in turn to finitely many values by the Carath\'eodory argument, and regard these values as additional measurement settings at its neighboring parties.
    Each remaining tree then has finite setting sets, so \zcref{thm:SourceTransferTreeEquivalence} applies to it.
    For each fixed choice of the classical source values, the joint correlation is the product of the tree correlations since they are disjoint.
    Summing these products with the probabilities of the classical source values recovers the original correlation.
\end{remark}

\section{A complete source-transfer SDP hierarchy for quantum tree networks}
\label{sec:CompleteSDPSourceTransfer}
The characterization in \zcref{thm:SourceTransferTreeEquivalence} replaces the more concrete Hilbert-space description of a quantum tree network by source $C^*$-algebras and states, positive root functionals, and completely positive transfer maps satisfying source-replacer conditions.
In this section, we use this characterization to construct a new complete SDP hierarchy for quantum tree networks.

Positivity of states and functionals can be imposed through Hankel and localizing matrices.
To handle nonlinear constraints on state values, the state-polynomial framework~\cite{klep2024state} provides a systematic approach, with many quantum applications~\cite{renou2026two,xu2026bulkspectralgapsemidecidable,moran2024uncertainty}.
Such constraints arise naturally in the source-replacer condition through products of source-state values.
This motivates extending the framework to multiple states on distinct source algebras.

A further extension is needed to incorporate the completely positive transfer maps between these algebras.
Completely positive maps have also been incorporated into SDP hierarchies for Bell scenarios \cite{baroni2026quantitativequantumsoundnessmultipartite}, where they are represented through dilating $*$-homomorphisms.
However, an arbitrary dilation of each transfer map need not preserve the source-party structure of a quantum network.
We therefore encode complete positivity directly through operator-valued Hankel matrices, allowing us to reconstruct the required maps with their prescribed input and output algebras unchanged.
The resulting framework extends noncommutative real algebraic geometry to jointly optimize over several states and the completely positive maps linking them, rather than over the moments of a single state.

We first construct the multi-state polynomial algebras associated with rooted source-party trees and define the source-transfer SDP hierarchy.
We then introduce its quadratic-module formulation and show how characters satisfying the constraints reconstruct source-transfer realizations for trees.
These results establish convergence to the source-transfer model and, by \zcref{thm:SourceTransferTreeEquivalence}, to the mixed quantum model.
A sufficient stopping criterion for extracting finite-dimensional realizations from finite-level SDP solutions is then established.
We conclude with a consequence for the computability of quantum tree game values: the corresponding promise problem is $\mathsf{coRE}$-complete.

Throughout this section, we fix a finite source-party tree $\Tree$, rooted at a party $r \in \PartyVertices{\Tree}$, together with finite setting and outcome sets $\mathsf{X}_v$ and $\mathsf{A}_v$ for every party $v \in \PartyVertices{\Tree}$.

\subsection{Multi-state polynomials with completely positive maps}\label{sec:MultiStatePolyRing}
We begin by rewriting the $C^*$-algebraic objects of \zcref{def:SourceTransferModelTree} in the language of noncommutative polynomials.
We then introduce the multi-state polynomial $*$-ring for our problem, generalizing~\cite{klep2024state} by allowing different state symbols on different source $*$-algebras.

The noncommutative $*$-algebras $\fA_{\alpha}^v$ are again constructed recursively from the leaves towards the root, alternating between party and source vertices.
If $\alpha$ is a leaf source, set $\JoinPolyAlg{\alpha} = \mathbb{C}$.
Let $v \neq r$ be a leaf party.
For its parent source $\Parent{v}$, define $\fA_{\Parent{v}}^v$ to be the unital $*$-algebra generated by the symbols
\begin{align}
    \bigl\{\meas{v}{a_v}{x_v}:a_v \in \mathsf{A}_v,\ x_v\in\mathsf{X}_v\bigr\}
\end{align}
satisfying the POVM relations $(\meas{v}{a_v}{x_v})^*=\meas{v}{a_v}{x_v}$ and $\sum_{a_v \in \mathsf{A}_v}\meas{v}{a_v}{x_v}=\id$.

Given an internal source $\alpha$, whenever the local source algebras $\fA_\alpha^v$ have been defined for every child party $v \in \Children{\alpha}$, define its joint formal source algebra by
\begin{align}\label{eq:FormalJointSourceAlgebra}
   \JoinPolyAlg{\alpha} \coloneqq \bigotimes_{v \in \Children{\alpha}} \fA_\alpha^v,
\end{align}
where $\bigotimes$ denotes the algebraic tensor product and each tensor factor is identified with its canonical unital copy in $\JoinPolyAlg{\alpha}$.

Now let $v \neq r$ be an internal party and suppose that $\JoinPolyAlg{\alpha}$ has already been defined for every child source $\alpha \in \Children{v}$.
Define the formal incoming algebra by
\begin{align}\label{eq:FormalIncomingAlgebra}
    \InPolyAlg{v} \coloneqq \bigotimes_{\alpha\in\Children{v}} \JoinPolyAlg{\alpha}.
\end{align}
For every monomial $w \in \InPolyAlg{v}$ and every $(a_v,x_v) \in \mathsf{A}_v \times \mathsf{X}_v$, introduce a formal symbol
\begin{align}
    \formalInstr{v}{a_v}{x_v}(w),
\end{align}
modeling the action of the transfer map at $v$.
We impose $\formalInstr{v}{a_v}{x_v}(w^*)=\formalInstr{v}{a_v}{x_v}(w)^*$ and extend the symbol linearly to all of $\InPolyAlg{v}$.
The algebra $\fA_{\Parent{v}}^v$ is then defined to be the unital $*$-algebra generated by the symbols $\formalInstr{v}{a_v}{x_v}(w)$.
At the root, after the formal algebras of its child sources have been defined, similarly let
\begin{align}\label{eq:FormalIncomingRootAlgebra}
    \InPolyAlg{r} \coloneqq \bigotimes_{\alpha\in\Children{r}} \JoinPolyAlg{\alpha}.
\end{align}

Next, we define the \emph{degree} on these formal algebras.
Every leaf POVM symbol has degree one, $\deg(\meas{v}{a_v}{x_v})\coloneqq1$, and the degree is extended additively to products.
For elementary tensors and transfer map symbols, set
\begin{align}
    \deg\!\left(\bigotimes_i p_i\right)\coloneqq\sum_i\deg(p_i), \qquad \deg\!\left(\formalInstr{v}{a_v}{x_v}(q)\right)\coloneqq1+\deg(q),
\end{align}
and let the degree of a polynomial be the largest degree of its monomials.
Since $\Tree$, the setting sets, and the outcome sets are finite, there are only finitely many formal monomials of degree at most $d$ for every $d\in\mathbb{N}$.

We now introduce the state-polynomial extension.
For every source $\alpha\in\SourceVertices{\Tree}$ and every monomial $w \in \JoinPolyAlg{\alpha}$, introduce a formal scalar state symbol $\varsigma_\alpha(w)$ modeling the joint source state $\omega_\alpha$.
At the root, for every $(a_r,x_r)\in\mathsf{A}_r\times\mathsf{X}_r$ and every monomial $w\in\InPolyAlg{r}$, introduce a formal scalar functional symbol $\rPolystate{a_r}{x_r}(w)$ modeling $\rstate{r}{a_r}{x_r}$.
Extend these symbols linearly to all polynomials and impose
\begin{align}
    \varsigma_\alpha(\id)=\id, \qquad \varsigma_\alpha(p^*)=\varsigma_\alpha(p)^*, \qquad \sum_{a_r}\rPolystate{a_r}{x_r}(\id)=\id, \qquad \rPolystate{a_r}{x_r}(q^*)=\rPolystate{a_r}{x_r}(q)^*.
\end{align}
All scalar symbols are required to commute.
Denote the resulting commutative unital $*$-algebra by
\begin{align}\label{eq:DefMultiStateRingThinS}
    \thinTreeS\coloneqq\mathbb{C}\bigl[\varsigma_\alpha(p),\ \rPolystate{a_r}{x_r}(q)\mid \alpha\in\SourceVertices{\Tree},\ (a_r,x_r)\in\mathsf{A}_r\times\mathsf{X}_r,\ p\in\JoinPolyAlg{\alpha},\ q\in\InPolyAlg{r}\bigr].
\end{align}
We call it the \emph{multi-state polynomial $*$-algebra}.

We further extend $\varsigma_\alpha$, $\rPolystate{a_r}{x_r}$, and $\formalInstr{v}{a_v}{x_v}$ by $\thinTreeS$-linearity.
That is, for $s\in\thinTreeS$, we impose
\begin{align}
    \varsigma_\alpha(sp)\coloneqq s\varsigma_\alpha(p), \qquad \rPolystate{a_r}{x_r}(sq)\coloneqq s\rPolystate{a_r}{x_r}(q), \qquad \formalInstr{v}{a_v}{x_v}(sq)\coloneqq s\formalInstr{v}{a_v}{x_v}(q).
\end{align}
The degrees of the scalar symbols are
\begin{align}
    \deg(\varsigma_\alpha(p))\coloneqq\deg(p), \qquad \deg(\rPolystate{a_r}{x_r}(q))\coloneqq1+\deg(q),
\end{align}
and are extended additively to all of $\thinTreeS$.
Denote by $\thinTreeS^d\subset\thinTreeS$ the subspace of polynomials of degree at most $d$.

For every source $\alpha$ and party $v$, define the \emph{noncommutative multi-state polynomial rings}
\begin{align}
    \fatS_\alpha\coloneqq\thinTreeS\otimes\JoinPolyAlg{\alpha}, \qquad \InfatS{v}\coloneqq\thinTreeS\otimes\InPolyAlg{v},
\end{align}
where the tensor products are algebraic.
The state symbol $\varsigma_\alpha$ is naturally a $\thinTreeS$-linear map on $\fatS_\alpha$, while the root symbols $\rPolystate{a_r}{x_r}$ are $\thinTreeS$-linear maps on $\InfatS{r}$.
The degree on these rings is again defined additively, and a superscript $d$ denotes the subspace of degree at most $d$.

For every party $v$, the product of the state symbols over its child sources defines a $\thinTreeS$-linear map
\begin{align}\label{eq:FormalIncomingProductState}
    \prod_{\alpha\in\Children{v}}\varsigma_\alpha:\InfatS{v}\to\thinTreeS
\end{align}
by
\begin{align}
    \left(\prod_{\alpha\in\Children{v}}\varsigma_\alpha\right)\!\left(s\otimes\bigotimes_{\alpha\in\Children{v}}q_\alpha\right)\coloneqq s\prod_{\alpha\in\Children{v}}\varsigma_\alpha(q_\alpha).
\end{align}
This formal product corresponds to the incoming product state $\InState{v}$ in \zcref{eq:TreeIncomingSourceData}.
We do not introduce yet another state symbol $\varsigma^{\mathrm{in}}$ to avoid difficulty of degree tracking.

For any $*$-algebra $\fA$ with a fixed set of generators, let $\Monomial{d}{\fA}$ denote the set of monomials of degree at most $d$, and let $\operatorname{Mon}(\fA)$ denote the set of all monomials.

The construction of $\fA_{\alpha}^v$ is clearly parallel to $\cA_{\alpha}^v$ in \zcref{def:SourceTransferModelTree}, so every abstract polynomial above has a concrete counterpart in an exact source-transfer realization.
We refer to this as evaluation.

\begin{definition}[Evaluation on a source-transfer realization]\label{def:EvaluationMultiStatePoly}
Let $\cR_{(\Tree, r)}$ be a source-transfer realization on $(\Tree,r)$ as in \zcref{def:SourceTransferModelTree}.
The formal $*$-algebras above admit a natural evaluation in $\cR_{(\Tree, r)}$, defined recursively from the leaves towards the root.

At a leaf party $v$, the measurement symbols are evaluated as the corresponding POVM elements in $\cA^v_{\Parent{v}}$.
At a source $\alpha$, the tensor factors of $\JoinPolyAlg{\alpha}$ are evaluated in the corresponding factor algebras $\cA_\alpha^v$ and combined in $\JoinAlg{\alpha}$.
At every non-root internal party, set
\begin{align}
    \formalInstr{v}{a_v}{x_v}(q) (\cR_{(\Tree, r)}) \coloneqq \instr{v}{a_v}{x_v}\!\left(q(\cR_{(\Tree, r)})\right).
\end{align}
This defines $p(\cR_{(\Tree, r)})$, the \emph{evaluation of $p$ on $\cR_{(\Tree, r)}$}, for every formal noncommutative polynomial $p$.

The scalar symbols are evaluated by
\begin{align}
    \varsigma_\alpha(p)(\cR_{(\Tree, r)}) \coloneqq \omega_\alpha\!\left(p(\cR_{(\Tree, r)})\right) 
    \in \mathbb{C}, \qquad \rPolystate{a_r}{x_r}(q) (\cR_{(\Tree, r)}) \coloneqq \rstate{r}{a_r}{x_r}\!\left(q(\cR_{(\Tree, r)})\right) \in \mathbb{C}.
\end{align}
Extending multiplicatively and linearly defines $f \mapsto f(\cR_{(\Tree, r)})\in\mathbb{C}$ for every $f\in\thinTreeS$.
We call $f(\cR_{(\Tree, r)})$ the \emph{evaluation of $f$ on $\cR_{(\Tree, r)}$}.
\end{definition}

\subsection{Source-transfer SDP hierarchy}\label{sec:SourceTransferSDPHierarchy}
We now define the hierarchy of SDP relaxations for \zcref{def:SourceTransferModelTree} using the multi-state polynomial rings above.

Fix $d \in \mathbb{N}$ sufficiently large so that all objects below are appropriately defined, and let $L_d : \thinTreeS^{2d} \to \mathbb{C}$ be a Hermitian linear functional.
We first define the matrices in our relaxation.
\begin{enumerate}
    \item \emph{Source Hankel and localizing matrices.}
    For every source $\alpha\in\SourceVertices{\Tree}$ and every self-adjoint $f\in\fatS_\alpha$, define
    \begin{align}\label{eq:SourceTransferGenericEdgeLocalizer}
        \left[H_{d-\left\lceil\frac{\deg(f)}{2}\right\rceil}^\alpha(f;L_d)\right]_{w_1,w_2}\coloneqq L_d\!\left(\varsigma_\alpha(w_1^*fw_2)\right), \qquad w_1,w_2\in\Monomial{d-\left\lceil\frac{\deg(f)}{2}\right\rceil}{\fatS_\alpha}.
    \end{align}
    The source Hankel matrix is $H_d^\alpha(L_d)\coloneqq H_d^\alpha(\id;L_d)$.

    \item \emph{Post-measurement Hankel and localizing matrices at the root $r$.}
    For every $(a_r,x_r)\in\mathsf{A}_r\times\mathsf{X}_r$ and every self-adjoint $f\in\InfatS{r}$, we similarly define
    \begin{align}\label{eq:SourceTransferGenericRootLocalizer}
        \left[H_{d-\left\lceil\frac{\deg(f)+1}{2}\right\rceil}^{a_r|x_r}(f;L_d)\right]_{w_1,w_2}\coloneqq L_d\!\left(\rPolystate{a_r}{x_r}(w_1^*fw_2)\right), \qquad w_1,w_2\in\Monomial{d-\left\lceil\frac{\deg(f)+1}{2}\right\rceil}{\InfatS{r}}.
    \end{align}
    The post-measurement Hankel matrix is $H_{d-1}^{a_r|x_r}(L_d)\coloneqq H_{d-1}^{a_r|x_r}(\id;L_d)$.

    \item \emph{Transfer map complete-positivity matrices.}
    For every non-root internal party $v$ and every $(a_v,x_v)\in\mathsf{A}_v\times\mathsf{X}_v$, define
    \begin{align}\label{eq:SourceTransferCPKernelMatrix}
        \left[C_{d-1}(\formalInstr{v}{a_v}{x_v};L_d)\right]_{(w_1,t_1),(w_2,t_2)}\coloneqq L_d\!\left(\varsigma_{\Parent{v}}\!\left(w_1^*\formalInstr{v}{a_v}{x_v}(t_1^*t_2)w_2\right)\right),
    \end{align}
    where $w_i\in\operatorname{Mon}(\fatS_{\Parent{v}})$ and $t_i\in\operatorname{Mon}(\InfatS{v})$ satisfy $\deg(w_i)+\deg(t_i)\leq d-1$ for $i=1,2$.
\end{enumerate}
Here $\formalInstr{v}{a_v}{x_v}(t_1^*t_2)\in\fA_{\Parent{v}}^v$ is identified with its canonical copy in the joint parent-source algebra $\fA_{\Parent{v}}$.
Next, define the following multi-state polynomials in $\thinTreeS$.
\begin{enumerate}[label=(\roman*)]
    \item \emph{Source-replacer polynomials at internal parties.}
    For every non-root internal party $v \in \PartyVertices{\Tree}$, $x_v\in\mathsf{X}_v$, $p_1,p_2\in\fatS_{\Parent{v}}$, and $q\in\InfatS{v}$, let
    \begin{align}\label{eq:SourceTransferInternalEqualityPolynomial}
        h_{x_v}^v(p_1,q,p_2)\coloneqq\varsigma_{\Parent{v}}\!\left(p_1^*\left(\sum_{a_v\in\mathsf{A}_v}\formalInstr{v}{a_v}{x_v}(q)-\left(\prod_{\alpha\in\Children{v}}\varsigma_\alpha\right)(q) \right)p_2\right).
    \end{align}

    \item \emph{Source-replacer polynomials at the root $r$.}
    For every $x_r\in\mathsf{X}_r$ and $q\in\InfatS{r}$, let
    \begin{align}\label{eq:SourceTransferRootEqualityPolynomial}
        h_{x_r}^r(q)\coloneqq\sum_{a_r\in\mathsf{A}_r}\rPolystate{a_r}{x_r}(q)-\left(\prod_{\alpha\in\Children{r}}\varsigma_\alpha\right)(q).
    \end{align}

    \item \emph{Correlation polynomials.}
    For every $(\mathbf{a},\mathbf{x})$, define the formal source and party messages recursively as in \zcref{eq:TreeSourceMessage,eq:TreeRecursiveTransfer}:
    \begin{align}
        \TransferMsg{\alpha}&\coloneqq\bigotimes_{v\in\Children{\alpha}}\TransferMsg{v}, & \TransferMsg{v}&\coloneqq\begin{cases}\meas{v}{a_v}{x_v},&v\text{ is a leaf party},\\[1mm]\formalInstr{v}{a_v}{x_v}\!\left(\displaystyle\bigotimes_{\alpha\in\Children{v}}\TransferMsg{\alpha}\right),&v\text{ is an internal party}.
        \end{cases}
    \end{align}
    Define
    \begin{align}\label{eq:SourceTransferCorrelationPolynomial}
        h_{\mathbf{a}|\mathbf{x}}\coloneqq\rPolystate{a_r}{x_r}\!\left(\bigotimes_{\alpha\in\Children{r}}\TransferMsg{\alpha}\right).
    \end{align}
\end{enumerate}

We are now ready to define the hierarchy.
\begin{definition}[Source-transfer multi-state SDP relaxation]\label{def:SourceTransferSDPRelaxation}
    Let $\Tree$ be a finite source-party tree rooted at $r \in \PartyVertices{\Tree}$, with finite setting and outcome sets $\mathsf{X}_v$ and $\mathsf{A}_v$ for every party $v\in\PartyVertices{\Tree}$, and with the associated multi-state polynomial algebras defined in \zcref{sec:MultiStatePolyRing}.
    Fix $d\in\mathbb{N}$ large enough so that all objects below are appropriately defined.
    
    Let $\fLsf_d(\Tree,r)$ be the set of Hermitian linear functionals $L_d: \thinTreeS^{2d} \to \mathbb{C}$ satisfying
    \begin{equation}\label{eq:SourceTransferSDPFree}
    \begin{alignedat}{3}
        &L_d(\id)=1, && &&\text{(normalization)}\\[1mm]
        &H_d^\alpha(L_d)\succeq0, &&\quad\forall\,\alpha\in\SourceVertices{\Tree}, &&\text{(source-state positivity)}\\[1mm]
        &H_{d-1}^{a_r|x_r}(L_d)\succeq0, &&\quad\forall\,(a_r,x_r)\in\mathsf{A}_r\times\mathsf{X}_r, &&\text{(root-functional positivity)}\\[1mm]
        &H_{d-\deg(g_\alpha)}^\alpha(\id-g_\alpha^*g_\alpha;L_d)\succeq0, &&\quad\substack{\forall\,\alpha\in\SourceVertices{\Tree},\\ \forall\,g_\alpha\text{ generators of }\JoinPolyAlg{\alpha},} &&\text{(source archimedean constraint)}\\[1mm]
        &H_{d-1-\deg(g_r)}^{a_r|x_r}(\id-g_r^*g_r;L_d)\succeq0, &&\quad\substack{\forall\,(a_r,x_r),\\ \forall\,g_r\text{ generators of }\InPolyAlg{r},} &&\text{(root archimedean constraint)}\\[1mm]
        &H_{d-1}^{\Parent{v}}(\meas{v}{a_v}{x_v};L_d)\succeq0, &&\quad\substack{\forall\,v\neq r\text{ a leaf party},\\ \forall\,(a_v,x_v)\in\mathsf{A}_v\times\mathsf{X}_v,} &&\text{(leaf POVM positivity)}\\[1mm]
        &C_{d-1}(\formalInstr{v}{a_v}{x_v};L_d)\succeq0, &&\quad\substack{\forall\,v\neq r\text{ an internal party},\\ \forall\,(a_v,x_v)\in\mathsf{A}_v\times\mathsf{X}_v,} &&\text{(transfer-map complete positivity)}\\[1mm]
        &L_d\!\left(h_{x_v}^v(w_1,t,w_2)\right)=0, &&\quad\substack{\forall\,v\neq r\text{ an internal party},\ \forall\,x_v,\\ w_1,w_2\in\operatorname{Mon}(\fatS_{\Parent{v}}),\ t\in\operatorname{Mon}(\InfatS{v}),\\ \deg(h_{x_v}^v(w_1,t,w_2))\leq2d,} &&\text{(internal source-replacer)}\\[1mm]
        &L_d\!\left(h_{x_r}^r(t)\right)=0, &&\quad\substack{\forall\,x_r,\ t\in\operatorname{Mon}(\InfatS{r}),\\ \deg(h_{x_r}^r(t))\leq2d,} &&\text{(root source-replacer).}
    \end{alignedat}
    \end{equation}

    As $d$ increases, the sets $\fLsf_d(\Tree,r)$ define a hierarchy of multi-state SDP relaxations for optimization over quantum tree networks.
    In particular, for a self-adjoint $f\in\thinTreeS^{2d}$, define
    \begin{align}
        \beta_d^{\min}(f)&\coloneqq\inf_{L_d\in\fLsf_d(\Tree,r)}L_d(f), & \beta_d^{\max}(f)&\coloneqq\sup_{L_d\in\fLsf_d(\Tree,r)}L_d(f).
    \end{align}

    To test a prescribed correlation $p(\mathbf{a}|\mathbf{x})$, let $\fLsf_d(\Tree,r,p)\subset\fLsf_d(\Tree,r)$ denote the subset further satisfying
    \begin{align}
        L_d \left( s_1^* \bigl(h_{\mathbf{a}|\mathbf{x}} - p(\mathbf{a}|\mathbf{x}) \bigr) s_2 \right) &= 0, && \substack{ \forall\,\mathbf{a},\mathbf{x},\ s_1,s_2\in\operatorname{Mon}(\thinTreeS),\\
        \deg(s_1)+\deg(s_2) +\deg(h_{\mathbf{a}|\mathbf{x}}) \leq 2d, } \qquad \text{(prescribed correlation).}
    \end{align}
    We say that $p(\mathbf{a}|\mathbf{x})$ passes the level-$d$ source-transfer SDP relaxation if $\fLsf_d(\Tree,r,p)\neq\emptyset$.
\end{definition}

Each level of relaxations can be cast as an SDP because all matrix entries and equality constraints are linear in the moments of $L_d$.
Moreover, the restriction of every $L_{d+1}\in\fLsf_{d+1}(\Tree,r)$ to $\thinTreeS^{2d}$ belongs to $\fLsf_d(\Tree,r)$, and the same holds for $\fLsf_d(\Tree,r,p)$.
Thus the hierarchy is monotone:
\begin{align}
    \beta_{d+1}^{\min}(f)\geq\beta_d^{\min}(f), \qquad \beta_{d+1}^{\max}(f)\leq\beta_d^{\max}(f)
\end{align}
for all sufficiently large $d$.

We refer to \zcref{sec:SimpleExample} for some concrete examples of the source-transfer SDP hierarchy.

\subsection{Quadratic-module formulation of the SDP hierarchy}\label{sec:QuadraticModuleSDPHierarchy}
We reformulate the SDP hierarchy in \zcref{def:SourceTransferSDPRelaxation} in the language of archimedean quadratic modules for the later convergence proof.

The positive-semidefinite matrix constraints in \zcref{def:SourceTransferSDPRelaxation} can equivalently be written as positivity conditions on multi-state polynomials.
For every self-adjoint $f\in\fatS_\alpha$,
\begin{align}\label{eq:EdgeLocalizerPolynomialForm}
    H_{d-\left\lceil\frac{\deg(f)}{2}\right\rceil}^\alpha(f;L_d)\succeq0 \quad \iff \quad L_d\!\left(\varsigma_\alpha(p^*fp)\right)\geq0
\end{align}
for every $p\in\fatS_\alpha^{\,d-\left\lceil\frac{\deg(f)}{2}\right\rceil}$.
Similarly, for every self-adjoint $f\in\InfatS{r}$,
\begin{align}\label{eq:RootLocalizerPolynomialForm}
    H_{d-\left\lceil\frac{\deg(f)+1}{2}\right\rceil}^{a_r|x_r}(f;L_d)\succeq 0 \quad \iff \quad L_d\!\left(\rPolystate{a_r}{x_r}(p^*fp)\right)\geq0
\end{align}
for every $p\in(\InfatS{r})^{d-\left\lceil\frac{\deg(f)+1}{2}\right\rceil}$.

The complete-positivity matrix has the analogous interpretation
\begin{align}\label{eq:CPKernelPolynomialForm}
    C_{d-1}(\formalInstr{v}{a_v}{x_v};L_d)\succeq0 \quad\Longleftrightarrow\quad L_d\!\left(\sum_{i,j=1}^m\varsigma_{\Parent{v}}\!\left(p_i^*\formalInstr{v}{a_v}{x_v}(q_i^*q_j)p_j\right)\right)\geq0
\end{align}
for every finite family $p_i\in\fatS_{\Parent{v}}$ and $q_i\in\InfatS{v}$ satisfying $\deg(p_i)+\deg(q_i)\leq d-1$.
In the limit, this is the multi-state polynomial form of the complete-positivity criterion in~\cite[\nopp II.6.9.8]{blackadar2006operator}.

The equality constraints have a similarly useful polynomial interpretation.
By linearity, the source-replacer constraints in \zcref{eq:SourceTransferSDPFree} are equivalent to requiring that $L_d$ vanishes on the complex linear span of the source-replacer polynomials $h_{x_v}^{v}$ and $h_{x_r}^{r}$.
By the $*$-preserving property of state symbols, this span is closed under involution, so vanishing on it is equivalent to vanishing on its Hermitian part.
We can therefore rewrite the source-replacer constraints as
\begin{align}\label{eq:SourceTransferEqualitySpace}
    L_d\bigl(\cE(\Tree,r)\bigr)=0, \qquad \cE(\Tree, r) \coloneqq \left\{ h=h^* \mid h\in \operatorname{span}_{\mathbb{C}} \left\{ h_{x_v}^{v}(p_1,q,p_2),\ h_{x_r}^{r}(q) \right\} \right\} \subset \thinTreeS.
\end{align}
The correlation test with $h_{\mathbf{a}|\mathbf{x}} - p(\mathbf{a}|\mathbf{x})$ can be treated analogously.

We now combine the preceding positivity and equality conditions into a
single quadratic module.
\begin{definition}[Source-transfer quadratic modules]
\label{def:SourceTransferQuadraticModule}
    Let the source-transfer quadratic module $\fQ(\Tree,r) \subset \thinTreeS$ be the following cone in the Hermitian part of $\thinTreeS$:
    \begin{equation}\label{eq:SourceTransferQuadraticModule}
    \begin{aligned}
        \fQ(\Tree,r)\coloneqq\operatorname{cone}\Biggl(&\left\{\varsigma_\alpha(p^*p)\ \middle|\ \substack{\alpha\in\SourceVertices{\Tree},\ p\in\fatS_\alpha}\right\}\\
        {}\cup{}&\left\{\rPolystate{a_r}{x_r}(q^*q)\ \middle|\ \substack{(a_r,x_r)\in\mathsf{A}_r\times\mathsf{X}_r,\ q\in\InfatS{r}}\right\}\\
        {}\cup{}&\left\{\varsigma_\alpha\!\left(p^*(\id-g_\alpha^*g_\alpha)p\right)\ \middle|\ \substack{\alpha\in\SourceVertices{\Tree},\ p\in\fatS_\alpha,\\ g_\alpha\text{ a generator of }\JoinPolyAlg{\alpha}}\right\}\\
        {}\cup{}&\left\{\rPolystate{a_r}{x_r}\!\left(q^*(\id-g_r^*g_r)q\right)\ \middle|\ \substack{(a_r,x_r)\in\mathsf{A}_r\times\mathsf{X}_r,\ q\in\InfatS{r},\\ g_r\text{ a generator of }\InPolyAlg{r}}\right\}\\
        {}\cup{}&\left\{\varsigma_{\Parent{v}}\!\left(p^*\meas{v}{a_v}{x_v}p\right)\ \middle|\ \substack{v\neq r\text{ a leaf party},\ (a_v,x_v)\in\mathsf{A}_v\times\mathsf{X}_v,\\ p\in\fatS_{\Parent{v}}}\right\}\\
        {}\cup{}&\left\{\sum_{i,j=1}^m\varsigma_{\Parent{v}}\!\left(p_i^*\formalInstr{v}{a_v}{x_v}(q_i^*q_j)p_j\right)\ \middle|\ \substack{v\neq r\text{ an internal party},\ (a_v,x_v)\in\mathsf{A}_v\times\mathsf{X}_v,\\ m\in\mathbb{N},\ p_i\in\fatS_{\Parent{v}},\ q_i\in\InfatS{v}}\right\}\Biggr)+\cE(\Tree,r).
    \end{aligned}
    \end{equation}
    Here $\operatorname{cone}$ denotes the set of all finite nonnegative real linear combinations of the displayed multi-state polynomials.

    For a prescribed correlation $p(\mathbf{a}|\mathbf{x})$, define
    \begin{equation}\label{eq:PrescribedCorrelationQuadraticModule}
    \begin{aligned}
        \fQ(\Tree,r,p)\coloneqq\fQ(\Tree,r)+\Bigl\{h=h^*\ \Bigm|\ h\in\operatorname{span}_{\mathbb{C}}\bigl\{&s_1^*(h_{\mathbf{a}|\mathbf{x}}-p(\mathbf{a}|\mathbf{x}))s_2\ \bigm|\\
        &\qquad s_1,s_2\in\thinTreeS,\ \mathbf{a},\mathbf{x}\bigr\}\Bigr\}.
    \end{aligned}
    \end{equation}
    Finally, let $\fQ_d(\Tree,r)$ and $\fQ_d(\Tree,r,p)$ denote the corresponding degree-$2d$ truncated cones.
    Their unions over $d$ give back $\fQ(\Tree,r)$ and $\fQ(\Tree,r,p)$, respectively.
\end{definition}

We now justify the name ``source-transfer quadratic module''.
For every $s\in\thinTreeS$, multiplication of the positivity generators by $s^*(\cdot)s$ can be absorbed into the corresponding test polynomials, namely $p\mapsto ps$, $q\mapsto qs$, and $p_i\mapsto p_i s$.
The equality space $\cE(\Tree,r)$ is closed under the same operation due to the $\thinTreeS$-linearity of the state symbols.
Hence
\begin{align}
    s^*\fQ(\Tree,r)s \subset \fQ(\Tree,r),
\end{align}
so $\fQ(\Tree,r)$ is indeed a quadratic module.
Then $\fQ(\Tree,r,p)$ is also a quadratic module since the added correlation span is also closed under $s^* (\cdot) s$.

\begin{remark}[Quadratic-module formulation of the  hierarchy]\label{rem:SourceTransferQuadraticModuleForm}
Consider a Hermitian linear functional $L_d:\thinTreeS^{2d}\to\mathbb{C}$.
By \zcref{eq:EdgeLocalizerPolynomialForm,eq:RootLocalizerPolynomialForm,eq:CPKernelPolynomialForm,eq:SourceTransferEqualitySpace}, the unconstrained source-transfer SDP hierarchy optimizing over quantum tree networks has the equivalent form
\begin{align}\label{eq:SourceTransferSDPQuadraticModuleForm}
    L_d \in \fLsf_d(\Tree,r) \quad \iff \quad L_d(\id)=1 \quad \text{and} \quad L_d(f) \geq 0 \quad \text{for every } f \in \fQ_d(\Tree,r).
\end{align}
Likewise, the hierarchy testing a prescribed correlation $p(\mathbf{a}|\mathbf{x})$ on a quantum tree network satisfies
\begin{align}\label{eq:PrescribedCorrelationQuadraticModuleForm}
    L_d \in \fLsf_d(\Tree,r,p) \quad \iff \quad L_d(\id)=1 \quad \text{and} \quad L_d(f) \geq 0 \quad \text{for every } f \in \fQ_d(\Tree,r,p).
\end{align}
\end{remark}

We now make a technical observation on $\fQ(\Tree,r)$ and $\fQ(\Tree,r,p)$ for the later convergence proof.
\begin{lemma}[Archimedean property of $\fQ(\Tree,r)$]\label{lem:SourceTransferQuadraticModuleArchimedean}
    The quadratic modules $\fQ(\Tree,r) \subset \thinTreeS$ and $\fQ(\Tree,r,p) \subset \thinTreeS$ are archimedean.
    That is, for every self-adjoint $f \in \thinTreeS$, there exists some $R_f > 0$ such that
    \begin{align}
        R_f\id\pm f \in \fQ(\Tree,r) \subset \fQ(\Tree,r,p).
    \end{align}
\end{lemma}
\begin{proof}
It is enough to prove that $\fQ(\Tree,r)$ is archimedean since $\fQ(\Tree,r) \subset \fQ(\Tree,r,p)$.
As in~\cite[Lemma~5.4 and~6.17]{klep2024state}, it suffices to show that the self-adjoint real and imaginary parts of every generator of the commutative multi-state polynomial ring $\thinTreeS$ are bounded with respect to $\fQ(\Tree,r)$.

We first consider a source-state symbol $\varsigma_\alpha(w)$ with monomials $w\in \JoinPolyAlg{\alpha}$.
For every generator $g_\alpha$ of $\JoinPolyAlg{\alpha}$, the source archimedean constraint implies that
\begin{align}
    \varsigma_\alpha(w^*w)-\varsigma_\alpha\!\left((g_\alpha w)^*(g_\alpha w)\right)=\varsigma_\alpha\!\left(w^*(\id-g_\alpha^*g_\alpha)w\right)\in\fQ(\Tree,r).
\end{align}
Induction on $\deg(w)$ and $\varsigma_\alpha(\id)=\id$ give
\begin{align}\label{eq:SourceWordBound}
    \id-\varsigma_\alpha(w^*w)\in\fQ(\Tree,r).
\end{align}
Therefore,
\begin{equation}
\begin{aligned}
    2\id\pm\left(\varsigma_\alpha(w)+\varsigma_\alpha(w)^*\right)&=\varsigma_\alpha\!\left((\id\pm w)^*(\id\pm w)\right)+\id-\varsigma_\alpha(w^*w)\in\fQ(\Tree,r),\\
    2\id\pm i\left(\varsigma_\alpha(w)-\varsigma_\alpha(w)^*\right)&=\varsigma_\alpha\!\left((\id\pm iw)^*(\id\pm iw)\right)+\id-\varsigma_\alpha(w^*w)\in\fQ(\Tree,r).
\end{aligned}
\end{equation}
Hence the real and imaginary parts of every source-state generator $\varsigma_{\alpha}(w)$ are bounded, as desired.

Next, let $w$ be a monomial in $\InPolyAlg{r}$ and consider $\rPolystate{a_r}{x_r}(w)$.
The same induction using the root archimedean constraint leads to
\begin{align}\label{eq:RootWordRelativeBound}
    \rPolystate{a_r}{x_r}(\id)-\rPolystate{a_r}{x_r}(w^*w)\in\fQ(\Tree,r).
\end{align}
Moreover, root-functional positivity and normalization $\sum_{b_r}\rPolystate{b_r}{x_r}(\id)=\id$ imply
\begin{align}
    \id-\rPolystate{a_r}{x_r}(\id)=\sum_{b_r\neq a_r}\rPolystate{b_r}{x_r}(\id)\in\fQ(\Tree,r).
\end{align}
Thus, by summing the above two, we have
\begin{align}\label{eq:RootWordBound}
    \id-\rPolystate{a_r}{x_r}(w^*w)\in\fQ(\Tree,r).
\end{align}
Consequently,
\begin{equation}
\begin{aligned}
    2\id\pm\left(\rPolystate{a_r}{x_r}(w)+\rPolystate{a_r}{x_r}(w)^*\right)&=\rPolystate{a_r}{x_r}\!\left((\id\pm w)^*(\id\pm w)\right)+\id-\rPolystate{a_r}{x_r}(\id)\\
    &\quad+\id-\rPolystate{a_r}{x_r}(w^*w)\in\fQ(\Tree,r),\\
    2\id\pm i\left(\rPolystate{a_r}{x_r}(w)-\rPolystate{a_r}{x_r}(w)^*\right)&=\rPolystate{a_r}{x_r}\!\left((\id\pm iw)^*(\id\pm iw)\right)+\id-\rPolystate{a_r}{x_r}(\id)\\
    &\quad+\id-\rPolystate{a_r}{x_r}(w^*w)\in\fQ(\Tree,r).
\end{aligned}
\end{equation}
Thus the real and imaginary parts of every root-functional generator are bounded as well.

Since the source-state and root-functional symbols generate $\thinTreeS$ as a commutative $*$-algebra, it follows that $\fQ(\Tree,r)$, and subsequently $\fQ(\Tree,r, p)$, are archimedean.
\end{proof}

\subsection{Characters of multi-state polynomials and source-transfer realizations}\label{sec:CharacterSourceTransferReconstruction}
We now show that characters of the multi-state polynomial algebra subjected to our SDP constraints can reconstruct source-transfer realizations.

Recall that a \emph{character} of the commutative multi-state polynomial $*$-algebra $\thinTreeS$ is a unital $*$-homomorphism
\begin{align}
    K:\thinTreeS \to \mathbb{C}, \qquad K(s_1^* s_2) = K(s_1)^* K(s_2).
\end{align}
In particular, a character is multiplicative on state symbols, regardless of whether they carry the same source labels or not.
For example, for an incoming elementary tensor $q=\bigotimes_{\alpha\in\Children{v}}q_\alpha\in\InPolyAlg{v}$, one has
\begin{align}\label{eq:CharacterIncomingProductState}
    K\!\left(\left(\prod_{\alpha\in\Children{v}}\varsigma_\alpha\right)(q)\right)=\prod_{\alpha\in\Children{v}}K\!\left(\varsigma_\alpha(q_\alpha)\right).
\end{align}
Note that this multiplicative property concerns only the commutative multi-state algebra $\thinTreeS$, and does not imply that the formal state symbols, or the states reconstructed from them, are multiplicative.

\begin{lemma}[Source-transfer realization from a character]\label{lem:SourceTransferCharacterReconstruction}
    Let $K:\thinTreeS \to \mathbb{C}$ be a character satisfying
    \begin{align}
        K(f)\geq 0 \qquad \text{for every } f \in \fQ(\Tree,r).
    \end{align}
    Then $K$ gives rise to a source-transfer realization on $(\Tree,r)$.
    If, in addition,
    \begin{align}
        K(f) \geq 0 \qquad \text{for every } f \in \fQ(\Tree,r,p),
    \end{align}
    then the reconstructed realization produces the prescribed correlation $p(\mathbf{a}|\mathbf{x})$.

    Furthermore, the reconstructed realization $\cR_K$ satisfies
    \begin{align}
        f(\cR_K)=K(f)
    \end{align}
    for every $f\in\thinTreeS$, in the sense of multi-state polynomial evaluation in \zcref{def:EvaluationMultiStatePoly}.
\end{lemma}

\begin{proof}
\noindent
\emph{Reconstructing the source algebras and states.}
For every source $\alpha\in\SourceVertices{\Tree}$, define
\begin{align}\label{eq:CharacterSourceState}
    \overline{\omega}_\alpha(q)\coloneqq K\!\left(\varsigma_\alpha(q)\right), \qquad q\in\JoinPolyAlg{\alpha}.
\end{align}
Source-state positivity and normalization imply $\overline{\omega}_\alpha(q^*q)\geq0$ and $\overline{\omega}_\alpha(\id)=1$.
Thus $\overline{\omega}_\alpha$ is a positive normalized functional on the unital $*$-algebra $\JoinPolyAlg{\alpha}$.
We apply the algebraic GNS construction (see e.g.,~\cite[Theorem~1.27]{burgdorf2016optimization}) with the source archimedean constraints.
Let $\mathcal{N}_\alpha\coloneqq\{q\in \JoinPolyAlg{\alpha} \mid\overline{\omega}_\alpha(q^*q)=0\}$.
For every generator $g_\alpha$ of $\JoinPolyAlg{\alpha} $,
\begin{align}
    \overline{\omega}_\alpha(q^*g_\alpha^*g_\alpha q)\leq\overline{\omega}_\alpha(q^*q),
\end{align}
so $\mathcal{N}_\alpha$ is a left ideal and left multiplication by each generator is contractive on $\JoinPolyAlg{\alpha} /\mathcal{N}_\alpha$.
After Hilbert-space completion, this gives a bounded GNS representation $\gns{\alpha}$ of $\JoinPolyAlg{\alpha}$.

For every child party $v\in\Children{\alpha}$, define
\begin{align}\label{eq:CharacterSourceCStarAlgebra}
    \cA_\alpha^v\coloneqq C^*\!\left(\pi_\alpha(\fA_\alpha^v)\right)\subset\Bof{\cH_\alpha}, \qquad \JoinAlg{\alpha}\coloneqq\bigotimes_{v\in\Children{\alpha}}^{\max}\cA_\alpha^v.
\end{align}
The algebras $\cA_\alpha^v$ mutually commute because the associated $\fA_{\alpha}^v$ commute in $\fA_\alpha$.
The inclusions into $\JoinAlg{\alpha}$ therefore induce a canonical unital $*$-representation of $\JoinAlg{\alpha}$ on $\cH_\alpha$, which we omit for simplicity.
Define the source state $\omega_\alpha$ by
\begin{align}
    \omega_\alpha(x)\coloneqq\sandwich{\Omega_\alpha}{x}{\Omega_\alpha}, \qquad x\in\JoinAlg{\alpha}.
\end{align}
For a leaf source, we take the convention $\JoinAlg{\alpha}=\mathbb{C}$ and $\omega_\alpha=\identitymap_{\mathbb{C}}$.

Define the dense unital $*$-homomorphism
\begin{align}\label{eq:FormalJointSourceToCStar}
    \widehat{\pi}_\alpha:\fA_\alpha\to\JoinAlg{\alpha}, \qquad \widehat{\pi}_\alpha\!\left(\bigotimes_{v\in\Children{\alpha}}q_v\right)\coloneqq\bigotimes_{v\in\Children{\alpha}}\pi_\alpha(q_v).
\end{align}
By construction,
\begin{align}\label{eq:CharacterJointSourceState}
    \omega_\alpha\!\left(\widehat{\pi}_\alpha(q)\right)=K\!\left(\varsigma_\alpha(q)\right), \qquad q\in\fA_\alpha.
\end{align}

\medskip
\noindent
\emph{Reconstructing the leaf POVMs.}
If $v\neq r$ is a leaf party, take
\begin{align}
    \pi_{\Parent{v}}\!\left(\meas{v}{a_v}{x_v}\right)\in\cA_{\Parent{v}}^v
\end{align}
to be the measurement POVMs.
The POVM normalization, $\sum_{a_v} \pi_{\Parent{v}}(\meas{v}{a_v}{x_v}) = \pi_{\Parent{v}}(\id) = I_{\cH_{\Parent{v}}}$, is imposed algebraically in $\fA_{\Parent{v}}^v$.
For positivity, the leaf positivity constraints give
\begin{align}
    \sandwich{\Omega_{\Parent{v}}}{\pi_{\Parent{v}}(p)^*\pi_{\Parent{v}}\!\left(\meas{v}{a_v}{x_v}\right)\pi_{\Parent{v}}(p)}{\Omega_{\Parent{v}}}=K\!\left(\varsigma_{\Parent{v}}\!\left(p^*\meas{v}{a_v}{x_v}p\right)\right)\geq0
\end{align}
for every $p\in\fA_{\Parent{v}}$.
Cyclicity of $\ket{\Omega_{\Parent{v}}}$ therefore implies $\pi_{\Parent{v}}(\meas{v}{a_v}{x_v}) \geq 0$.

\medskip
\noindent
\emph{Incoming algebras at the parties.}
For every party $v$, define
\begin{align}
    \InAlg{v}\coloneqq\bigotimes_{\alpha\in\Children{v}}^{\max}\JoinAlg{\alpha}, \qquad \InState{v}\coloneqq\bigotimes_{\alpha\in\Children{v}}\omega_\alpha,
\end{align}
where $\InState{v}$ is associated with the vector $\ket{\InVec{v}}=\bigotimes_{\alpha \in \Children{v}}\ket{\Omega_{\alpha}}$.
The maps in \zcref{eq:FormalJointSourceToCStar} induce a dense unital $*$-homomorphism
\begin{align}\label{eq:FormalIncomingToCStar}
    \InRep{v}:\InPolyAlg{v}\to\InAlg{v}, \qquad \InRep{v}\!\left(\bigotimes_{\alpha\in\Children{v}}q_\alpha\right)\coloneqq\bigotimes_{\alpha\in\Children{v}}\widehat{\pi}_\alpha(q_\alpha).
\end{align}
Using \zcref{eq:CharacterIncomingProductState,eq:CharacterJointSourceState}, we have
\begin{align}\label{eq:CharacterIncomingState}
    \InState{v}\!\left(\InRep{v}(q)\right)=K\!\left(\left(\prod_{\alpha\in\Children{v}}\varsigma_\alpha\right)(q)\right), \qquad q\in\InPolyAlg{v}.
\end{align}

\medskip
\noindent
\emph{Properties of the represented formal transfer map symbols.}
For an internal party $v\neq r$, we make two useful observations on the symbol $\pi_{\Parent{v}}\!\left(\formalInstr{v}{a_v}{x_v}(q)\right)$.

First, the complete-positivity constraints imply that, for every finite family $q_1,\ldots,q_m\in\InPolyAlg{v}$ and $p_1,\ldots,p_m\in\fA_{\Parent{v}}$,
\begin{equation}
    \begin{aligned}
        &\sum_{i,j=1}^m\sandwich{\Omega_{\Parent{v}}}{\pi_{\Parent{v}}(p_i)^*\pi_{\Parent{v}}\!\left(\formalInstr{v}{a_v}{x_v}(q_i^*q_j)\right)\pi_{\Parent{v}}(p_j)}{\Omega_{\Parent{v}}}\\
        &\qquad = K \left( \sum_{i,j=1}^{m} \varsigma_{\Parent{v}} \left( p_i^* \formalInstr{v}{a_v}{x_v}(q_i^*q_j) p_j \right) \right) \geq0.
    \end{aligned}
\end{equation}
Since $\pi_{\Parent{v}}(\fA_{\Parent{v}})\ket{\Omega_{\Parent{v}}}$ is dense in $\cH_{\Parent{v}}$, it follows that
\begin{align}\label{eq:CharacterFormalTransferCP}
    \left[\pi_{\Parent{v}}\!\left(\formalInstr{v}{a_v}{x_v}(q_i^*q_j)\right)\right]_{i,j=1}^m\succeq0
\end{align}
as an operator-valued matrix on $\cH_{\Parent{v}}^{\oplus m}$, with each entry belonging to the algebra $\cA_{\Parent{v}}^v$.

Second, the internal source-replacer equalities similarly imply, for every $q\in\InPolyAlg{v}$ and $p_1,p_2\in\fA_{\Parent{v}}$,
\begin{equation}
    \begin{aligned}
        \sandwich{\Omega_{\Parent{v}}}{\pi_{\Parent{v}}(p_1)^*\left(\sum_{a_v}\pi_{\Parent{v}}\!\left(\formalInstr{v}{a_v}{x_v}(q)\right)-\InState{v}\!\left(\InRep{v}(q)\right)I_{\cH_{\Parent{v}}}\right)\pi_{\Parent{v}}(p_2)}{\Omega_{\Parent{v}}} \\
        = K \left( h_{x_v}^{v}(p_1,q,p_2) \right) =0.
    \end{aligned}
\end{equation}
Again by cyclicity, for every $q\in\InPolyAlg{v}$, we have
\begin{align}\label{eq:CharacterFormalTransferReplacer}
    \sum_{a_v\in\mathsf{A}_v}\pi_{\Parent{v}}\!\left(\formalInstr{v}{a_v}{x_v}(q)\right)=\InState{v}\!\left(\InRep{v}(q)\right)I_{\cH_{\Parent{v}}}.
\end{align}

\medskip
\noindent
\emph{Well-definedness with respect to the incoming algebra representation.}
Suppose that $\InRep{v}(q)=0$.
We show that $\pi_{\Parent{v}}(\formalInstr{v}{a_v}{x_v}(q))=0$.
By \zcref{eq:CharacterFormalTransferReplacer},
\begin{align}
    \sum_{a_v}\pi_{\Parent{v}}\!\left(\formalInstr{v}{a_v}{x_v}(q^*q)\right)=\InState{v}\!\left(\InRep{v}(q)^*\InRep{v}(q)\right)I_{\cH_{\Parent{v}}}=0.
\end{align}
Every summand is positive by \zcref{eq:CharacterFormalTransferCP}, and hence
\begin{align}\label{eq:TransferQStarQZero}
    \pi_{\Parent{v}}\!\left(\formalInstr{v}{a_v}{x_v}(q^*q)\right)=0
\end{align}
for every $a_v$ as well.
Applying \zcref{eq:CharacterFormalTransferCP} to $q_1=\id$ and $q_2=q$ gives
\begin{align}
    \begin{pmatrix}\pi_{\Parent{v}}(\formalInstr{v}{a_v}{x_v}(\id))&\pi_{\Parent{v}}(\formalInstr{v}{a_v}{x_v}(q))\\[1mm]\pi_{\Parent{v}}(\formalInstr{v}{a_v}{x_v}(q))^*&0\end{pmatrix}\succeq0.
\end{align}
Positivity implies that a zero diagonal entry forces the corresponding row and column to vanish.
Hence $\pi_{\Parent{v}}(\formalInstr{v}{a_v}{x_v}(q)) = 0$.

In other words, $\pi_{\Parent{v}}(\formalInstr{v}{a_v}{x_v}(q))$ depends only on $\InRep{v}(q)$ rather than $q$.

\medskip
\noindent
\emph{Reconstructing the transfer maps.}
On the dense $*$-subalgebra $\InRep{v}(\InPolyAlg{v})\subset\InAlg{v}$, define the transfer map $\instr{v}{a_v}{x_v}: \InRep{v}(\InPolyAlg{v}) \to \pi_{\Parent{v}}(\fA_{\Parent{v}}^v)$ by
\begin{align}\label{eq:CharacterTransferMapDense}
    \instr{v}{a_v}{x_v}\!\left(\InRep{v}(q)\right)\coloneqq\pi_{\Parent{v}}\!\left(\formalInstr{v}{a_v}{x_v}(q)\right)\in\cA_{\Parent{v}}^v.
\end{align}
This is well defined by the above argument, showing that $\pi_{\Parent{v}}(\formalInstr{v}{a_v}{x_v}(q))$ depends only on $\InRep{v}(q)$ rather than $q$.
\zcref[S]{eq:CharacterFormalTransferCP} becomes
\begin{align}\label{eq:CharacterTransferDenseCP}
    \left[\instr{v}{a_v}{x_v}(x_i^*x_j)\right]_{i,j}\succeq0
\end{align}
for every finite family $x_i\in\InRep{v}(\InPolyAlg{v})$, while \zcref{eq:CharacterFormalTransferReplacer} becomes
\begin{align}\label{eq:CharacterTransferReplacerDense}
    \sum_{a_v\in\mathsf{A}_v}\instr{v}{a_v}{x_v}(x)=\InState{v}(x)I_{\cH_{\Parent{v}}}, \qquad x\in\InRep{v}(\InPolyAlg{v}).
\end{align}

We next show that $\instr{v}{a_v}{x_v}$ is contractive.
Indeed, \zcref{eq:CharacterTransferDenseCP} implies
\begin{align}
    \begin{pmatrix}\instr{v}{a_v}{x_v}(\id)&\instr{v}{a_v}{x_v}(x)\\[1mm]\instr{v}{a_v}{x_v}(x)^*&\instr{v}{a_v}{x_v}(x^*x)\end{pmatrix}\succeq0,
\end{align}
and \zcref{eq:CharacterTransferReplacerDense} implies
\begin{align}
    0\leq\instr{v}{a_v}{x_v}(\id)\leq I_{\cH_{\Parent{v}}}, \qquad 0\leq\instr{v}{a_v}{x_v}(x^*x)\leq\InState{v}(x^*x)I_{\cH_{\Parent{v}}}\leq\|x\|^2I_{\cH_{\Parent{v}}}.
\end{align}
Therefore, by the Cauchy--Schwarz inequality for a positive $2 \times 2$ operator-valued matrix,
\begin{align}\label{eq:CharacterTransferContractive}
    \left\|\instr{v}{a_v}{x_v}(x)\right\|^2\leq\left\|\instr{v}{a_v}{x_v}(\id)\right\|\left\|\instr{v}{a_v}{x_v}(x^*x)\right\|\leq\|x\|^2,
\end{align}
proving the claim.

Hence, $\instr{v}{a_v}{x_v}$ extends uniquely by continuity to
\begin{align}
    \instr{v}{a_v}{x_v}:\InAlg{v}\to\cA_{\Parent{v}}^v.
\end{align}
\zcref[S]{eq:CharacterTransferDenseCP} extends to the norm closure, so this extension map is completely positive due to~\cite[\nopp II.6.9.8]{blackadar2006operator}.
Likewise, \zcref{eq:CharacterTransferReplacerDense} extends to
\begin{align}
    \sum_{a_v\in\mathsf{A}_v}\instr{v}{a_v}{x_v}(x)=\InState{v}(x)\id_{\cA_{\Parent{v}}^v}, \qquad x\in\InAlg{v},
\end{align}
which is precisely the source-replacer condition with respect to the input state $\InState{v}$.

\medskip
\noindent
\emph{Reconstructing the root post-measurement functionals.}
Define at $r$
\begin{align}\label{eq:CharacterRootFunctional}
    \overline{\omega}_{a_r|x_r}^{(r)}(q)\coloneqq K\!\left(\rPolystate{a_r}{x_r}(q)\right), \qquad q\in\InPolyAlg{r}.
\end{align}
Root-functional positivity gives $\overline{\omega}_{a_r|x_r}^{(r)}(q^*q)\geq0$.
Moreover, the root source-replacer equality and \zcref{eq:CharacterIncomingState} give
\begin{equation}\label{eq:CharacterRootDomination}
\begin{aligned}
    \sum_{a_r\in\mathsf{A}_r}\overline{\omega}_{a_r|x_r}^{(r)}(q)=\InState{r}\!\left(\InRep{r}(q)\right), \qquad 
    0\leq\overline{\omega}_{a_r|x_r}^{(r)}(q^*q)\leq\InState{r}\!\left(\InRep{r}(q^*q)\right).
\end{aligned}
\end{equation}
Define
\begin{align}
    \rstate{r}{a_r}{x_r}\!\left(\InRep{r}(q)\right)\coloneqq\overline{\omega}_{a_r|x_r}^{(r)}(q)
\end{align}
on the dense $*$-subalgebra $\InRep{r}(\InPolyAlg{r}) \subset \InAlg{r}$.
It is well-defined by an analogous argument to that for the transfer maps $\instr{v}{a_v}{x_v}$.
Indeed, \zcref{eq:CharacterRootDomination} implies
\begin{align}
    \InRep{r}(q)=0 \quad \implies \quad \overline{\omega}_{a_r|x_r}^{(r)}(q^*q)=0 \quad \implies \quad \left| \overline{\omega}_{a_r|x_r}^{(r)}(q) \right|^2 \leq \overline{\omega}_{a_r|x_r}^{(r)}(\id)\, \overline{\omega}_{a_r|x_r}^{(r)}(q^*q) = 0.
\end{align}
Furthermore, for $x=\InRep{r}(q)$,
\begin{align}
    \left|\rstate{r}{a_r}{x_r}(x)\right|^2\leq\rstate{r}{a_r}{x_r}(\id)\rstate{r}{a_r}{x_r}(x^*x)\leq\InState{r}(x^*x)\leq\|x\|^2.
\end{align}
Hence $\rstate{r}{a_r}{x_r}$ extends uniquely to a bounded positive functional on $\InAlg{r}$, and
\begin{align}
    \sum_{a_r\in\mathsf{A}_r}\rstate{r}{a_r}{x_r}=\InState{r},
\end{align}
as desired.

\medskip
\noindent
\emph{Prescribed correlations and multi-state evaluation.}
We have therefore reconstructed all the objects of \zcref{def:SourceTransferModelTree}.
If furthermore $K$ is nonnegative on $\fQ(\Tree,r,p)$, then the prescribed-correlation equalities imply
\begin{align}
    \rstate{r}{a_r}{x_r}\!\left(\bigotimes_{\alpha\in\Children{r}}\TransferMsg{\alpha}\right)=p(\mathbf{a}|\mathbf{x})
\end{align}
for every $\mathbf{a},\mathbf{x}$, by the same GNS cyclicity argument.
Finally, $f(\cR_K)=K(f)$, where $\cR_K$ is the above reconstructed realization, follows recursively from the definitions.
\end{proof}

\subsection{Convergence of the SDP hierarchy to the source-transfer model}
\label{sec:ConvergenceSDPHierarchy}
We are now ready to prove that the SDP hierarchy in \zcref{def:SourceTransferSDPRelaxation}, equivalently \zcref{rem:SourceTransferQuadraticModuleForm}, is complete for the source-transfer model on quantum tree networks.
It follows that we obtain a complete SDP hierarchy for the mixed quantum model on trees in \zcref{def:MixedQuantumModelSourcePartyGraph} due to \zcref{thm:SourceTransferTreeEquivalence}.

\begin{theorem}[The source-transfer SDP hierarchy is complete]\label{thm:SourceTransferSDPConvergence}
    Let $\Tree$ be a finite source-party tree rooted at a party $r\in\PartyVertices{\Tree}$.
    Then a correlation $p(\mathbf{a}|\mathbf{x})$ passes the level-$d$
    source-transfer SDP relaxation as in \zcref{def:SourceTransferSDPRelaxation} for every sufficiently large $d$ if and only if it admits a source-transfer realization on $(\Tree,r)$ as in \zcref{def:SourceTransferModelTree}.
    
    Furthermore, fix a self-adjoint multi-state polynomial $f=f^*\in\thinTreeS$, and define
    \begin{align}\label{eq:ExactSourceTransferOptimization}
        \beta^{\min}(f) &\coloneqq \inf_{\cR} f(\cR), & \beta^{\max}(f) &\coloneqq \sup_{\cR} f(\cR),
    \end{align}
    where $\cR$ ranges over all source-transfer realizations on $(\Tree,r)$ and $f(\cR)$ denotes the multi-state polynomial evaluation as in \zcref{def:EvaluationMultiStatePoly}.
    Then
    \begin{align}\label{eq:SourceTransferOptimizationConvergence}
        \beta_d^{\min}(f) &\nearrow \beta^{\min}(f), & \beta_d^{\max}(f) &\searrow \beta^{\max}(f).
    \end{align}
    The same convergence holds after prescribing a correlation $p(\mathbf{a}|\mathbf{x})$, with the infimum and supremum in \zcref{eq:ExactSourceTransferOptimization} restricted to source-transfer realizations $\cR$ reproducing $p(\mathbf{a}|\mathbf{x})$.
\end{theorem}

\begin{proof}
\emph{Necessity.}
Let $\cR$ be a source-transfer realization on $(\Tree,r)$.
For every sufficiently large $d$, define
\begin{align}\label{eq:ExactRealizationTruncatedFunctional}
    L_d(f) \coloneqq f(\cR), \qquad f \in \thinTreeS^{2d},
\end{align}
using the evaluation of \zcref{def:EvaluationMultiStatePoly}.
Then $L_d$ is a Hermitian linear functional satisfying $L_d(\id)=1$.

We verify that $L_d(\fQ_d(\Tree,r)) \subset \mathbb{R}_{\geq 0}$.
By construction, the internal and root source-replacer polynomials, $h_{x_v}^v$ and $h_{x_r}^r$, vanish under the evaluation of $\cR$.
The source-state, root-functional, and leaf-POVM generators are nonnegative
since the corresponding states $\omega_\alpha$ and functionals $\rstate{r}{a_r}{x_r}$ are positive and the leaf $\meas{v}{a_v}{x_v}$ are positive.
For an internal vertex $v \neq r$, complete positivity of $\instr{v}{a_v}{x_v}$ also gives
\begin{align}
    \sum_{i,j=1}^{m} \omega_{\Parent{v}} \left( p_i(\cR)^* \instr{v}{a_v}{x_v} \left( q_i(\cR)^*q_j(\cR) \right) p_j(\cR) \right) \geq 0
\end{align}
for every finite family $p_i\in\fatS_{\Parent{v}}$ and $q_i\in\InfatS{v}$. 
This corresponds to complete-positivity generators in \zcref{eq:SourceTransferQuadraticModule}.

The archimedean generators, $\id - g_\alpha^* g_\alpha$ and $\id - g_r^* g_r$, are also nonnegative.
Indeed, every leaf measurement generator $\meas{v}{a_v}{x_v}$ is a contraction.
For every internal transfer map $\instr{v}{a_v}{x_v}$, the source-replacer condition and complete positivity imply $0 \leq \instr{v}{a_v}{x_v}(\id) \leq \id$.
Since $\left\| \instr{v}{a_v}{x_v} \right\| = \left\| \instr{v}{a_v}{x_v}(\id) \right\| \leq 1$, $\instr{v}{a_v}{x_v}$ is also contractive.
It follows recursively from the leaves towards the root that every formal generator is evaluated as a contraction, fulfilling the archimedean condition.

Consequently, $L_d(\fQ_d(\Tree,r)) \subset \mathbb{R}_{\geq 0}$, and thus $L_d \in \fLsf_d(\Tree,r)$ by \zcref{rem:SourceTransferQuadraticModuleForm}.
Moreover, if $\cR$ reproduces the prescribed correlation $p(\mathbf{a}|\mathbf{x})$, then by construction $h_{\mathbf{a}|\mathbf{x}}(\cR) = p(\mathbf{a}|\mathbf{x})$ for every $\mathbf{a},\mathbf{x}$.
Therefore, $L_d$ is further nonnegative on $\fQ_d(\Tree,r,p)$ and belongs to $\fLsf_d(\Tree,r,p)$.
This proves necessity of the source-transfer SDP hierarchy.
(In fact, for a source-transfer realization $\cR$, the evaluation $L(f) = f(\cR)$ defines a character $L$.
This corresponds to the observation in \zcref{lem:SourceTransferCharacterReconstruction}.)

\medskip
\noindent
\emph{Sufficiency.}
Suppose that $p(\mathbf{a}|\mathbf{x})$ passes the level-$d$ source-transfer SDP hierarchy for every $d$ that is sufficiently large.

Choose
\begin{align}
    L_d \in \fLsf_d(\Tree,r,p) \neq \emptyset
\end{align}
for each sufficiently large $d$.
We aim to extract a normalized linear functional $L$ on all of $\thinTreeS$ such that $L(\fQ(\Tree,r,p)) \subset \mathbb{R}_{\geq 0}$ using the standard diagonal extraction argument, thanks to \zcref{lem:SourceTransferQuadraticModuleArchimedean}.

Indeed, as a vector space of countable dimension, we can enumerate a basis $\{b_1, b_2, \dots\}$ for $\thinTreeS$.
We can assume without loss of generality that these basis elements are self-adjoint, $b_i^* = b_i$.
By \zcref{lem:SourceTransferQuadraticModuleArchimedean}, $\fQ(\Tree,r,p) \subset \thinTreeS$ is archimedean.
Then, for each $b_i$, there exists a constant $R_i > 0$ such that $R_i \pm b_i \in \fQ(\Tree,r,p)$.
In particular, as $\fQ(\Tree,r,p)$ is the union of all $\fQ_d(\Tree,r,p)$, for all sufficiently large $d$ we have also $R_i \pm b_i \in \fQ_d(\Tree,r,p)$.
Consequently, $\left|L_d(b_i)\right| \leq R_i$ for all sufficiently large $d$.
It follows that the sequence $\{ L_d(b_i) \}_d$ is uniformly bounded and has a convergent subsequence by compactness.
A standard diagonal argument over the countable basis $\{b_1, b_2, \dots\}$ gives a convergent subsequence of $L_{d_k}$ such that $L_{d_k}(b_i)$ is convergent for all $i$ as $k \to \infty$.
Define
\begin{align}
    L: \thinTreeS \to \mathbb{C}, \qquad L(b_i) = \lim_{k \to \infty} L_{d_k}(b_i)
\end{align}
and extend $L$ linearly to all of $\thinTreeS$.
By construction, $L(\id) = 1$ and $L(f) \geq 0$ for every $f \in \fQ(\Tree,r,p)$, as desired.

By the Kadison--Dubois representation theorem for the archimedean quadratic module  $\fQ(\Tree,r,p)$ \cite[Theorem~5.4.4]{marshall2008positive}, there exists a probability measure $\mu$ supported on characters $K: \thinTreeS \to \mathbb{C}$ satisfying
\begin{align}
    K(p) \geq 0 \qquad \text{for every } p\in\fQ(\Tree,r,p),
\end{align}
such that
\begin{align}\label{eq:KadisonDuboisSourceTransferRepresentation}
    L(p) = \int K(p)\,\mathrm{d}\mu(K)
\end{align}
for every $p\in\thinTreeS$.
Although $L$ may be a mixture of characters, the prescribed-correlation constraints in $\fQ(\Tree,r,p)$ force each of these characters to reproduce the same $p(\mathbf{a}|\mathbf{x})$.
Choosing one $K$ and applying \zcref{lem:SourceTransferCharacterReconstruction} therefore gives the required source-transfer realization.
This proves sufficiency and hence the first statement of the theorem.

\medskip
\noindent
\emph{Optimization.}
Let $f=f^*\in\thinTreeS$ be a self-adjoint multi-state polynomial.
Recall that every exact source-transfer realization $\cR$ induces, for all sufficiently large $d$, a feasible functional $L_d\in\fLsf_d(\Tree,r)$ satisfying $L_d(f)=f(\cR)$.
Hence, as outer relaxations,
\begin{align}\label{eq:SourceTransferOptimizationOuterBounds}
    \beta_d^{\min}(f) \leq \beta^{\min}(f), \qquad \beta_d^{\max}(f) \geq \beta^{\max}(f).
\end{align}
By monotonicity and the archimedean property, the sequence $\beta_d^{\min}(f)$ converges to some $\beta_\infty^{\min}(f) < \infty$.

Choose $L_d\in\fLsf_d(\Tree,r)$ such that $L_d(f) \leq \beta_d^{\min}(f)+1/d$.
Then the same compactness argument as in the sufficiency part gives a limiting normalized functional $L$ that is nonnegative on $\fQ(\Tree,r)$ and satisfying $L(f) = \beta_\infty^{\min}(f)$.
Applying \zcref{eq:KadisonDuboisSourceTransferRepresentation} and \zcref{lem:SourceTransferCharacterReconstruction},
\begin{align}
    \beta_\infty^{\min}(f)=L(f)=\int K(f)\,\mathrm{d}\mu(K)=\int f(\cR_K)\,\mathrm{d}\mu(K)\geq \int \beta^{\min}(f) \,\mathrm{d}\mu(K) = \beta^{\min}(f).
\end{align}
Combining with \zcref{eq:SourceTransferOptimizationOuterBounds}, this proves $\beta_d^{\min}(f)\nearrow\beta^{\min}(f)$.
Applying the same conclusion to $-f$ leads to $\beta_d^{\max}(f)\searrow\beta^{\max}(f)$.

Finally, the prescribed-correlation case is identical: one simply replaces $\fLsf_d(\Tree,r)$ and $\fQ(\Tree,r)$ by $\fLsf_d(\Tree,r,p)$ and $\fQ(\Tree,r,p)$, then applies the correlation-constrained case of \zcref{lem:SourceTransferCharacterReconstruction}.
\end{proof}

\begin{corollary}[Completeness for mixed quantum tree networks]\label{cor:MixedTreeSDPConvergence}
Let $\Tree$ be a finite source-party tree and let $r\in\PartyVertices{\Tree}$ be any root party.
For a correlation $p(\mathbf{a}|\mathbf{x})$, the following are equivalent:
\begin{enumerate}[label=(\roman*)]
    \item $p(\mathbf{a}|\mathbf{x})$ admits a mixed quantum realization on $\Tree$ as in \zcref{def:MixedQuantumModelSourcePartyGraph};
    \item $p(\mathbf{a}|\mathbf{x})$ passes the level-$d$ source-transfer SDP relaxation rooted at $r$ for every sufficiently large $d$.
\end{enumerate}

Moreover, suppose that a physical objective $f$ admits, for every root party $r$, a self-adjoint multi-state polynomial representation $f_r$ in the corresponding source-transfer description.
Let $\beta_{\mathrm{mix}}^{\min}(f)$ and $\beta_{\mathrm{mix}}^{\max}(f)$ denote its infimum and supremum over mixed quantum realizations on $\Tree$.
Then, for every root party $r$,
\begin{align}
    \beta_{d,r}^{\min}(f_r)&\nearrow\beta_{\mathrm{mix}}^{\min}(f), & \beta_{d,r}^{\max}(f_r)&\searrow\beta_{\mathrm{mix}}^{\max}(f).
\end{align}
In particular, the limiting optimal values are independent of the root party $r$.
\end{corollary}

For example, any real polynomial $f$ in the correlation entries $p(\mathbf{a}|\mathbf{x})$ has such a self-adjoint representation $f_r$, obtained by identifying $h_{\mathbf{a}|\mathbf{x}}$ from \zcref{eq:SourceTransferCorrelationPolynomial} for each entry $p(\mathbf{a}|\mathbf{x})$.

\begin{proof}
This follows directly from \zcref{thm:SourceTransferTreeEquivalence,thm:SourceTransferSDPConvergence}.
\end{proof}

The character reconstruction \zcref{lem:SourceTransferCharacterReconstruction} also implies that the mixed-model correlation set $C_{\mathrm{mix}}(\Tree)$ is compact for every finite source-party tree $\Tree$ with finite setting and outcome sets, so optimal game values are attained.
\begin{remark}[Compactness and attainment]\label{rem:MixedTreeCompactness}
Fix a finite source-party tree $\Tree$ with finite setting and outcome sets.
We first show that the characters $K:\thinTreeS\to\mathbb{C}$ satisfying $K(\fQ(\Tree,r))\subset\mathbb{R}_{\geq0}$ form a compact set in the topology of pointwise convergence.
The archimedean property in \zcref{lem:SourceTransferQuadraticModuleArchimedean} gives, for every self-adjoint $s\in\thinTreeS$, a constant $R_s>0$ such that $|K(s)|\leq R_s$ for every such $K$.
Applying this to real and imaginary parts gives a uniform bound on $K(s)$ for every polynomial $s \in \thinTreeS$.
Identify each $K$ with the tuple $(K(s))_{s \in \thinTreeS}$.
The bounds $R_s$ show that $(K(s))_{s \in \thinTreeS} \in \prod_{s \in \thinTreeS} \{z \in \mathbb{C} \mid \lvert z \rvert \leq R_s\}$, which is a compact set in the topology of pointwise convergence by Tychonoff's theorem.
Moreover, pointwise limits preserve linearity, normalization, and nonnegativity on $\fQ(\Tree,r)$, as well as the character identity $K(s^*t)=\lim_i K_i(s^*t)=\lim_i\overline{K_i(s)}K_i(t)=\overline{K(s)}K(t)$.
It follows that these $K$ form a closed subset of the compact set $\prod_{s \in \thinTreeS} \{z \in \mathbb{C} \mid \lvert z \rvert \leq R_s\}$ and hence compact.

By \zcref{lem:SourceTransferCharacterReconstruction,thm:SourceTransferTreeEquivalence}, the continuous map
\begin{align}\label{eq:CharacterCorrelationMap}
K\longmapsto\bigl(K(h_{\mathbf{a}|\mathbf{x}})\bigr)_{\mathbf{a},\mathbf{x}}
\end{align}
takes characters $K$ nonnegative on $\fQ(\Tree,r)$ to mixed-model correlations.
Conversely, every mixed realization gives a source-transfer realization whose evaluation defines such a character $K$.
The image of this map is therefore exactly $C_{\mathrm{mix}}(\Tree)$, which is consequently compact.
Every continuous real-valued function of the correlation entries therefore attains its minimum and maximum over $C_{\mathrm{mix}}(\Tree)$.
It follows that any continuous real-valued $f$ of the correlation entries $p(\mathbf{a}|\mathbf{x})$, including the mixed-model game value of a tree network, attains its minimum and maximum for $p \in C_{\mathrm{mix}}(\Tree)$.
\end{remark}

We finish the subsection by an observation on the choice of root in practical implementations.
\begin{remark}[Finite-level dependence on the root]\label{rem:FiniteLevelRootDependence}
The choice of root party fixes the direction of the source-transfer description, and therefore changes the finite-level SDP hierarchy.
In particular, it changes the recursively defined algebras $\fA_\alpha^v$, the joint source algebras $\fA_\alpha$, the source-state and root-functional symbols $\varsigma_\alpha$ and $\rPolystate{a_r}{x_r}$, the formal transfer symbols $\formalInstr{v}{a_v}{x_v}$, and the corresponding constraints.

Thus, at a fixed degree $d$, both the size and the strength of the relaxation may depend on the root $r$.
The limiting correlation set and optimal values are nevertheless independent of this choice by \zcref{cor:MixedTreeSDPConvergence}.
Thus the root may be chosen to simplify the finite-level SDP or to obtain tighter optimal values.
\end{remark}

\subsection{Stopping criteria and finite-dimensional extraction}\label{sec:SourceTransferFlatness}
We give a stopping criterion for the SDP hierarchy: a sufficient condition under which a finite-level SDP solution yields a finite-dimensional source-transfer realization.
By \zcref{rem:FiniteDimensionalTransferToMixed}, this also gives finite-dimensional mixed and standard tensor-product realizations on the network.
We follow~\cite[Definition~6.9]{klep2024state} to define the notion of flatness.

\begin{definition}[$\delta$-flat extension]\label{def:SourceTransferFlatness}
Suppose that $\Tree$ has at least one source, let $d,\delta\geq1$, and let $L_{d+\delta}\in\fLsf_{d+\delta}(\Tree,r)$.
Write $L_d\coloneqq L_{d+\delta}|_{\thinTreeS^{2d}}$ for its restriction.
We say that $L_{d+\delta}$ is a \emph{$\delta$-flat extension of $L_d$} if, for every source $\alpha$,
\begin{align}\label{eq:SourceTransferFlatnessCondition}
    \operatorname{rank}H_{d+\delta}^\alpha(L_{d+\delta})=\operatorname{rank}H_d^\alpha(L_d)\eqqcolon n_\alpha.
\end{align}
\end{definition}

We now define a degree bound $\delta_r$ that ensures the extension contains the polynomial relations needed for finite-dimensional extraction.
The bound is computed from the leaves to the root, using polynomial degree and the ranks $n_\alpha$ above.
At a leaf party $v$, the POVM generators have degree one, so set $\delta_v=1$.
At a source $\alpha$, degrees add across its child-party tensor factors, so we set
\begin{align}\label{eq:FlatnessSourceDegreeBound}
    \delta_\alpha\coloneqq\max\left\{1,\sum_{v\in\Children{\alpha}}\delta_v\right\}.
\end{align}
To motivate these bounds, consider the finite-dimensional source algebras that we aim to reconstruct.
We assign polynomial expressions in $\JoinAlg{\alpha}$ the degrees of their formal counterparts in $\JoinPolyAlg{\alpha}$.
Suppose that $\dim\cH_\alpha=n_\alpha$ for the source GNS space and that $\cA_\alpha^v\subset\Bof{\cH_\alpha}$ is generated by polynomials of degree at most $\delta_v$.
Starting from the span of $\id$, successively include words of length one, two, and so on in these generators.
If one step does not increase the span, multiplication by any generator preserves that span, so it already equals $\cA_\alpha^v$.
Since $\dim\cA_\alpha^v\leq n_\alpha^2$, there can be at most $n_\alpha^2-1$ strict increases.
We can therefore choose a basis represented by words of degree at most $n_\alpha^2\delta_v$.
Tensoring these bases over $v\in\Children{\alpha}$ gives a basis of $\JoinAlg{\alpha}$ represented in degree at most $n_\alpha^2\delta_\alpha$.

At an internal party $v$, tensoring the bases of its incoming source algebras gives representatives $b_0^v=\id,b_1^v,\ldots,b_{m_v}^v$ of a basis of $\InAlg{v}$, each of degree at most $\sum_{\alpha\in\Children{v}}n_\alpha^2\delta_\alpha$.
Multiplication in $\InAlg{v}$ is determined by expressing each product of two basis elements in this basis.
Thus, for each $i,j$, there are coefficients $c_{ij,k}\in\mathbb{C}$ such that the polynomial
\begin{align}\label{eq:FlatnessBasisRelationDegrees}
    q_{ij}\coloneqq (b_i^v)^*b_j^v-\sum_k c_{ij,k}b_k^v, \qquad \deg(q_{ij})\leq2\sum_{\alpha\in\Children{v}}n_\alpha^2\delta_\alpha,
\end{align}
evaluates to zero in $\InAlg{v}$.
To make the transfer map respect $q_{ij}=0$, the proof of \zcref{lem:SourceTransferCharacterReconstruction} uses $q_{ij}^*q_{ij}$, whose degree is at most $4\sum_{\alpha\in\Children{v}}n_\alpha^2\delta_\alpha$.
To recover the data of $L_d$ for $w$ with $\deg(w)\leq2d$, the relations $q(w)=w-\sum_k\lambda_k(w)b_k^v$ expressing $w$ in the basis similarly require $\deg(q(w)^*q(w))\leq4\max\{d,\sum_{\alpha\in\Children{v}}n_\alpha^2\delta_\alpha\}$.
A transfer-map or root-functional symbol adds one to the degree, so we set
\begin{align}\label{eq:FlatnessPartyDegreeBound}
    \delta_v\coloneqq1+4\max\left\{d,\sum_{\alpha\in\Children{v}}n_\alpha^2\delta_\alpha\right\}
\end{align}
for every internal party, including the root.

By construction, these bounds increase towards $r$:
\begin{align}\label{eq:FlatnessDegreeMonotonicity}
    \delta_v\leq\delta_{\Parent{v}}\leq\delta_r \quad(v\neq r), \qquad \delta_\alpha\leq\delta_{\Parent{\alpha}}\leq\delta_r \quad(\alpha\in\SourceVertices{\Tree}).
\end{align}
The theorem below uses an extension of length $\delta\geq2\delta_r$.

\begin{theorem}[Character extension and finite-dimensional extraction]\label{thm:SourceTransferFlatExtraction}
Let $L_{d+\delta}\in\fLsf_{d+\delta}(\Tree,r)$ be a $\delta$-flat extension of $L_d$, and let $\delta_r$ be defined recursively by \zcref{eq:FlatnessSourceDegreeBound,eq:FlatnessPartyDegreeBound}.
Suppose that $\delta\geq2\delta_r$ and that $L_{d + \delta}$ satisfies the $*$-multiplicativity condition
\begin{align}\label{eq:FlatExtractionScalarMultiplicativity}
    L_{d+\delta}(s^*u)=\overline{L_{d+\delta}(s)}L_{d+\delta}(u) \qquad(s,u\in\thinTreeS^{d+\delta}).
\end{align}

Then $L_d$ extends to a character $K:\thinTreeS\to\mathbb{C}$ nonnegative on $\fQ(\Tree,r)$.
Consequently, there is a finite-dimensional source-transfer realization $\cR$ whose source states have GNS spaces $\cH_\alpha$ with $\dim\cH_\alpha=n_\alpha$ for every source $\alpha$, and which satisfies
\begin{align}\label{eq:FlatExtractionMomentRecovery}
    f(\cR)=K(f)=L_d(f) \qquad(f\in\thinTreeS^{2d}).
\end{align}
If $L_{d+\delta}\in\fLsf_{d+\delta}(\Tree,r,p)$ and all correlation polynomials have degree at most $2d$, this realization reproduces $p$.

For an objective $f\in\thinTreeS$ with $f^*=f$ and $\deg(f)\leq2d$, the reconstructed finite-dimensional realization $\cR$ minimizes (respectively maximizes) $f$ over all source-transfer realizations on $(\Tree,r)$ whenever $L_{d+\delta}$ minimizes (respectively maximizes) $L(f)$ over $L\in\fLsf_{d+\delta}(\Tree,r)$.
In the minimization case,
\begin{align}\label{eq:FlatExtractionOptimalValue}
    \beta_{d+\delta}^{\min}(f)=L_{d+\delta}(f)=f(\cR)=\beta^{\min}(f).
\end{align}
The analogous equality holds with $\min$ replaced by $\max$ for maximization.
\end{theorem}

Note that the $*$-multiplicativity condition in \zcref{eq:FlatExtractionScalarMultiplicativity} is equivalent to
\begin{align}\label{eq:FlatExtractionScalarRank}
    \operatorname{rank}\left[L_{d+\delta}(s^*u)\right]_{s,u\in\Monomial{d+\delta}{\thinTreeS}}=1.
\end{align}
Since $L_{d+\delta}(\id)=1$, rank one gives $\det\bigl(\begin{smallmatrix}1&L_{d+\delta}(u)\\L_{d+\delta}(s^*)&L_{d+\delta}(s^*u)\end{smallmatrix}\bigr)=0$, so $L_{d+\delta}(s^*u)=\overline{L_{d+\delta}(s)}L_{d+\delta}(u)$ for every pair of scalar monomials $s,u$ indexing this matrix.
This condition could be satisfied, for example, when $L_d$ corresponds to an extreme optimal solution.
Note, however, that in practice interior-point SDP solvers converge to optimal solutions of maximum rank~\cite[Remark~6.11]{klep2024state}.

\begin{proof}
\emph{Finite-dimensional Gram representations.}
A Gram decomposition of $H_{d+\delta}^\alpha(L_{d+\delta})$ gives vectors $\ket{q}_\alpha$ with $\langle p\mid q\rangle_\alpha=L_{d+\delta}(\varsigma_\alpha(p^*q))$.
By flatness,
\begin{align}\label{eq:FlatExtractionGramSpaces}
    \cH_\alpha=\operatorname{span}\{\ket{q}_\alpha:\deg(q)\leq d\}, \qquad \dim\cH_\alpha=n_\alpha, \qquad \ket{\Omega_\alpha}=\ket{\id}_\alpha.
\end{align}
Let $g$ be a local generator of $\JoinPolyAlg{\alpha}$ with $\deg(g)\leq\delta_v$, and let $p,q$ be monomials of degree at most $d$.
Since $\deg(g^*p)\leq\delta_v+d\leq d+\delta$, the vector $\ket{g^*p}_\alpha$ is defined, and
\begin{align}\label{eq:FlatGNSAdjointShift}
    \braket{p}{gq}_\alpha=L_{d+\delta}\bigl(\varsigma_\alpha(p^*gq)\bigr)=\braket{g^*p}{q}_\alpha.
\end{align}
In particular, $\ket{q}_\alpha=0$ forces $\braket{p}{gq}_\alpha=0$ for every such $p$, hence $\ket{gq}_\alpha=0$ by \zcref{eq:FlatExtractionGramSpaces}.
Multiplication $\pi_\alpha(g)\ket{q}_\alpha\coloneqq\ket{gq}_\alpha$ is therefore well defined on $\cH_\alpha$, and \zcref{eq:FlatGNSAdjointShift} gives $\pi_\alpha(g)^*=\pi_\alpha(g^*)$.
Since $\deg(g_1g_2q)\leq2\delta_v+d\leq2\delta_r+d\leq d+\delta$ by \zcref{eq:FlatnessDegreeMonotonicity} and $\delta\geq2\delta_r$, we also have $\pi_\alpha(g_1)\pi_\alpha(g_2)=\pi_\alpha(g_1g_2)$ for any two such generators.
The local generators $g_v$ with $\deg(g_v)\leq\delta_v$ thus give a $*$-representation $\pi_\alpha$ of a unital $*$-subalgebra of $\JoinPolyAlg{\alpha}$ on the $n_\alpha$-dimensional space $\cH_\alpha$; compare the flat GNS constructions of~\cite[Theorem~1.69]{burgdorf2016optimization} and~\cite[Sections~6.3 and~6.5]{klep2024state}.
Iterating this identity from $\ket{\Omega_\alpha}=\ket{\id}_\alpha$, we obtain $\pi_\alpha(q)\ket{\Omega_\alpha}=\ket{q}_\alpha$ for every polynomial $q$ in this subalgebra with $\deg(q)\leq d+\delta$.

For $s\in\thinTreeS^{d+\delta}$, the Gram inner product and \zcref{eq:FlatExtractionScalarMultiplicativity} give
\begin{align}\label{eq:FlatExtractionScalarVector}
    \bigl\|\ket{s}_\alpha-L_{d+\delta}(s)\ket{\Omega_\alpha}\bigr\|^2=L_{d+\delta}(s^*s)-\bigl|L_{d+\delta}(s)\bigr|^2=0.
\end{align}
Hence $\ket{s}_\alpha=L_{d+\delta}(s)\ket{\Omega_\alpha}$.
For a scalar monomial $s$ and a source monomial $q\in\JoinPolyAlg{\alpha}$ with $\deg(sq)\leq d$, multiplication gives $\ket{sq}_\alpha=\pi_\alpha(q)\ket{s}_\alpha=L_{d+\delta}(s)\pi_\alpha(q)\ket{\Omega_\alpha}$.
Every monomial in $\fatS_\alpha$ has the form $sq$, so $\cH_\alpha=\operatorname{span}\{\pi_\alpha(q)\ket{\Omega_\alpha}:q\in\Monomial{d}{\JoinPolyAlg{\alpha}}\}$.

\medskip
\noindent
\emph{Constructing a finite-dimensional source-transfer realization.}
We follow the construction in the proof of \zcref{lem:SourceTransferCharacterReconstruction}.
However, here we must obtain its positivity and source-replacer identities from the finite-level SDP.
Define
\begin{equation}\label{eq:FlatExtractionSourceAlgebras}
    \begin{aligned}
        \cA_\alpha^v&\coloneqq C^*(\pi_\alpha(g_v):g_v\text{ a local generator},\ \deg(g_v)\leq\delta_v)\subset\Bof{\cH_\alpha},\\
        \JoinAlg{\alpha}&\coloneqq\bigotimes_{v\in\Children{\alpha}}^{\max}\cA_\alpha^v, \qquad
        \InAlg{v}\coloneqq\bigotimes_{\alpha\in\Children{v}}^{\max}\JoinAlg{\alpha}.
    \end{aligned}
\end{equation}
The source states, incoming product states, and incoming polynomial evaluation into $\InAlg{v}$ are defined by
\begin{equation}\label{eq:FlatExtractionProductStates}
    \begin{aligned}
        \omega_\alpha\!\left(\bigotimes_{v\in\Children{\alpha}}x_v\right)&\coloneqq\sandwich{\Omega_\alpha}{\prod_{v\in\Children{\alpha}}x_v}{\Omega_\alpha}, &&x_v\in\cA_\alpha^v,\\
        \InState{v}\!\left(\bigotimes_{\alpha\in\Children{v}}x_\alpha\right)&\coloneqq\prod_{\alpha\in\Children{v}}\omega_\alpha(x_\alpha), &&x_\alpha\in\JoinAlg{\alpha},\\
        \pi_v^{\mathrm{in}}\!\left(\bigotimes_{\alpha\in\Children{v}}\bigotimes_{u\in\Children{\alpha}}q_\alpha^u\right)&\coloneqq\bigotimes_{\alpha\in\Children{v}}\bigotimes_{u\in\Children{\alpha}}\pi_\alpha(q_\alpha^u).
    \end{aligned}
\end{equation}
Here $q_\alpha^u$ is a polynomial in the local generators $g_u\in\fA_\alpha^u$ with $\deg(g_u)\leq\delta_u$.

Since $\InAlg{v}$ is finite-dimensional, for each internal party $v$ we can find polynomials $b_0^v=\id,\ldots,b_{m_v}^v$ with $\deg(b_i^v)\leq\sum_{\alpha\in\Children{v}}n_\alpha^2\delta_\alpha$ such that $\{\pi_v^{\mathrm{in}}(b_i^v)\}_{i=0}^{m_v}$ is a basis of $\InAlg{v}$.

As in \zcref{eq:FlatnessBasisRelationDegrees}, for each $i,j$ we write $q_{ij}\coloneqq(b_i^v)^*b_j^v-\sum_k c_{ij,k}b_k^v$, with coefficients chosen so that $\pi_v^{\mathrm{in}}(q_{ij})=0$.
Likewise, for each word $w\in\InPolyAlg{v}$ with $\deg(w)\leq2d$, we write $q(w)\coloneqq w-\sum_k\lambda_k(w)b_k^v$ with $\pi_v^{\mathrm{in}}(q(w))=0$.
By the definition of $\delta_v$, every $q=q_{ij}$ or $q=q(w)$ satisfies
\begin{equation}\label{eq:FlatExtractionConstraintDegrees}
    \begin{aligned}
        \deg(q)&\leq(\delta_v-1)/2, & 1+\deg(q^*q)&\leq\delta_v,\\
        d+\deg(q)&\leq d+\delta-1, & 2d+1+\deg(q^*q)&\leq2(d+\delta).
    \end{aligned}
\end{equation}

Define the transfer maps and root functionals on these bases by
\begin{equation}\label{eq:FlatExtractionBasisEvaluations}
    \begin{aligned}
        \instr{v}{a_v}{x_v}\!\left(\pi_v^{\mathrm{in}}(b_i^v)\right)&\coloneqq\pi_{\Parent{v}}\!\left(\formalInstr{v}{a_v}{x_v}(b_i^v)\right) \quad(v\neq r),\\
        \rstate{r}{a_r}{x_r}\!\left(\pi_r^{\mathrm{in}}(b_i^r)\right)&\coloneqq L_{d+\delta}\!\left(\rPolystate{a_r}{x_r}(b_i^r)\right),
    \end{aligned}
\end{equation}
and extend linearly.
These extensions are unique because $\{\pi_v^{\mathrm{in}}(b_i^v)\}_i$ is a basis of $\InAlg{v}$.
These are the finite-basis versions of \zcref{eq:CharacterTransferMapDense,eq:CharacterRootFunctional}.
Fix $v\neq r$ and let $q$ be any of the relations $q_{ij}$ or $q(w)$ above.
The matrix $C_{d+\delta-1}(\formalInstr{v}{a_v}{x_v};L_{d+\delta})\succeq0$ can be tested on $(p,\id)$ and $(p,q)$ for every $\deg(p)\leq d$, by \zcref{eq:FlatExtractionConstraintDegrees}.
Since these $\ket{p}_{\Parent{v}}$ span $\cH_{\Parent{v}}$, we obtain
\begin{align}\label{eq:FlatExtractionOperatorCP}
    \begin{pmatrix}
        \pi_{\Parent{v}}(\formalInstr{v}{a_v}{x_v}(\id))&\pi_{\Parent{v}}(\formalInstr{v}{a_v}{x_v}(q))\\
        \pi_{\Parent{v}}(\formalInstr{v}{a_v}{x_v}(q))^*&\pi_{\Parent{v}}(\formalInstr{v}{a_v}{x_v}(q^*q))
    \end{pmatrix}\succeq0.
\end{align}

Likewise, $\deg(h_{x_v}^v(p_1,t,p_2))\leq2(d+\delta)$ for $\deg(p_1),\deg(p_2)\leq d$ and $t=b_i^v$ or $t=q^*q$, where $q=q_{ij}$ or $q=q(w)$.
The corresponding source-replacer equalities and \zcref{eq:FlatExtractionScalarMultiplicativity} therefore give
\begin{align}\label{eq:FlatExtractionOperatorReplacer}
    \sum_{a_v}\pi_{\Parent{v}}\!\left(\formalInstr{v}{a_v}{x_v}(t)\right)=\InState{v}\!\left(\pi_v^{\mathrm{in}}(t)\right)I_{\cH_{\Parent{v}}}, \qquad t=b_i^v\text{ or }q^*q.
\end{align}

These are the finite-level versions of \zcref{eq:CharacterFormalTransferCP,eq:CharacterFormalTransferReplacer}.
The well-definedness argument in the proof of \zcref{lem:SourceTransferCharacterReconstruction} therefore gives
\begin{align}\label{eq:FlatExtractionTransferKernel}
    \pi_v^{\mathrm{in}}(q)=0 \quad\Longrightarrow\quad \pi_{\Parent{v}}\!\left(\formalInstr{v}{a_v}{x_v}(q)\right)=0 \qquad(q=q_{ij}\text{ or }q=q(w)).
\end{align}
By \zcref{eq:FlatExtractionTransferKernel}, we have $\instr{v}{a_v}{x_v}(\pi_v^{\mathrm{in}}(t))=\pi_{\Parent{v}}(\formalInstr{v}{a_v}{x_v}(t))$ for $t=(b_i^v)^*b_j^v$ and for every word $t$ with $\deg(t)\leq2d$.

The condition $C_{d+\delta-1}(\formalInstr{v}{a_v}{x_v};L_{d+\delta})\succeq0$ also implies
\begin{align}\label{eq:FlatExtractionBasisCP}
    \left[\instr{v}{a_v}{x_v}\!\left(\pi_v^{\mathrm{in}}((b_i^v)^*b_j^v)\right)\right]_{i,j}
    =\left[\pi_{\Parent{v}}\!\left(\formalInstr{v}{a_v}{x_v}((b_i^v)^*b_j^v)\right)\right]_{i,j}\succeq0.
\end{align}
Thus $\instr{v}{a_v}{x_v}$ is completely positive, since $\{\pi_v^{\mathrm{in}}(b_i^v)\}_i$ is a basis of $\InAlg{v}$.
The source-replacer identity extends from the basis by linearity.
At the root, the same basis argument from \zcref{lem:SourceTransferCharacterReconstruction} gives $\rstate{r}{a_r}{x_r}(x^*x)\geq0$ for $x\in\InAlg{r}$ and $\sum_{a_r}\rstate{r}{a_r}{x_r}=\InState{r}$.
Together with the leaf POVMs $\pi_{\Parent{v}}(\meas{v}{a_v}{x_v})$, whose positivity follows directly from the POVM positivity constraints, these objects form a source-transfer realization $\cR$.
It is finite-dimensional since each $\cA_\alpha^v\subset\Bof{\cH_\alpha}$ is finite-dimensional.
The source GNS dimension is exactly $n_\alpha$: by the first part of the proof, the representation of $\JoinAlg{\alpha}$ on $\cH_\alpha$ is cyclic with vector $\ket{\Omega_\alpha}$, so it is a GNS representation of $\omega_\alpha$.

\medskip
\noindent
\emph{Character and agreement with $L_d$.}
Define $K(f)\coloneqq f(\cR)$ for $f\in\thinTreeS$ using \zcref{def:EvaluationMultiStatePoly}.
As in the necessity part of \zcref{thm:SourceTransferSDPConvergence}, this is a character satisfying $K(\fQ(\Tree,r))\geq0$.

Finally, we need to check that $K$ agrees with $L_d$ on $\thinTreeS^{2d}$.
We first verify, from the leaves towards the root, that $w(\cR)=\pi_\alpha(w)$ for every $w\in\Monomial{2d}{\fA_\alpha^v}$ and $v\in\Children{\alpha}$.
This holds for the leaf generators by their definition.
At an internal party $v\neq r$, suppose it holds for the incoming factors, so that $w(\cR)=\pi_v^{\mathrm{in}}(w)$ for every $w\in\Monomial{2d}{\InPolyAlg{v}}$.
Using $q(w)=w-\sum_k\lambda_k(w)b_k^v$, \zcref{eq:FlatExtractionBasisEvaluations,eq:FlatExtractionTransferKernel} give
\begin{equation}\label{eq:FlatExtractionWordRecovery}
    \begin{aligned}
        \formalInstr{v}{a_v}{x_v}(w)(\cR)&=\instr{v}{a_v}{x_v}\!\left(\pi_v^{\mathrm{in}}(w)\right) =\sum_k\lambda_k(w)\pi_{\Parent{v}}\!\left(\formalInstr{v}{a_v}{x_v}(b_k^v)\right) = \pi_{\Parent{v}}\!\left(\formalInstr{v}{a_v}{x_v}(w)\right).
    \end{aligned}
\end{equation}
For each generator $g_v=\formalInstr{v}{a_v}{x_v}(t)$ with $\deg(t)\leq2d-1$, \zcref{eq:FlatExtractionWordRecovery} gives $g_v(\cR)=\pi_\alpha(g_v)$, where $\alpha=\Parent{v}$.
Since both maps preserve products and adjoints, this extends to $w(\cR)=\pi_\alpha(w)$ for every $w\in\Monomial{2d}{\fA_\alpha^v}$.
Consequently, by definition, for $p\in\JoinPolyAlg{\alpha}$ with $\deg(p)\leq2d$,
\begin{align}\label{eq:FlatExtractionSourceScalarRecovery}
    K(\varsigma_\alpha(p))=\sandwich{\Omega_\alpha}{\pi_\alpha(p)}{\Omega_\alpha}
    =\langle\Omega_\alpha\mid p\rangle_\alpha=L_{d+\delta}(\varsigma_\alpha(p)).
\end{align}

At the root, applying \zcref{eq:FlatExtractionBasisEvaluations} to $w=\sum_k\lambda_k(w)b_k^r+q(w)$ and using $L_{d+\delta}(\rPolystate{a_r}{x_r}(q(w)))=0$, as in \zcref{eq:FlatExtractionTransferKernel}, gives
\begin{align}\label{eq:FlatExtractionRootScalarRecovery}
    K(\rPolystate{a_r}{x_r}(w))=\sum_k\lambda_k(w)L_{d+\delta}(\rPolystate{a_r}{x_r}(b_k^r))
    =L_{d+\delta}(\rPolystate{a_r}{x_r}(w))
\end{align}
whenever $\deg(\rPolystate{a_r}{x_r}(w))\leq2d$.

Thus $K(s)=L_{d+\delta}(s)$ for $s=\varsigma_\alpha(q)$ or $s=\rPolystate{a_r}{x_r}(w)$ with $\deg(s)\leq2d$.
For a scalar monomial $s_1\cdots s_\ell$ of total degree at most $2d$, multiplicativity of $K$ and \zcref{eq:FlatExtractionScalarMultiplicativity} give $K(s_1\cdots s_\ell)=\prod_jL_{d+\delta}(s_j)=L_{d+\delta}(s_1\cdots s_\ell)$.
Here we apply \zcref{eq:FlatExtractionScalarMultiplicativity} successively to $(s_1\cdots s_j)^*$ and $s_{j+1}$; their degrees are at most $2d\leq d+\delta$, so each application is within its stated degree bounds.
By linearity,
\begin{align}\label{eq:FlatExtractionLowDegreeAgreement}
    K|_{\thinTreeS^{2d}}=L_d.
\end{align}
Therefore $f(\cR)=K(f)=L_d(f)$ for every $f\in\thinTreeS^{2d}$, including any prescribed correlation polynomials.
Finally, \zcref{eq:SourceTransferOptimizationOuterBounds} implies that $\beta_{d+\delta}^{\min}(f)\leq\beta^{\min}(f)\leq f(\cR)=L_{d+\delta}(f)$.
If $L_{d+\delta}$ minimizes $L(f)$ over $L\in\fLsf_{d+\delta}(\Tree,r)$, $\beta_{d+\delta}^{\min}(f)=L_{d+\delta}(f)$, so all inequalities are equalities; the maximization case follows by considering $-f$.
\end{proof}

Thus, a flat extension as in \zcref{def:SourceTransferFlatness} satisfying the hypotheses of \zcref{thm:SourceTransferFlatExtraction} leads to a finite-dimensional source-transfer realization.
By \zcref{rem:FiniteDimensionalTransferToMixed}, we also obtain a finite-dimensional mixed model and, equivalently, a finite-dimensional tensor-product model realizing the same correlation.

Conversely, it is straightforward to check that a finite-dimensional source-transfer or mixed realization leads to a solution satisfying the flatness condition.
Nonetheless, we remark that the existence of a finite-dimensional source-transfer realization that is optimal for a given objective does not guarantee that an SDP solver will find a flat optimal solution in practice.

\subsection{A computability consequence}\label{sec:ComputabilitySourceTransferSDP}
The convergence result also has a consequence for the computability of quantum network game values.
A game $G$ specifies a distribution over measurement-setting sets $\mathsf{X}_v$ and a rule for accepting the parties' outcomes.
We use the usual finite verifier description of games, as in \cite[Definition~6.1]{lin2025mipcocore}.
Write $v_{\Tree}^{\mathrm{mix}}(G)$ for the supremum of its winning probability over mixed quantum realizations on $\Tree$.

The game-value promise problem stated below is already $\mathsf{coRE}$-hard for trees consisting of two parties connected to a single quantum source.
Indeed, their mixed model is exactly the commuting-operator model for Bell correlations, so hardness follows from \cite[Corollary~6.10]{lin2025mipcocore}.
We now prove $\mathsf{coRE}$-completeness for arbitrary finite tree networks.

\begin{corollary}[Computability of quantum tree game values]\label{cor:MixedTreeGameComputability}
Given a finite tree network $\Tree$ and a game $G$ with finite setting and outcome sets, the promise problem of distinguishing
\begin{align}\label{eq:MixedTreeGamePromise}
    \mathrm{YES}:\quad v_{\Tree}^{\mathrm{mix}}(G)=1,
    \qquad
    \mathrm{NO}:\quad v_{\Tree}^{\mathrm{mix}}(G)\leq\tfrac12
\end{align}
is $\mathsf{coRE}$-complete.
In particular, there exists an algorithm that detects every NO instance in finite time and never rejects a YES instance.
\end{corollary}
\begin{proof}
In view of the $\mathsf{coRE}$-hardness established above, it remains to prove membership in $\mathsf{coRE}$.
Fix any root party $r$, and let $f_{G,r}$ be the multi-state polynomial representing the winning probability.
At each level $d$, we test whether there exists $L_d\in\fLsf_d(\Tree,r)$ satisfying $L_d(f_{G,r})\geq 3/4$.
It is in principle decidable with quantifier elimination over the reals (the Tarski--Seidenberg theorem), for instance via cylindrical algebraic decomposition (CAD).
We output NO when the test is infeasible.
A YES instance admits realizations with winning probability above $3/4$, so it is never rejected.
For a NO instance, convergence implies that $\beta_{d,r}^{\max}(f_{G,r})<3/4$ at some finite level, where the test is infeasible.
This proves membership in $\mathsf{coRE}$.
\end{proof}

\section{Convergence of inflation-NPA on quantum tree networks}\label{sec:InflationTreeConvergence}
Quantum inflation and the inflation-NPA hierarchy~\cite{wolfe2021quantum} are the main tools for characterizing general quantum network scenarios, but it remains open whether they converge to the correct quantum network set.
Alternatively, the authors of~\cite{ligthart2023convergent} introduce a modified hierarchy with explicit operators acting on the source algebras and bounds on the Schmidt rank of measurement operators, thereby restricting the models covered at each fixed bound.
Convergence of the original hierarchy is known for the bilocal scenario~\cite{ligthart2023inflation}, with an extension to bipartite-source star networks in~\cite[Corollary~3.3.17]{ligthart2024semidefinite}.
In this section, we advance this line of work by proving that inflation-NPA converges to the mixed quantum model on every finite source-party tree.

\subsection{POVM inflation tests on source-party graphs}\label{sec:InflationAlgebrasPOVM}
First, let $\Tree$ be a finite source-party graph without assuming that it is a tree.
The finite setting and outcome sets at party $v$ are $\mathsf X_v$ and $\mathsf A_v$, as in \zcref{def:SourceTransferModelTree}.
For $m\geq1$, write $[m]\coloneqq\{1,\ldots,m\}$.
Each inflated measurement operator for party $v$ carries a tuple $\mathbf i=(i_\alpha)_{\alpha\in\Neighbors{v}}\in[m]^{\Neighbors{v}}$, with one label $i_{\alpha}$ for each neighboring source $\alpha$.
For example, on a triangle with sources $\alpha=AB$, $\beta=BC$, and $\gamma=CA$, the measurement copies are $E^A_{a\mid x}(i_\alpha,i_\gamma)$, $E^B_{b\mid y}(i_\alpha,i_\beta)$, and $E^C_{c\mid z}(i_\beta,i_\gamma)$.
Two copies at $A$ are required to commute when both their $\alpha$ and $\gamma$ labels differ; copies at different parties always commute.

We follow the \emph{inflation algebra} construction in~\cite[Section~2.3]{ligthart2024semidefinite}.
Let $\cE_m(\Tree)$ be the universal unital $C^*$-algebra generated by the positive contractions $E^v_{a_v\mid x_v}(\mathbf i)$, subject to
\begin{equation}\label{eq:InflationPOVMRelations}
    \begin{aligned}
        E^v_{a_v\mid x_v}(\mathbf{i})&\geq0, \qquad \sum_{a_v\in\mathsf A_v}E^v_{a_v\mid x_v}(\mathbf i)=\id,\\
        [E^v_{a_v\mid x_v}(\mathbf i),E^w_{a_w\mid x_w}(\mathbf j)]&=0 \qquad
        \text{if }v\neq w\text{, otherwise }i_\alpha\neq j_\alpha\text{ for every }\alpha\in\Neighbors{v}.
    \end{aligned}
\end{equation}
Consider the source permutation group $\Gamma_m\coloneqq\prod_{\alpha\in\SourceVertices{\Tree}}S_m$, which acts by independently permuting the labels of each source.
For $g=(g_\alpha)_\alpha\in\Gamma_m$, we define automorphism $\alpha_g$ on $\cE_m(\Tree)$ by
\begin{align}\label{eq:InflationPermutationAction}
    \alpha_g\!\left(E^v_{a_v\mid x_v}(\mathbf i)\right)\coloneqq E^v_{a_v\mid x_v}\!\left((g_\alpha(i_\alpha))_{\alpha\in\Neighbors{v}}\right).
\end{align}

Let $e=(\mathbf{a} | \mathbf{x})$ denote a complete event, and let $D_e^{(t)}$ be the product of measurements acting on the $t$-th copy of the source:
\begin{align}\label{eq:InflationDiagonalEvent}
    D_e^{(t)}\coloneqq\prod_{v\in\PartyVertices{\Tree}}E^v_{a_v\mid x_v}\!\left((t)_{\alpha\in\Neighbors{v}}\right),\qquad t\in[m].
\end{align}

\begin{definition}[POVM inflation test]\label{def:InflationTestPOVM}
Let $\Tree$ be a finite source-party graph.
A correlation $p(\mathbf{a} | \mathbf{x})$ passes the level-$m$ inflation test on $\Tree$ if there is a state $\omega_m:\cE_m(\Tree)\to\mathbb C$ such that
\begin{equation}\label{eq:InflationStateConstraints}
    \begin{aligned}
        \omega_m\circ\alpha_g&=\omega_m, \quad \forall g\in\Gamma_m, && \text{(permutation invariance/symmetry)}\\
        \omega_m\!\left(\prod_{t=1}^kD_{e_t}^{(t)}\right)&=\prod_{t=1}^kp(e_t), \quad 1\leq k\leq m, &&\text{(diagonal independence factorization)}
    \end{aligned}
\end{equation}
for every choice of the complete events $e_1,\ldots,e_k$.
Equivalently by the universal property, one may use any unital $C^*$-algebra containing POVMs satisfying \zcref{eq:InflationPOVMRelations} and a state satisfying \zcref{eq:InflationStateConstraints}.
\end{definition}

\begin{remark}[Limit of POVM inflation tests]\label{rem:LimitInflationTest}
    The inclusion $\cE_m(\Tree)\subset\cE_{m+1}(\Tree)$ is canonical.
    We therefore define the inflation algebra in the limit and the full source-permutation group by
    \begin{align}\label{eq:InflationInfiniteAlgebra}
        \cE_\infty(\Tree)\coloneqq\overline{\bigcup_{m\geq1}\cE_m(\Tree)}^{\|\cdot\|},\qquad \Gamma\coloneqq\bigcup_{m\geq1}\Gamma_m.
    \end{align}
    This limit algebra is separable since it has countably many generators.
    
    If $p$ passes every finite inflation test, the standard compactness argument used in the proof of~\cite[Lemma~5]{ligthart2023convergent} gives a $\Gamma$-invariant state ${\omega}$ on $\cE_\infty(\Tree)$ satisfying the constraints in \zcref{eq:InflationStateConstraints}.
    Conversely, restricting such a state to $\cE_m(\Tree)$ gives a feasible state for every finite inflation test.
\end{remark}

At this stage, the test uses states on an operator algebra; the inflation-NPA hierarchy later introduced in \zcref{sec:InflationNPASDP} will also have a parameter on the degree of words.

\subsection{Invariant-state decomposition and permutation averages}\label{sec:InflationSymmetryLimits}
We now make some useful technical observations on source-permutation-invariant states, permutation group averages of operators, and the limits of completely positive maps.

We call $z\in\cE_\infty(\Tree)$ a \emph{finite-support element} if it belongs to the $C^*$-algebra generated by finitely many inflated measurement operators $E^v_{a_v\mid x_v}(\mathbf{i})$.
For each source $\alpha$, let $I_\alpha(z)$ be the set of labels $i_\alpha$ appearing in $\mathbf{i}$ of these generators, with $I_\alpha(z)=\emptyset$ if none of these generators involves the source $\alpha$.
Two finite-support elements $z,w$ have \emph{disjoint source labels} if their generating families can be chosen so that $I_\alpha(z)\cap I_\alpha(w)=\emptyset$ for every source $\alpha$.
The commutation relations in \zcref{eq:InflationPOVMRelations} then imply $[z,w]=0$.

\begin{lemma}[Existence of a factorizing inflation state with the prescribed correlation]\label{lem:InflationFactorizingState}
Suppose that $p$ passes the level-$m$ test in \zcref{def:InflationTestPOVM} for every $m$.
Then there is a $\Gamma$-invariant state $\omega$ on $\cE_\infty(\Tree)$ satisfying \zcref{eq:InflationStateConstraints} and
\begin{align}\label{eq:InflationDisjointFactorization}
    \omega(zw)=\omega(z)\omega(w)
\end{align}
for all finite-support elements $z,w$ with disjoint source labels.
\end{lemma}
\begin{proof}
By the de Finetti theorem~\cite[Theorem~2.3.2]{ligthart2024semidefinite}, there is a probability measure $\mu$ on $\Gamma$-invariant states $\varphi$ satisfying \zcref{eq:InflationDisjointFactorization}, whose average satisfies \zcref{eq:InflationStateConstraints}.
For every complete event $e$, factorization and the first two diagonal constraints $D_e^{(1)}$, $D_e^{(2)}$ give
\begin{align}\label{eq:InflationComponentCorrelationVariance}
    \int\bigl(\varphi(D_e^{(1)})-p(e)\bigr)^2\,\mathrm{d}\mu(\varphi)=0.
\end{align}
Thus, for each $e$, the vanishing integral implies $\varphi(D_e^{(1)})=p(e)$ for $\mu$-almost every $\varphi$.
Since there are finitely many complete events $e$, these equalities hold simultaneously on a set of full $\mu$-measure.
Choose $\omega$ from this set.
Its symmetry and factorization then imply all of \zcref{eq:InflationStateConstraints}.
\end{proof}

For $z\in\cE_\infty(\Tree)$, define the finite group average
\begin{align}\label{eq:InflationGroupAverage}
    \mathsf{Av}_m(z)\coloneqq\frac1{|\Gamma_m|}\sum_{g\in\Gamma_m}\alpha_g(z).
\end{align}
For two finite sets of source labels, we first bound the probability that a random permutation of one set intersects the other.
We then use this bound to estimate the commutator of $\mathsf{Av}_m(z)$ with another operator and its difference from $\omega(z)I$ on the GNS vector.

\begin{lemma}[Source-label intersection and average estimates]\label{lem:InflationCollisionBounds}
Let $I_\alpha,J_\alpha\subset[m]$ be finite label sets for each source $\alpha$, and choose $g\in\Gamma_m$ uniformly.
Then
\begin{align}\label{eq:InflationCollisionProbability}
    \Pr\!\left\{g_\alpha(I_\alpha)\cap J_\alpha\neq\emptyset\text{ for some }\alpha\right\}\leq\frac1m\sum_\alpha|I_\alpha|\,|J_\alpha|.
\end{align}
In particular, if $z,w$ are finite-support elements and $[m]$ contains all their source labels, then
\begin{align}\label{eq:InflationAverageCommutator}
    \|[\mathsf{Av}_m(z),w]\|\leq\frac{2\|z\|\|w\|}{m}\sum_\alpha|I_\alpha(z)|\,|I_\alpha(w)|.
\end{align}
If $\omega$ is $\Gamma_m$-invariant and satisfies the disjoint-source factorization condition in \zcref{eq:InflationDisjointFactorization}, let $(\cH,\pi,\ket{\Omega})$ be its GNS representation.
Then
\begin{align}\label{eq:InflationAverageVectorBound}
    \|\pi(\mathsf{Av}_m(z-\omega(z) \id))\ket{\Omega}\|^2\leq\frac{\|z-\omega(z) \id\|^2}{m}\sum_\alpha|I_\alpha(z)|^2.
\end{align}
\end{lemma}
\begin{proof}
Since $g_\alpha$ is chosen uniformly at random, the image of each fixed $i\in I_\alpha$ is uniformly distributed and $\Pr(g_\alpha(i)=j)=1/m$ for every $j\in J_\alpha$.
Thus, by the union bound, we have
\begin{align}
    \Pr\!\left(\exists\alpha:\ g_\alpha(I_\alpha)\cap J_\alpha\neq\emptyset\right)
    \leq\sum_\alpha\sum_{i\in I_\alpha}\sum_{j\in J_\alpha}\Pr(g_\alpha(i)=j)
    =\frac1m\sum_\alpha|I_\alpha|\,|J_\alpha|.
\end{align}

The elements $\alpha_g(z)$ and $w$ commute when their source index sets are disjoint.
In the case that the index sets are not disjoint, we have $\|[\alpha_g(z),w]\| \leq\|\alpha_g(z)w\|+\|w\alpha_g(z)\|\leq2\|z\|\|w\|.$
Thus, \zcref{eq:InflationAverageCommutator} follows, since the number of permutations $g\in\Gamma_m$ for which $\alpha_g(z)$ and $w$ have intersecting source-label sets, divided by $|\Gamma_m|$, is precisely the probability bounded in \zcref{eq:InflationCollisionProbability}.

For the last estimate, write $z_0\coloneqq z-\omega(z) \id$ so that we have $\omega(z_0)=0$.
By $\Gamma_m$-invariance, we compute
\begin{equation}\label{eq:InflationAverageSecondMoment}
    \begin{aligned}
        \|\pi(\mathsf{Av}_m(z_0))\ket{\Omega}\|^2 &= \omega(\mathsf{Av}_m(z_0)^* \mathsf{Av}_m(z_0)) \\
        &=\frac1{|\Gamma_m|^2}\sum_{g,h\in\Gamma_m}\omega\!\left(\alpha_g(z_0)^*\alpha_h(z_0)\right)\\
        &=\frac1{|\Gamma_m|^2}\sum_{g,h\in\Gamma_m}\omega\!\left(z_0^*\alpha_{g^{-1}h}(z_0)\right) =\frac1{|\Gamma_m|}\sum_{h\in\Gamma_m}\omega\!\left(z_0^*\alpha_h(z_0)\right).
    \end{aligned}
\end{equation}
Note that the summand $\omega\!\left(z_0^*\alpha_h(z_0)\right) = \omega\!\left(z_0^*\right) \omega\!\left(\alpha_h(z_0)\right) = 0$ when $z_0$ and $\alpha_h(z_0)$ have disjoint source-label sets.
Otherwise, $| \omega\!\left(z_0^*\alpha_h(z_0)\right) | \leq \| z_0 \|^2$, and we are done by \zcref{eq:InflationCollisionProbability} with $I_\alpha=J_\alpha=I_\alpha(z)$.
\end{proof}

\begin{lemma}[Permutation averages converge to scalars in the GNS representation]\label{lem:InflationScalarAverages}
Let $\omega$ be a $\Gamma$-invariant state satisfying \zcref{eq:InflationDisjointFactorization}, with GNS representation $(\cH,\pi,\ket{\Omega})$.
Let $U_g$ be the unitaries associated with $\alpha_g$ such that $U_g\pi(z)\ket{\Omega}\coloneqq\pi(\alpha_g(z))\ket{\Omega}$.
Then the vectors fixed by every $U_g$ are exactly the scalar multiples of $\ket{\Omega}$.
Moreover, for every $z\in\cE_\infty(\Tree)$, we have strong operator convergence
\begin{align}\label{eq:InflationScalarStrongLimit}
    \pi(\mathsf{Av}_m(z))\to\omega(z)I_{\cH}\qquad\text{as }m\to\infty.
\end{align}

Define the von Neumann algebra and associated state extension by
\begin{align}\label{eq:InflationAmbientVonNeumann}
    \cM_\infty\coloneqq\pi(\cE_\infty(\Tree))''\subset\Bof{\cH},\qquad \omega(Q)\coloneqq\sandwich{\Omega}{Q}{\Omega}\quad(Q\in\cM_\infty).
\end{align}
Then for every $Q\in\cM_\infty$, we also have
\begin{align}\label{eq:InflationVonNeumannVectorAverage}
    \left(\frac1{|\Gamma_m|}\sum_{g\in\Gamma_m}U_gQU_g^*\right)\ket{\Omega}\longrightarrow\omega(Q)\ket{\Omega}.
\end{align}
\end{lemma}
\begin{proof}
By \zcref{lem:InflationCollisionBounds}, we have
\begin{align}\label{eq:InflationAverageOnGNSVector}
    \pi(\mathsf{Av}_m(z))\ket{\Omega}\to\omega(z)\ket{\Omega}
\end{align}
for every finite-support element $z$.
Since these elements are norm dense in $\cE_\infty(\Tree)$ and $\|\mathsf{Av}_m\|\leq1$, this convergence extends by continuity to every $z\in\cE_\infty(\Tree)$.

Consider the finite group averages of $U_g$
\begin{align}\label{eq:InflationFiniteGroupProjection}
    P_m\coloneqq\frac1{|\Gamma_m|}\sum_{g\in\Gamma_m}U_g.
\end{align}
Each $P_m$ is the orthogonal projection onto the vectors fixed by every $U_g$ with $g\in\Gamma_m$.
We claim that $P_m$ converges strongly to $\ketbra{\Omega}{\Omega}$; i.e., $\ketbra{\Omega}{\Omega}$ is the projection onto vectors fixed by every $U_g$ with $g \in \Gamma$.
Indeed, \zcref{eq:InflationAverageOnGNSVector} gives
\begin{align}\label{eq:InflationInvariantProjection}
    P_m\pi(z)\ket{\Omega}
    =\pi(\mathsf{Av}_m(z))\ket{\Omega}
    \to\omega(z)\ket{\Omega}
    =\ketbra{\Omega}{\Omega}\pi(z)\ket{\Omega}.
\end{align}
Since these vectors are dense in $\cH$ and the $P_m$ are uniformly bounded, the convergence extends to all of $\cH$.
It follows that only scalar multiples of $\ket{\Omega}$ are fixed by every $U_g$.

For $Q\in\cM_\infty$, the identity $U_g^*\ket{\Omega}=\ket{\Omega}$ now gives
\begin{align}\label{eq:InflationVonNeumannAverageCalculation}
    \left(\frac1{|\Gamma_m|}\sum_{g\in\Gamma_m}U_gQU_g^*\right)\ket{\Omega}
    =P_mQ\ket{\Omega}
    \to\ketbra{\Omega}{\Omega}Q\ket{\Omega}
    =\omega(Q)\ket{\Omega},
\end{align}
proving \zcref{eq:InflationVonNeumannVectorAverage}.

It remains to show \zcref{eq:InflationScalarStrongLimit}.
First let $z,w$ be finite-support elements.
Thanks to the estimates given by \zcref{lem:InflationCollisionBounds}, we have
\begin{equation}\label{eq:InflationAverageDenseVectors}
    \begin{aligned}
        \|(\pi(\mathsf{Av}_m(z))-\omega(z)I)\pi(w)\ket{\Omega}\|
        \leq\|[\mathsf{Av}_m(z),w]\| +\|w\|\,\|(\pi(\mathsf{Av}_m(z))-\omega(z)I)\ket{\Omega}\| \to 0,
    \end{aligned}
\end{equation}
as $m \to \infty$.
By cyclicity of $\pi(w)\ket{\Omega}$, the above convergence extends to every vector in $\cH$, which we can extend further to every $z\in\cE_\infty(\Tree)$ by norm-density and continuity.
\end{proof}

The consideration of the von Neumann algebra $\cM_\infty$ above is due to the fact that a $C^*$-algebra need not contain the weak limits of its bounded sequences, which we need for our construction.
We finish the subsection by recording a general compactness statement for later convenience.

\begin{lemma}[Limits of completely positive contractions between separable spaces]\label{lem:InflationCPLimits}
Let $\cA$ be a separable $C^*$-algebra, let $\cH$ be a separable Hilbert space, and let $\cM\subset\Bof{\cH}$ be a von Neumann algebra.
For each $m\geq1$ and each index $\ell$ in a countable set, suppose that $F_m^\ell:\cA\to\cM$ is a completely positive contraction.
Then there exists a single increasing subsequence $m_k$ and completely positive contractions $F^\ell:\cA\to\cM$ such that, simultaneously for every $\ell$,
\begin{align}\label{eq:InflationPointUltraweakLimit}
    \sandwich{\xi}{F_{m_k}^\ell(q)}{\eta}\longrightarrow\sandwich{\xi}{F^\ell(q)}{\eta},\qquad q\in\cA,\quad \ket{\xi}, \ket{\eta}\in\cH.
\end{align}
\end{lemma}
\begin{proof}
By separability, choose a countable norm-dense set $\{q_t\}$ in $\cA$ and an orthonormal basis $\{e_i\}$ of $\cH$.
For every $\ell,t,i,j$, the sequence $\sandwich{e_i}{F_m^\ell(q_t)}{e_j}$ is bounded.
The standard diagonal extraction over these countably many sequences gives one subsequence $m_k$ for which they all converge.

Since $F_m^\ell$ are uniformly contractive, we can extend these limits by density to every $q\in\cA$ and $\ket{\xi},\ket{\eta}\in\cH$, defining linear contractions $F^\ell:\cA\to\Bof{\cH}$ satisfying \zcref{eq:InflationPointUltraweakLimit}.
Their values belong to $\cM$, since $\cM$ is closed in the weak operator topology.
Finally, for every positive operator-valued matrix $[q_{ij}]$ over $\cA$, the positive operator-valued matrices $[F_{m_k}^\ell(q_{ij})]$ converge in the weak operator topology to $[F^\ell(q_{ij})]$, which is therefore positive.
Thus each $F^\ell$ is completely positive.
\end{proof}

\subsection{From inflation tests to the source-transfer model}\label{sec:InflationSourceTransferEquivalence}
We now assume that $\Tree$ is a source-party tree and fix a root party $r$.
To obtain a source-transfer model, an appropriate construction for internal transfer maps in \zcref{eq:TreeTransferMapType} is key.
Recall that when reconstructing these transfer maps from the mixed model, we use \zcref{eq:TreeMixedToTransferMap}, which takes a partial sandwich of the expression $\bigl(\Meas{v}{a_v}{x_v}\bigr)^{1/2}(q\otimes I)\bigl(\Meas{v}{a_v}{x_v}\bigr)^{1/2}$.
To reconstruct transfer maps from inflation, we average/symmetrize over the source labels of the input instead of partial sandwiching.

To build intuition of this averaging construction, consider the line-of-four network $1-\alpha-2-\beta-3-\gamma-4$ rooted at $3$ as in \zcref{fig:SourceTransferLineOfFour}.
Let $q$ be an element generated by $E^1_{a_1\mid x_1}(1)$, and let $q^i$ be its source-$\alpha$ copy with label $i$ replacing $1$.
The transfer map at party $2$ is defined via a pointwise ultraweak limit along a subsequence of $m$:
\begin{align}\label{eq:InflationPathTransferAverage}
    \Phi^{2,j;m}_{a_2\mid x_2}(q)\coloneqq \frac1m\sum_{i=1}^m q^i E^2_{a_2\mid x_2}(i,j) = \frac1m\sum_{i=1}^m (E^2_{a_2\mid x_2}(i,j))^{1/2} q^i (E^2_{a_2\mid x_2}(i,j))^{1/2}.
\end{align}
It is completely positive by construction.
By POVM completeness, summing over $a_2$ leaves $1/m\sum_{i=1}^m q^i$, which we will show converges strongly to $\omega_{\alpha}(q) I$ thanks to \zcref{lem:InflationScalarAverages}, giving the desired source-replacer condition.

\begin{theorem}[Inflation, source-transfer, and mixed models on trees]\label{thm:InflationSourceTransferTreeEquivalence}
Let $\Tree$ be a finite source-party tree, with finite setting and outcome sets, and let $r\in\PartyVertices{\Tree}$ be a root party.
For a correlation $p(\mathbf a|\mathbf x)$, the following are equivalent:
\begin{enumerate}[label=(\roman*)]
    \item $p$ passes the level-$m$ POVM inflation test in \zcref{def:InflationTestPOVM} for every $m\geq1$;
    \item $p$ admits a source-transfer realization on $(\Tree,r)$ as in \zcref{def:SourceTransferModelTree};
    \item $p$ admits a mixed quantum realization on $\Tree$ as in \zcref{def:MixedQuantumModelSourcePartyGraph}.
\end{enumerate}
\end{theorem}

\begin{proof}
The equivalence of (ii) and (iii) is already established by \zcref{thm:SourceTransferTreeEquivalence}.
We shall finish the proof by first showing that (iii) implies (i) and then that (i) implies (ii).

\medskip
\noindent\emph{Mixed model $\Rightarrow$ inflation.}
Given a mixed realization, take $m$ independent copies of every source:
\begin{align}\label{eq:InflationMixedCopies}
    \cH^{(m)}\coloneqq\bigotimes_{\alpha\in\SourceVertices{\Tree}}\bigotimes_{i=1}^m\cH_\alpha^{(i)},\qquad \ket{\Omega^{(m)}}\coloneqq\bigotimes_{\alpha\in\SourceVertices{\Tree}}\bigotimes_{i=1}^m\ket{\Omega_\alpha^{(i)}}.
\end{align}
Let $E^v_{a_v\mid x_v}(\mathbf i)$ be the copy of $\Meas{v}{a_v}{x_v}$ acting on the source algebra copy indexed by $(\alpha,i_\alpha)$ for source $\alpha$ neighboring $v$.
It is straightforward to check that $E^v_{a_v\mid x_v}(\mathbf i)$ satisfies the required algebraic properties and \zcref{eq:InflationPOVMRelations,eq:InflationStateConstraints} are fulfilled, proving (iii)$\Rightarrow$(i).

\medskip
\noindent\emph{Inflation $\Rightarrow$ source-transfer: construction.}
Choose $\omega$ on $\cE_{\infty}(\Tree)$ to be the factorizing $\Gamma$-invariant state by \zcref{lem:InflationFactorizingState}.
Consider its GNS representation $(\cH, \pi, \ket{\Omega})$ and the von Neumann $\cM_\infty = \pi(\cE_{\infty}(\Tree))''$ as in \zcref{eq:InflationAmbientVonNeumann}.
For convenience, in the rest of this proof, we use $E^v_{a_v\mid x_v}(\mathbf i)$ to refer to the represented operator $\pi(E^v_{a_v\mid x_v}(\mathbf i)) \in \Bof{\cH}$.
Let source permutation act on $\cM_\infty$ by $\alpha_g(Q)\coloneqq U_gQU_g^*$.
We first give all the formulas for the recursive construction.
Their well-definedness, including the existence of the completely positive limits, is proved later in \zcref{lem:InflationRecursiveConstruction}.

At a non-root leaf party $v$, define the generated $C^*$-algebra and the leaf measurements by
\begin{align}\label{eq:InflationLeafAlgebra}
    \cA_{\Parent{v}}^{v,j}\coloneqq C^*\!\left(E^v_{a_v\mid x_v}(j):a_v\in\mathsf{A}_v,\ x_v\in\mathsf{X}_v\right), \qquad \meas{v}{a_v}{x_v}\coloneqq E^v_{a_v\mid x_v}(1).
\end{align}
Here $j$ specifies the copy of source $\Parent{v}$.
By $\Gamma$-invariance, we may choose the first copy when defining the source algebras.
At a source $\alpha$, once the algebras $\cA_\alpha^{v,j}$ have been constructed for every child party $v\in\Children{\alpha}$, define
\begin{align}\label{eq:InflationConstructedSourceAlgebra}
    \cA_\alpha^v\coloneqq\cA_\alpha^{v,1}, \qquad \JoinAlg{\alpha}\coloneqq\bigotimes_{v\in\Children{\alpha}}^{\max}\cA_\alpha^v.
\end{align}

For $c\in\cA_\alpha^v$, let $c^j$ be its image under the transposition of source-$\alpha$ labels $1$ and $j$, with $c^1=c$.
The map $\rho_\alpha^j$ evaluating the joint source algebra of $\alpha$ at copy $j$ is defined to be
\begin{align}\label{eq:InflationSourceCopyEvaluation}
    \rho_\alpha^j:\JoinAlg{\alpha}\to\cM_\infty,\qquad \rho_\alpha^j\!\left(\bigotimes_{v\in\Children{\alpha}}c_v\right)\coloneqq\prod_{v\in\Children{\alpha}}c_v^j.
\end{align}
Note that $\rho_\alpha^j$ is a unital $*$-homomorphism that is not necessarily injective.
The $\alpha$-source state $\omega_\alpha$ is then the restriction to $\JoinAlg{\alpha}$ via this map at the first copy:
\begin{align}\label{eq:InflationConstructedSourceState}
    \omega_\alpha\coloneqq\omega\circ\rho_\alpha^1.
\end{align}
By our convention, for a leaf source, the empty tensor product for joint source algebra is $\mathbb{C}$, and $\rho_\alpha^j(c)=cI$.

At a party $v$ whose joint source algebras in $\Children{v}$ are already constructed, define the input algebra and state by
\begin{align}\label{eq:InflationIncomingData}
    \InAlg{v}\coloneqq\bigotimes_{\alpha\in\Children{v}}^{\max}\JoinAlg{\alpha},\qquad \InState{v}\coloneqq\bigotimes_{\alpha\in\Children{v}}\omega_\alpha.
\end{align}
For a tuple of child-source labels $\mathbf i=(i_\alpha)_{\alpha\in\Children{v}}$, define the unital $*$-homomorphism $\iota_{v,\mathbf i}$, induced by the maps $\rho_\alpha^{i_\alpha}$, by
\begin{align}\label{eq:InflationIncomingEvaluation}
    \iota_{v,\mathbf i}:\InAlg{v}\to\cM_\infty,\qquad \iota_{v,\mathbf i}\!\left(\bigotimes_{\alpha\in\Children{v}}q_\alpha\right)\coloneqq\prod_{\alpha\in\Children{v}}\rho_\alpha^{i_\alpha}(q_\alpha).
\end{align}

For an internal party $v\neq r$, let $d_v \coloneqq |\Children{v}|$ and write $E^v_{a_v\mid x_v}(j,\mathbf i)$ with the parent-source label $j$ first.
Write $[m]^{\Children{v}}$ for the set of tuples $\mathbf{i}=(i_\alpha)_{\alpha\in\Children{v}}$ with $i_\alpha\in[m]$ for every child source $\alpha$.
We define, for $q \in \InAlg{v}$,
\begin{align}\label{eq:InflationInternalCPAverage}
    \Phi^{v,j;m}_{a_v\mid x_v}(q)\coloneqq\frac1{m^{d_v}}\sum_{\mathbf i\in[m]^{\Children{v}}}\iota_{v,\mathbf i}(q) E^v_{a_v\mid x_v}(j,\mathbf i).
\end{align}
As we shall show in \zcref{lem:InflationRecursiveConstruction}, there exists a subsequence $m_k^v$ at $v$ for which these maps converge pointwise ultraweakly, simultaneously for all $j,a_v,x_v$:
\begin{align}\label{eq:InflationInternalLimitMaps}
    \Phi^{v,j}_{a_v\mid x_v}(q)\coloneqq\lim_{k\to\infty}\Phi^{v,j;m_k^v}_{a_v\mid x_v}(q),\qquad q\in\InAlg{v}.
\end{align}
The output parent source algebras and transfer maps are thus defined with the first copy of $\Parent{v}$ as
\begin{equation}\label{eq:InflationConstructedTransferData}
    \begin{aligned}
        \cA_{\Parent{v}}^{v,j}&\coloneqq C^*\!\left(\Phi^{v,j}_{a_v\mid x_v}(q):a_v\in\mathsf A_v,\ x_v\in\mathsf X_v,\ q\in\InAlg{v}\right),\\
        \cA_{\Parent{v}}^v&\coloneqq\cA_{\Parent{v}}^{v,1},\qquad \instr{v}{a_v}{x_v}\coloneqq\Phi^{v,1}_{a_v\mid x_v}:\InAlg{v}\to\cA_{\Parent{v}}^v,
    \end{aligned}
\end{equation}
which are the needed data for the recursive construction towards $\Parent{v}$.
Since $\Tree$ has only finitely many internal parties, we choose the subsequences $(m_k^v)_k$ recursively, so that a single $(m_k)_k$ satisfies \eqref{eq:InflationInternalLimitMaps} for every $j,a_v,x_v$, and every internal party $v$ at once.

Finally, at the root $r$, let $\mathbf1 = (1_{\alpha})_{\alpha \in \Children{r}}$ and define
\begin{align}\label{eq:InflationConstructedRootFunctional}
    \rstate{r}{a_r}{x_r}(q)\coloneqq\omega\!\left(\iota_{r,\mathbf1}(q)E^r_{a_r\mid x_r}(\mathbf1)\right),\qquad q\in\InAlg{r}.
\end{align}
By \zcref{lem:InflationRecursiveConstruction}, the transfer maps are completely positive contractions, the source states are states, and the root functionals are positive.
It remains to verify the source-replacer conditions and the correlation formula.

\medskip
\noindent\emph{Source-replacer conditions.}
Summing the finite averages over $a_v$ removes the measurement at $v$ by POVM completeness.
Since $\sum_{a_v} E^v_{a_v\mid x_v} = I$, \zcref{lem:InflationConstructedAverages} therefore gives, for every $q\in\InAlg{v}$,
\begin{equation}\label{eq:InflationTransferNormalizationLimit}
    \begin{aligned}
        \sum_{a_v}\Phi^{v,j;m}_{a_v\mid x_v}(q)
        &=\frac1{m^{d_v}}\sum_{\mathbf i\in[m]^{\Children{v}}}\iota_{v,\mathbf i}(q) \to \InState{v}(q)I
    \end{aligned}
\end{equation}
strongly as $m\to\infty$.
Taking the above subsequence $m_k$ on both sides proves
\begin{align}\label{eq:InflationConstructedTransferReplacer}
    \sum_{a_v}\instr{v}{a_v}{x_v}(q)=\InState{v}(q)I.
\end{align}
At the root, the same POVM completeness and \zcref{lem:InflationConstructedAverages} imply
\begin{align}\label{eq:InflationConstructedRootReplacer}
    \sum_{a_r}\rstate{r}{a_r}{x_r}(q)=\omega(\iota_{r,\mathbf1}(q))=\InState{r}(q).
\end{align}
These are exactly the source-replacer conditions of \zcref{eq:TreeTransferReplacer,eq:TreeRootFunctionalReplacer}.

\medskip
\noindent\emph{Correlation recovery.}
Fix a complete event $(\mathbf a|\mathbf x)$.
With the above constructed operators and transfer maps, the recursive formulas for the transfer messages are
\begin{equation}\label{eq:InflationRecursiveMessages}
    \begin{aligned}
        \TransferMsg{v}&=\meas{v}{a_v}{x_v} &&\text{at a leaf party},\\
        \TransferMsg{\alpha}&=\bigotimes_{u\in\Children{\alpha}}\TransferMsg{u} &&\text{at a source},\\
        \TransferMsg{v}&=\instr{v}{a_v}{x_v}\!\left(\bigotimes_{\alpha\in\Children{v}}\TransferMsg{\alpha}\right) &&\text{at an internal non-root party}.
    \end{aligned}
\end{equation}

Consider $E^v_{a_v\mid x_v}(\mathbf{1})$ for the measurement copy at party $v$ with every neighboring source label equal to $1$.
(For a non-root party $v$, we also write $E^v_{a_v\mid x_v}(\mathbf{1})=E^v_{a_v\mid x_v}(1,\mathbf{1})$ by splitting off the parent-source label as in \zcref{eq:InflationInternalCPAverage} with $\mathbf{1}$ the all-ones tuple over $\Children{v}$.)
By \zcref{eq:InflationDiagonalEvent}, $p(\mathbf{a}|\mathbf{x})=\omega\!\left(\prod_{v\in\PartyVertices{\Tree}}E^v_{a_v\mid x_v}(\mathbf{1})\right)$.
Now, for each subtree $\Tree_v$, we will replace all measurements $E^u_{a_u\mid x_u}(\mathbf{1})$ for $u \in \PartyVertices{\Tree_v}$ by the transfer message $\TransferMsg{v}$, for $v$ starting from the leaves towards the root $r$.

Specifically, let $B$ be a set of non-root parties such that the subtrees $\Tree_w$, $w\in B$, have pairwise disjoint party sets.
Consider the set of parties outside these subtrees $U\coloneqq\PartyVertices{\Tree}\setminus\bigcup_{w\in B}\PartyVertices{\Tree_w}$.
We claim that
\begin{align}\label{eq:InflationCorrelationInduction}
    p(\mathbf{a} | \mathbf{x})=\omega\!\left(\left(\prod_{u\in U}E^u_{a_u\mid x_u}(\mathbf{1})\right)\left(\prod_{w\in B}\TransferMsg{w}\right)\right),
\end{align}
i.e., the measurements at parties in these $\Tree_w$ can be replaced by their corresponding transfer message $\TransferMsg{w}$ while preserving the correlations.

We prove the claim by induction on the number of replaced subtrees, or equivalently on the size of $B$.
In the base case, no subtree has been replaced, so $B = \emptyset$ and $U = \PartyVertices{\Tree}$.
The claim then follows directly from \zcref{eq:InflationDiagonalEvent}, after moving $v$ from $U$ to $B$.
For the induction step, let $v \in U$ be a leaf party of the remaining subtree.
The required replacement of $v$ follows immediately from \zcref{eq:InflationLeafAlgebra}.
Now let $v\in U$ be an internal non-root party, and denote
\begin{align}\label{eq:InflationChildPartySet}
    C_v\coloneqq\bigcup_{\alpha\in\Children{v}}\Children{\alpha}
\end{align}
for the parties immediately below its child sources.
Suppose that \zcref{eq:InflationCorrelationInduction} holds with $C_v\subset B$, that is, all measurements in $\Tree_w$ for $w \in C_v$ can be replaced by $\TransferMsg{w}$ while preserving the correlation.

Denote $q_v\coloneqq\bigotimes_{\alpha\in\Children{v}}\TransferMsg{\alpha}$.
By the definitions of the homomorphisms and transfer messages, we have
\begin{align}\label{eq:InflationIncomingMessageProduct}
    \iota_{v,\mathbf{1}}(q_v)= \prod_{\alpha \in \Children{v}} \rho_{\alpha}^1(\TransferMsg{\alpha}) = \prod_{\alpha\in\Children{v}}\prod_{w\in\Children{\alpha}}\TransferMsg{w}=\prod_{w\in C_v}\TransferMsg{w}.
\end{align}
Hence, as measurement operators in disjoint party sets commute, we have
\begin{equation}\label{eq:InflationCorrelationInnerSubstitution}
    \begin{aligned}
        \left(\prod_{u\in U}E^u_{a_u\mid x_u}(\mathbf{1})\right)\left(\prod_{w\in B}\TransferMsg{w}\right) &= \left(\prod_{u\in U \setminus \{v\}}E^u_{a_u\mid x_u}(\mathbf{1})\right) E^v_{a_v\mid x_v}(\mathbf{1}) \left(\prod_{w\in B \setminus C_v}\TransferMsg{w}\right) \left(\prod_{w\in C_v}\TransferMsg{w}\right) \\
        &= \left(\prod_{u\in U \setminus \{v\}}E^u_{a_u\mid x_u}(\mathbf{1})\right)  \left(\prod_{w\in B \setminus C_v}\TransferMsg{w}\right) \iota_{v,\mathbf{1}}(q_v) E^v_{a_v\mid x_v}(\mathbf{1}).
    \end{aligned}
\end{equation}

We wish to apply \zcref{lem:InflationExpectationReplacement} with $R \coloneqq \left(\prod_{u\in U \setminus \{v\}}E^u_{a_u\mid x_u}(\mathbf{1})\right) \left(\prod_{w\in B \setminus C_v}\TransferMsg{w}\right)$.
To this end, we need to check that $R$ is indeed invariant under permutations of every child source $\alpha \in \Children{v}$.
Indeed, any $\alpha \in \Children{v}$ is connected only to $v$ and parties in $C_v$, so the permutations do not affect any $E^u_{a_u\mid x_u}(\mathbf{1})$ with $u \in U \setminus \{v\}$.
For $w\in B\setminus C_v$, we have $\Parent{w}\neq\alpha$, so these permutations on $\alpha$ also fix $\TransferMsg{w}$ by \zcref{lem:InflationRecursiveConstruction}(ii).
Thus, \zcref{lem:InflationExpectationReplacement,eq:InflationRecursiveMessages,eq:InflationCorrelationInnerSubstitution} imply that
\begin{equation}\label{eq:InflationCorrelationReplacement}
    \begin{aligned}
        p(\mathbf{a}|\mathbf{x}) &= \omega\!\left(R \iota_{v,\mathbf{1}}(q_v)E^v_{a_v\mid x_v}(1,\mathbf{1})\right) \\
        &= \omega\!\left(R \Phi^{v,1}_{a_v\mid x_v}(q_v)\right) \\
        &= \omega\!\left(\left(\prod_{u\in U\setminus\{v\}}E^u_{a_u\mid x_u}(\mathbf{1})\right) \left(\prod_{w\in B\setminus C_v}\TransferMsg{w}\right) \TransferMsg{v}\right),
    \end{aligned}
\end{equation}
replacing the whole subtree $\Tree_v$ by its transfer message $\TransferMsg{v}$.
This proves the claim by replacing $U$ with $U\setminus\{v\}$ and $B$ with $(B\setminus C_v)\cup\{v\}$.

Applying \zcref{eq:InflationCorrelationInduction} with $U=\{r\}$ and $B=\bigcup_{\alpha\in\Children{r}}\Children{\alpha}$ gives
\begin{equation}\label{eq:InflationRecoveredCorrelation}
    \begin{aligned}
        p(\mathbf{a}|\mathbf{x}) =\omega\!\left(E^r_{a_r\mid x_r}(\mathbf{1})\prod_{\alpha\in\Children{r}}\prod_{w\in\Children{\alpha}}\TransferMsg{w}\right) =\rstate{r}{a_r}{x_r}\!\left(\bigotimes_{\alpha\in\Children{r}}\TransferMsg{\alpha}\right),
    \end{aligned}
\end{equation}
where the second equality follows from the definitions of the transfer messages, the input evaluation map, the root functional, and \zcref{eq:InflationIncomingMessageProduct}.
This proves the desired correlation recovery formula.
\end{proof}

We now prove the three lemmas required by the above proof, where the first one shows that the recursive construction is well-defined, the second corresponds to the source-replacer condition, and the third is used for the correlation recovery.

\begin{lemma}[Properties of the recursive construction]\label{lem:InflationRecursiveConstruction}
In the GNS representation $(\cH, \pi, \ket{\Omega})$ used in \zcref{thm:InflationSourceTransferTreeEquivalence}, the recursive constructions from \zcref{eq:InflationLeafAlgebra,eq:InflationConstructedSourceAlgebra,eq:InflationSourceCopyEvaluation,eq:InflationConstructedSourceState,eq:InflationIncomingData,eq:InflationIncomingEvaluation,eq:InflationInternalCPAverage,eq:InflationInternalLimitMaps,eq:InflationConstructedTransferData,eq:InflationConstructedRootFunctional} produce separable unital $C^*$-algebras, unital $*$-homomorphisms $\rho_\alpha^j$ and $\iota_{v,\mathbf i}$, source states $\omega_\alpha$, completely positive contractions $\Phi^{v,j}_{a_v\mid x_v}$, and positive root functionals $\rstate{r}{a_r}{x_r}$.

For each non-root party $v$, let $\cM_v^j$ be the von Neumann algebra generated by all measurements at $v$ with parent-source label fixed at $j$ and all measurement copies at its descendant parties:
\begin{equation}\label{eq:InflationSubtreeVonNeumannAlgebra}
    \begin{aligned}
        \cM_v^j \coloneqq \Bigl( &\{E^v_{a_v\mid x_v}(j,\mathbf i):\text{all }a_v,x_v,\mathbf i\} {}\cup{} \{E^w_{a_w\mid x_w}(\mathbf k):w\in\PartyVertices{\Tree_v}\setminus\{v\},\ \text{all }a_w,x_w,\mathbf k\} \Bigr)^{\prime\prime}.
    \end{aligned}
\end{equation}
That is, the parent-source label $j$ is fixed, while all other source labels may vary freely.
Then the constructed algebras $\cA_{\Parent{v}}^{v,j}$ satisfy:
\begin{enumerate}[label=(\roman*)]
    \item $\cA_{\Parent{v}}^{v,j}\subset\cM_v^j$.
    \item For every source permutation $g=(g_\alpha)_\alpha\in\Gamma$
    \begin{align}\label{eq:InflationAlgebraPermutation}
        \alpha_g\!\left(\cA_{\Parent{v}}^{v,j}\right)=\cA_{\Parent{v}}^{v,g_{\Parent{v}}(j)}
    \end{align}
    and $\alpha_g(q)=q$ for all $q \in \cA_{\Parent{v}}^{v,j}$ if $g_{\Parent{v}}(j)=j$.
    That is, the algebra is invariant under permutations of sources except for those of $\Parent{v}$.
    \item For every $q \in \cA_{\Parent{v}}^{v,j}$ and every inflation measurement $E^u_{a_u\mid x_u}(\mathbf{i})$,
    \begin{align}\label{eq:InflationAlgebraGeneratorCommutation}
        [q,E^u_{a_u\mid x_u}(\mathbf{i})]=0 \quad\text{if }u\neq v\text{ or }\bigl(u=v\text{ and }i_{\Parent{v}}\neq j\bigr).
    \end{align}
    That is, elements of $\cA_{\Parent{v}}^{v,j}$ commute with every measurement operator except for $E^v_{a_v\mid x_v}( j, \mathbf{i})$.
\end{enumerate}
\end{lemma}
\begin{proof}
We proceed by induction, reflecting the recursive nature of the construction.

\medskip
\noindent
\emph{Leaf parties.}
At a leaf party $v$, the algebra $\cA_{\Parent{v}}^{v,j}$ is generated by the finitely many $E^v_{a_v\mid x_v}(j)$ and is therefore separable; here $j$ labels the copy of the parent source $\Parent{v}$.
Property (i) follows directly from the definition.
The equality
\begin{align}\label{eq:InflationLeafPermutationProof}
    \alpha_g\!\left(E^v_{a_v\mid x_v}(j)\right)=E^v_{a_v\mid x_v}\!\left(g_{\Parent{v}}(j)\right)
\end{align}
proves (ii), and the inflation commutation relations give (iii).

Next, suppose that all the descendant source algebras of a given source $\alpha$ have been constructed and satisfy (i)--(iii).
We first verify that the source and input maps are well-defined, and then construct the transfer maps and verify (i)--(iii) for their output parent source algebras.

\medskip
\noindent
\emph{Source algebras, states and input $*$-homomorphisms.}
Given a source $\alpha$, consider parties $u, v \in \Children{\alpha}$ such that $u \neq v$.
Then the subtrees $\Tree_u, \Tree_v$ have disjoint party sets.
The definition of $\cM_u^i, \cM_v^j$, \zcref{eq:InflationSubtreeVonNeumannAlgebra}, involves only the measurements of the parties in the subtrees $\Tree_u, \Tree_v$, respectively.
Thus the inflation commutation relations imply
\begin{align}\label{eq:InflationSubtreeAlgebraCommutation}
    [\cM_u^i,\cM_v^j]=0,\qquad i,j\geq1.
\end{align}
By the induction hypothesis (i), the subalgebras $\cA_\alpha^{u,i}$ and $\cA_\alpha^{v,j}$ commute as well.
Moreover by (ii), the map $c=c^1\mapsto c^j$ is a unital $*$-isomorphism from $\cA_\alpha^{u,1}$ onto $\cA_\alpha^{u,j}$.
Since their ranges $\cA_\alpha^{u,j}$ mutually commute, the universal property of the maximal tensor product gives a unital $*$-homomorphism
\begin{align}\label{eq:InflationSourceMapProof}
    \rho_\alpha^j:\JoinAlg{\alpha} = \bigotimes_{v\in\Children{\alpha}}^{\max}\cA_\alpha^v \to\cM_\infty,\qquad \rho_\alpha^j\!\left(\bigotimes_{u\in\Children{\alpha}}c_u\right)=\prod_{u\in\Children{\alpha}}c_u^j.
\end{align}
Consequently, the $\alpha$-source state $\omega_\alpha=\omega\circ\rho_\alpha^1$ is indeed a state, as $\rho_\alpha^1$ preserves positivity and the unit.

Similarly, consider $\InAlg{v}=\bigotimes_{\beta\in\Children{v}}^{\max}\JoinAlg{\beta}$.
For distinct child sources $\beta,\gamma$, the ranges $\rho_\beta^{i_\beta}(\JoinAlg{\beta})$ and $\rho_\gamma^{i_\gamma}(\JoinAlg{\gamma})$ involve measurements in disjoint party subtrees and therefore commute.
The universal property of the maximal tensor product gives a unital $*$-homomorphism
\begin{align}\label{eq:InflationInputMapProof}
    \iota_{v,\mathbf{i}}:\InAlg{v}\to\cM_\infty,\qquad
    \iota_{v,\mathbf{i}}\!\left(\bigotimes_{\beta\in\Children{v}}q_\beta\right)=\prod_{\beta\in\Children{v}}\rho_\beta^{i_\beta}(q_\beta).
\end{align}
By the induction hypothesis (i), every $\iota_{v,\mathbf{i}}(q)$ belongs to the von Neumann algebra generated by measurements from parties that are not $v$.
Thus, these measurements commute with $E^v_{a_v\mid x_v}(j,\mathbf{k})$ for every choice of source labels, and so
\begin{align}\label{eq:InflationInputMeasurementCommutation}
    [\iota_{v,\mathbf{i}}(q),E^v_{a_v\mid x_v}(j,\mathbf{k})]=0,\qquad q\in\InAlg{v},
\end{align}
for all source labels $j,\mathbf{i},\mathbf{k}$.

For $g\in\Gamma$, denote $g\mathbf{i}\coloneqq(g_\beta(i_\beta))_{\beta\in\Children{v}}$.
We claim for the input $*$-homomorphism that
\begin{align}\label{eq:InflationInputMapCovariance}
    \alpha_g\!\left(\iota_{v,\mathbf{i}}(q)\right)=\iota_{v,g\mathbf{i}}(q),\qquad q\in\InAlg{v}.
\end{align}
Indeed, the induction hypothesis (ii) gives $\alpha_g(c_u^i)=c_u^{g_\alpha(i)}$.
By the definition of $\iota_{v,\mathbf{i}}$, this proves \zcref{eq:InflationInputMapCovariance} for every elementary tensor $q=\bigotimes_{\beta\in\Children{v}}\bigl(\bigotimes_{u\in\Children{\beta}}c_u\bigr)$, with $c_u\in\cA_\beta^{u,1}$.
Linearity and continuity extend the identity to all $q\in\InAlg{v}$.

\medskip
\noindent\emph{Transfer maps from ultraweak limits.}
Fix an internal non-root party $v$ and set $d_v\coloneqq|\Children{v}|$.
\zcref[S]{eq:InflationInputMeasurementCommutation} implies
\begin{align}\label{eq:InflationAverageCPSandwich}
    \iota_{v,\mathbf{i}}(q) E^v_{a_v\mid x_v}(j,\mathbf{i})=(E^v_{a_v\mid x_v}(j,\mathbf{i}))^{1/2}\iota_{v,\mathbf{i}}(q)(E^v_{a_v\mid x_v}(j,\mathbf{i}))^{1/2}.
\end{align}
Each summand in \zcref{eq:InflationInternalCPAverage} defining $\Phi^{v,j;m}_{a_v\mid x_v}$ is therefore completely positive.
Moreover,
\begin{align}\label{eq:InflationAverageContractivity}
    0\leq\Phi^{v,j;m}_{a_v\mid x_v}(\id) =\frac1{m^{d_v}}\sum_{\mathbf{i}\in[m]^{\Children{v}}}E^v_{a_v\mid x_v}(j,\mathbf{i}) \leq I,
\end{align}
so these maps are contractions.
The GNS space $\cH$ is separable, and $\InAlg{v}$ as a finite maximal tensor product of separable $C^*$-algebras is again separable.
Thus \zcref{lem:InflationCPLimits} gives a common subsequence $m_k$ and completely positive contractions $\Phi^{v,j}_{a_v\mid x_v}$ as ultraweak limits, simultaneously for all $j,a_v,x_v$.

Lastly, since $\{q_t\}_{t\geq1}$ is norm-dense in $\InAlg{v}$, continuity gives
\begin{align}\label{eq:InflationOutputSeparability}
    \cA_{\Parent{v}}^{v,j}
    =C^*\!\left(\Phi^{v,j}_{a_v\mid x_v}(q_t):a_v\in\mathsf{A}_v,\ x_v\in\mathsf{X}_v,\ t\geq1\right).
\end{align}
Hence every output parent algebra $\cA_{\Parent{v}}^{v,j}$ is separable, allowing the further recursive construction.
We now verify the three properties for these new objects.

\medskip
\noindent\emph{Property (i).}
The input operators $q \in \InAlg{v}$ and the measurement at $v$ belong to $\cM_v^j$ by construction, so
\begin{align}\label{eq:InflationOutputContainment}
    \Phi^{v,j;m}_{a_v\mid x_v}(q)\in\cM_v^j
    \quad\Longrightarrow\quad
    \Phi^{v,j}_{a_v\mid x_v}(q)\in\cM_v^j
\end{align}
thanks to the ultraweak closedness of $\cM_v^j$.
Taking the generated $C^*$-algebra $\cA_{\Parent{v}}^{v,j}$ proves (i).

\medskip
\noindent\emph{Property (ii).}
Fix $g\in\Gamma$.
Since each $g_\beta$ moves only finitely many labels, $g_\beta([m])=[m]$ for every $\beta\in\Children{v}$ when $m$ is sufficiently large.
Thus $\mathbf{i}\mapsto g\mathbf{i}$ is a bijection of $[m]^{\Children{v}}$, and \zcref{eq:InflationInputMapCovariance} gives
\begin{equation}\label{eq:InflationFiniteAverageCovariance}
    \begin{aligned}
        \alpha_g\!\left(\Phi^{v,j;m}_{a_v\mid x_v}(q)\right)
        =\frac1{m^{d_v}}\sum_{\mathbf{i}\in[m]^{\Children{v}}}
        \iota_{v,g\mathbf{i}}(q)E^v_{a_v\mid x_v}\!\left(g_{\Parent{v}}(j),g\mathbf{i}\right) =\Phi^{v,g_{\Parent{v}}(j);m}_{a_v\mid x_v}(q).
    \end{aligned}
\end{equation}
Thus, the ultraweak limit along $m_k$, which is the same subsequence for every $j$ and $g_{\Parent{v}}(j)$, implies
\begin{align}\label{eq:InflationLimitMapCovariance}
    \alpha_g\!\left(\Phi^{v,j}_{a_v\mid x_v}(q)\right)=\Phi^{v,g_{\Parent{v}}(j)}_{a_v\mid x_v}(q),
\end{align}
since $\alpha_g(\cdot)=U_g \cdot U_g^*$ is ultraweakly continuous.

Therefore, the generating family of $\cA_{\Parent{v}}^{v,j}$ is mapped to the generating family of $\cA_{\Parent{v}}^{v,g_{\Parent{v}}(j)}$ by $\alpha_g$, proving \zcref{eq:InflationAlgebraPermutation}.
If also $g_{\Parent{v}}(j)=j$, it fixes every generator and hence every element, proving (ii).

\medskip
\noindent\emph{Property (iii).}
Let $q\in\InAlg{v}$ and $Z\coloneqq E^u_{a_u\mid x_u}(\mathbf{k})$, with either $u\neq v$ or $k_{\Parent{v}}\neq j$ when $u=v$.
It suffices to prove $[\Phi^{v,j}_{a_v\mid x_v}(q),Z]=0$ for every $q,a_v,x_v$, since these outputs generate $\cA_{\Parent{v}}^{v,j}$ and $Z=Z^*$.

To this end, first consider the summand $\iota_{v,\mathbf{i}}(q) E^v_{a_v\mid x_v}(j,\mathbf{i})$ in $\Phi^{v,j;m}_{a_v\mid x_v}(q)$ and check its commutation with $Z$.
If $u=v$, $\iota_{v,\mathbf{i}}(q)$ commutes with $Z$ by \zcref{eq:InflationInputMeasurementCommutation}.
Since $k_{\Parent{v}}\neq j$, the inflation relations also give
\begin{align}\label{eq:InflationSamePartyCommutation}
    [E^v_{a_v\mid x_v}(j,\mathbf{i}),Z]=0
    \quad\text{if }i_\beta\neq k_\beta\text{ for every }\beta\in\Children{v}.
\end{align}
Thus the summand can fail to commute with $Z$ only when some child-source labels coincide.

If $u\neq v$, $E^v_{a_v\mid x_v}(j,\mathbf{i})$ commutes with $Z$ by the inflation relations.
For the input factor, the induction hypothesis (iii) gives
\begin{align}\label{eq:InflationDescendantInputCommutation}
    [\iota_{v,\mathbf{i}}(q),Z]=0
    \quad\text{if }i_\beta\neq k_\beta
    \text{ whenever }\beta\in\Children{v}\text{ and }u\in\Children{\beta}.
\end{align}
Indeed, on an elementary tensor the input is a product of elements $c_w^{i_\beta}\in\cA_\beta^{w,i_\beta}$.
By (iii), each $c_w^{i_\beta}$ commutes with $Z$ unless $w=u$ and $i_\beta=k_\beta$ with $\beta = \Parent{w} = \Parent{u}$.
The conclusion then extends to every $q$ by linearity and continuity.
If $u$ is not a child party of any child source of $v$, the input commutes with $Z$ for every $\mathbf{i}$ by the induction hypothesis (iii).

We now bound the number of indices $\mathbf{i}$ for which noncommutation $[\iota_{v,\mathbf{i}}(q) E^v_{a_v\mid x_v}(j,\mathbf{i}), Z] \neq 0$ occurs.
Note that, for each child source $\beta\in\Children{v}$, fixing $i_\beta=k_\beta$ leaves the other $d_v-1$ coordinates free to take any value in $[m]$.
Hence
\begin{align}\label{eq:InflationSingleLabelCount}
    \left|\left\{ \mathbf{i}\in[m]^{\Children{v}} \mid i_\beta=k_\beta \right\}\right|\leq m^{d_v-1}.
\end{align}
As $| [m]^{\Children{v}}| = m^{d_v}$, the fraction of $\mathbf{i}$ such that $i_\beta=k_\beta$ is at most $1/m$.
When $u=v$, noncommutation requires at least one of the $d_v$ equalities $i_\beta=k_\beta$, with $\beta\in\Children{v}$.
The union bound therefore gives a fraction of noncommutation at most $d_v/m$.
When $u\neq v$, at most one child source $\beta$ satisfies $u\in\Children{\beta}$, so the fraction is at most $1/m$, and is zero if there is no such source.
Thus, in every case, at most a fraction $d_v/m$ of the summands $\iota_{v,\mathbf{i}}(q) E^v_{a_v\mid x_v}(j,\mathbf{i})$ can fail to commute with $Z$.

Now, each summand $\iota_{v,\mathbf{i}}(q) E^v_{a_v\mid x_v}(j,\mathbf{i})$ has norm at most $\|q\|$ by contractivity, so its commutator with $Z$ has norm at most $2\|q\|\|Z\|$.
Consequently, by the triangle inequality,
\begin{align}\label{eq:InflationTransferCommutatorBound}
    \|[\Phi^{v,j;m}_{a_v\mid x_v}(q),Z]\| \leq\frac{2d_v}{m}\|q\|\|Z\| \to 0.
\end{align}
Taking the ultraweak limit $m_k \to \infty$ gives $[\Phi^{v,j}_{a_v\mid x_v}(q),Z]=0$, proving (iii).

\medskip
\noindent\emph{Root functionals.}
Finally, \zcref{eq:InflationInputMeasurementCommutation} also holds at the root, without a parent-source label.
Then, for every $q\in\InAlg{r}$,
\begin{align}\label{eq:InflationRootPositivityProof}
    \rstate{r}{a_r}{x_r}(q^* q)
    =\omega\!\left( \iota_{r,\mathbf{1}}(q^*q) E^r_{a_r\mid x_r}(\mathbf{1})\right) = \omega\!\left(\left(E^r_{a_r\mid x_r}(\mathbf{1})\right)^{1/2} \iota_{r,\mathbf{1}}(q)^* \iota_{r,\mathbf{1}}(q) \left(E^r_{a_r\mid x_r}(\mathbf{1})\right)^{1/2}\right) \geq 0,
\end{align}
which proves positivity of the root functionals.
\end{proof}

\begin{lemma}[Averages of the constructed input operators]\label{lem:InflationConstructedAverages}
For the construction from \zcref{eq:InflationLeafAlgebra,eq:InflationConstructedSourceAlgebra,eq:InflationSourceCopyEvaluation,eq:InflationConstructedSourceState,eq:InflationIncomingData,eq:InflationIncomingEvaluation,eq:InflationInternalCPAverage,eq:InflationInternalLimitMaps,eq:InflationConstructedTransferData,eq:InflationConstructedRootFunctional} and every $q\in\JoinAlg{\alpha}$, we have
\begin{align}\label{eq:InflationConstructedSourceMean}
    \frac1m\sum_{j=1}^m\rho_\alpha^j(q) \to \omega_\alpha(q)I
\end{align}
strongly as $m\to\infty$.
Consequently, for every party $v$, with $d_v\coloneqq|\Children{v}|$, and every $q\in\InAlg{v}$,
\begin{align}\label{eq:InflationIncomingAverage}
    \frac1{m^{d_v}}\sum_{\mathbf i\in[m]^{\Children{v}}}\iota_{v,\mathbf i}(q) \to \InState{v}(q)I
\end{align}
strongly, and
\begin{align}\label{eq:InflationIncomingProductState}
    \omega\circ\iota_{v,\mathbf i}=\InState{v}
\end{align}
for every tuple $\mathbf{i}$.
\end{lemma}
\begin{proof}
We first apply \zcref{lem:InflationScalarAverages} to $\rho_\alpha^1(q)\in\cM_\infty$.
By \zcref{lem:InflationRecursiveConstruction}(ii), we can compute the group average of $\rho_\alpha^1(q)$ as
\begin{equation}\label{eq:InflationSourceGroupAverage}
    \begin{aligned}
        \frac1{|\Gamma_m|}\sum_{g\in\Gamma_m}U_g\rho_\alpha^1(q)U_g^*
        =\frac1{|\Gamma_m|}\sum_{g\in\Gamma_m}\rho_\alpha^{g_\alpha(1)}(q) =\frac1m\sum_{j=1}^m\rho_\alpha^j(q),
    \end{aligned}
\end{equation}
where the second equality holds because each value $g_\alpha(1)=j$ occurs for exactly $|\Gamma_m|/m$ permutations.
Thus \zcref{eq:InflationVonNeumannVectorAverage} from \zcref{lem:InflationScalarAverages} leads to
\begin{align}\label{eq:InflationSourceAverageOnVector}
    \left(\frac1m\sum_{j=1}^m\rho_\alpha^j(q)\right)\ket{\Omega}
    \to\omega(\rho_\alpha^1(q))\ket{\Omega}
    =\omega_\alpha(q)\ket{\Omega}.
\end{align}

To prove strong operator convergence, let $z$ be a finite-support word generated by the represented measurement in $E^v_{a_v\mid x_v}(\mathbf{i})$.
By \zcref{lem:InflationRecursiveConstruction}(iii), $[\rho_\alpha^j(q),z]=0$ if $j$ does not belong to the label set $I_\alpha(z)$.
Since $\|\rho_\alpha^j(q)\|\leq\|q\|$, it follows that
\begin{equation}\label{eq:InflationConstructedSourceCommutator}
    \begin{aligned}
        \left\|\left[\frac1m\sum_{j=1}^m\rho_\alpha^j(q),z\right]\right\| \leq\frac1m\sum_{j\in[m]\cap I_\alpha(z)}\|[\rho_\alpha^j(q),z]\| &\leq\frac{2|I_\alpha(z)|}{m}\|q\|\|z\| \to 0
    \end{aligned}
\end{equation}
for sufficiently large $m$.
Combining this with \zcref{eq:InflationSourceAverageOnVector} implies
\begin{equation}\label{eq:InflationSourceAverageDenseVectors}
    \begin{aligned}
        \left\|\left(\frac1m\sum_{j=1}^m\rho_\alpha^j(q)- \omega_\alpha(q)I\right)z\ket{\Omega}\right\| \leq\frac{2|I_\alpha(z)|}{m}\|q\|\|z\| +\|z\|\left\|\left(\frac1m\sum_{j=1}^m\rho_\alpha^j(q)-\omega_\alpha(q)I\right)\ket{\Omega}\right\| \to 0.
    \end{aligned}
\end{equation}
It follows that the convergence holds for all vectors in $\cH$ by cyclicity of $\ket{\Omega}$, proving \zcref{eq:InflationConstructedSourceMean}.

Next, for an elementary tensor $q=\bigotimes_{\beta\in\Children{v}}q_\beta$, it follows from \zcref{eq:InflationConstructedSourceMean} that
\begin{equation}\label{eq:InflationIncomingAverageFactorization}
    \begin{aligned}
        \frac1{m^{d_v}}\sum_{\mathbf{i}\in[m]^{\Children{v}}}\iota_{v,\mathbf{i}}(q)
        =\prod_{\beta\in\Children{v}}\left(\frac1m\sum_{i=1}^m\rho_\beta^i(q_\beta)\right) \to \left(\prod_{\beta\in\Children{v}}\omega_\beta(q_\beta)\right)I
        =\InState{v}(q)I
    \end{aligned}
\end{equation}
strongly.
Since $q \mapsto \frac1{m^{d_v}}\sum_{\mathbf{i}\in[m]^{\Children{v}}}\iota_{v,\mathbf{i}}(q)$ is contractive, \zcref{eq:InflationIncomingAverageFactorization} extends to all $q\in\InAlg{v}$ by continuity and linearity.
This proves \zcref{eq:InflationIncomingAverage}.

Finally, for any tuple $\mathbf{i}$, choose $g\in\Gamma$ with $g\mathbf{1}=\mathbf{i}$.
$\Gamma$-invariance of $\omega$ and \zcref{eq:InflationInputMapCovariance} give
\begin{align}\label{eq:InflationInputStateInvariance}
    \omega(\iota_{v,\mathbf{i}}(q))
    =\omega\!\left(\alpha_g(\iota_{v,\mathbf{1}}(q))\right)
    =\omega(\iota_{v,\mathbf{1}}(q)).
\end{align}
Consequently, \zcref{eq:InflationIncomingAverage} implies that
\begin{align}\label{eq:InflationInputStateFromAverage}
    \omega(\iota_{v,\mathbf{1}}(q)) =\omega\!\left(\frac1{m^{d_v}}\sum_{\mathbf{i}\in[m]^{\Children{v}}}\iota_{v,\mathbf{i}}(q)\right) \to \InState{v}(q).
\end{align}
So $\omega(\iota_{v,\mathbf{i}}(q)) = \omega(\iota_{v,\mathbf{1}}(q)) = \InState{v}(q)$, which shows \zcref{eq:InflationIncomingProductState}.
\end{proof}

\begin{lemma}[Replacing a measurement by its transfer map]\label{lem:InflationExpectationReplacement}
Let $v\neq r$ be an internal party, $j\geq1$, and $q\in\InAlg{v}$.
Suppose that $R\in\cM_\infty$ is fixed by every permutation of every child source of $v$; i.e., $\alpha_g(R)=R$ for every $g=(g_\alpha)_\alpha\in\Gamma$ with $g_\alpha=\mathrm{id}$ for all $\alpha\notin\Children{v}$.
Then
\begin{align}\label{eq:InflationExpectationReplacementIdentity}
    \omega\!\left(R\iota_{v,\mathbf1}(q)E^v_{a_v\mid x_v}(j,\mathbf1)\right)=\omega\!\left(R\Phi^{v,j}_{a_v\mid x_v}(q)\right).
\end{align}
\end{lemma}
\begin{proof}
For each tuple $\mathbf{i}$, choose $g\in\Gamma$ such that $g_\beta(1)=i_\beta$ for $\beta\in\Children{v}$ and $g_\alpha=\mathrm{id}$ for $\alpha\notin\Children{v}$.
Then $\alpha_g(R)=R$, $\alpha_g(E^v_{a_v\mid x_v}(j,\mathbf{1}))=E^v_{a_v\mid x_v}(j,\mathbf{i})$, and $\alpha_g(\iota_{v,\mathbf{1}}(q))=\iota_{v,\mathbf{i}}(q)$ by \zcref{eq:InflationInputMapCovariance}.
Using $\Gamma$-invariance $\omega\circ\alpha_g=\omega$, we calculate
\begin{equation}\label{eq:InflationCorrelationRelabeling}
    \begin{aligned}
        \omega\!\left(R\iota_{v,\mathbf{1}}(q)E^v_{a_v\mid x_v}(j,\mathbf{1})\right) = \omega\!\left(\alpha_g\!\left(R\iota_{v,\mathbf{1}}(q)E^v_{a_v\mid x_v}(j,\mathbf{1})\right)\right) = \omega\!\left(R\iota_{v,\mathbf{i}}(q)E^v_{a_v\mid x_v}(j,\mathbf{i})\right).
    \end{aligned}
\end{equation}
It follows that the average over $\mathbf{i}\in[m]^{\Children{v}}$ satisfies
\begin{equation}\label{eq:InflationReplacementAverage}
    \begin{aligned}
        \omega\!\left(R\iota_{v,\mathbf{1}}(q)E^v_{a_v\mid x_v}(j,\mathbf{1})\right) = \frac1{m^{d_v}}\sum_{\mathbf{i}\in[m]^{\Children{v}}} \omega\!\left(R\iota_{v,\mathbf{i}}(q)E^v_{a_v\mid x_v}(j,\mathbf{i})\right) =\omega\!\left(R\Phi^{v,j;m}_{a_v\mid x_v}(q)\right),
    \end{aligned}
\end{equation}
and we are done by taking the subsequence $m_k$ defining $\Phi^{v,j}_{a_v\mid x_v}$.
\end{proof}

\subsection{Inflation-NPA SDP hierarchy and its convergence}\label{sec:InflationNPASDP}
We finish by showing that the inflation-NPA hierarchy converges to the mixed quantum model on source-party trees.
For $m,d\geq1$, let $\cE_m^d(\Tree)$ denote the linear span of formal words of length at most $d$ in $E^v_{a_v\mid x_v}(\mathbf{i}) \in \cE_m(\Tree)$, including the empty word $\id$.
The inflation-NPA test is the degree-truncated version of \zcref{def:InflationTestPOVM}, with positivity expressed by moment and localizing matrices.

\begin{definition}[Inflation-NPA test]\label{def:InflationNPATest}
A correlation $p$ passes the level-$(m,d)$ inflation-NPA test if there exists a Hermitian linear functional $L:\cE_m^{2d}(\Tree)\to\mathbb{C}$ with $L(\id)=1$ satisfying
\begin{equation}\label{eq:InflationNPAPositivity}
    \begin{aligned}
        L(f^*f)&\geq0 &&\text{for every }f\in\cE_m^d(\Tree),\\
        L(h^* E^v_{a_v\mid x_v}(\mathbf{i}) h)&\geq0,\qquad L(h^*(\id-E^v_{a_v\mid x_v}(\mathbf{i})^2)h)\geq0
        &&\text{for every }h\in\cE_m^{d-1}(\Tree)
    \end{aligned}
\end{equation}
and every measurement generator $E^v_{a_v\mid x_v}(\mathbf{i})$.
The POVM-completeness and commutation relations are already incorporated into $\cE_m^{2d}(\Tree)$.
We additionally impose the symmetry and diagonal-correlation constraints in \zcref{eq:InflationStateConstraints} whenever the expressions involved have degree at most $2d$.

For optimization, we follow~\cite[Sections~IV and~VII.C.2]{wolfe2021quantum}.
Write the real polynomial objective $f$ using real coefficients $c_0,\ldots,c_N$ and complete events $e_{ij}$, and define its inflated counterpart $\widehat{f}$ by
\begin{equation}\label{eq:InflationPolynomialObjectiveLift}
    \begin{aligned}
        f(p)&=c_0+\sum_{i=1}^N c_i\prod_{j=1}^{k_i}p(e_{ij}),\\
        \widehat{f}&=c_0\id+\sum_{i=1}^N c_i\prod_{j=1}^{k_i}D_{e_{ij}}^{(j)},
    \end{aligned}
\end{equation}
where $k_i$ is the degree of the $i$-th monomial.
For sufficiently large $m$ and $d$, define
\begin{equation}\label{eq:InflationOptimizationBounds}
    \begin{aligned}
        \beta_{m,d}^{\min}(f)&\coloneqq\inf_L L(\widehat{f}), \qquad \beta_{m,d}^{\max}(f)\coloneqq\sup_L L(\widehat{f}),
    \end{aligned}
\end{equation}
where $L:\cE_m^{2d}(\Tree)\to\mathbb{C}$ ranges over Hermitian linear functionals satisfying $L(\id)=1$, \zcref{eq:InflationNPAPositivity}, and source-permutation invariance, without the prescribed-correlation equations.
\end{definition}

Clearly, the tests of \zcref{def:InflationNPATest} can be implemented by SDPs.

\begin{corollary}[Convergence of inflation-NPA on source-party trees]\label{cor:InflationNPATreeConvergence}
Let $\Tree$ be a finite source-party tree with finite setting and outcome sets.
A correlation $p$ passes every level-$(m,d)$ inflation-NPA test if and only if it admits a mixed quantum realization on $\Tree$.
\end{corollary}
\begin{proof}
A mixed realization gives inflation states $\omega_m$ by \zcref{thm:InflationSourceTransferTreeEquivalence}, whose restrictions to $2d$-truncation satisfy every $(m,d)$-test.
Conversely, for each fixed $m$, the Banach--Alaoglu argument of \cite[Section~2.2.3 and Lemma~5]{ligthart2023convergent} gives a state $\omega_m$ on $\cE_m(\Tree)$ satisfying \zcref{def:InflationTestPOVM} whenever all degree-$d$ truncations are feasible.
Thus $p$ passes every POVM inflation test, and \zcref{thm:InflationSourceTransferTreeEquivalence} gives a mixed realization.
\end{proof}

\begin{remark}[POVM and projective formulations]\label{rem:InflationPOVMProjective}
The original inflation-NPA formulation imposes projective measurements.
The local dilation argument in~\cite[p.~4]{wolfe2021quantum} shows that every mixed-model correlation passes the projective tests.
Conversely, every correlation passing all projective tests also passes all POVM tests.
Thus \zcref{cor:InflationNPATreeConvergence} also establishes completeness of the original projective formulation.
\end{remark}

\begin{remark}[Flatness and the inflation order]\label{rem:InflationFlatnessStopping}
For a fixed inflation level $m$, the ordinary flatness condition of the NPA hierarchy~\cite[Theorem~10]{navascues2008convergent} applies to the inflated measurement algebra.
However, our reconstruction of a source-transfer realization relies on the de Finetti theorem~\cite{ligthart2023convergent}, which requires the limit $m\to\infty$ to extract a $\Gamma$-invariant state factorizing over disjoint inflated copies.
It is therefore unclear whether a sufficient flatness condition can be established at a fixed level $(m,d)$ to extract a source-transfer realization, or equivalently a mixed-model realization.
\end{remark}

\begin{remark}[Computability and convergence of optimal values]\label{rem:InflationGameComputability}
The same argument as in \zcref{cor:MixedTreeGameComputability} establishes $\mathsf{coRE}$-completeness of the mixed quantum tree game-value promise problem using the inflation-NPA hierarchy subject to the constraint that the winning probability is at least $3/4$.

Furthermore, the inflation-NPA optimization bounds $\beta_{m,d}^{\min}(f)$ and $\beta_{m,d}^{\max}(f)$ in \zcref{eq:InflationOptimizationBounds} converge to the mixed-model infimum and supremum of every real polynomial objective $f$ in the correlation entries as in \zcref{eq:InflationPolynomialObjectiveLift}.
The limiting state obtained by the same compactness argument as in \zcref{cor:InflationNPATreeConvergence} has a de Finetti decomposition $\omega=\int\varphi\,d\mu(\varphi)$ into $\Gamma$-invariant states factorizing over disjoint source labels.
For each component, put $p_\varphi(e)\coloneqq\varphi(D_e^{(1)})$.
Symmetry and factorization imply the inflation constraints for $p_\varphi$, which therefore admits a mixed realization by \zcref{thm:InflationSourceTransferTreeEquivalence}.
By \zcref{eq:InflationPolynomialObjectiveLift,eq:InflationOptimizationBounds}, we obtain
\begin{equation}\label{eq:InflationOptimizationDeFinettiBound}
    \begin{aligned}
        \varphi(\widehat{f})
        &=c_0+\sum_{i=1}^N c_i\prod_{j=1}^{k_i}\varphi(D_{e_{ij}}^{(j)})
        =f(p_\varphi),\\
        \inf_{p\in C_{\mathrm{mix}}(\Tree)}f(p)
        &\geq\lim_{m,d\to\infty}\beta_{m,d}^{\min}(f)
        =\omega(\widehat{f})
        =\int f(p_\varphi)\,d\mu(\varphi)
        \geq\inf_{p\in C_{\mathrm{mix}}(\Tree)}f(p).
    \end{aligned}
\end{equation}
Thus all inequalities are equalities, proving convergence for minimization.
Applying the same argument to $-f$ proves convergence for maximization.
\end{remark}

\section*{Acknowledgments}
We thank David Gross for insightful feedback.
This work was performed within the project COMPUTE, which is funded within the QuantERA II Programme that has received funding from the EU's H2020 research and innovation programme under the GA No.~101017733 {\normalsize\euflag}.
XX and MOR acknowledge funding by the ANR for the JCJC grants LINKS (No.~ANR-23-CE47-0003) and the T-ERC QNET (No.~ANR-24-ERCS-0008).
IK also acknowledges support of the Slovenian Research Agency program P1-0222 and grants J1-50002, N1-0217, J1-60011, J1-50001, J1-3004 and J1-60025.

\section*{Data availability statement}
Data sharing is not applicable to this article as no datasets were generated or analyzed during the current study.

\section*{Conflict of interest}
The authors have no competing interests to declare that are relevant to the content of this article.

\printbibliography

\appendix

\section{Explicit examples of source-transfer models and SDP relaxations}\label{sec:SimpleExample}
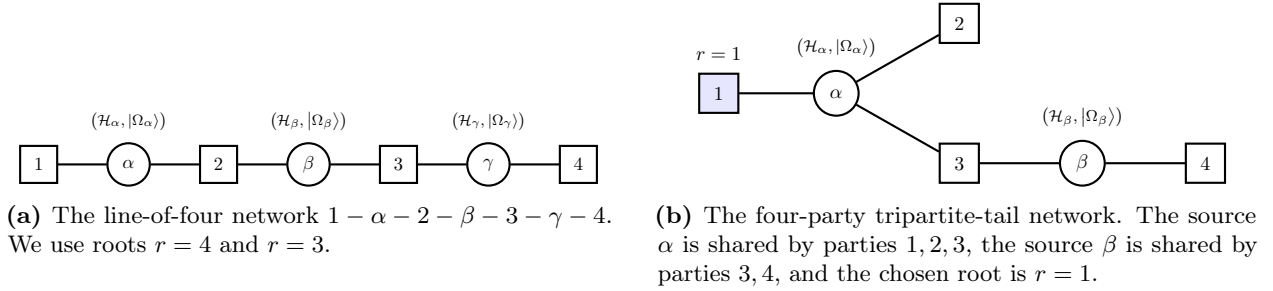
\begin{figure}[htbp]
    \centering
    \begin{subfigure}[t]{0.48\textwidth}
        \centering
        \begin{tikzpicture}[
            scale=0.70,
            transform shape,
            every node/.style={font=\small},
            party/.style={rectangle,draw,thick,minimum width=7mm,minimum height=7mm,inner sep=0pt},
            source/.style={circle,draw,thick,minimum size=8mm,inner sep=0pt},
            sourceinfo/.style={font=\scriptsize,align=center}
        ]
            \node[party]  (p1) at (0,0) {$1$};
            \node[source] (a)  at (1.7,0) {$\alpha$};
            \node[party]  (p2) at (3.4,0) {$2$};
            \node[source] (b)  at (5.1,0) {$\beta$};
            \node[party]  (p3) at (6.8,0) {$3$};
            \node[source] (c)  at (8.5,0) {$\gamma$};
            \node[party]  (p4) at (10.2,0) {$4$};

            \draw[thick] (p1)--(a)--(p2)--(b)--(p3)--(c)--(p4);
            \node[sourceinfo,above=3pt of a] {$\bigl(\cH_\alpha,\ket{\Omega_\alpha}\bigr)$};
            \node[sourceinfo,above=3pt of b] {$\bigl(\cH_\beta,\ket{\Omega_\beta}\bigr)$};
            \node[sourceinfo,above=3pt of c] {$\bigl(\cH_\gamma,\ket{\Omega_\gamma}\bigr)$};
        \end{tikzpicture}
        \caption{The line-of-four network $1-\alpha-2-\beta-3-\gamma-4$. We use roots $r=4$ and $r=3$. }
        \label{fig:appendix-line-four}
    \end{subfigure}
    \hfill
    \begin{subfigure}[t]{0.48\textwidth}
        \centering
        \begin{tikzpicture}[
            scale=0.74,
            transform shape,
            every node/.style={font=\small},
            party/.style={rectangle,draw,thick,minimum width=7mm,minimum height=7mm,inner sep=0pt},
            root/.style={party,fill=blue!10},
            source/.style={circle,draw,thick,minimum size=8mm,inner sep=0pt},
            sourceinfo/.style={font=\scriptsize,align=center}
        ]
            \node[root]   (p1) at (0,0) {$1$};
            \node[source] (a)  at (2.1,0) {$\alpha$};
            \node[party]  (p2) at (4.3,1.25) {$2$};
            \node[party]  (p3) at (4.3,-1.25) {$3$};
            \node[source] (b)  at (6.5,-1.25) {$\beta$};
            \node[party]  (p4) at (8.7,-1.25) {$4$};

            \draw[thick] (p1)--(a);
            \draw[thick] (a)--(p2);
            \draw[thick] (a)--(p3);
            \draw[thick] (p3)--(b)--(p4);
            \node[sourceinfo,above=3pt of a] {$\bigl(\cH_\alpha,\ket{\Omega_\alpha}\bigr)$};
            \node[sourceinfo,above=3pt of b] {$\bigl(\cH_\beta,\ket{\Omega_\beta}\bigr)$};
            \node[above=3pt of p1] {$r=1$};
        \end{tikzpicture}
        \caption{The four-party tripartite-tail network. The source $\alpha$ is shared by parties $1,2,3$, the source $\beta$ is shared by parties $3,4$, and the chosen root is $r=1$.}
        \label{fig:appendix-multipartite-four}
    \end{subfigure}
    \caption{
    Two source-party trees used to illustrate the source-transfer model and its SDP hierarchy explicitly.
    Squares denote measurement parties and circles denote independent sources.
    In the mixed model, each source $\delta$ carries a Hilbert space $\cH_\delta$ and a source vector $\ket{\Omega_\delta}$.
    }
    \label{fig:appendix-examples}
\end{figure}

In this appendix, we focus on two simple source-party trees to provide concrete examples for the main text. 
The line-of-four network in \zcref{fig:appendix-line-four} contains four measurement parties, who are arranged on a line sharing bipartite sources consecutively; we consider two possible choices of the root.
The four-party tripartite-tail network in \zcref{fig:appendix-multipartite-four} contains a tripartite source shared by parties $1,2,3$ and a bipartite source shared by parties $3,4$; we choose $1$ as the root.

For each example, we first derive the source-transfer model from a mixed realization, then reconstruct the mixed model explicitly from these source-transfer data, and finally write the corresponding specialization of the new multi-state SDP hierarchy in \zcref{def:SourceTransferSDPRelaxation}.
Even though the general results have already been proved in the main text, we make the discussion explicit and self-contained to motivate the definition of the general source-transfer model, clarify the mathematical proofs and calculations using special cases, and provide SDP programs for these concrete scenarios.

\subsection{Line-of-four network}\label{sec:AppendixLineFour}
Consider the source-party tree in \zcref{fig:appendix-line-four}
\begin{align}\label{eq:AppendixLineGraph}
    1-\alpha-2-\beta-3-\gamma-4.
\end{align}

\paragraph{The mixed model.}
By \zcref{def:MixedQuantumModelSourcePartyGraph}, its mixed realization consists of source Hilbert spaces $\cH_\alpha,\cH_\beta,\cH_\gamma$, the product source vector
\begin{align}
    \ket{\Omega_{\mathrm{line}}}\coloneqq\ket{\Omega_\alpha\otimes\Omega_\beta\otimes\Omega_\gamma},
\end{align}
and POVMs satisfying
\begin{equation}\label{eq:AppendixLineMixedLocations}
    \begin{aligned}
        \Meas{1}{a_1}{x_1}&\in\cM_\alpha^1, &
        \Meas{2}{a_2}{x_2}&\in\cM_\alpha^2\vnotimes\cM_\beta^2,\\
        \Meas{3}{a_3}{x_3}&\in\cM_\beta^3\vnotimes\cM_\gamma^3, &
        \Meas{4}{a_4}{x_4}&\in\cM_\gamma^4
    \end{aligned}
\end{equation}
with the required commutation on the von Neumann algebras.
Using the natural amplification of the POVMs to the global Hilbert space, the correlation is
\begin{align}\label{eq:AppendixLineMixedCorrelation}
    p_{\mathrm{line}}(\mathbf{a}|\mathbf{x})=\sandwich{\Omega_{\mathrm{line}}}{\prod_{j=1}^4\Meas{j}{a_j}{x_j}}{\Omega_{\mathrm{line}}}.
\end{align}

\subsubsection{Rooted at party \texorpdfstring{$4$}{4}}\label{sec:AppendixLineRootFour}
With root $r=4$, the parents of the parties are $\Parent{1} = \alpha$, $\Parent{2} = \beta$, and $\Parent{3} = \gamma$.
The child-party sets of the three sources are $\Children{\alpha}=\{1\}$, $\Children{\beta}=\{2\}$, and $\Children{\gamma}=\{3\}$, hence the joint source algebras have only one tensor factor.
Starting from the mixed realization, let
\begin{align}\label{eq:AppendixLineRootFourAlgebras}
    \cA_\alpha^1&\coloneqq\cM_\alpha^1, & \JoinAlg{\alpha}&=\cA_\alpha^1, &
    \cA_\beta^2&\coloneqq\cM_\beta^2, & \JoinAlg{\beta}&=\cA_\beta^2, &
    \cA_\gamma^3&\coloneqq\cM_\gamma^3, & \JoinAlg{\gamma}&=\cA_\gamma^3.
\end{align}
Let $\omega_\alpha,\omega_\beta,\omega_\gamma$ be the states induced by $\ket{\Omega_{\alpha}}, \ket{\Omega_{\beta}}, \ket{\Omega_{\gamma}}$ and restricted to $\JoinAlg{\alpha}, \JoinAlg{\beta}, \JoinAlg{\gamma}$, respectively, i.e.,
\begin{align}
	\omega_{\delta}: \JoinAlg{\delta} \to \mathbb{C}, \qquad q \mapsto \sandwich{\Omega_{\delta}}{q}{\Omega_{\delta}}
\end{align}
for $\delta = \alpha, \beta, \gamma$.

On the only non-root leaf $1$, we let
\begin{align}
    \meas{1}{a_1}{x_1}\coloneqq\Meas{1}{a_1}{x_1}\in\cA_\alpha^1.
\end{align}
The internal parties $2$ and $3$ define transfer maps
\begin{align}
    \instr{2}{a_2}{x_2}&:\JoinAlg{\alpha}\to\cA_\beta^2, &
    \instr{3}{a_3}{x_3}&:\JoinAlg{\beta}\to\cA_\gamma^3,
\end{align}
by
\begin{equation}\label{eq:AppendixLineRootFourTransferMaps}
    \begin{aligned}
        \instr{2}{a_2}{x_2}(q)&\coloneqq(\bra{\Omega_\alpha})\bigl(\Meas{2}{a_2}{x_2}\bigr)^{1/2}(q\otimes I_{\cH_\beta})\bigl(\Meas{2}{a_2}{x_2}\bigr)^{1/2}(\ket{\Omega_\alpha}),\\
        \instr{3}{a_3}{x_3}(s)&\coloneqq(\bra{\Omega_\beta})\bigl(\Meas{3}{a_3}{x_3}\bigr)^{1/2}(s\otimes I_{\cH_\gamma})\bigl(\Meas{3}{a_3}{x_3}\bigr)^{1/2}(\ket{\Omega_\beta}).
    \end{aligned}
\end{equation}
Their ranges lie in $\cA_{\beta}^2$ and $\cA_{\gamma}^3$, respectively.
They are completely positive since tensoring with identity, conjugation by positive operators, and sandwiching by vectors are all completely positive operations.
Note also the connection to the partial sandwich form of Radon--Nikodym derivative in \zcref{rem:RadonNikodymSandwichForm}.
POVM completeness gives the source-replacer condition with respect to $\omega_{\alpha}$ and $\omega_{\beta}$.
\begin{align}\label{eq:AppendixLineRootFourReplacers}
    \sum_{a_2}\instr{2}{a_2}{x_2}(q)&=\omega_\alpha(q)\id_{\cA_\beta^2}, &
    \sum_{a_3}\instr{3}{a_3}{x_3}(s)&=\omega_\beta(s)\id_{\cA_\gamma^3}.
\end{align}
At the root $4$, define
\begin{align}\label{eq:AppendixLineRootFourFunctional}
    \rstate{4}{a_4}{x_4}(t)\coloneqq\sandwich{\Omega_\gamma}{\bigl(\Meas{4}{a_4}{x_4}\bigr)^{1/2}t\bigl(\Meas{4}{a_4}{x_4}\bigr)^{1/2}}{\Omega_\gamma},
\end{align}
so that $\sum_{a_4}\rstate{4}{a_4}{x_4}=\omega_\gamma$.
The source-transfer expression for the correlation then reads
\begin{align}\label{eq:AppendixLineRootFourCorrelation}
    p_{\mathrm{line}}(\mathbf{a}|\mathbf{x})=\rstate{4}{a_4}{x_4}\!\left(\instr{3}{a_3}{x_3}\!\left(\instr{2}{a_2}{x_2}\!\left(\meas{1}{a_1}{x_1}\right)\right)\right).
\end{align}

\subsubsection{Rooted at party \texorpdfstring{$3$}{3}}\label{sec:AppendixLineRootThree}
We consider a different choice of root, $r = 3$.
In this case, the leaf party $4$ transfers its information to $3$ through the source $\gamma$, while the left branch starts from the leaf party $1$, transferring through party $2$ and eventually reaching party $3$.
We write
\begin{align}
    \cA_\alpha^1&\coloneqq\cM_\alpha^1, & \JoinAlg{\alpha}&=\cA_\alpha^1, &
    \cA_\beta^2&\coloneqq\cM_\beta^2, & \JoinAlg{\beta}&=\cA_\beta^2, &
    \cA_\gamma^4&\coloneqq\cM_\gamma^4, & \JoinAlg{\gamma}&=\cA_\gamma^4.
\end{align}
The leaf POVMs are
\begin{align}
    \meas{1}{a_1}{x_1}&\coloneqq\Meas{1}{a_1}{x_1}\in\cA_\alpha^1, &
    \meas{4}{a_4}{x_4}&\coloneqq\Meas{4}{a_4}{x_4}\in\cA_\gamma^4.
\end{align}
The transfer map for party $2$ $\instr{2}{a_2}{x_2}:\JoinAlg{\alpha}\to\cA_\beta^2$ is defined as in \zcref{eq:AppendixLineRootFourTransferMaps}, satisfying the first source-replacer condition in \zcref{eq:AppendixLineRootFourReplacers}.
The root functionals are
\begin{align}\label{eq:AppendixLineRootThreeFunctional}
    \rstate{3}{a_3}{x_3}(s)\coloneqq\sandwich{\Omega_\beta\otimes\Omega_\gamma}{\bigl(\Meas{3}{a_3}{x_3}\bigr)^{1/2}s\bigl(\Meas{3}{a_3}{x_3}\bigr)^{1/2}}{\Omega_\beta\otimes\Omega_\gamma},
\end{align}
for $s\in\JoinAlg{\beta}\maxotimes\JoinAlg{\gamma}$, and satisfy
\begin{align}
    \sum_{a_3}\rstate{3}{a_3}{x_3}=\omega_\beta\otimes\omega_\gamma.
\end{align}
Thus, one checks that the correlation can be rewritten as
\begin{align}\label{eq:AppendixLineRootThreeCorrelation}
    p_{\mathrm{line}}(\mathbf{a}|\mathbf{x})=\rstate{3}{a_3}{x_3}\!\left(\instr{2}{a_2}{x_2}\!\left(\meas{1}{a_1}{x_1}\right)\otimes\meas{4}{a_4}{x_4}\right).
\end{align}

We consider the following specialization of \zcref{thm:SourceTransferTreeEquivalence} to line-of-four networks.

\begin{proposition}[Two source-transfer descriptions of the line of four]\label{prop:AppendixLineEquivalence}
For a correlation $p_{\mathrm{line}}(\mathbf{a}|\mathbf{x})$ on the network described by \zcref{eq:AppendixLineGraph}, the following are equivalent:
\begin{enumerate}[label=(\roman*)]
    \item $p_{\mathrm{line}}$ admits a mixed realization satisfying \zcref{eq:AppendixLineMixedLocations,eq:AppendixLineMixedCorrelation};
    \item $p_{\mathrm{line}}$ admits the source-transfer realization rooted at $4$ described by \zcref{eq:AppendixLineRootFourAlgebras,eq:AppendixLineRootFourTransferMaps,eq:AppendixLineRootFourReplacers,eq:AppendixLineRootFourFunctional,eq:AppendixLineRootFourCorrelation};
    \item $p_{\mathrm{line}}$ admits the source-transfer realization rooted at $3$ described by \zcref{eq:AppendixLineRootThreeFunctional,eq:AppendixLineRootThreeCorrelation} and the preceding source algebras, leaf POVMs, and the transfer map.
\end{enumerate}
\end{proposition}
\begin{proof}
The constructions in \zcref{sec:AppendixLineRootFour,sec:AppendixLineRootThree} prove $(\mathrm{i})\Rightarrow(\mathrm{ii})$ and $(\mathrm{i})\Rightarrow(\mathrm{iii})$.
We spell out the converse calculations, which are completely analogous to those of \zcref{thm:SourceTransferTreeEquivalence}.

Suppose first that the source-transfer realization rooted at $4$ is given.
Let $\gns{\alpha}$, $\gns{\beta}$, and $\gns{\gamma}$ be the GNS representations of the three source states, and set
\begin{align}
    \cM_\alpha&\coloneqq\pi_\alpha(\JoinAlg{\alpha})'', &
    \cM_\beta&\coloneqq\pi_\beta(\JoinAlg{\beta})'', &
    \cM_\gamma&\coloneqq\pi_\gamma(\JoinAlg{\gamma})''.
\end{align}
Identify the endpoint algebras by
\begin{align}\label{eq:AppendixLineRootFourEndpointAlgebras}
    \cM_\alpha^1&=\cM_\alpha, & \cM_\alpha^2&=\cM_\alpha', &
    \cM_\beta^2&=\cM_\beta, & \cM_\beta^3&=\cM_\beta', &
    \cM_\gamma^3&=\cM_\gamma, & \cM_\gamma^4&=\cM_\gamma'.
\end{align}
Applying \zcref{cor:RadonNikodymSourceProductReplacer,rem:FunctionalProductRadonNikodymIdentity} gives
\begin{equation}\label{eq:AppendixLineRootFourRecoveredLocations}
    \begin{aligned}
        \Meas{1}{a_1}{x_1}&\coloneqq\pi_\alpha\!\left(\meas{1}{a_1}{x_1}\right)\in\cM_\alpha, & \qquad
        \Meas{2}{a_2}{x_2}&\in\cM_\alpha'\vnotimes\cM_\beta,\\
        \Meas{3}{a_3}{x_3}&\in\cM_\beta'\vnotimes\cM_\gamma, &
        \Meas{4}{a_4}{x_4}&\in\cM_\gamma'.
    \end{aligned}
\end{equation}
The correlation is recovered by applying the expectation-value identities successively:
\begin{equation}\label{eq:AppendixLineRootFourRecoveredCorrelation}
\begin{aligned}
    p_{\mathrm{line}}(\mathbf{a}|\mathbf{x})
    &=\rstate{4}{a_4}{x_4}\!\left(\instr{3}{a_3}{x_3}\!\left(\instr{2}{a_2}{x_2}\!\left(\meas{1}{a_1}{x_1}\right)\right)\right)\\
    &=\sandwich{\Omega_\gamma}{\Meas{4}{a_4}{x_4}\,\pi_\gamma\!\left(\instr{3}{a_3}{x_3}\!\left(\instr{2}{a_2}{x_2}\!\left(\meas{1}{a_1}{x_1}\right)\right)\right)}{\Omega_\gamma}\\
    &=\sandwich{\Omega_\beta\otimes\Omega_\gamma}{\Meas{3}{a_3}{x_3}\left(\pi_\beta\!\left(\instr{2}{a_2}{x_2}\!\left(\meas{1}{a_1}{x_1}\right)\right)\otimes\Meas{4}{a_4}{x_4}\right)}{\Omega_\beta\otimes\Omega_\gamma}\\
    &=\sandwich{\Omega_{\mathrm{line}}}{\bigl(\Meas{2}{a_2}{x_2}\otimes I_{\cH_\gamma}\bigr)\left(\pi_\alpha\!\left(\meas{1}{a_1}{x_1}\right)\otimes\Meas{3}{a_3}{x_3}\bigl(I_{\cH_\beta}\otimes\Meas{4}{a_4}{x_4}\bigr)\right)}{\Omega_{\mathrm{line}}}\\
    &=\sandwich{\Omega_{\mathrm{line}}}{\prod_{j=1}^4\Meas{j}{a_j}{x_j}}{\Omega_{\mathrm{line}}}.
\end{aligned}
\end{equation}
Due to the commutation relations in \zcref{eq:AppendixLineRootFourRecoveredLocations}, the second equality uses \zcref{eq:FunctionalProductRadonNikodymIdentity}; the third and fourth use \zcref{eq:ProductRadonNikodymExpectationIdentity} at parties $3$ and $2$, respectively.

Now suppose that the source-transfer realization rooted at $3$ is given instead.
Use the same GNS notation, but identify
\begin{align}\label{eq:AppendixLineRootThreeEndpointAlgebras}
    \cM_\alpha^1&=\cM_\alpha, & \cM_\alpha^2&=\cM_\alpha', &
    \cM_\beta^2&=\cM_\beta, & \cM_\beta^3&=\cM_\beta', &
    \cM_\gamma^4&=\cM_\gamma, & \cM_\gamma^3&=\cM_\gamma'.
\end{align}
The reconstructed POVMs by \zcref{cor:RadonNikodymSourceProductReplacer,rem:FunctionalProductRadonNikodymIdentity} satisfy
\begin{equation}\label{eq:AppendixLineRootThreeRecoveredLocations}
    \begin{aligned}
        \Meas{1}{a_1}{x_1}&\coloneqq\pi_\alpha\!\left(\meas{1}{a_1}{x_1}\right)\in\cM_\alpha, & \qquad
        \Meas{2}{a_2}{x_2}&\in\cM_\alpha'\vnotimes\cM_\beta,\\
        \Meas{3}{a_3}{x_3}&\in\cM_\beta'\vnotimes\cM_\gamma', &
        \Meas{4}{a_4}{x_4}&\coloneqq\pi_\gamma\!\left(\meas{4}{a_4}{x_4}\right)\in\cM_\gamma.
    \end{aligned}
\end{equation}
The corresponding calculation is
\begin{equation}\label{eq:AppendixLineRootThreeRecoveredCorrelation}
\begin{aligned}
    p_{\mathrm{line}}(\mathbf{a}|\mathbf{x})
    &=\rstate{3}{a_3}{x_3}\!\left(\instr{2}{a_2}{x_2}\!\left(\meas{1}{a_1}{x_1}\right)\otimes\meas{4}{a_4}{x_4}\right)\\
    &=\sandwich{\Omega_\beta\otimes\Omega_\gamma}{\Meas{3}{a_3}{x_3}\left(\pi_\beta\!\left(\instr{2}{a_2}{x_2}\!\left(\meas{1}{a_1}{x_1}\right)\right)\otimes\pi_\gamma\!\left(\meas{4}{a_4}{x_4}\right)\right)}{\Omega_\beta\otimes\Omega_\gamma}\\
    &=\sandwich{\Omega_{\mathrm{line}}}{\bigl(\Meas{2}{a_2}{x_2}\otimes I_{\cH_\gamma}\bigr)\left(\pi_\alpha\!\left(\meas{1}{a_1}{x_1}\right)\otimes\Meas{3}{a_3}{x_3}\bigl(I_{\cH_\beta}\otimes\pi_\gamma\!\left(\meas{4}{a_4}{x_4}\right)\bigr)\right)}{\Omega_{\mathrm{line}}}\\
    &=\sandwich{\Omega_{\mathrm{line}}}{\prod_{j=1}^4\Meas{j}{a_j}{x_j}}{\Omega_{\mathrm{line}}}.
\end{aligned}
\end{equation}
Here the first equality after the source-transfer expression is \zcref{eq:FunctionalProductRadonNikodymIdentity}, and the next equality, corresponding to the transfer-map at $2$, follows from \zcref{eq:ProductRadonNikodymExpectationIdentity}.
This proves both converses.
\end{proof}

\subsubsection{Source-transfer SDP relaxations for the two roots}\label{sec:AppendixLineSDP}
We now specialize \zcref{def:SourceTransferSDPRelaxation} to the two root choices for the line-of-four network.

\paragraph{Root $r=4$.}
The formal source algebras are constructed sequentially:
\begin{equation}
    \begin{aligned}
        \fA_\alpha&=\fA_\alpha^1=\!{}^*\!\operatorname{alg}\!\left(\meas{1}{a_1}{x_1}\right), &
    		\fA_\beta&=\fA_\beta^2=\!{}^*\!\operatorname{alg}\!\left(\formalInstr{2}{a_2}{x_2}(q):q\in\fA_\alpha\right),\\
    		\fA_\gamma&=\fA_\gamma^3=\!{}^*\!\operatorname{alg}\!\left(\formalInstr{3}{a_3}{x_3}(s):s\in\fA_\beta\right).
    \end{aligned}
\end{equation}
Thus $\InPolyAlg{2}=\fA_\alpha$, $\InPolyAlg{3}=\fA_\beta$, and $\InPolyAlg{4}=\fA_\gamma$.
Introduce formal state and root functional symbols $\varsigma_{\alpha}$, $\varsigma_{\beta}$, $\varsigma_{\gamma}$, $\rPolystate{a_r}{x_r}$, and the associated multi-state polynomial algebras for the corresponding formal $*$-algebras.
The source-replacer and correlation polynomials are
\begin{equation}\label{eq:AppendixLineRootFourSDPPolynomials}
	\begin{aligned}
		h_{x_2}^{2}(p_1,q,p_2)&=\varsigma_\beta\!\left(p_1^*\left[\sum_{a_2}\formalInstr{2}{a_2}{x_2}(q)-\varsigma_\alpha(q)\right]p_2\right),\\
   	 	h_{x_3}^{3}(c_1,s,c_2)&=\varsigma_\gamma\!\left(c_1^*\left[\sum_{a_3}\formalInstr{3}{a_3}{x_3}(s)-\varsigma_\beta(s)\right]c_2\right),\\
    		h_{x_4}^{4}(t)&=\sum_{a_4}\rPolystate{a_4}{x_4}(t)-\varsigma_\gamma(t),\\
    		h_{\mathbf{a}|\mathbf{x}}^{(4)}&=\rPolystate{a_4}{x_4}\!\left(\formalInstr{3}{a_3}{x_3}\!\left(\formalInstr{2}{a_2}{x_2}\!\left(\meas{1}{a_1}{x_1}\right)\right)\right).
	\end{aligned}
\end{equation}
The two complete-positivity matrices are
\begin{equation}\label{eq:AppendixLineRootFourCPMatrices}
\begin{aligned}
    \left[C_{d-1}\!\left(\formalInstr{2}{a_2}{x_2};L_d\right)\right]_{(p_1,q_1),(p_2,q_2)}
    &=L_d\!\left(\varsigma_\beta\!\left(p_1^*\formalInstr{2}{a_2}{x_2}(q_1^*q_2)p_2\right)\right),\\
    \left[C_{d-1}\!\left(\formalInstr{3}{a_3}{x_3};L_d\right)\right]_{(c_1,s_1),(c_2,s_2)}
    &=L_d\!\left(\varsigma_\gamma\!\left(c_1^*\formalInstr{3}{a_3}{x_3}(s_1^*s_2)c_2\right)\right),
\end{aligned}
\end{equation}
where $p_i\in\operatorname{Mon}(\fatS_\beta)$, $q_i\in\operatorname{Mon}(\fatS_\alpha)$ with $\deg(p_i)+\deg(q_i)\leq d-1$, and $c_i\in\operatorname{Mon}(\fatS_\gamma)$, $s_i\in\operatorname{Mon}(\fatS_\beta)$ with $\deg(c_i)+\deg(s_i)\leq d-1$.

At level $d$, by identifying the generators iteratively, the source and root archimedean conditions are
\begin{equation}\label{eq:AppendixLineRootFourArchimedean}
\begin{alignedat}{2}
    &H_{d-1}^{\alpha}\!\left(\id-\bigl(\meas{1}{a_1}{x_1}\bigr)^*\meas{1}{a_1}{x_1};L_d\right)\succeq0,
    &&\quad(a_1,x_1)\in\mathsf{A}_1\times\mathsf{X}_1,\\
    &H_{d-1-\deg(q)}^{\beta}\!\left(\id-\bigl(\formalInstr{2}{a_2}{x_2}(q)\bigr)^*\formalInstr{2}{a_2}{x_2}(q);L_d\right)\succeq0,
    &&\quad\substack{q\in\operatorname{Mon}(\fA_\alpha),\\(a_2,x_2)\in\mathsf{A}_2\times\mathsf{X}_2},\\
    &H_{d-1-\deg(s)}^{\gamma}\!\left(\id-\bigl(\formalInstr{3}{a_3}{x_3}(s)\bigr)^*\formalInstr{3}{a_3}{x_3}(s);L_d\right)\succeq0,
    &&\quad\substack{s\in\operatorname{Mon}(\fA_\beta),\\(a_3,x_3)\in\mathsf{A}_3\times\mathsf{X}_3},\\
    &H_{d-2-\deg(s)}^{a_4|x_4}\!\left(\id-\bigl(\formalInstr{3}{a_3}{x_3}(s)\bigr)^*\formalInstr{3}{a_3}{x_3}(s);L_d\right)\succeq0,
    &&\quad\substack{s\in\operatorname{Mon}(\fA_\beta),\\(a_3,x_3)\in\mathsf{A}_3\times\mathsf{X}_3,\\(a_4,x_4)\in\mathsf{A}_4\times\mathsf{X}_4}.
\end{alignedat}
\end{equation}
The remaining level-$d$ constraints are
\begin{equation}\label{eq:AppendixLineRootFourSDP}
\begin{alignedat}{3}
    &H_d^\delta(L_d)\succeq0, &&\quad\delta\in\{\alpha,\beta,\gamma\}, &&\text{(source states)}\\
    &H_{d-1}^{a_4|x_4}(L_d)\succeq0, &&\quad(a_4,x_4)\in\mathsf{A}_4\times\mathsf{X}_4, &&\text{(root functionals)}\\
    &H_{d-1}^{\alpha}\!\left(\meas{1}{a_1}{x_1};L_d\right)\succeq0, &&\quad(a_1,x_1)\in\mathsf{A}_1\times\mathsf{X}_1, &&\text{(leaf POVM)}\\
    &C_{d-1}\!\left(\formalInstr{j}{a_j}{x_j};L_d\right)\succeq0, &&\quad j\in\{2,3\},\ (a_j,x_j)\in\mathsf{A}_j\times\mathsf{X}_j, &&\text{(complete positivity)}\\
    &L_d\!\left(h_{x_2}^{2}(p_1,q,p_2)\right)=0, &&\quad p_1,p_2\in\operatorname{Mon}(\fatS_\beta),\ q\in\operatorname{Mon}(\fatS_\alpha), &&\text{(replacer at $2$)}\\
    &L_d\!\left(h_{x_3}^{3}(c_1,s,c_2)\right)=0, &&\quad c_1,c_2\in\operatorname{Mon}(\fatS_\gamma),\ s\in\operatorname{Mon}(\fatS_\beta), &&\text{(replacer at $3$)}\\
    &L_d\!\left(h_{x_4}^{4}(t)\right)=0, &&\quad t\in\operatorname{Mon}(\fatS_\gamma), &&\text{(replacer at $4$).}
\end{alignedat}
\end{equation}
All equalities are imposed whenever the total degree is at most $2d$.
The prescribed-correlation constraints are
\begin{align}
    L_d\!\left(s_1^*\left[h_{\mathbf{a}|\mathbf{x}}^{(4)}-p_{\mathrm{line}}(\mathbf{a}|\mathbf{x})\right]s_2\right)=0
\end{align}
for all compatible multi-state monomials $s_1,s_2$.
Since $\deg(h_{\mathbf{a}|\mathbf{x}}^{(4)})=4$, the correlation is first imposable at $d=2$.

\paragraph{Root $r=3$.}
Now $\fA_\alpha=\fA_\alpha^1$ is generated by $\meas{1}{a_1}{x_1}$, $\fA_\beta=\fA_\beta^2$ is generated by $\formalInstr{2}{a_2}{x_2}(q)$, and $\fA_\gamma=\fA_\gamma^4$ is generated by $\meas{4}{a_4}{x_4}$.
Introduce the corresponding formal state and root functional symbols $\varsigma_{\alpha}$, $\varsigma_{\beta}$, $\varsigma_{\gamma}$, and $\rPolystate{a_r}{x_r}$, and the associated multi-state polynomial algebras.
The source-replacer and correlation polynomials are
\begin{equation}\label{eq:AppendixLineRootThreeSDPPolynomials}
	\begin{aligned}
		h_{x_2}^{2}(p_1,q,p_2)&=\varsigma_\beta\!\left(p_1^*\left[\sum_{a_2}\formalInstr{2}{a_2}{x_2}(q)-\varsigma_\alpha(q)\right]p_2\right),\\
    		h_{x_3}^{3}(t)&=\sum_{a_3}\rPolystate{a_3}{x_3}(t)-(\varsigma_\beta\varsigma_\gamma)(t),\\
    		h_{\mathbf{a}|\mathbf{x}}^{(3)}&=\rPolystate{a_3}{x_3}\!\left(\formalInstr{2}{a_2}{x_2}\!\left(\meas{1}{a_1}{x_1}\right)\otimes\meas{4}{a_4}{x_4}\right),
	\end{aligned}
\end{equation}
where $(\varsigma_\beta\varsigma_\gamma)(s\otimes t)=\varsigma_\beta(s)\varsigma_\gamma(t)$.
The complete-positivity matrix is
\begin{align}\label{eq:AppendixLineRootThreeCPMatrix}
    \left[C_{d-1}\!\left(\formalInstr{2}{a_2}{x_2};L_d\right)\right]_{(p_1,q_1),(p_2,q_2)}
    =L_d\!\left(\varsigma_\beta\!\left(p_1^*\formalInstr{2}{a_2}{x_2}(q_1^*q_2)p_2\right)\right),
\end{align}
where $p_i\in\operatorname{Mon}(\fatS_\beta)$ and $q_i\in\operatorname{Mon}(\fatS_\alpha)$ satisfy $\deg(p_i)+\deg(q_i)\leq d-1$.

At level $d$, the source and root archimedean conditions are
\begin{equation}\label{eq:AppendixLineRootThreeArchimedean}
\begin{alignedat}{2}
    &H_{d-1}^{\alpha}\!\left(\id-\bigl(\meas{1}{a_1}{x_1}\bigr)^*\meas{1}{a_1}{x_1};L_d\right)\succeq0,
    &&\quad(a_1,x_1)\in\mathsf{A}_1\times\mathsf{X}_1,\\
    &H_{d-1-\deg(q)}^{\beta}\!\left(\id-\bigl(\formalInstr{2}{a_2}{x_2}(q)\bigr)^*\formalInstr{2}{a_2}{x_2}(q);L_d\right)\succeq0,
    &&\quad\substack{q\in\operatorname{Mon}(\fA_\alpha),\\(a_2,x_2)\in\mathsf{A}_2\times\mathsf{X}_2},\\
    &H_{d-1}^{\gamma}\!\left(\id-\bigl(\meas{4}{a_4}{x_4}\bigr)^*\meas{4}{a_4}{x_4};L_d\right)\succeq0,
    &&\quad(a_4,x_4)\in\mathsf{A}_4\times\mathsf{X}_4,\\
    &H_{d-2-\deg(q)}^{a_3|x_3}\!\left(\id-\bigl(\formalInstr{2}{a_2}{x_2}(q)\otimes\id\bigr)^*\bigl(\formalInstr{2}{a_2}{x_2}(q)\otimes\id\bigr);L_d\right)\succeq0,
    &&\quad\substack{q\in\operatorname{Mon}(\fA_\alpha),\\(a_2,x_2)\in\mathsf{A}_2\times\mathsf{X}_2,\\(a_3,x_3)\in\mathsf{A}_3\times\mathsf{X}_3},\\
    &H_{d-2}^{a_3|x_3}\!\left(\id-\bigl(\id\otimes\meas{4}{a_4}{x_4}\bigr)^*\bigl(\id\otimes\meas{4}{a_4}{x_4}\bigr);L_d\right)\succeq0,
    &&\quad\substack{(a_4,x_4)\in\mathsf{A}_4\times\mathsf{X}_4,\\(a_3,x_3)\in\mathsf{A}_3\times\mathsf{X}_3}.
\end{alignedat}
\end{equation}
The remaining level-$d$ constraints are
\begin{equation}\label{eq:AppendixLineRootThreeSDP}
\begin{alignedat}{3}
    &H_d^\delta(L_d)\succeq0, &&\quad\delta\in\{\alpha,\beta,\gamma\}, &&\text{(source states)}\\
    &H_{d-1}^{a_3|x_3}(L_d)\succeq0, &&\quad(a_3,x_3)\in\mathsf{A}_3\times\mathsf{X}_3, &&\text{(root functionals)}\\
    &H_{d-1}^{\alpha}\!\left(\meas{1}{a_1}{x_1};L_d\right)\succeq0, &&\quad(a_1,x_1)\in\mathsf{A}_1\times\mathsf{X}_1, &&\text{(leaf at $1$)}\\
    &H_{d-1}^{\gamma}\!\left(\meas{4}{a_4}{x_4};L_d\right)\succeq0, &&\quad(a_4,x_4)\in\mathsf{A}_4\times\mathsf{X}_4, &&\text{(leaf at $4$)}\\
    &C_{d-1}\!\left(\formalInstr{2}{a_2}{x_2};L_d\right)\succeq0, &&\quad(a_2,x_2)\in\mathsf{A}_2\times\mathsf{X}_2, &&\text{(complete positivity)}\\
    &L_d\!\left(h_{x_2}^{2}(p_1,q,p_2)\right)=0, &&\quad p_1,p_2\in\operatorname{Mon}(\fatS_\beta),\ q\in\operatorname{Mon}(\fatS_\alpha), &&\text{(replacer at $2$)}\\
    &L_d\!\left(h_{x_3}^{3}(t)\right)=0, &&\quad t\in\operatorname{Mon}(\InfatS{3}), &&\text{(replacer at $3$).}
\end{alignedat}
\end{equation}
The prescribed-correlation constraints are generated by $h_{\mathbf{a}|\mathbf{x}}^{(3)}-p_{\mathrm{line}}(\mathbf{a}|\mathbf{x})$ and again first testable at $d=2$.

This explicitly shows how the two roots lead to different finite SDP descriptions of the same limiting correlation set.

\subsection{Four-party tripartite-tail network}\label{sec:AppendixMultipartiteFour}
Consider the source-party tree in \zcref{fig:appendix-multipartite-four}:
\begin{align}\label{eq:AppendixMultipartiteGraph}
    1-\alpha-2, \qquad 1-\alpha-3-\beta-4,
\end{align}
where the notation means that the source $\alpha$ is shared by parties $1,2,3$, while $\beta$ is shared by parties $3,4$.
Fix party $r = 1$ as the root.
Then $\Parent{4}=\beta$, $\Parent{\beta}=3$, $\Parent{3}=\alpha$, and $\Children{\alpha}=\{2,3\}$.
The non-root leaf parties are $2$ and $4$.

\paragraph{The mixed model.}
A mixed realization consists of source spaces $\cH_\alpha,\cH_\beta$, the vector $\ket{\Omega_{\mathrm{tail}}}\coloneqq\ket{\Omega_\alpha\otimes\Omega_\beta}$, and mutually commuting von Neumann algebras
\begin{align}
    [\cM_\alpha^u,\cM_\alpha^v]=0, \quad u\neq v, u,v\in\{1,2,3\},  \qquad \quad [\cM_\beta^3,\cM_\beta^4]=0.
\end{align}
The POVMs satisfy
\begin{align}\label{eq:AppendixMultipartiteMixedLocations}
    \Meas{1}{a_1}{x_1}&\in\cM_\alpha^1, &
    \Meas{2}{a_2}{x_2}&\in\cM_\alpha^2, &
    \Meas{3}{a_3}{x_3}&\in\cM_\alpha^3\vnotimes\cM_\beta^3, &
    \Meas{4}{a_4}{x_4}&\in\cM_\beta^4,
\end{align}
and the correlation reads
\begin{align}\label{eq:AppendixMultipartiteMixedCorrelation}
    p_{\mathrm{tail}}(\mathbf{a}|\mathbf{x})=\sandwich{\Omega_{\mathrm{tail}}}{\prod_{j=1}^4\Meas{j}{a_j}{x_j}}{\Omega_{\mathrm{tail}}}.
\end{align}

\subsubsection{Source-transfer description rooted at party \texorpdfstring{$1$}{1}}\label{sec:AppendixMultipartiteRootOne}
The source $\beta$ has child party $4$, while the source $\alpha$ has child parties $2$ and $3$.
We define the source-transfer algebras by
\begin{equation}
    \begin{aligned}
        \cA_\beta^4&\coloneqq\cM_\beta^4, & \JoinAlg{\beta}&=\cA_\beta^4,\\
        \cA_\alpha^2&\coloneqq\cM_\alpha^2, & \cA_\alpha^3&\coloneqq\cM_\alpha^3, &
        \JoinAlg{\alpha}&\coloneqq\cA_\alpha^2\maxotimes\cA_\alpha^3.
    \end{aligned}
\end{equation}
By construction, one has a natural representation of $\cA_\alpha^2$ and $\cA_\alpha^3$ in $\Bof{\cH_\alpha}$ 
\begin{align}\label{eq:AppendixMultipartiteMultiplicationRepresentation}
    \iota_\alpha:\JoinAlg{\alpha}\to\Bof{\cH_\alpha}, \qquad \iota_\alpha(q_2\otimes q_3)=q_2q_3.
\end{align}
(Recall the proof of \zcref{thm:SourceTransferTreeEquivalence} where we omit this representation for simplicity.)
Define the vector-induced source states to be
\begin{align}
    \omega_\alpha(q)&\coloneqq\sandwich{\Omega_\alpha}{\iota_\alpha(q)}{\Omega_\alpha}, &
    \omega_\beta(s)&\coloneqq\sandwich{\Omega_\beta}{s}{\Omega_\beta}.
\end{align}

At the two non-root leaves $2$ and $4$, we simply let
\begin{align}
    \meas{2}{a_2}{x_2}&\coloneqq\Meas{2}{a_2}{x_2}\in\cA_\alpha^2, &
    \meas{4}{a_4}{x_4}&\coloneqq\Meas{4}{a_4}{x_4}\in\cA_\beta^4.
\end{align}
Define the completely positive transfer map for the internal party $3$
\begin{align}
    \instr{3}{a_3}{x_3}:\JoinAlg{\beta} \to \cA_{\Parent{3}}^3 = \cA_\alpha^3 \subset \JoinAlg{\alpha}
\end{align}
by
\begin{align}\label{eq:AppendixMultipartiteTransferMap}
    \instr{3}{a_3}{x_3}(s)\coloneqq(\bra{\Omega_\beta})\bigl(\Meas{3}{a_3}{x_3}\bigr)^{1/2}(I_{\cH_\alpha}\otimes s)\bigl(\Meas{3}{a_3}{x_3}\bigr)^{1/2}(\ket{\Omega_\beta}).
\end{align}
The range is contained in $\cA_\alpha^3$ by construction, and it satisfies the source-replacer condition
\begin{align}\label{eq:AppendixMultipartiteTransferReplacer}
    \sum_{a_3}\instr{3}{a_3}{x_3}(s)=\omega_\beta(s)\id_{\cA_\alpha^3}.
\end{align}
Finally, the root functionals are
\begin{align}\label{eq:AppendixMultipartiteRootFunctional}
    \rstate{1}{a_1}{x_1}(q)\coloneqq\sandwich{\Omega_\alpha}{\bigl(\Meas{1}{a_1}{x_1}\bigr)^{1/2}\iota_\alpha(q)\bigl(\Meas{1}{a_1}{x_1}\bigr)^{1/2}}{\Omega_\alpha}, \qquad q\in\JoinAlg{\alpha},
\end{align}
with $\sum_{a_1}\rstate{1}{a_1}{x_1}=\omega_\alpha$.
The transfer message at $\alpha$ is
\begin{align}
    \TransferMsg{\alpha}=\meas{2}{a_2}{x_2}\otimes\instr{3}{a_3}{x_3}\!\left(\meas{4}{a_4}{x_4}\right),
\end{align}
giving the correlation
\begin{align}\label{eq:AppendixMultipartiteTransferCorrelation}
    p_{\mathrm{tail}}(\mathbf{a}|\mathbf{x})=\rstate{1}{a_1}{x_1}\!\left(\meas{2}{a_2}{x_2}\otimes\instr{3}{a_3}{x_3}\!\left(\meas{4}{a_4}{x_4}\right)\right).
\end{align}

The following proposition is the specialization of the general results in \zcref{thm:SourceTransferTreeEquivalence} to the tripartite-tail network.
\begin{proposition}[Source-transfer description of the four-party tripartite-tail network]\label{prop:AppendixMultipartiteEquivalence}
A correlation $p_{\mathrm{tail}}(\mathbf{a}|\mathbf{x})$ admits the mixed realization \zcref{eq:AppendixMultipartiteMixedLocations,eq:AppendixMultipartiteMixedCorrelation} if and only if it admits the source-transfer realization rooted at $1$ described in \zcref{sec:AppendixMultipartiteRootOne}.
\end{proposition}
\begin{proof}
The forward implication was constructed already above.
Conversely, as in the proof of \zcref{thm:SourceTransferTreeEquivalence}, let $\gns{\alpha}$ be the GNS representation of $(\JoinAlg{\alpha},\omega_\alpha)$, and let $\gns{\beta}$ be the GNS representation of $(\JoinAlg{\beta},\omega_\beta)$.
Let
\begin{align}
    \cM_\alpha^2&\coloneqq\pi_\alpha(\cA_\alpha^2)'', &
    \cM_\alpha^3&\coloneqq\pi_\alpha(\cA_\alpha^3)'', &
    \cM_\alpha&\coloneqq\pi_\alpha(\JoinAlg{\alpha})''=\cM_\alpha^2\vee\cM_\alpha^3, &
    \cM_\beta&\coloneqq\pi_\beta(\JoinAlg{\beta})''.
\end{align}
Define the measurement observable algebras of the parties:
\begin{align}\label{eq:AppendixMultipartiteEndpointAlgebras}
    \cM_\alpha^1&\coloneqq\cM_\alpha', &
    \cM_\beta^3&\coloneqq\cM_\beta', &
    \cM_\beta^4&\coloneqq\cM_\beta.
\end{align}
By construction from the maximal tensor products and commutants they are indeed mutually commuting.
The represented leaf POVMs are 
\begin{align}
    \Meas{2}{a_2}{x_2}&\coloneqq\pi_\alpha\!\left(\meas{2}{a_2}{x_2}\right)\in\cM_\alpha^2, &
    \Meas{4}{a_4}{x_4}&\coloneqq\pi_\beta\!\left(\meas{4}{a_4}{x_4}\right)\in\cM_\beta.
\end{align}
Applying \zcref{cor:RadonNikodymSourceProductReplacer} to the transfer map $\instr{3}{a_3}{x_3}: \JoinAlg{\beta} \to \cA_{\Parent{3}}^3 \subset \JoinAlg{\alpha}$ gives
\begin{align}\label{eq:AppendixMultipartiteRecoveredLocations}
    \Meas{3}{a_3}{x_3}\in\cM_\alpha^3\vnotimes\cM_\beta' = \cM_\alpha^3\vnotimes\cM_\beta^3.
\end{align}
The root functionals by \zcref{rem:FunctionalProductRadonNikodymIdentity} give
\begin{align}
    \Meas{1}{a_1}{x_1}\in\cM_\alpha' = \cM_{\alpha}^1.
\end{align}
Thus all four measurement POVMs have exactly the locations as in \zcref{eq:AppendixMultipartiteMixedLocations}.

To recover the correlation, write $\Meas{2}{a_2}{x_2}=\pi_\alpha(\meas{2}{a_2}{x_2}\otimes\id)$ and identify each operator with its natural amplification.
Then
\begin{equation}\label{eq:AppendixMultipartiteRecoveredCorrelation}
\begin{aligned}
    p_{\mathrm{tail}}(\mathbf{a}|\mathbf{x})
    &=\rstate{1}{a_1}{x_1}\!\left(\meas{2}{a_2}{x_2}\otimes\instr{3}{a_3}{x_3}\!\left(\meas{4}{a_4}{x_4}\right)\right)\\
    &=\sandwich{\Omega_\alpha}{\Meas{1}{a_1}{x_1}\Meas{2}{a_2}{x_2}\,\pi_\alpha\!\left(\instr{3}{a_3}{x_3}\!\left(\meas{4}{a_4}{x_4}\right)\right)}{\Omega_\alpha}\\
    &=\sandwich{\Omega_\beta\otimes\Omega_\alpha}{\Meas{3}{a_3}{x_3}\left(\pi_\beta\!\left(\meas{4}{a_4}{x_4}\right)\otimes\Meas{1}{a_1}{x_1}\Meas{2}{a_2}{x_2}\right)}{\Omega_\beta\otimes\Omega_\alpha}\\
    &=\sandwich{\Omega_{\mathrm{tail}}}{\prod_{j=1}^4\Meas{j}{a_j}{x_j}}{\Omega_{\mathrm{tail}}}.
\end{aligned}
\end{equation}
The third equality uses \zcref{eq:ProductRadonNikodymExpectationIdentity} with $Z=\Meas{1}{a_1}{x_1}\Meas{2}{a_2}{x_2}\in(\cM_\alpha^3)'$.
\end{proof}

\subsubsection{Source-transfer SDP relaxation}\label{sec:AppendixMultipartiteSDP}
Let the formal leaf $*$-algebras be
\begin{align}
    \fA_\beta^4&={}^*\!\operatorname{alg}\!\left(\meas{4}{a_4}{x_4}\right), &
    \fA_\alpha^2&={}^*\!\operatorname{alg}\!\left(\meas{2}{a_2}{x_2}\right).
\end{align}
Define transfer symbols $\formalInstr{3}{a_3}{x_3}(q)$, $q\in\fA_\beta$, and let $\fA_\alpha^3$ be the $*$-algebra generated by them.
Write
\begin{align}
    \fA_\beta&=\fA_\beta^4, &
    \fA_\alpha&=\fA_\alpha^2\otimes\fA_\alpha^3, &
    \InPolyAlg{3}&=\fA_\beta, &
    \InPolyAlg{1}&=\fA_\alpha.
\end{align}
We introduce a joint state symbol $\varsigma_\alpha$ on $\fA_\alpha$, $\varsigma_{\beta}$ on $\fA_{\beta}$, the root functional symbols $\rPolystate{a_1}{x_1}$, and the associated multi-state polynomial algebras.
The internal and root source-replacer polynomials are
\begin{equation}\label{eq:AppendixMultipartiteSDPPolynomials}
    \begin{aligned}
        h_{x_3}^{3}(p_1,q,p_2)&=\varsigma_\alpha\!\left(p_1^*\left[\sum_{a_3}\formalInstr{3}{a_3}{x_3}(q)-\varsigma_\beta(q)\right]p_2\right),\\
    	h_{x_1}^{1}(t)&=\sum_{a_1}\rPolystate{a_1}{x_1}(t)-\varsigma_\alpha(t),\\
    	h_{\mathbf{a}|\mathbf{x}}^{(1)}&=\rPolystate{a_1}{x_1}\!\left(\meas{2}{a_2}{x_2}\otimes\formalInstr{3}{a_3}{x_3}\!\left(\meas{4}{a_4}{x_4}\right)\right).
    \end{aligned}
\end{equation}
Here $p_1,p_2$ range over the full noncommutative multi-state source ring $\fatS_\alpha$.
Accordingly, the complete-positivity matrix is
\begin{align}\label{eq:AppendixMultipartiteCPMatrix}
    \left[C_{d-1}\!\left(\formalInstr{3}{a_3}{x_3};L_d\right)\right]_{(p_1,q_1),(p_2,q_2)}=L_d\!\left(\varsigma_\alpha\!\left(p_1^*\formalInstr{3}{a_3}{x_3}(q_1^*q_2)p_2\right)\right),
\end{align}
with $p_i\in\operatorname{Mon}(\fatS_\alpha)$ and $q_i\in\operatorname{Mon}(\fatS_\beta)$ such that $\deg(p_i q_i) \leq d - 1$.

At level $d$, the source and root archimedean conditions are
\begin{equation}\label{eq:AppendixMultipartiteArchimedean}
\begin{alignedat}{2}
    &H_{d-1}^{\beta}\!\left(\id-\bigl(\meas{4}{a_4}{x_4}\bigr)^*\meas{4}{a_4}{x_4};L_d\right)\succeq0,
    &&\quad(a_4,x_4)\in\mathsf{A}_4\times\mathsf{X}_4,\\
    &H_{d-1}^{\alpha}\!\left(\id-\bigl(\meas{2}{a_2}{x_2}\otimes\id\bigr)^*\bigl(\meas{2}{a_2}{x_2}\otimes\id\bigr);L_d\right)\succeq0,
    &&\quad(a_2,x_2)\in\mathsf{A}_2\times\mathsf{X}_2,\\
    &H_{d-1-\deg(q)}^{\alpha}\!\left(\id-\bigl(\id\otimes\formalInstr{3}{a_3}{x_3}(q)\bigr)^*\bigl(\id\otimes\formalInstr{3}{a_3}{x_3}(q)\bigr);L_d\right)\succeq0,
    &&\quad\substack{q\in\operatorname{Mon}(\fA_\beta),\\(a_3,x_3)\in\mathsf{A}_3\times\mathsf{X}_3},\\
    &H_{d-2}^{a_1|x_1}\!\left(\id-\bigl(\meas{2}{a_2}{x_2}\otimes\id\bigr)^*\bigl(\meas{2}{a_2}{x_2}\otimes\id\bigr);L_d\right)\succeq0,
    &&\quad\substack{(a_2,x_2)\in\mathsf{A}_2\times\mathsf{X}_2,\\(a_1,x_1)\in\mathsf{A}_1\times\mathsf{X}_1},\\
    &H_{d-2-\deg(q)}^{a_1|x_1}\!\left(\id-\bigl(\id\otimes\formalInstr{3}{a_3}{x_3}(q)\bigr)^*\bigl(\id\otimes\formalInstr{3}{a_3}{x_3}(q)\bigr);L_d\right)\succeq0,
    &&\quad\substack{q\in\operatorname{Mon}(\fA_\beta),\\(a_3,x_3)\in\mathsf{A}_3\times\mathsf{X}_3,\\(a_1,x_1)\in\mathsf{A}_1\times\mathsf{X}_1}.
\end{alignedat}
\end{equation}
The last two lines are imposed for every $(a_1,x_1)\in\mathsf{A}_1\times\mathsf{X}_1$.
The remaining level-$d$ constraints are
\begin{equation}\label{eq:AppendixMultipartiteSDP}
\begin{alignedat}{3}
    &H_d^\alpha(L_d),\ H_d^\beta(L_d)\succeq0, && &&\text{(source states)}\\
    &H_{d-1}^{a_1|x_1}(L_d)\succeq0, &&\quad(a_1,x_1)\in\mathsf{A}_1\times\mathsf{X}_1, &&\text{(root functionals)}\\
    &H_{d-1}^{\alpha}\!\left(\meas{2}{a_2}{x_2};L_d\right)\succeq0, &&\quad(a_2,x_2)\in\mathsf{A}_2\times\mathsf{X}_2, &&\text{(leaf at $2$)}\\
    &H_{d-1}^{\beta}\!\left(\meas{4}{a_4}{x_4};L_d\right)\succeq0, &&\quad(a_4,x_4)\in\mathsf{A}_4\times\mathsf{X}_4, &&\text{(leaf at $4$)}\\
    &C_{d-1}\!\left(\formalInstr{3}{a_3}{x_3};L_d\right)\succeq0, &&\quad(a_3,x_3)\in\mathsf{A}_3\times\mathsf{X}_3, &&\text{(complete positivity)}\\
    &L_d\!\left(h_{x_3}^{3}(p_1,q,p_2)\right)=0, &&\quad p_1,p_2\in\operatorname{Mon}(\fatS_\alpha),\ q\in\operatorname{Mon}(\fatS_\beta), &&\text{(replacer at $3$)}\\
    &L_d\!\left(h_{x_1}^{1}(t)\right)=0, &&\quad t\in\operatorname{Mon}(\fatS_\alpha), &&\text{(replacer at $1$).}
\end{alignedat}
\end{equation}
All equalities are truncated by total degree as in \zcref{def:SourceTransferSDPRelaxation}.
The prescribed-correlation constraints are generated by $h_{\mathbf{a}|\mathbf{x}}^{(1)}-p_{\mathrm{tail}}(\mathbf{a}|\mathbf{x})$.
Since $\deg(h_{\mathbf{a}|\mathbf{x}}^{(1)})=4$, they first appear at level $d=2$.

\section{Measurement-algebra factorization is insufficient for the mixed model}\label{sec:AppendixWhySourceTransfer}
The main text establishes two complete SDP hierarchies for the mixed quantum model on source-party trees: the source-transfer hierarchy in \zcref{cor:MixedTreeSDPConvergence} and the inflation-NPA hierarchy in \zcref{cor:InflationNPATreeConvergence}.
The former is constructed directly from source algebras, source states, and completely positive transfer maps.

A natural question is whether a complete hierarchy could instead be obtained by applying the existing state-polynomial framework directly to the party measurement algebras and their factorization constraints.

The original scalar-extension method introduced SDP-compatible relaxations of measurement-algebra factorization~\cite{PozasKerstjens2019Bounding}.
For the bilocal network, suitable additional constraints lead to complete SDP hierarchies~\cite{ligthart2023inflation,renou2026two}.
These formulations are closely related to state-polynomial optimization~\cite{klep2024state}, which provides convergent SDP relaxations for polynomial constraints involving a state.

In this approach, one works with the party measurement algebras and a global state $\omega$.
Source independence is expressed through identities of the form
\begin{align}\label{eq:AppendixGenericStateFactorization}
    \omega(q_1\cdots q_k)=\omega(q_1)\cdots\omega(q_k).
\end{align}
Such constraints can be encoded by state-polynomial optimization.
In the Renou--Xu model~\cite{ligthart2023inflation,renou2026two}, this condition is sufficient to recover the mixed model for the bilocal network.
For general networks, this measurement-factorization description has recently been studied by~\cite{yang2026mutuallycommutingvonneumannalgebra}.

Already on the line of four, this measurement-factorization model does not coincide with the mixed quantum model.
Indeed, it admits a correlation whose conditional distribution at the two middle parties, given the endpoint outcomes, is a PR box, which excludes a mixed realization (\zcref{prop:AppendixPRFactorizationCounterexample}).
Consequently, imposing only these factorization constraints cannot yield a complete hierarchy for the mixed model.

\subsection{One-state measurement-factorization model}\label{sec:AppendixMeasurementFactorizationModel}
Consider the line-of-four source-party tree in \zcref{fig:appendix-line-four}, whose vertices we now denote by
\begin{align}\label{eq:AppendixFactorizationP4Graph}
    X-\alpha-A-\beta-B-\gamma-Y.
\end{align}
All parties have one measurement setting and binary outcomes $x, a, b, y \in \{0,1\}$.
The pairs of parties that do not share a source are $(X,B)$, $(X,Y)$, and $(A,Y)$.
Hence they have independent measurement algebras.
Its natural factorization-based model, which is a line-of-four specialization of~\cite[Definition~2.2]{yang2026mutuallycommutingvonneumannalgebra}, reads as follows.

\begin{definition}[Measurement-factorization model for the line-of-four]\label{def:AppendixMeasurementFactorizationLine}
A correlation $p(x,a,b,y)$ admits a \emph{measurement-factorization realization} on the line-of-four network~\eqref{eq:AppendixFactorizationP4Graph} if there exist a Hilbert space $\cH$, a unit vector $\ket{\Omega}\in\cH$, mutually commuting von Neumann algebras
\begin{align}
    \cM_X,\cM_A,\cM_B,\cM_Y\subset\Bof{\cH},
\end{align}
and binary POVMs
\begin{equation}
\begin{aligned}
    \{X_x\}_{x=0}^1&\subset\cM_X, &
    \{A_a\}_{a=0}^1&\subset\cM_A,\\
    \{B_b\}_{b=0}^1&\subset\cM_B, &
    \{Y_y\}_{y=0}^1&\subset\cM_Y.
\end{aligned}
\end{equation}
such that, with $\omega(q)\coloneqq\sandwich{\Omega}{q}{\Omega}$,
\begin{align}\label{eq:AppendixMeasurementFactorizationCorrelation}
    p(x,a,b,y)=\omega(X_xA_aB_bY_y).
\end{align}
Moreover, the state factorizes over every pair of parties that does not share a source:
\begin{align}\label{eq:AppendixIndependentPartyFactorizations}
    \omega(q_Xq_B)&=\omega(q_X)\omega(q_B), &
    \omega(q_Xq_Y)&=\omega(q_X)\omega(q_Y), &
    \omega(q_Aq_Y)&=\omega(q_A)\omega(q_Y)
\end{align}
for all $q_X\in\cM_X$, $q_A\in\cM_A$, $q_B\in\cM_B$, and $q_Y\in\cM_Y$.
\end{definition}
All conditions in~\zcref{def:AppendixMeasurementFactorizationLine} can be readily adapted to the original state polynomial optimization.

\subsection{Counterexample on the line of four}\label{sec:AppendixFactorizationCounterexample}
Related to~\cite[Example~3.3.18]{ligthart2024semidefinite}, we show such a factorization-based model admits correlations that are incompatible with the mixed model on the line of four.
\begin{proposition}[A post-quantum correlation satisfying the independent-party factorizations]\label{prop:AppendixPRFactorizationCounterexample}
Define
\begin{align}\label{eq:AppendixPRCorrelation}
    p_{\mathrm{PR}}(x,a,b,y)\coloneqq
    \begin{cases}
        \frac{1}{8}, & a\oplus b=xy,\\[1mm]
        0, & a\oplus b\neq xy,
    \end{cases}
    \qquad x,a,b,y\in\{0,1\}.
\end{align}
Then $p_{\mathrm{PR}}$ admits a measurement-factorization realization in the sense of~\zcref{def:AppendixMeasurementFactorizationLine}, but it does not admit a mixed quantum realization on the line of four as in \zcref{def:MixedQuantumModelSourcePartyGraph}.
\end{proposition}
\begin{proof}
We first construct a measurement-factorization realization for $p_{\mathrm{PR}}$ explicitly.
Let $\cH = \mathbb{C}^2 \otimes \mathbb{C}^2 \otimes \mathbb{C}^2$ with computational basis $\{\ket{x,y,\lambda}:x,y,\lambda\in\{0,1\}\}$, and write
\begin{align}\label{eq:AppendixPRVector}
    \ket{\Omega_{\mathrm{PR}}}\coloneqq\frac{1}{\sqrt{8}}\sum_{x,y,\lambda\in\{0,1\}}\ket{x,y,\lambda}.
\end{align}
Consider the binary PVMs,
\begin{equation}\label{eq:AppendixPRMeasurements}
\begin{aligned}
    X_x&\coloneqq\sum_{y,\lambda}\ketbra{x,y,\lambda}{x,y,\lambda}, &
    Y_y&\coloneqq\sum_{x,\lambda}\ketbra{x,y,\lambda}{x,y,\lambda},\\
    B_b&\coloneqq\sum_{x,y}\ketbra{x,y,b}{x,y,b}, &
    A_a&\coloneqq\sum_{x,y}\ketbra{x,y,a\oplus xy}{x,y,a\oplus xy},
\end{aligned}
\end{equation}
which are mutually commuting as operators diagonal in the same basis.
Let $\omega_{\mathrm{PR}}(\cdot) \coloneqq \sandwich{\Omega_{\mathrm{PR}}}{\cdot}{\Omega_{\mathrm{PR}}}$.
One directly computes that
\begin{align}\label{eq:AppendixPRCorrelationCalculation}
    \omega_{\mathrm{PR}}(X_xA_aB_bY_y)=\frac{1}{8}\delta_{a\oplus b,xy}=p_{\mathrm{PR}}(x,a,b,y).
\end{align}
Moreover, we can check that
\begin{equation}\label{eq:AppendixPRPairFactorizations}
    \begin{aligned}
        \omega_{\mathrm{PR}}(X_xB_b)&=\frac{1}{4}=\omega_{\mathrm{PR}}(X_x)\omega_{\mathrm{PR}}(B_b), \\
        \omega_{\mathrm{PR}}(X_xY_y)&=\frac{1}{4}=\omega_{\mathrm{PR}}(X_x)\omega_{\mathrm{PR}}(Y_y), \\
        \omega_{\mathrm{PR}}(A_aY_y)&=\frac{1}{4}=\omega_{\mathrm{PR}}(A_a)\omega_{\mathrm{PR}}(Y_y),
    \end{aligned}
\end{equation}
on the generators of the measurement algebras.
Since these operators generate the measurement algebras,~\zcref{eq:AppendixIndependentPartyFactorizations} follows by linearity.

It remains to exclude a mixed quantum realization for $p_{\mathrm{PR}}(x,a,b,y)$.
We use the observation made in~\cite[Theorem~2.4]{fritz2012beyond}---any quantum realization of a line-of-four correlation must give a quantum Bell correlation after conditioning on the two end outcomes---and extend it to the mixed model.
Indeed, conditioned on the marginal $p_{\mathrm{PR}}(x,y)=\frac{1}{4}$, we obtain the PR box correlation in the Bell scenario
\begin{align}\label{eq:AppendixConditionalPRBox}
    p_{\mathrm{PR}}(a,b|x,y)=
    \begin{cases}
        \frac{1}{2}, & a\oplus b=xy,\\[1mm]
        0, & a\oplus b\neq xy.
    \end{cases}
\end{align}
Suppose that $p_{\mathrm{PR}}$ admitted a mixed realization with source spaces $\cH_\alpha,\cH_\beta,\cH_\gamma$ and product source vector $\ket{\Omega_\alpha\otimes\Omega_\beta\otimes\Omega_\gamma}$. Define operators on the middle source space $\cH_\beta$ by
\begin{align}\label{eq:AppendixEffectiveBellMeasurements}
    \widetilde A_{a|x}&\coloneqq 2\sandwich{\Omega_\alpha}{X_xA_a}{\Omega_\alpha}, &
    \widetilde B_{b|y}&\coloneqq 2\sandwich{\Omega_\gamma}{B_bY_y}{\Omega_\gamma}.
\end{align}
The mixed-model commutation relations make these operators positive and mutually commuting.
Since $p_{\mathrm{PR}}(x)=p_{\mathrm{PR}}(y)=1/2$, we have $\sum_a\widetilde A_{a|x}=\sum_b\widetilde B_{b|y}=I_{\cH_\beta}$.
These commuting POVMs therefore satisfy
\begin{align}\label{eq:AppendixMixedImpliesBell}
    \sandwich{\Omega_\beta}{\widetilde A_{a|x}\widetilde B_{b|y}}{\Omega_\beta} = 4p_{\mathrm{PR}}(x,a,b,y) = p_{\mathrm{PR}}(a,b|x,y).
\end{align}
It follows that the conditional PR box in~\zcref{eq:AppendixConditionalPRBox} is a commuting-operator quantum Bell correlation attaining the CHSH value $4$, contradicting the Tsirelson bound of $2\sqrt{2}$.
\end{proof}

\begin{remark}[Further factorization conditions do not help]\label{rem:AppendixStrongerCutFactorizations}
The counterexample persists even when we impose factorization between groups of parties that depend on disjoint sources.
Indeed, the counterexample also satisfies
\begin{align}
    \omega_{\mathrm{PR}}(q_Xq_Aq_Y)&=\omega_{\mathrm{PR}}(q_Xq_A)\omega_{\mathrm{PR}}(q_Y), &
    \omega_{\mathrm{PR}}(q_Xq_Bq_Y)&=\omega_{\mathrm{PR}}(q_X)\omega_{\mathrm{PR}}(q_Bq_Y)
\end{align}
for all $q_X\in\cM_X$, $q_A\in\cM_A$, $q_B\in\cM_B$, and $q_Y\in\cM_Y$.
Together with \zcref{eq:AppendixIndependentPartyFactorizations}, these identities include all nontrivial factorizations between groups of parties whose incident source sets are disjoint in the line-of-four network.
\end{remark}
\enlargethispage{2\baselineskip}
By \zcref{prop:AppendixPRFactorizationCounterexample}, even a complete state-polynomial hierarchy for the measurement-factorization model would admit $p_{\mathrm{PR}}$, so extending the bilocal approach to all source-party trees requires more than these factorization constraints.
Our source-transfer hierarchy instead encodes the source-transfer model, which is equivalent to the mixed model by \zcref{thm:SourceTransferTreeEquivalence}, incorporating its source states and completely positive transfer maps through the multi-state extension of state-polynomial optimization in \zcref{def:SourceTransferSDPRelaxation}.
This gives a direct and complete hierarchy for the mixed model on trees, and the same characterization identifies the model reconstructed from inflation in \zcref{thm:InflationSourceTransferTreeEquivalence}, establishing convergence of inflation-NPA.

\end{document}